\documentclass{lmcs} 
\pdfoutput=1
\usepackage[utf8]{inputenc}

\usepackage{lastpage}
\lmcsdoi{19}{2}{17}
\lmcsheading{}{\pageref{LastPage}}{}{}%
{Mar.~23,~2022}{Jun.~29,~2023}{}


\usepackage{afterpage}
\usepackage{amsmath,amsthm}
\usepackage{amssymb}
\usepackage{babel}
\usepackage{bussproofs}
\usepackage{calc}
\usepackage{centernot}
\usepackage{crossreftools}
\usepackage{chngcntr}
\usepackage[dmyyyy]{datetime}
\usepackage{esvect}
\usepackage{fancybox}
\usepackage[colorlinks=true,urlcolor=blue,citecolor=blue,linkcolor=blue]{hyperref}
\usepackage{multirow}
\usepackage{mathtools}
\usepackage{placeins}
\usepackage{soul}   
\usepackage{stmaryrd}
\usepackage{url}
\usepackage{ushort}
\usepackage[normalem]{ulem}
\let\savebigtimes\bigtimes
\let\bigtimes\relax
\usepackage{mathabx}
\let\bigtimes\savebigtimes
\usepackage{graphicx}
\usepackage{tikz,fp,tikz-cd}
\usetikzlibrary{arrows.meta,arrows,calc,automata,angles,quotes,matrix,
                positioning,through,shapes.geometric,shapes.symbols,shapes.multipart,through,trees,
                decorations.markings,decorations.text,
                }

\pgfdeclaredecoration{funbisim}{final}{
  \state{final}[width=\pgfdecoratedpathlength]{
    \draw[->]
      (0,\pgfdecorationsegmentamplitude+0.35mm) -- +(\pgfdecoratedpathlength,0);
    \draw[-]
      (0,-\pgfdecorationsegmentamplitude-0.35mm) -- +(\pgfdecoratedpathlength,0);}}
\tikzset{
  funbisim/.style={
    decoration={funbisim, amplitude=0.25ex},
    decorate,
    funbisim options/.style={#1}    
  }}

\pgfdeclaredecoration{funbisimright}{final}{
  \state{final}[width=\pgfdecoratedpathlength]{
    \draw[->]
      (0,\pgfdecorationsegmentamplitude+0.35mm) -- +(\pgfdecoratedpathlength,0);
    \draw[-]
      (0,-\pgfdecorationsegmentamplitude-0.35mm) -- +(\pgfdecoratedpathlength,0);}}
\tikzset{
  funbisimright/.style={
    decoration={funbisimright, amplitude=0.25ex},
    decorate,
    funbisimright options/.style={#1}    
  }}

\pgfdeclaredecoration{funbisimleft}{final}{
  \state{final}[width=\pgfdecoratedpathlength]{
    \draw[->]
      (0,-\pgfdecorationsegmentamplitude-0.35mm) -- +(\pgfdecoratedpathlength,0);
    \draw[-]
      (0,\pgfdecorationsegmentamplitude+0.35mm) -- +(\pgfdecoratedpathlength,0);}}
\tikzset{
  funbisimleft/.style={
    decoration={funbisimleft, amplitude=0.25ex},
    decorate,
    funbisimleft options/.style={#1}    
  }}

\pgfdeclaredecoration{bisim}{final}{
  \state{final}[width=\pgfdecoratedpathlength]{
    \draw[<->]
      (0,\pgfdecorationsegmentamplitude+0.35mm) -- +(\pgfdecoratedpathlength,0);
    \draw[-]
      (0,-\pgfdecorationsegmentamplitude-0.35mm) -- +(\pgfdecoratedpathlength,0);}}
\tikzset{
  bisim/.style={
    decoration={bisim, amplitude=0.25ex},
    decorate,
    bisim options/.style={#1}    
  }}

\pgfdeclaredecoration{funonebisim}{final}{
  \state{final}[width=\pgfdecoratedpathlength]{
    \draw[->,firebrick]
      (0,\pgfdecorationsegmentamplitude+0.35mm) -- +(\pgfdecoratedpathlength,0);
    \draw[-,firebrick]
      (0,-\pgfdecorationsegmentamplitude-0.35mm) -- +(\pgfdecoratedpathlength,0);}}
\tikzset{
  funonebisim/.style={
    decoration={funonebisim, amplitude=0.25ex},
    decorate,
    funonebisim options/.style={#1}    
  }}

\pgfdeclaredecoration{funonebisimdashed}{final}{
  \state{final}[width=\pgfdecoratedpathlength]{
    \draw[->,densely dashed,firebrick]
      (0,\pgfdecorationsegmentamplitude+0.35mm) -- +(\pgfdecoratedpathlength,0);
    \draw[-,densely dashed,firebrick]
      (0,-\pgfdecorationsegmentamplitude-0.35mm) -- +(\pgfdecoratedpathlength,0);}}
\tikzset{
  funonebisimdashed/.style={
    decoration={funonebisimdashed, amplitude=0.25ex},
    decorate,
    funonebisimdashed options/.style={#1}    
  }}

\pgfdeclaredecoration{funonebisimright}{final}{
  \state{final}[width=\pgfdecoratedpathlength]{
    \draw[->,firebrick]
      (0,\pgfdecorationsegmentamplitude+0.35mm) -- +(\pgfdecoratedpathlength,0);
    \draw[-,firebrick]
      (0,-\pgfdecorationsegmentamplitude-0.35mm) -- +(\pgfdecoratedpathlength,0);}}
\tikzset{
  funonebisimright/.style={
    decoration={funonebisimright, amplitude=0.25ex},
    decorate,
    funonebisimright options/.style={#1}    
  }}

\pgfdeclaredecoration{funonebisimleft}{final}{
  \state{final}[width=\pgfdecoratedpathlength]{
    \draw[->,firebrick]
      (0,-\pgfdecorationsegmentamplitude-0.35mm) -- +(\pgfdecoratedpathlength,0);
    \draw[-,firebrick]
      (0,\pgfdecorationsegmentamplitude+0.35mm) -- +(\pgfdecoratedpathlength,0);}}
\tikzset{
  funonebisimleft/.style={
    decoration={funonebisimleft, amplitude=0.25ex},
    decorate,
    funonebisimleft options/.style={#1}    
  }}

\pgfdeclaredecoration{onebisim}{final}{
  \state{final}[width=\pgfdecoratedpathlength]{
    \draw[<->,firebrick]
      (0,\pgfdecorationsegmentamplitude+0.35mm) -- +(\pgfdecoratedpathlength,0);
    \draw[-,firebrick]
      (0,-\pgfdecorationsegmentamplitude-0.35mm) -- +(\pgfdecoratedpathlength,0);}}
\tikzset{
  onebisim/.style={
    decoration={onebisim, amplitude=0.25ex},
    decorate,
    onebisim options/.style={#1}    
  }}

\makeatletter
\def\calcLength(#1,#2)#3{%
\pgfpointdiff{\pgfpointanchor{#1}{center}}%
             {\pgfpointanchor{#2}{center}}%
\pgf@xa=\pgf@x%
\pgf@ya=\pgf@y%
\FPeval\@temp@a{\pgfmath@tonumber{\pgf@xa}}%
\FPeval\@temp@b{\pgfmath@tonumber{\pgf@ya}}%
\FPeval\@temp@sum{(\@temp@a*\@temp@a+\@temp@b*\@temp@b)}%
\FProot{\FPMathLen}{\@temp@sum}{2}%
\FPround\FPMathLen\FPMathLen5\relax
\global\expandafter\edef\csname #3\endcsname{\FPMathLen}
}
\makeatother

\DeclareRobustCommand*{\mycirc}[1]{%
  \tikz[baseline=(C.base)]
    \node[draw,circle,inner sep=1pt](C){#1};}

\tikzset{
  my dash/.style={dash pattern=on 5pt off 2pt}
         }

\theoremstyle{definition}

  \newtheorem{nonexa}[thm]{Non-Example}
\setul{1pt}{.4pt}  

\DeclareFontFamily{U}{mathx}{\hyphenchar\font45}   
\DeclareFontShape{U}{mathx}{m}{n}{
      <5> <6> <7> <8> <9> <10>
      <10.95> <12> <14.4> <17.28> <20.74> <24.88>
      mathx10
      }{}
\DeclareSymbolFont{mathx}{U}{mathx}{m}{n}
\DeclareFontSubstitution{U}{mathx}{m}{n}
\DeclareMathAccent{\widecheck}{0}{mathx}{"71}
\DeclareMathAccent{\wideparen}{0}{mathx}{"75}

\definecolor{azure}{rgb}{0.94,1.00,1.00}
\definecolor{brown}{rgb}{.75,.25,.25}
\definecolor{cyan}{rgb}{0.25,0.88,0.82}
\definecolor{chocolate}{rgb}{0.82,0.41,0.12}
\definecolor{darkcyan}{rgb}{0.5,0,1}
\definecolor{darkgreen}{rgb}{0,0.39,0}
\definecolor{darkmagenta}{rgb}{0.5,0,0.5}
\definecolor{darkgoldenrod}{RGB}{184,134,11}
\definecolor{firebrick}{RGB}{175,25,25}
\definecolor{forestgreen}{rgb}{0.13,0.55,0.13}
\definecolor{goldenrod}{RGB}{218,165,32}
\definecolor{lightcyan}{rgb}{0.88,1.00,1.00}
\definecolor{lightpink}{rgb}{1.00,0.71,0.76}
\definecolor{myyellow}{RGB}{235,235,0}
\definecolor{lightyellow}{rgb}{1.00,1.00,0.88}
\definecolor{lightgoldenrod}{rgb}{0.83,0.97,0.51}
\definecolor{lightgoldenrodyellow}{rgb}{0.98,0.98,0.82}
\definecolor{lightskyblue}{rgb}{0.53,0.81,0.98}
\definecolor{moccasin}{rgb}{1.00,0.89,0.71}
\definecolor{magenta}{rgb}{1,0,1}
\definecolor{navyblue}{rgb}{0,0,0.5}
\definecolor{orange}{rgb}{1.0,0.65,0.0}
\definecolor{orangered}{rgb}{1.0,0.27,0.0}
\definecolor{palegreen}{rgb}{0.60,0.98,0.60}
\definecolor{paleblue}{rgb}{0.678,0.847,0.902}
\definecolor{powderblue}{rgb}{0.69,0.88,0.90}
\definecolor{purple}{rgb}{1,0.5,1}
\definecolor{royalblue}{RGB}{65,105,225}
\definecolor{mediumblue}{RGB}{0,0,205}
\definecolor{cornflowerblue}{RGB}{100,149,237}
\definecolor{springgreen}{rgb}{0.0,1.0,0.5}
\definecolor{turquoise}{rgb}{0.25,0.88,0.82}
\definecolor{snow}{rgb}{1.00,0.98,0.98}
\definecolor{tan}{rgb}{0.82,0.71,0.55}
\definecolor{red}{rgb}{1,0,0}
\definecolor{violetred}{RGB}{208,32,144}
\definecolor{white}{RGB}{255,255,255}

\newcommand{\colorin}[1]{\textcolor{#1}}

\newcommand{\black}{\colorin{black}}

\newcommand{\chocolate}{\colorin{chocolate}}
\newcommand{\colorred}{\colorin{red}}
\newcommand{\alert}{\colorred}

\newcommand{\darkcyan}{\colorin{darkcyan}}

\newcommand{\forestgreen}{\colorin{forestgreen}}
\newcommand{\firebrick}{\colorin{firebrick}}

\newcommand{\mediumblue}{\colorin{mediumblue}}

\newcommand{\royalblue}{\colorin{royalblue}}

\newcommand{\nb}{\nobreakdash}
\newcommand{\punc}[1]{\ensuremath{\hspace*{1.5pt}{#1}}}

\newcommand{\nf}{\normalfont}
\newcommand{\textnf}[1]{\text{\nf{#1}}}

\newcommand{\bs}[1]{\boldsymbol{#1}}

\newenvironment{new}{\color{chocolate}}{\color{black}}
\newenvironment{newer}{\color{firebrick}}{\color{black}}
\newenvironment{newest}{\color{red}}{\color{black}}
\newenvironment{revised}{\color{violetred}}{\color{black}}
  \newenvironment{changed}{\color{mediumblue}}{\color{black}}

  \newenvironment{changedalert}{\color{red}}{\color{black}}

\newcommand{\isnewer}[1]{\begin{newer}{#1}\ed{nd{newer}}}

\newcommand{\ischanged}[1]{\begin{changed}{#1}\end{changed}}

\newcommand{\funin}{\mathrel{:}}
\newcommand{\fap}[2]{{#1}(\hspace*{-0.5pt}{#2}\hspace*{-0.5pt})}

\newcommand{\bfap}[3]{{#1}({#2},\hspace*{0.5pt}{#3})}

\newcommand{\iap}[2]{#1_{#2}}
\newcommand{\bap}[2]{#1_{#2}}
\newcommand{\pap}[2]{#1^{#2}}
\newcommand{\bpap}[3]{#1_{#2}^{#3}}
\newcommand{\pbap}[3]{#1_{#3}^{#2}}
\newcommand{\sidfun}{\textrm{\nf id}}
\newcommand{\sidfunon}{\iap{\sidfun}}

\newcommand{\sproj}{\pi}
\newcommand{\sprojwrt}{\iap{\pi}}

\newcommand{\proj}{\fap{\sproj}}

\newcommand{\sdefdby}{{:=}}
\newcommand{\defdby}{\mathrel{\sdefdby}}

\newcommand\tuple[1]{\langle #1 \rangle}

\newcommand\tuplespace{\hspace*{0.5pt}}
\newcommand\pair[2]{\tuple{#1, \tuplespace #2}}
\newcommand\triple[3]{\tuple{#1, \tuplespace #2, \tuplespace #3}}

\newcommand{\nat}{\mathbb{N}}
\newcommand{\natplus}{\pap{\nat}{+}} 

\newcommand{\sgraphof}{\mathit{graph}}
\newcommand{\graphof}{\fap{\sgraphof}}

\newcommand{\BNFor}{\:\mid\:}
\newcommand{\BNFdefdby}{\:{::=}\:}

\newcommand{\family}[2]{\setexp{#1}_{#2}}

\newcommand{\slexspo}{{<}_{\text{\nf lex}}}
\newcommand{\lexspo}{\mathrel{\slexspo}}

\newcommand{\eqcl}[1]{\iap{\left[{#1}\right]}}

\newcommand{\ssynteq}{{\equiv}} 
\newcommand{\synteq}{\mathrel{\ssynteq}} 
\newcommand{\snotsynteq}{{\not\equiv}}
\newcommand{\notsynteq}{\mathrel{\snotsynteq}}

\newcommand{\sformeq}{\royalblue{=}}
\newcommand{\formeq}{\mathrel{\sformeq}}

\newcommand{\sred}{{\to}}

\newcommand{\sredi}[1]{{\iap{\sred}{#1}}}
\newcommand{\redi}[1]{\mathrel{\sredi{#1}}}

\newcommand{\sconvred}{{\leftarrow}}

\newcommand{\sconvredi}[1]{{\iap{\leftarrow}{#1}}}
\newcommand{\convredi}[1]{\mathrel{\sconvredi{#1}}}

\newcommand{\sconvredtci}[1]{{\pbap{\leftarrow}{+}{#1}}}
\newcommand{\convredtci}[1]{\mathrel{\sconvredtci{#1}}}
\newcommand{\sredrtc}{\sred^{*}}

\newcommand{\sredrtci}[1]{{\iap{\sredrtc}{#1}}}
\newcommand{\redrtci}[1]{\mathrel{\sredrtci{#1}}}

\newcommand{\sredb}[1]{{\bap{\sred}{#1}}}

\newcommand{\sredbp}[2]{{\bpap{\sred}{#1}{#2}}}
\newcommand{\sconvredb}[1]{{\bap{\sconvred}{#1}}}

  \newcommand{\subotr}{\hspace*{-1pt}{\scriptscriptstyle (\sone)}}

\newcommand{\sotelimred}{\sredb{\subotr}}
\newcommand{\otelimred}{\mathrel{\sotelimred}}

\newcommand{\sotelimredtc}{\sredbp{\subotr}{+}}

\newcommand{\sotelimconvred}{\sconvredb{\subotr}}

\newcommand{\bodylabcol}{\text{\nf\darkcyan{bo}}}
\newcommand{\bodylab}{\bodylabcol}

\newcommand{\bodytransition}{body tran\-si\-tion}
\newcommand{\bodytransitions}{\bodytransition{s}}

\newcommand{\descsetexpmid}{\mathrel{\vert}}
\newcommand{\descsetexpbigmid}{\mathrel{\big\vert}}
\newcommand{\descsetexpBigmid}{\mathrel{\Big\vert}}

\newcommand{\descsetexp}[2]{\left\{{#1}\descsetexpmid{#2}\right\}}

\newcommand{\descsetexpbig}[2]{\bigl\{{#1}\descsetexpbigmid{#2}\bigr\}}
\newcommand{\descsetexpBig}[2]{\Bigl\{{#1}\descsetexpBigmid{#2}\Bigr\}}

\newcommand{\saeqrel}{{\simeq}}
\newcommand{\aeqrel}{\mathrel{\saeqrel}}

\newcommand{\sacongrel}{{\simeq}}
\newcommand{\acongrel}{\mathrel{\sacongrel}}

\newcommand{\sbinrelcomp}[2]{{#1}\cdot{#2}}

\newcommand{\sphifun}{\phi}
\newcommand{\phifun}{\fap{\sphifun}}

\newcommand{\scompfuns}[2]{{#1}\circ{#2}}
\newcommand{\compfuns}[2]{\fap{\scompfuns{#1}{#2}}}

\newcommand{\sabinrel}{{R}}
\newcommand{\abinrel}{\mathrel{\sabinrel}}

\newcommand{\setexp}[1]{\left\{{#1}\right\}}

\newcommand{\setexpbig}[1]{\bigl\{{#1}\bigr\}}

\newcommand{\factorset}[2]{{{#1}/_{#2}}}
\newcommand{\factor}{\factorset}
\renewcommand{\emptyset}{\varnothing}

\newcommand{\slogand}{\wedge}
\newcommand{\slogor}{\vee}
\newcommand{\logand}{\mathrel{\slogand}}
\newcommand{\logor}{\mathrel{\slogor}} 
\newcommand{\slognot}{\neg}
\newcommand{\lognot}[1]{\slognot{#1}}

\newcommand{\length}[1]{\left|{#1}\right|}

\newcommand{\depth}{\length}

\newcommand{\tightfbox}[1]{{\fboxsep=1.5pt\fbox{#1}}}

\newcommand{\actiontarget}{ac\-tion-tar\-get}

\newcommand{\backlink}{backlink}
\newcommand{\backlinks}{\backlink{s}}

\newcommand{\equationallogic}{equa\-ti\-onal-log\-ic}

\newcommand{\eqlogicbased}{$\eqlogic$\nb-based}
\newcommand{\eqlogicbasedover}[1]{$\eqlogicover{#1}$\nb-based}

\newcommand{\equationbasedover}[1]{$\StExpEqover{#1}$-based}

\newcommand{\equivalenceclasseswrt}[1]{${#1}$\nb-equivalence classes}

\newcommand{\finitestate}{fi\-nite-state}

\newcommand{\immediatetermination}{im\-me\-di\-ate-ter\-mi\-na\-tion}

\newcommand{\loopentry}{loop-en\-try}
\newcommand{\loopbody}{loop-body}

\newcommand{\onederivative}{1\nb-de\-riv\-a\-tive}
\newcommand{\onederivatives}{\onederivative{s}}

\newcommand{\entrybodylabeling}{en\-try\discretionary{/}{}{/}body-la\-be\-ling}
\newcommand{\entrybodylabelings}{\entrybodylabeling{s}}

\newcommand{\LEEshape}{\LEE\nb-shape}

\newcommand{\LEEshaped}{\LEE\nb-shaped}

\newcommand{\LEEwitness}{$\LEE$\nb-wit\-ness}

\newcommand{\LEEwitnesses}{$\LEE$\hspace*{1.25pt}\nb-wit\-nes\-ses}

\newcommand{\LLEEwitness}{{\nf LLEE}\nb-wit\-ness}
\newcommand{\LLEEwitnesses}{{\nf LLEE}\nb-wit\-nesses}
\newcommand{\LLEEwitnessed}{{\nf LLEE}\nb-wit\-nessed}

\newcommand{\LLEEcharts}{{\nf LLEE}\nb-chart{s}}

\newcommand{\LLEEonechart}{{\nf LLEE}\hspace*{1pt}-\onechart}

\newcommand{\LLEEonecharts}{{LLEE}\hspace*{1pt}-\onechart{s}}

\newcommand{\nontrivial}{non-triv\-i\-al}

\newcommand{\onetransition}{$\sone$\nb-tran\-si\-tion}
\newcommand{\onetransitions}{\onetransition{s}}

\newcommand{\onetransitionfree}{\onetransition-free}

\newcommand{\onefree}{$1$\nb-free}
\newcommand{\sonefree}{$\sone$\nb-free}

\newcommand{\onechart}{$\sone$\nb-chart}

\newcommand{\onecharts}{\onechart{s}}

\newcommand{\onebisimulation}{$\sone$\nb-bi\-si\-mu\-la\-tion}
\newcommand{\onebisimulations}{\onebisimulation{s}}

\newcommand{\onebisimilar}{$\sone$\nb-bi\-si\-mi\-lar}

\newcommand{\onebisimilarity}{$\sone$\nb-bi\-si\-mi\-la\-ri\-ty}

\newcommand{\perpetual}{per\-pet\-u\-al}
\newcommand{\perpetualloop}{\perpetual-loop}

\newcommand{\processsemantics}{pro\-cess-se\-man\-tics}

\newcommand{\provablein}[1]{{$#1$}\nb-pro\-vable}
\newcommand{\Provablein}[1]{{$#1$}\nb-Pro\-vable}
\newcommand{\provablyin}[1]{{$#1$}\nb-pro\-vably}

\newcommand{\provabilityin}[1]{{$#1$}\nb-pro\-va\-bi\-li\-ty}

\newcommand{\sidecondition}{side-con\-di\-tion}
\newcommand{\sideconditions}{\sidecondition{s}}

\newcommand{\subonechart}{sub-$\sone$\nb-chart}
\newcommand{\subonecharts}{\subonechart{s}}

\newcommand{\starexpression}{star-ex\-pres\-sion}

\newcommand{\theoremsubsumed}{the\-o\-rem-sub\-sumed}
\newcommand{\theoremequivalent}{the\-o\-rem-equiv\-a\-lent}
\newcommand{\theoremsubsumption}{the\-o\-rem-sub\-sump\-tion}
\newcommand{\theoremequivalence}{the\-o\-rem-equiv\-a\-lence}

\newcommand{\transitionact}[1]{{${#1}$}\nb-tran\-si\-tion}
\newcommand{\transitionsact}[1]{\transitionact{#1}s}

\newcommand{\weaklyguarded}{weak\-ly guard\-ed}

\newcommand{\welldefined}{well-de\-fined}

\newcommand{\wellfounded}{well-found\-ed}

\newcommand{\overlinebar}[1]{\mathbf{\overline{\text{$#1$}}}}

\newcommand{\aspec}{\mathcal{S}}

\newcommand{\astexp}{e}
\newcommand{\bstexp}{f}
\newcommand{\cstexp}{g}
\newcommand{\dstexp}{h}

\newcommand{\astexpi}{\iap{\astexp}}
\newcommand{\bstexpi}{\iap{\bstexp}}
\newcommand{\cstexpi}{\iap{\cstexp}}
\newcommand{\dstexpi}{\iap{\dstexp}}

\newcommand{\asstexp}{E}

\newcommand{\csstexp}{G}

\newcommand{\asstexpi}{\iap{\asstexp}}

\newcommand{\asstexpacc}{\asstexp'}

\newcommand{\asstexpacci}{\iap{\asstexpacc}}

\newcommand{\snormedplus}{{\textit{nd}^+}}
\newcommand{\normedplus}{\fap{\snormedplus}}

\newcommand{\astexpacc}{\astexp'}

\newcommand{\astexpacci}{\iap{\astexpacc}}

\newcommand{\slanguage}{L}
\newcommand{\languageof}{\fap{\slanguage}}

\newcommand{\aDFA}{M}

\newcommand{\StExp}{\textit{StExp}}
\newcommand{\StExpover}{\fap{\StExp}}

\newcommand{\StExpEq}{\textit{Eq}}
\newcommand{\StExpEqover}{\fap{\StExpEq}}

\newcommand{\stackStExp}{\mathit{StExp}^{{\scriptscriptstyle(}\sstexpstackprod{\scriptscriptstyle )}}}
\newcommand{\stackStExpover}{\fap{\stackStExp\hspace*{-1.5pt}}}

\newcommand{\slangsemeq}{{=_{\textit{L}}}}

\newcommand{\procsemscriptscript}[1]{\scriptscriptstyle \llbracket{#1}\rrbracket_{\hspace*{-0.5pt}\sprocsem}}
\newcommand{\sprocsemeq}{{\iap{=}{\procsemscriptscript{\cdot}}}} 
\newcommand{\procsemeq}{\mathrel{\sprocsemeq}}

\newcommand{\acxt}{C}

\newcommand{\cxtap}[2]{{#1}[{#2}]}
\newcommand{\acxtap}{\cxtap{\acxt}}

\newcommand{\iets}{\hspace*{0.5pt}}

\newcommand{\stexpzero}{0}
\newcommand{\stexpone}{1}

\newcommand{\sstexpit}{\sstar}
\newcommand{\stexpit}[1]{{#1^{\sstexpit}}}

\newcommand{\sstexpprod}{{\cdot}}
\newcommand{\stexpprod}[2]{{#1}\mathrel{\sstexpprod}{#2}}
\newcommand{\sstexpsum}{+}
\newcommand{\stexpsum}[2]{{#1}\sstexpsum{#2}}
\newcommand{\sstexpstackprod}{\chocolate{\varstar}} 
\newcommand{\stexpstackprod}[2]{{#1}\mathrel{\sstackprod}{#2}}
  
\newcommand{\ssprod}{{\cdot}}
\renewcommand{\prod}{\,{\ssprod}\,}%
\newcommand{\sstackprod}{\chocolate{\varstar}}
\newcommand{\stackprod}{\,{\sstackprod}\,}%

\newcommand{\sth}[1]{|{#1}|_{\scalebox{0.8}{$\sstar$}}}

\newcommand{\sdescrelstexpit}{\sredi{\scriptscriptstyle(\sstar)}}

\newcommand{\descrelstexpit}[1]{\mathrel{\sdescrelstexpit}}

\newcommand{\sconvdescrelstexpit}{\sconvredi{\scriptscriptstyle(\sstar)}}

\newcommand{\convdescrelstexpit}[1]{\mathrel{\sconvdescrelstexpit}}

\newcommand{\sspartderivs}{\partial}

\newcommand{\spartonederivs}{{\underline{\partial}}}
\newcommand{\partonederivs}{\fap{\spartonederivs}}
\newcommand{\sitpartonederivs}{{\underline{\partial}^+}}
\newcommand{\itpartonederivs}{\fap{\sitpartonederivs}}

\newcommand{\soneactderivs}{\textit{\underline{A\hspace*{-0.25pt}$\sspartderivs$}}}
\newcommand{\oneactderivs}{\fap{\soneactderivs}}

\newcommand{\sprocsem}{P}
\newcommand{\procsem}[1]{\llbracket{#1}\rrbracket_{\hspace*{-1.5pt}\sprocsem}}

\newcommand{\slangsem}{L}
\newcommand{\langsem}[1]{\llbracket{#1}\rrbracket_{\slangsem}}

\newcommand{\slt}[1]{{\xrightarrow{#1}}}
  \newcommand{\sltfact}[2]{{\factor{\xrightarrow{#1}}{#2}}}
\newcommand{\slti}[2]{{\xrightarrow{#1}}{_{#2}}}

\newcommand{\lt}[1]{\mathrel{\slt{#1}}}
  \newcommand{\ltfact}[2]{\mathrel{\sltfact{#1}{#2}}}
\newcommand{\lti}[2]{\mathrel{\slti{#1}{#2}}}

\newcommand{\sone}{\protect\firebrick{1}}
\newcommand{\sstar}{*}

\newcommand{\looplab}[1]{{\darkcyan{[#1]}}}  
\newcommand{\loopsteplab}{\looplab}

\newcommand{\loopnsteplab}[1]{\darkcyan{[{#1}]}}
\newcommand{\sloopnstepto}[1]{{\iap{\rightarrow}{\loopnsteplab{#1}}}}
\newcommand{\loopnstepto}[1]{\mathrel{\sloopnstepto{#1}}}

\newcommand{\loopentrytransition}{loop-entry tran\-si\-tion}
\newcommand{\loopentrytransitions}{\loopentrytransition{s}}

\newcommand{\aLname}{n}

\newcommand{\sterminates}{{\downarrow}}
\newcommand{\terminates}[1]{{#1}{\sterminates}}
\newcommand{\sterminatesbig}{{\big\downarrow}}
\newcommand{\terminatesbig}[1]{{#1}{\sterminatesbig}}

\newcommand{\snotterminates}{\ndownarrow}  
\newcommand{\notterminates}[1]{{#1}{\snotterminates}}

\newcommand{\sterminatesfact}[1]{{\factor{\downarrow}{#1}}}
\newcommand{\terminatesfact}[2]{{#1}{\sterminatesfact{#2}}}

\newcommand{\onescriptbs}{\scalebox{0.75}{$\scriptstyle\bs{1}$}}
\newcommand{\onebrackscript}{\scalebox{0.75}{$\scriptstyle (1)$}}

\newcommand{\soneterminates}{{\pap{\downarrow}{\hspace*{-1.5pt}\firebrick{\onebrackscript}}}}
\newcommand{\oneterminates}[1]{{#1}{\soneterminates}}
\newcommand{\soneterminatesi}[1]{{\pbap{\downarrow}{\hspace*{-1.25pt}\firebrick{\onebrackscript}}{#1}}}
\newcommand{\oneterminatesi}[2]{{#2}{\soneterminatesi{#1}}}

\newcommand{\achart}{\mathcal{C}}
\newcommand{\acharti}{\iap{\achart}}
\newcommand{\bchart}{\mathcal{D}}
\newcommand{\bcharti}{\iap{\bchart}}

\newcommand{\acharthat}{\hspace*{0.75pt}\Hat{\hspace*{-0.75pt}\achart}\hspace*{-0pt}} 

\newcommand{\aonechart}{\underline{\mathcal{C}}}
\newcommand{\aonecharti}{\iap{\aonechart}}
\newcommand{\bonechart}{\underline{\mathcal{D}}}
\newcommand{\bonecharti}{\iap{\bonechart}}
\newcommand{\conechart}{\underline{\mathcal{E}}}
\newcommand{\conecharti}{\iap{\conechart}}

\newcommand{\aonechartacc}{\underline{\mathcal{C}}\hspace*{-0.4pt}'}
\newcommand{\bonechartacc}{\underline{\mathcal{D}}\hspace*{0pt}'}
\newcommand{\aonechartacci}{\iap{\aonechartacc}}
\newcommand{\bonechartacci}{\iap{\bonechartacc}}
\newcommand{\aonechartdacc}{\underline{\mathcal{C}}''}
\newcommand{\aonechartdacci}{\iap{\aonechartdacc}}
\newcommand{\aonecharthat}{\hspace*{0.2pt}\Hat{\hspace*{-0.75pt}\aonechart}\hspace*{-0.75pt}} 
\newcommand{\aonecharthati}[1]{\hspace*{0.2pt}\iap{\Hat{\hspace*{-0.75pt}\aonechart}}{#1}\hspace*{-0.75pt}}
\newcommand{\aonecharthatjp}[1]{\hspace*{0.2pt}\pap{\widehat{\hspace*{-0.75pt}\aonechart}}{\scriptscriptstyle (#1)}\hspace*{-0.75pt}} 
\newcommand{\bonecharthati}[1]{\hspace*{0.2pt}\iap{\Hat{\hspace*{-0.75pt}\bonechart}}{#1}\hspace*{-0.75pt}} 
 
\newcommand{\conecharthatjp}[1]{\hspace*{0.2pt}\pap{\widehat{\hspace*{-0.75pt}\conechart}}{\scriptscriptstyle (#1)}\hspace*{-0.75pt}} 
\newcommand{\conecharthatijp}[2]{\hspace*{0.2pt}\bpap{\widehat{\hspace*{-0.75pt}\conechart}}{#1}{\scriptscriptstyle (#2)}\hspace*{-0.75pt}} 

\newcommand{\aonecharthatacci}[1]{\hspace*{0.2pt}\aonecharthat\hspace*{0.4pt}'_{#1}\hspace*{-0.75pt}}
\newcommand{\aonecharthatdacci}[1]{\hspace*{0.2pt}\aonecharthat\hspace*{0.4pt}''_{#1}\hspace*{-0.75pt}}

\newcommand{\bonecharthatsubscript}{\widehat{\smash{\bonechart}\rule{0pt}{4.75pt}}} 
\newcommand{\bonecharthataccsubscript}{\widehat{\smash{\bonechartacc}\rule{0pt}{4.75pt}}} 
\newcommand{\bonecharthatsubscripti}[1]{\widehat{\smash{\bonechart}\rule{0pt}{4.75pt}}_{#1}} 
\newcommand{\bonecharthataccsubscripti}[1]{\widehat{\smash{\bonechart}\rule{0pt}{4.75pt}}{}'_{#1}} 

\newcommand{\chartof}{\fap{\achart}}
\newcommand{\onechartof}{\fap{\aonechart}}

\newcommand{\onecharthatof}{\fap{\aonecharthat}}

\newcommand{\indscchartof}[1]{\iap{#1}{\subscriptindtrans}} 
\newcommand{\indchartof}{\indscchartof}
\newcommand{\indchart}{\indscchartof}

\newcommand{\LTS}{LTS}

\newcommand{\oneLTS}{$\sone$\nb-\LTS}

\newcommand{\saTSS}{{\mathcal{T}}}
\newcommand{\StExpTSS}{\text{$\saTSS$}}
\newcommand{\StExpTSSover}[1]{\text{$\fap{\StExpTSS}{#1}$}}

  \newcommand{\subscriptindtrans}{\mediumblue{\scriptscriptstyle\pmb{\otind{\cdot}}}}

\newcommand{\aoneloop}{\underline{\mathcal{L}}}
\newcommand{\aoneloopi}{\iap{\aoneloop}}

\newcommand{\indsubonechartinat}[1]{\fap{\aonecharti{#1}}}

\newcommand{\otind}[1]{({#1}]}

\newcommand{\iact}[1]{\mediumblue{{\scriptscriptstyle\pmb (}\hspace*{-0pt}{\black{#1}}\hspace*{0.4pt}{\scriptscriptstyle\pmb ]}}}  
\newcommand{\silt}[1]{\slt{\iact{#1}}}

\newcommand{\silti}[2]{{\xrightarrow{\iact{#1}}}{_{#2}}}
\newcommand{\ilt}[1]{\mathrel{\silt{#1}}}

\newcommand{\ilti}[2]{\mathrel{\silti{#1}{#2}}}

\newcommand{\sasol}{s}
\newcommand{\sasoli}{\iap{\sasol}}
\newcommand{\asol}{\fap{\sasol}}
\newcommand{\asoli}[1]{\fap{\iap{\sasol}{#1}}}
\newcommand{\sbsol}{t}

\newcommand{\sextrsol}{\sasol}
\newcommand{\sextrsolof}{\iap{\sextrsol}}
\newcommand{\sextrsoluntil}{\sbsol}
\newcommand{\sextrsoluntilof}{\iap{\sextrsoluntil}}
\newcommand{\extrsol}{\fap{\sextrsol}}

\newcommand{\extrsolof}[1]{\fap{\sextrsolof{#1}}}
\newcommand{\extrsoluntilof}[1]{\bfap{\sextrsoluntilof{#1}}}

\newcommand{\sterminatesconst}{\tau}

\newcommand{\sterminatesconstof}{\iap{\sterminatesconst}}

\newcommand{\terminatesconstof}[1]{\fap{\sterminatesconstof{#1}}}
\newcommand{\actions}{\mathit{A}}

\newcommand{\oneactions}{\firebrick{\underline{\actions}}} 

\newcommand{\aact}{a}
\newcommand{\bact}{b}
\newcommand{\cact}{c}
\newcommand{\dact}{d}

\newcommand{\aacti}{\iap{\aact}}

\newcommand{\cacti}{\iap{\cact}}

\newcommand{\aoneact}{\firebrick{\underline{a}}}
\newcommand{\boneact}{\firebrick{\underline{b}}}

\newcommand{\aoneacti}[1]{\firebrick{\iap{\aoneact}{#1}}}
\newcommand{\boneacti}[1]{\firebrick{\iap{\boneact}{#1}}}

\newcommand{\verts}{V}

\newcommand{\start}{\averti{\hspace*{-0.5pt}\text{\nf s}}}
\newcommand{\transs}{{\sred}}
\newcommand{\exts}{{\sterminates}}
\newcommand{\termexts}{{\sterminates}}

\newcommand{\onetranss}{\silt{\cdot}}
\newcommand{\onetermexts}{\soneterminates}
\newcommand{\ats}{{AT}}

\newcommand{\alab}{\darkcyan{l}}
\newcommand{\alabi}[1]{\iap{\alab}{\darkcyan{#1}}}

\newcommand{\vertsof}{\fap{\verts\hspace*{-1pt}}}

\newcommand{\atsiof}[1]{\fap{\atsi{#1}}}

\newcommand{\vertsi}[1]{\iap{\verts}{\hspace*{-0.25pt}{#1}}}

\newcommand{\starti}[1]{\averti{\text{\nf s},#1}}

\newcommand{\termextsi}[1]{\iap{\termexts}{{#1}}}
\newcommand{\atsi}[1]{\iap{\ats}{\hspace*{-1.5pt}{#1}}}

\newcommand{\stransitions}{T}
\newcommand{\transitionsinfrom}[1]{\fap{\iap{\stransitions}{#1}}}

\newcommand{\transshat}{\hat{\transs}}

\newcommand{\ainst}{\iota}

\newcommand{\arule}{R}

\newcommand{\asettranss}{U}
\newcommand{\asettranssi}{\iap{\asettranss}}

\newcommand{\avert}{v}
\newcommand{\bvert}{w}
\newcommand{\cvert}{u}

\newcommand{\averti}{\iap{\avert}}
\newcommand{\bverti}{\iap{\bvert}}
\newcommand{\cverti}{\iap{\cvert}}

\newcommand{\avertacc}{\avert'}

\newcommand{\avertacci}{\iap{\avertacc}}

\newcommand{\averttilde}{\tilde{\avert}}

\newcommand{\averttildei}{\iap{\averttilde}}

\newcommand{\bvertbar}{\overlinebar{\bvert}}

\newcommand{\bvertbari}{\iap{\bvertbar}}

\newcommand{\atrans}{\tau}

\newcommand{\sloopelim}{\textnf{elim}}
\newcommand{\sloopelimred}{\iap{\Rightarrow}{\sloopelim}}
\newcommand{\loopelimred}{\mathrel{\sloopelimred}}
\newcommand{\sloopelimredrtc}{\pbap{\Rightarrow}{*}{\sloopelim}}
\newcommand{\loopelimredrtc}{\mathrel{\sloopelimredrtc}}
\newcommand{\sloopelimredtc}{\pbap{\Rightarrow}{+}{\sloopelim}}
\newcommand{\loopelimredtc}{\mathrel{\sloopelimredtc}}
\newcommand{\sloopelimllarrow}{\mbox{$\Longrightarrow$\hspace*{-0.45cm}\raisebox{-0.275pt}{\tikz{\draw(0,0)circle[radius=0.9mm];}}\hspace*{0.26cm}}}
\newcommand{\sloopelimllred}{\bap{\sloopelimllarrow\hspace*{-1pt}}{\sloopelim}}
\newcommand{\loopelimllred}{\mathrel{\sloopelimllred}}
\newcommand{\sloopelimllredrtc}{\bpap{\sloopelimllarrow\hspace*{-1pt}}{\sloopelim}{*}}
\newcommand{\loopelimllredrtc}{\mathrel{\sloopelimllredrtc}}

\newcommand{\gensubchartofby}[2]{{#1}{\downarrow}^{\hspace*{-0.75pt}#2}_{\hspace*{-0.5pt}*}}

\newcommand{\sfunbisim}{%
    \setbox0=\hbox{\kern-.1ex{$\rightarrow$}\kern-.1ex}
    \setbox1=\vbox{\hbox{\raise .1ex \box0}\hrule}%
    {\hbox{\kern.05ex\box1\kern.1ex}}
  }
\newcommand{\funbisim}{\hspace*{-0pt}\mathrel{\sfunbisim}}

\newcommand{\sshaftfunbisim}{
    \setbox0=\hbox{\kern-.1ex{{---}}\kern-.1ex}
    \setbox1=\vbox{\hbox{\raise .1ex \box0}\hrule}%
    {\hbox{\kern.05ex\box1\kern.1ex}}
  }

\newcommand{\sconvfunbisim}{%
    \setbox0=\hbox{\kern-.1ex{$\leftarrow$}\kern-.1ex}
    \setbox1=\vbox{\hbox{\raise .1ex \box0}\hrule}%
    {\hbox{\kern.05ex\box1\kern.1ex}}
  }

\newcommand{\sbisim}{%
    \setbox0=\hbox{\kern-.1ex{$\leftrightarrow$}\kern-.1ex}
    \setbox1=\vbox{\hbox{\raise .1ex \box0}\hrule}%
    \hbox{\kern.1ex\box1\kern.1ex}
  }
\newcommand{\bisim}{\mathrel{\sbisim\hspace*{1pt}}}

\newcommand{\sfunonebisim}{%
    \setbox0=\hbox{\kern-.1ex{\firebrick{$\rightarrow$}}\kern-.1ex}
    \setbox1=\vbox{\hbox{\raise .1ex \box0}\hrule}%
    \ensuremath{\firebrick{\hbox{\kern.05ex\box1\kern.1ex}}}
  }
\newcommand{\funonebisim}{\mathrel{\sfunonebisim}}

\newcommand{\sconvfunonebisim}{%
    \setbox0=\hbox{\kern-.1ex{$\leftarrow$}\kern-.1ex}
    \setbox1=\vbox{\hbox{\raise .1ex \box0}\hrule}%
    \ensuremath{\firebrick{\hbox{\kern.05ex\box1\kern.1ex}}}
  }
\newcommand{\convfunonebisim}{\mathrel{\sconvfunonebisim}}

\newcommand{\sonebisim}{%
    \setbox0=\hbox{\kern-.1ex{$\leftrightarrow$}\kern-.1ex}
    \setbox1=\vbox{\hbox{\raise .1ex \box0}\hrule}%
    \ensuremath{\firebrick{\hbox{\kern.1ex\box1\kern.1ex}}}
  }
\newcommand{\onebisim}{\mathrel{\sonebisim}}

\newcommand{\sfunbisimvia}[1]{{\iap{\sfunbisim}{#1}}}
\newcommand{\funbisimvia}[1]{\hspace*{-0.5pt}\mathrel{\sfunbisimvia{#1}}}

\newcommand{\sfunonebisimvia}[1]{%
    \setbox0=\hbox{\kern-.1ex{$\rightarrow$}\kern-.1ex}
    \setbox1=\vbox{\hbox{\raise .1ex \box0}\hrule}%
    {\bap{\firebrick{\hbox{\kern.05ex\box1\kern.1ex}}}{#1}}
  }
\newcommand{\funonebisimvia}[1]{\hspace*{-1.5pt}\mathrel{\sfunonebisimvia{#1}}}
\newcommand{\sconvfunonebisimvia}[1]{%
    \setbox0=\hbox{\kern-.1ex{$\leftarrow$}\kern-.1ex}
    \setbox1=\vbox{\hbox{\raise .1ex \box0}\hrule}%
    {\bap{\firebrick{\hbox{\kern.05ex\box1\kern.1ex}}}{#1}}
  }

\newcommand{\sonebisimvia}[1]{%
    \setbox0=\hbox{\kern-.1ex{$\leftrightarrow$}\kern-.1ex}
    \setbox1=\vbox{\hbox{\raise .1ex \box0}\hrule}%
    {\bap{\firebrick{\hbox{\kern.05ex\box1\kern.1ex}}}{\hspace*{-1.5pt}#1}}
  }
\newcommand{\onebisimvia}[1]{\mathrel{\sonebisimvia{#1}}}  

\newcommand{\sbisimvia}[1]{\iap{\sbisim}{\hspace*{-1.5pt}#1}}
\newcommand{\bisimvia}[1]{\hspace*{-0.5pt}\mathrel{\sbisimvia{#1}}}

\newcommand{\abisim}{B}

\newcommand{\sbehinc}{{\sqsubseteq}}

\newcommand{\sbehinca}[1]{{\prescript{#1}{}{\sbehinc}}}
\newcommand{\behinca}[1]{\mathrel{\sbehinca}}

\newcommand{\sonebehinc}{{\pap{\sbehinc}{\subotr}}}

\newcommand{\sonebehinca}[1]{{{}_{#1}\sonebehinc}}
\newcommand{\onebehinca}[1]{\mathrel{\sonebehinca}}
\newcommand{\REFL}{\textrm{\nf Ref\/l}}
\newcommand{\SYMM}{\textrm{\nf Symm}}
\newcommand{\TRANS}{\textrm{\nf Trans}}

\newcommand{\CXT}{\textrm{\nf Cxt}}

\newcommand{\sRSP}{\textrm{\nf RSP}}
\newcommand{\RSPstar}{\text{$\sRSP^{*}\hspace*{-1pt}$}}

\newcommand{\snotRSP}{\textrm{\nf\sout{RSP}}}
\newcommand{\notRSPstar}{\text{$\snotRSP^{*}\hspace*{-1pt}$}}

\newcommand{\sUSP}{\textrm{\nf USP}}

\newcommand{\USP}{\sUSP}
\newcommand{\USPone}{\text{$\bap{\sUSP\hspace*{-0.5pt}}{1}$}}

\newcommand{\assocstexpsum}{\fap{\text{\nf assoc}}{\sstexpsum}}
\newcommand{\assocstexpprod}{\fap{\text{\nf assoc}}{\sstexpprod}}
\newcommand{\commstexpsum}{\fap{\text{\nf comm}}{\sstexpsum}}
\newcommand{\neutralstexpsum}{\fap{\text{\nf neutr}}{\sstexpsum}}  
\newcommand{\idempotstexpsum}{\fap{\text{\nf idempot}}{\sstexpsum}}

\newcommand{\rdistr}{\fap{\text{\nf r-distr}}{\sstexpsum,\sstexpprod}}

\newcommand{\stexpzerostexpprod}{\text{\nf deadlock}}
\newcommand{\rightidstexpprod}{\fap{\text{\nf id}_{\text{\nf r}}}{\sstexpprod}}
\newcommand{\leftidstexpprod}{\fap{\text{\nf id}_{\text{\nf l}}}{\sstexpprod}}
\newcommand{\recdefstexpit}{\fap{\text{\nf rec}}{{}^{*}}}
\newcommand{\termstexpit}{\fap{\text{\nf trm-body}}{{}^{*}}}

\newcommand{\ACI}{\text{\sf ACI}}
\newcommand{\ACIover}{\fap{\ACI}}

\newcommand{\seqin}[1]{{\iap{=}{#1}\hspace*{1pt}}}
\newcommand{\eqin}[1]{\mathrel{\seqin{#1}}}
\newcommand{\seqinsol}[1]{{\pbap{=}{\text{\scriptsize (sol)}}{#1}\hspace*{1pt}}}
\newcommand{\eqinsol}[1]{\mathrel{\seqinsol{#1}}}

\newcommand{\ACIeq}{\eqin{\ACI}}

\newcommand{\milnersyseq}{\eqin{\milnersys}}

\newcommand{\milnersysmineq}{\eqin{\milnersysmin}}

\newcommand{\sderivablein}[1]{\vdash_{#1}}
\newcommand{\derivablein}[1]{\sderivablein{#1}}
\newcommand{\sthmequiv}{{\sim}}
\newcommand{\thmequiv}{\mathrel{\sim}}
\newcommand{\sisthmsubsumedby}{{\precsim}}
\newcommand{\isthmsubsumedby}{\mathrel{\sisthmsubsumedby}}
\newcommand{\sthmsubsumes}{{\succsim}}
\newcommand{\thmsubsumes}{\mathrel{\sthmsubsumes}}
\newcommand{\ssubsystem}{{\sqsubseteq}}
\newcommand{\subsystem}{\mathrel{\ssubsystem}}
\newcommand{\ssupsystem}{{\sqsupseteq}}
\newcommand{\supsystem}{\mathrel{\ssupsystem}}

\newcommand{\CoindProofi}[1]{\text{\nf CoProof$_{#1}$}}
\newcommand{\LCoindProofi}[1]{\text{\nf LCoProof$_{#1}$}}

\newcommand{\saeqfun}{L}
\newcommand{\aeqfun}{\fap{\saeqfun}}
\newcommand{\saeqfuni}{\iap{\saeqfun}}
\newcommand{\aeqfuni}[1]{\fap{\iap{\saeqfun}{#1}}}

\newcommand{\scoindproofeqin}[1]{{\bap{\overset{\text{\nf coind}}{=\joinrel=\joinrel=}}{#1}}} 
\newcommand{\sLLEEcoindproofeqin}[1]{{\bap{\overset{\text{\nf\scriptsize LLEE}}{=\joinrel=\joinrel=}}{#1}}} 
\newcommand{\coindproofeqin}[1]{\mathrel{\scoindproofeqin{#1}}}
\newcommand{\LLEEcoindproofeqin}[1]{\mathrel{\sLLEEcoindproofeqin{#1}}}

\newcommand{\aLLEECoProof}{\mathcal{L\hspace*{-0.5pt}C\hspace*{-0.75pt}P}}
\newcommand{\aLLEECoProofoverof}[1]{\fap{\iap{\aLLEECoProof}{#1}}}
\newcommand{\aLLEECoProofacc}{\mathcal{L\hspace*{-0.5pt}C\hspace*{-0.75pt}P'}}

\newcommand{\aCoProof}{\mathcal{C\hspace*{-0.75pt}P}}

\newcommand{\sdescendsinloopto}{{\curvearrowright}}  
\newcommand{\descendsinloopto}{\mathrel{\sdescendsinloopto}}

\newcommand{\sconvdescendsinloopto}{{\curvearrowleft}} 
\newcommand{\convdescendsinloopto}{\mathrel{\sconvdescendsinloopto}}

\newcommand{\sconvdescendsinlooptotc}{{\pap{\sconvdescendsinloopto}{\hspace*{-0.8pt}+}}}
\newcommand{\convdescendsinlooptotc}{\mathrel{\sconvdescendsinlooptotc}}

\newcommand{\sconvdescendsinlooptosc}{{\pap{\sconvdescendsinloopto}{\hspace*{-0.8pt}=}}}
\newcommand{\convdescendsinlooptosc}{\mathrel{\sconvdescendsinlooptosc}}
\newcommand{\muterm}{$\mu$\nb-term}

\newcommand{\txtdescendsinloopto}{de\-scends-in-loop-to}

\newcommand{\premisen}[1]{{$#1$}\nb-pre\-mise}

\newcommand{\prooftheoretic}{proof-the\-o\-re\-tic}

\newcommand{\sloopsbackto}{{\lefttorightarrow}} 
\newcommand{\loopsbackto}{\mathrel{\sloopsbackto}}

\newcommand{\eqlogic}{\text{$\mathcal{E\hspace*{-1pt}L}$}}

\newcommand{\eqlogicover}[1]{\fap{\eqlogic}{#1}}

\newcommand{\Fone}{\text{$\bap{\textsf{F}}{\hspace*{-1pt}\textsf{1}}$}}

\newcommand{\ACeq}{$\textsf{AC}^{=}$}
\newcommand{\AKeq}{$\textsf{AK}^{=}$}
\newcommand{\BHeq}{$\textsf{BH}^{=}$}

\newcommand{\BHle}{$\textsf{BH}^{\le}$}
\newcommand{\Geq}{$\textsf{G}^{\slangsemeq}$}
\newcommand{\HS}{$\textsf{HS}$}

\newcommand{\milnersys}{\text{\nf\sf Mil}}

\newcommand{\milnersysof}[1]{\text{$\fap{\text{\nf\sf Mil}}{#1}$}}
\newcommand{\milnersysover}{\milnersysof}

\newcommand{\milnersysacc}{\text{\text{\nf\sf Mil}$^{\prime}$}}
\newcommand{\milnersysaccover}[1]{\text{$\fap{\milnersysacc}{#1}$}}
\newcommand{\milnersysaccbar}{\text{$\overline{\nf\sf Mil}{}^{\prime}$}}
\newcommand{\milnersysaccbarover}[1]{\text{$\fap{\milnersysaccbar}{#1}$}}

\newcommand{\milnersysmin}{\text{\text{\nf\sf Mil}$^{\pmb{-}}$}}
\newcommand{\milnersysminover}[1]{\text{$\fap{\milnersysmin}{#1}$}}

\newcommand{\coindmilnersys}{\text{\nf\sf cMil}}      
\newcommand{\coindmilnersysof}[1]{\text{$\fap{\coindmilnersys}{#1}$}}
\newcommand{\coindmilnersysover}{\coindmilnersysof}

\newcommand{\coindmilnersysbar}{\text{$\overline{\text{\sf cMil}}$}}
\newcommand{\coindmilnersysbarof}[1]{\text{$\fap{\coindmilnersysbar}{#1}$}}
\newcommand{\coindmilnersysbarover}{\coindmilnersysbarof}

\newcommand{\coindmilnersysone}{{\text{{\nf\sf cMil}$_1$}}}
\newcommand{\coindmilnersysoneof}[1]{\text{$\fap{\coindmilnersysone}{#1}$}}
\newcommand{\coindmilnersysoneover}{\coindmilnersysoneof}

\newcommand{\CC}{\text{\sf C\hspace*{-0.75pt}C}}

\newcommand{\CCover}[1]{\text{$\fap{\CC}{#1}$}}

\newcommand{\CLC}{\text{\sf C\hspace*{-0.5pt}L\hspace*{-0.5pt}C}}     
\newcommand{\CLCover}[1]{\text{$\fap{\CLC}{#1}$}}

\newcommand{\BBP}{\text{$\text{\sf BBP}$}}

\newcommand{\asys}{{\mathcal{S}}}
\newcommand{\asysi}{\iap{\asys}}

\newcommand{\aDeriv}{\mathcal{D}}
\newcommand{\aDerivi}{\iap{\mathcal{D}}}
\newcommand{\aDerivacc}{\aDeriv'}
\newcommand{\aDerivacci}{\iap{\aDerivacc}}

\newcommand{\aDerivtilde}{\tilde{\aDeriv}}

\newcommand{\sLEE}{\text{\nf LEE}}

\newcommand{\LEE}{\sLEE}

\newcommand{\sLLEE}{\text{\nf LLEE}}

\newcommand{\LLEE}{\sLLEE}

\newcommand{\thplus}[2]{{#1}{+}{#2}}

\newcommand{\aseteqs}{\Gamma}
\newcommand{\aseteqsacc}{\aseteqs'}
\newcommand{\bseteqs}{\Delta}

\newcommand{\picarrowstart}{\raisebox{2pt}{\begin{tikzpicture}%
                                             \draw[<-,very thick,>=latex,chocolate,shorten <=2pt](0,0) -- ++ (180:{12pt});%
                                           \end{tikzpicture}}}
  
\newcommand{\pictermvert}{\begin{tikzpicture}%
                           \node[draw,chocolate,very thick,circle,minimum width=2.5pt,fill,inner sep=0pt,outer sep=2pt](v){};%
                           \draw[thick,chocolate] (v) circle (0.12cm);%
                         \end{tikzpicture}}

\DeclareRobustCommand*{\mycirc}[1]{%
  \tikz[baseline=(C.base)]
    \node[draw,circle,inner sep=1pt](C){#1};}

\begin{document}

\title[A Coinductive Reformulation of Milner's Proof System]
      {A Coinductive Reformulation of Milner's Proof System\texorpdfstring{\\}{}
       for Regular Expressions Modulo Bisimilarity}
\titlecomment{{\lsuper*}This is a special-journal version of the paper ``A Coinductive Version of Milner's Proof System for Regular Expressions Modulo Bisimilarity''
                        \cite{grab:2021:calco} presented at CALCO~2021.}

\author[C.~Grabmayer]{Clemens Grabmayer\lmcsorcid{0000-0002-2414-1073}}  
\address{Dept.\ of Computer Science, Gran Sasso Science Institute, Viale F.\ Crispi,~7, 67100~L'Aquila~AQ,~Italy}  
\email{\texttt{clemens.grabmayer@gssi.it}}  





\begin{abstract}\sloppy
  Milner 
         defined 
                        an operational semantics for regular expressions as \finitestate\ processes.
  In order to axiomatize bisimilarity of regular expressions under this process semantics,  
    he adapted Salomaa's complete proof system for equality of regular expressions under the language semantics.  
  Apart from most equational axioms, Milner's system $\milnersys$ inherits from Salomaa's system
    a non-algebraic rule for solving fixed-point equations.   
  Recognizing distinctive properties of the process semantics 
    that render Salomaa's proof strategy inapplicable, Milner posed completeness of the system \milnersys\ as an~open~question. 
      
  As a proof-theoretic approach to this problem
    we characterize
      the derivational power that the fixed-point rule adds
        to the purely equational part \milnersysmin\ of \milnersys.
  We do so by means of a coinductive rule
    that permits cyclic derivations that consist 
      of a finite process graph (maybe with empty steps) that satisfies the layered loop existence and elimination property LLEE,
      and two of its \provablein{\milnersysmin} solutions.
  By adding this rule instead of the fixed-point rule in \milnersys, we define 
    the coinductive reformulation $\coindmilnersys$ as an extension of \milnersysmin.
  For showing that $\coindmilnersys$ and $\milnersys$ are theorem equivalent
    we develop effective proof transformations from \milnersys\ to \coindmilnersys, and vice versa.
  Since it is located half-way in between bisimulations and proofs in Milner's system \milnersys,
    $\coindmilnersys$ may become a beachhead for a completeness proof~of~\milnersys.
  
  This article extends our contribution to the CALCO~2021 proceedings.
  Here we refine the proof transformations by framing them as eliminations of derivable and admissible rules,
  and we link coinductive proofs to a coalgebraic formulation of solutions of process~graphs. 
\end{abstract}

\maketitle

\section{Introduction}
  \label{intro}  

Milner introduced  in \cite{miln:1984} a process semantics $\procsem{\cdot}$ for regular expressions $\astexp$ as \finitestate\ process graphs $\procsem{\astexp}$.
Informally the process interpretation is defined as follows,
  for regular expressions built from the constants $\stexpzero$ and $\stexpone$ by using the regular operators $\sstexpsum$, $\sstexpprod$, and $\stexpit{(\cdot)}$:
\begin{itemize}[itemsep=0.25ex]
  \item 
    $\stexpzero$ stands for deadlock, $\stexpone$ for successful termination, letters $a$ for atomic actions,
  \item  
    the operators $\sstexpsum$ and $\sstexpprod$ are interpreted as choice and concatenation of processes, respectively,
  \item    
    (unary) Kleene star $\stexpit{(\cdot)}$ denotes iteration with the option to terminate successfully 
    before each execution of the iteration body
    (then even infinitely many iterations are possible).
\end{itemize}  

Milner called regular expressions `star expressions' when they are interpreted as processes.
He formulated this semantics 
  after developing a complete equational proof system for equality of `a class of regular behaviors'. 
    By that he understood 
      the bisimilarity equivalence classes of \finitestate\ processes that are represented by \muterm{s}.
Then he defined the process semantics by interpreting regular expressions as \muterm\ representations of \finitestate\ processes. 
In doing so, he defined `star behaviors',  
  the bisimilarity equivalence classes of the interpretations of star expressions
  as a subclass of `regular behaviors'.
As an af\-ter\-thought to the complete proof system for regular behaviors,
  he was interested in
   an axiomatization of equality of `star behavior' directly on star expressions (instead of on \muterm\ representations).
For this purpose he appropriately adapted Salomaa's complete proof system \cite{salo:1966} for language equivalence on regular expressions
to a system $\milnersys$ that is sound for equality of denoted star behaviors.
But Milner noticed that completeness of 
                                        \milnersys\ cannot be shown in analogy to Salomaa's completeness proof.
He formulated completeness of \milnersys\ as an open problem,
  because he realized a significant difficulty
    due to a peculiarity by which the process semantics contrasts starkly with the language semantics of regular expressions.
  
The process semantics of regular expressions is incomplete in the following sense. 
  While for every \finitestate\ automaton $\aDFA$ there is a regular expression $\astexp$
    whose language in\-ter\-pre\-ta\-tion $\langsem{\astexp}$ coincides with the language accepted by $\aDFA$ (formally $\languageof{\aDFA} = \langsem{\astexp}$),
  it is not the case that every \finitestate\ process is the process interpretation of some star expression,   
    not even modulo bisimilarity.  
Giving a counterexample that demonstrates this,        
  Milner proved in \cite{miln:1984} that the process graph $G_1$ below with linear recursive equational specification $\fap{\aspec}{G_1}$
    does not define a star behavior,
      and hence is not bisimilar to the process interpretation of a star expression. 
He conjectured that the same is true for the process graph $G_2$ below
  with specification $\fap{\aspec}{G_2}$.
That was confirmed later by~Bosscher~\cite{boss:1997}.    
  \begin{center}\vspace{-1ex}\label{fig:milner-bosscher-expressible}
    \scalebox{0.9}{\begin{tikzpicture}

%
\matrix[anchor=center,row sep=0cm,column sep=1.5cm,
        every node/.style={draw,very thick,circle,minimum width=2.5pt,fill,inner sep=0pt,outer sep=2pt}] at (6,0.75) {
  \node[color=chocolate](C-1-0){};  &   \node[color=chocolate](C-1-1){};
  \\
};
\draw[<-,very thick,>=latex,chocolate,shorten <=2pt](C-1-0) -- ++ (180:0.58cm);
\path(C-1-0) ++ (0cm,-0.4cm) node{$Y_1$};
\path(C-1-0) ++ (-0.325cm,0.5cm) node{\Large $\iap{G}{2}$};

\path (C-1-1) ++ (0cm,-0.45cm) node{$Y_2$};

\draw[thick,chocolate] (C-1-1) circle (0.12cm);
\draw[thick,chocolate] (C-1-0) circle (0.12cm);
\draw[->,bend left,distance=0.65cm,shorten <=2pt,shorten >=2pt] (C-1-0) to node[above]{$a$} (C-1-1); 
\draw[->,bend left,distance=0.65cm,shorten <=2pt,shorten >=2pt] (C-1-1) to node[below]{$b$} (C-1-0); 


\path (C-1-0) ++ (0cm,-0.9cm) node[below]{$
  \fap{\aspec}{G_2} = 
  \left\{\,
  \begin{aligned}
    Y_1 & {} = 1 + a \prod Y_2
    \\ 
    Y_2 & {} = 1 + b \prod Y_1
  \end{aligned}
  \,\right.
  $};

%
\matrix[anchor=center,row sep=0.924cm,column sep=0.75cm,
        every node/.style={draw,very thick,circle,minimum width=2.5pt,fill,inner sep=0pt,outer sep=2pt}] at (0,0) {
                   &                  &  \node(C-2-2){};
  \\
  \node(C-2-1){};  &                  &                  
  \\
                   &                  &  \node(C-2-3){};  
  \\
};
\draw[<-,very thick,>=latex,color=chocolate](C-2-1) -- ++ (180:0.5cm);  

\path(C-2-1) ++ (-0.325cm,0.5cm) node{\Large $\iap{G}{1}$};
\path(C-2-1) ++ (-0.2cm,-0.4cm) node{$X_1$};
\draw[->,bend right,distance=0.65cm] (C-2-1) to node[above]{$\aacti{2}$} (C-2-2); 
\draw[->,bend right,distance=0.65cm] (C-2-1) to node[left]{$\aacti{3}$}  (C-2-3);

\path(C-2-2) ++ (0cm,0.4cm) node{$X_2$};
\draw[->,bend right,distance=0.65cm]  (C-2-2) to node[above]{$\aacti{1}$} (C-2-1); 
\draw[->,bend left,distance=0.65cm]  (C-2-2) to node[right,xshift=-1pt]{$\aacti{3}$} (C-2-3);

\path(C-2-3) ++ (0cm,-0.4cm) node{$X_3$};
\draw[->,bend right,distance=0.65cm] (C-2-3) to node[left]{$\aacti{1}$}  ($(C-2-1)+(+0.15cm,-0.05cm)$);
\draw[->,bend left,distance=0.65cm]  (C-2-3) to node[right,xshift=-1.5pt]{$\aacti{2}$} (C-2-2);

\path (C-2-1) ++ (-1cm,0cm) node[left]{$
  \fap{\aspec}{G_1} = 
  \left\{\,
  \begin{aligned}
    X_1 & {} = a_2 \prod X_2  +  a_3 \prod X_3
    \\ 
    X_2 & {} = a_1 \prod X_1  +  a_3 \prod X_3
    \\
    X_3 & {} = a_1 \prod X_1  +  a_2 \prod X_2
  \end{aligned}
  \,\right.
  $};

\end{tikzpicture}}
  \end{center}\vspace{-2ex}
(Here and later we highlight the start vertex of a process graph by a brown arrow~\picarrowstart,
and emphasize a vertex $\avert$ with immediate termination in brown as \pictermvert\ including a boldface ring.)
It follows that 
  the systems of $\fap{\aspec}{G_1}$ and $\fap{\aspec}{G_2}$ of guarded equations with star expressions cannot be solved by star expressions modulo bisimilarity.
Due to soundness of $\milnersys$, also the specifications $\fap{\aspec}{G_1}$ and $\fap{\aspec}{G_2}$ are unsolvable by star expressions
  when equality is interpreted as provability in $\milnersys$.
However, if all actions in the process graphs $G_1$ and $G_2$ are replaced by a single action $\aact$, obtaining graphs $G_1^{\scriptscriptstyle (a)}$ and $G_2^{\scriptscriptstyle (a)}$,  
  then the arising specifications $\fap{\aspec}{G_1^{\scriptscriptstyle (a)}}$ and $\fap{\aspec}{G_2^{\scriptscriptstyle (a)}}$ are solvable, modulo bisimilarity, 
  and also with respect to provability in $\milnersys$.
Indeed it is easy to verify that solutions are obtained
  by letting $X_1 \defdby X_2 \defdby X_3 \defdby \stexpprod{\stexpit{\aact}}{\stexpzero}$ in $\fap{\aspec}{G_1^{\scriptscriptstyle (a)}}$,
    and by letting $Y_1 \defdby Y_2 \defdby \stexpit{\aact}$ in $\fap{\aspec}{G_2^{\scriptscriptstyle (a)}}$.
    
The extraction procedure of solutions of specifications in Salomaa's proof completeness is able to solve every linear system of recursion equations,
  independently of the actions occurring.
It follows that an analogous procedure is not possible for solving systems of linear recursion equations in the process semantics. 
The extraction procedure for linear specifications with respect to the language semantics
  is possible because both laws for distributing $\sstexpprod$ over $\sstexpsum$ are available, and indeed are part of Salomaa's proof system.
But \milnersys\ does not contain   
  the \emph{left}-distributivity law  $x \cdot (y + z) = x \cdot y + x \cdot z$, because it famously is not sound under bisimilarity.
In the presence of only \emph{right}-distributivity $ (x + y) \cdot z = x \cdot z + y \cdot z$ in $\milnersys$
  the extraction procedure from Salomaa's proof does not work,
    because failure of \emph{left}-distributivity oftentimes prevents expressions to be rewritten
      in such a way that the fixed-point rule $\RSPstar$ in $\milnersys$ can be applied successfully.
%
But if $\RSPstar$ is replaced in $\milnersys$ by a general unique-solvability rule scheme for guarded systems of equations 
  (see Definition~\ref{def:milnersys}), 
then a complete system arises (noted in \cite{grab:2006}).
Therefore completeness of $\milnersys$ hinges on whether the fixed-point rule $\RSPstar$ enables to prove equal
  any two star-expression solutions of a given guarded system of equations, 
  on the basis of the purely equational part~$\milnersysmin$~of~$\milnersys$.
   
As a stepping stone for tackling this difficult question, we characterize 
  the derivational power that the fixed-point rule~$\RSPstar$ 
  adds to the subsystem $\milnersysmin$~of~$\milnersys$.
We do so by means of `coinductive proofs' whose shapes have  
the `loop existence and elimination property' \LEE\ from \cite{grab:fokk:2020:lics}. 
This property stems from the interpretation of (\onefree) star expressions,
  which is defined by induction on syntax trees, creating a hierarchy of `loop subgraphs'. 
Crucially for our purpose, linear guarded systems of equations that correspond to finite process graphs with \LEE\
  \emph{are} uniquely solvable modulo provability~in~$\milnersysmin$.
The reason is that process graphs with \LEE,
  which need not be in the image of the process semantics,
are amenable to applying \emph{right}-distributivity and the 
                                                      rule \RSPstar\ for an extraction procedure like in Salomaa's proof (see Section~\ref{coindmilnersys:2:milnersys}).
These process graphs can be expressed modulo bisimilarity by some star expression,
  which can be used to show that any two solutions modulo $\milnersysmin$ of a specification of \LEEshape\ are \provablyin{\milnersys} equal. 
This is a crucial step 
                       in the completeness proof by Fokkink and myself in \cite{grab:fokk:2020:lics} 
  for the tailored restriction \BBP\ of Milner's system $\milnersys$ to~`\onefree'~star~expressions.

\begin{figure}[t!]
\begin{tikzpicture}
 
\matrix[anchor=center,row sep=1.25cm,column sep=2.65cm,
        every node/.style={draw=none}
        ] at (0,0) {
    \node(v1){}; &             & \node(v2){};
    \\
                 & \node(v){}; 
    \\
  };
  
\path (v1) node(v1-label){$ (1 \cdot g^*) \cdot 0
                              \formeq
                            ((1 \cdot (a + b)) \cdot h^*)  \cdot 0 $};       
\draw[->,shorten <=2.8cm,shorten >=1.95cm] (v1) to node[above,pos=0.5725]{$a, b$} (v2);

\path (v2) node(v2-label){$ (1 \cdot g^*) \cdot 0
                              \formeq
                             (1 \cdot h^*) \cdot 0 $}; 
\draw[->,thick,densely dotted,out=-45,in=0,distance=0.75cm,shorten <= 0.2cm] (v2) to node[right]{$\sone$} ($(v) + (3.2cm,0.3cm)$);

\path (v) node(v-label){$ (\underbrace{
                           a + b
                                        }_{g})^* 
                                       \,\cdot\, 0
                            \formeq
                          (\underbrace{
                           a \cdot (a + b) + b
                                        }_{h})^* 
                            \,\cdot\, 0 $}; 
\draw[->,thick,darkcyan,shorten <= 1.1cm,shorten >= 0.5cm] (v) to node[above,pos=0.45]{$\loopnsteplab{1}$}
                                                                  node[below,pos=0.7]{$\black{a}$} (v1);
\draw[->,thick,darkcyan,shorten <= 1.1cm,shorten >= 0.5cm] (v) to node[above,pos=0.45]{$\loopnsteplab{1}$}
                                                                  node[below,pos=0.7]{$\black{b}$} (v2);

\matrix[anchor=center,row sep=1.25cm,column sep=1.1cm,
        every node/.style={draw,very thick,circle,minimum width=2.5pt,fill,inner sep=0pt,outer sep=2pt}
        ] at (7,0) {
    \node(v1--2){}; &             & \node(v2--2){};
    \\
                 & \node(v--2){}; 
    \\
  };  
  
\draw[->,thick,darkcyan] (v--2) to node[right,pos=0.7]{$\loopnsteplab{1}$}
                                   node[left,pos=0.45]{$\black{a}$} (v1--2);
\draw[->,thick,darkcyan] (v--2) to node[left,pos=0.7]{$\loopnsteplab{1}$}
                                   node[right,pos=0.45]{$\black{b}$} (v2--2);  

\draw[->] (v1--2) to node[above]{$a, b$} (v2--2);

\draw[<-,very thick,>=latex,chocolate,shorten <= 2pt](v--2) -- ++ (180:0.55cm); 
\draw[->,thick,densely dotted,out=-45,in=0,distance=0.75cm] 
  (v2--2) to node[right]{$\sone$} (v--2);

\path (v--2) ++ (-1.4cm,0cm) node(v--2-label) {\Large $\mathcal{G},\,\Hat{\mathcal{G}}$};

\end{tikzpicture}
  \vspace*{-1.5ex}
  \caption{\label{fig:ex:1:coindproof}%
           A \protect\LLEEwitnessed\ coinductive proof of $(a + b)^* \cdot 0 \formeq (a \cdot (a + b) + b)^* \cdot 0\,$~%
           with\vspace*{-1.5pt} underlying \protect\onechart~$\mathcal{G}$ and \protect\LLEEwitness~$\Hat{\mathcal{G}}$ 
             (with \protect\darkcyan{colored, indexed \protect\loopentry\ transitions}). 
           The proof uses the abbreviations $\,g \defdby a + b\,$ and $\,h \defdby a \cdot (a + b) + b\,$.
           }
\end{figure}%
Thus motivated, we define a `\LLEEwitnessed\ coinductive proof' as
  a process graph $\mathcal{G}$ with `layered' \LEE\ (\LLEE) whose vertices are labeled by equations between star expressions.
  The left- and the right-hand sides of the equations in the vertices of $\mathcal{G}$ 
    have to form a solution vector of a specification corresponding to the process graph $\mathcal{G}$.
  That specification, however, needs to be satisfied only up to provability in $\milnersysmin$ from sound assumptions. 
Such coinductive derivations are typically circular,  
  like the one depicted in Figure~\ref{fig:ex:1:coindproof}
    of the semantically valid equation $(a + b)^* \cdot 0 \formeq (a \cdot (a + b) + b)^* \cdot 0\,$.
That example is intended to give a first impression of the concepts involved, 
  despite of the fact that some details can only be appreciated later,
    when this example will be revisited in Example~\ref{ex:1:coindproof}. We describe these concepts below.

The process graph $\mathcal{G}$ in Figure~\ref{fig:ex:1:coindproof}, 
  which is given together with a labeling $\Hat{\mathcal{G}}$ that is a `\LLEEwitness'~of~$\mathcal{G}$.
  The \darkcyan{colored} transitions with marking labels $\loopsteplab{n}$, for $n\in\natplus$,
    indicate the LLEE-structure of $\mathcal{G}$, see Section~\ref{LEE}. 
The graph $\mathcal{G}$ underlies the coinductive proof on the left (see Example~\ref{ex:1:coindproof} for a justification). 
$\mathcal{G}$ is a `\onechart' that~is, 
  a process graph with \onetransitions\ that represent empty steps.
Here and later we depict \onetransitions\ as dotted arrows.
For \onecharts, `\onebisimulation' is the adequate concept of bisimulation~(Definition~\ref{def:onebisim}).

We showed in \cite{grab:2021:TERMGRAPH2020-postproceedings,grab:2020:TERMGRAPH2020-postproceedings:arxiv} 
  that the process (chart) interpretation $\chartof{\astexp}$ of a star expression $\astexp$ is
  the image of a \onechart\ $\onechartof{\astexp}$ with \LLEE\ under a functional \onebisimulation.
In this example, $\mathcal{G} = \onechartof{h^* \cdot 0}$ maps by a functional \onebisimulation\
  to interpretations of both expressions in the conclusion. 
The correctness conditions for such coinductive proofs are formed by the requirement that the left-, and respectively, the right-hand sides of formal equations
form `\provablein{\milnersysmin} solutions' of the underlying process graph: an expression at a vertex~$\avert$ 
can be reconstructed, provably in~$\milnersysmin$, from the transitions to, and the expressions at, immediate successor vertices of $\avert$.
Crucially we establish in Section~\ref{coindmilnersys:2:milnersys},
by a generalization of arguments in \cite{grab:fokk:2020:lics,grab:fokk:2020:lics:arxiv} using $\RSPstar$,
  that every \LLEEwitnessed\ coinductive proof over $\milnersysmin$ can be transformed into a derivation~in~$\milnersys$~with~the~same~conclusion.    
  
This raises the question of whether the fixed-point rule \RSPstar\ of $\milnersys$ adds any derivational power to $\milnersysmin$
  that goes beyond those of \LLEEwitnessed\ coinductive proofs over $\milnersysmin$, and if so, how far precisely.
In Section~\ref{milnersys:2:coindmilnersys} we show that every instance of the fixed-point rule \RSPstar\
  can be mimicked by a \LLEEwitnessed\ coinductive proof over $\milnersysmin$ in which 
  also the premise of the rule may be used.%
  \begin{figure}[tb!]
  \begin{center}
    \AxiomC{$ \overbrace{(a + b)^*}^{\chocolate{{e _0 ^*}}} \chocolate{\mathrel{\cdot} 0} 
                \:\formeq\:
              \overbrace{(a \cdot (a + b) + b)}^{\alert{f}} \cdot (\overbrace{(a + b)^*}^{\chocolate{{e _0 ^*}}} \chocolate{{} \cdot 0}) + \forestgreen{0} $}
    \RightLabel{$\ainst$, $\RSPstar$}
    \UnaryInfC{$ (a + b)^* \cdot 0
                   \:\formeq\:
                 (a \cdot (a + b) + b)^* \,\cdot\, \forestgreen{0} $}
    \DisplayProof
    %
    %
    \AxiomC{$ \chocolate{{e _0 ^*} \cdot 0} \:\formeq\: \alert{f} \cdot (\chocolate{{e _0 ^*} \cdot 0}) + \forestgreen{0} \rule{0pt}{23.5pt}$}
    \UnaryInfC{$ \chocolate{{e _0 ^*} \cdot 0} \:\formeq\: \alert{f^*} \,\cdot\, \forestgreen{0}$}
    \DisplayProof 
  \end{center}
  \vspace*{0ex}
  \begin{center}  
    \begin{tikzpicture}
      \matrix[anchor=center,row sep=1.25cm,column sep=3.5cm,
              every node/.style={draw=none}
              ] at (0,0) {
          \node(v1){}; &[1.5cm]             &[-1.5cm] \node(v2){};
          \\
                       & \node(v){}; 
          \\
        };
        
      \path (v1) node(v1-label){$ (1 \cdot (a + b)) \cdot (\chocolate{{e _0 ^*} \cdot 0}) 
                                    \formeq
                                  ((1 \cdot (a + b)) \cdot \alert{f^*})  \cdot \forestgreen{0} $};       
      \draw[->,shorten <=3.75cm,shorten >=1.95cm] (v1) to node[above,pos=0.6]{$a, b$} (v2);

      \path (v2) node(v2-label){$ 1 \cdot (\chocolate{{e _0 ^*} \cdot 0})
                                    \formeq
                                   (1 \cdot \alert{f^*}) \cdot \forestgreen{0} $}; 
      \draw[->,thick,densely dotted,out=-45,in=0,distance=0.75cm,shorten <= 0.2cm] (v2) to node[right]{$\sone$} ($(v) + (2.55cm,0.125cm)$);

      \path (v) node[xshift=-3.5cm,yshift=-0.15cm](v-label)
                             {$ \underbrace{\chocolate{{e _0 ^*} \cdot 0}}
                                           _{\text{(by the premise of $\ainst$)} \;\;
                                               (a \cdot (a + b) + b) \cdot (\chocolate{{e _0 ^*} \cdot 0}) + \forestgreen{0}
                                                 \:\synteq\:
                                               \alert{f} \cdot (\chocolate{{e _0 ^*} \cdot 0}) + \forestgreen{0} 
                                                 \: = \:
                                             \phantom{
                                                 \: = \:
                                               \alert{f} \cdot (\chocolate{{e _0 ^*} \cdot 0}) + \forestgreen{0}   
                                                 \:\synteq\: 
                                               (a \cdot (a + b) + b) \cdot (\chocolate{{e _0 ^*} \cdot 0}) + \forestgreen{0} 
                                               \text{(by the premise of $\ainst$)} }}
                                    \hspace*{-45ex}
                                  \formeq
                                \underbrace{(a \cdot (a + b) + b)^*}_{\alert{f^*}} \,\cdot\, \forestgreen{0} $}; 
      \draw[->,thick,darkcyan,shorten <= 1.5cm,shorten >= 1cm] (v) to node[above,pos=0.45]{$\loopnsteplab{1}$}
                                                                        node[below,pos=0.525]{$\black{a}$} (v1);
      \draw[->,thick,darkcyan,shorten <= 0.6cm,shorten >= 0.5cm] (v) to node[above,pos=0.45]{$\loopnsteplab{1}$}
                                                                        node[below,pos=0.575]{$\black{b}$} (v2); 
                                                                        
      \path (v) ++ (-5.5cm,0.5cm) node{\Large $ {\onecharthatof{\rule{0pt}{8pt}\alert{f^*}\cdot \forestgreen{0}}} $};                                                                     
    \end{tikzpicture}
  \end{center}
  \vspace*{-1.5ex}
\caption{\protect\label{fig:ex:1:RSPstar:to:coindproof}%
         Mimicking an instance $\protect\ainst$ of the fixed-point rule \protect\RSPstar\ (above) in Milner's system $\protect\milnersys = \protect\thplus{\protect\milnersysmin}{\protect\RSPstar}$
           by a coinductive proof (below) over $\protect\thplus{\protect\milnersysmin}{\protect\setexp{\protect\text{premise of $\protect\ainst$}}}$ with \protect\LLEEwitness~$\protect\onecharthatof{\protect\alert{f^*}\protect\cdot \protect\forestgreen{0}}$.} 
\end{figure}%
It follows that the derivational power that \RSPstar\ adds to $\milnersysmin$ within $\milnersys$ consists of iterating such \LLEEwitnessed\ coinductive proofs
along finite (meta-)prooftrees.    
The example in Figure~\ref{fig:ex:1:RSPstar:to:coindproof}
  is intended to give a first idea 
  of the construction that we will use (in the proof of Lemma~\ref{lem:lem:mimic:RSPstar})
  to mimic instances of \RSPstar.
  Here this construction results in a coinductive proof that only differs slightly from the one with the same underlying \LLEEonechart\ we saw earlier.
We will revisit this example in Example~\ref{ex:fig:ex:1:RSPstar:to:coindproof}. 

Based on the two transformations from coinductive proofs to derivations in $\milnersys$, and of applications of the fixed-point rule to coinductive proofs,  
 we reformulate Milner's system $\milnersys$
   as a \theoremequivalent\ proof system $\coindmilnersys$.  
For this, we replace the fixed-point rule $\RSPstar$ in $\milnersys$ 
  with a rule that permits to infer an equation $\astexp \formeq \bstexp$ from a finite set $\aseteqs$ of equations
    if there is a \LLEEwitnessed\ coinductive proof over $\milnersysmin$ plus the equations in $\aseteqs$ 
      that has conclusion $\astexp \formeq \bstexp$.  
We also define a \theoremequivalent\ system \CLC\ (`combining \LLEEwitnessed\ coinductive provability') with the equational coinductive proof rule alone. 
In the formalization of these systems
  we depart from the the exposition in \cite{grab:2021:calco:arxiv,grab:2021:calco}.
There, we used a hybrid concept of formulas that included entire coinductive proofs,
  which then could be used as specific rule premises.
Here, the proof systems are purely equational, and coinductive proofs occur only as \sideconditions\ of rules that formalize coinductive~provability. 
  
Additionally, we formulate proof systems $\coindmilnersysbar$ and \CC\
  that arise from $\coindmilnersys$ and \CLC\ by dropping `\LLEE-witnessed' as a requirement for coinductive proofs.
These systems are (obviously) complete for bisimilarity of process interpretations,
  because they can mimic the unique solvability rule scheme for guarded systems of specifications mentioned before.

\paragraph{Inspiration for cyclic proofs from related and previous work.}  
  Apart from their origin from a question in process theory,
    the results described here were inspired by co\-in\-duc\-tively motivated proof systems with derivations of cyclic form, 
      and by our previous work on their proof-theoretic links to traditional equational proof systems.
  This is a brief account of those direct influences, without drawing wider connections to work on cyclic proofs.
  
  About proofs by cyclic arguments that express bisimulations we learned
    from Rutten and Jacobs \cite{jaco:rutt:1997,rutt:1998},
      and about formalized cyclic derivations
        via the coinductively motivated proof systems by Brandt and Henglein~\cite{bran:heng:1998}:
          \BHeq\ for un\-win\-ding-equivalence, and \BHle\ for the subtyping relation between recursive types in $\mu$\nb-term notation. 
  Derivations in \BHeq\ roughly represent bisimulations up to transitivity and symmetry. 
  Via a connection of \BHeq\ to a tableaux-like system \AKeq\ by Ariola and Klop \cite{ario:klop:1996} with cyclic deductions
    we later recognized that a tableaux system~\HS\ with loop-detecting deductions of cyclic form similar to \AKeq\
      had already been used earlier by H\"{u}ttel and Stirling
        in \cite{huet:stir:1991,huet:stir:1998} 
          to show that bisimilarity of normed context-free processes is decidable.
  In \cite{grab:2005:calco} we developed a simple coinductively motivated proof system \Geq\ for language equivalence $\slangsemeq$ of regular expressions.
    That system was later refined substantially (using more flexible rules, similar to \BHeq), and generalized (similar to as \BHle\ generalized \BHeq) 
      to one for language containment of regular expressions by Henglein and Nielsen in \cite{heng:niel:2011}.  
      
  However, all of these proof systems use derivations in the form of proof-trees.
    Thus they permit cyclic derivations only of `palm-tree' form (ordered trees with \backlinks, \cite{tarj:1972}).
  In contrast, 
    we will permit cyclic derivations to have the form of general transition graphs.
      In this manner we `free' proof-graphs from the requirement to only exhibit `vertical sharing' \cite{blom:2001},
        and move close to informal reasoning as used in coalgebra like in \cite{jaco:rutt:1997,rutt:1998}.
  
  The transformations that we construct in Section~\ref{coindmilnersys:2:milnersys} and Section~\ref{milnersys:2:coindmilnersys}
    have been inspired by the  proof-theoretic interpretations that we developed in \cite{grab:2005:prfschr} between
  the proof system \ACeq\ for un\-win\-ding-equality of recursive types 
    by Amadio and Cardelli~\cite{amad:card:1993}
      (a Hilbert-style proof system with a fixed-point rule analogous to \RSPstar)
      and 
    the system \BHeq\ by Brandt and Henglein.
  The transformation from $\coindmilnersys$ to $\milnersys$ in Section~\ref{coindmilnersys:2:milnersys} is also similar to
    one we described in \cite{grab:2005:calco} that transforms
      derivations in \Geq\ 
      into derivations in Salomaa's system \Fone\ \cite{salo:1966} for $\slangsemeq$ 
                                                                   (where \Fone\ contains a fixed-point rule just like $\RSPstar$).

\paragraph{Relation with the conference article.}
  This article provides significantly more details and explanations than the article \cite{grab:2021:calco} in the proceedings of CALCO~2021.
  Furthermore it contains the following additions of content:
  \begin{itemize}[label=$\triangleright$,itemsep=0.25ex]
    \item
      Detailed proofs for the proof transformations from $\coindmilnersys$ to $\milnersys$ (in Section~\ref{coindmilnersys:2:milnersys}),
                                                and from $\milnersys$ to $\coindmilnersys$ (in Section~\ref{milnersys:2:coindmilnersys}). 
    
    \item
      Proof-theoretic explanation of the transformations 
        as the elimination of rules that are derivable or admissible 
          (based on Definition~\ref{def:derivable:admissible:rules}, Lemma~\ref{lem:derivable:admissible:rules} in Section~\ref{prelims}).
      \begin{itemize}
        \item
          For the proof transformation from the coinductive reformulation 
                                            $\coindmilnersys$ to Milner's system $\milnersys$
            we show that circular coinductive proofs over the purely equational part $\milnersysmin$ of $\milnersys$
              are admissible in $\milnersys$ (see Lemma~\ref{lem:LCoindProof:admissible:milnersys}).
        \item   
          For the proof transformation from Milner's system $\milnersys$ to its coinductive reformulation $\coindmilnersys$ 
            we show that the fixed-point rule \RSPstar\ of $\milnersys$ is derivable in $\coindmilnersys$
            (see Lemma~\ref{lem:RSPstar:derivable:coindmilnersysone}).
      \end{itemize}
      
  \item
      A statement that illustrates that the transformation 
                                                           from \coindmilnersys\ to \milnersys\
        can provide inroads for a completeness proof of Milner's system \milnersys\ (see Corollary~\ref{cor:milnersys:compl:expansion:minimization}).
      
  \item
    A diagram that gives an overview of all developed proof transformations
      (see Figure~\ref{fig:transformations}).
      
  \item
    An example that provides a sanity-check on the proof transformation from $\coindmilnersys$~to~$\milnersys$.
    It demonstrates that this transformation cannot work for mimicking instances of the fixed-point rule 
        without guardedness \sidecondition\ (see Non-Example~\ref{nonex:RSPstar:2:LLEEcoindproof}).
        
  \item
    An illustration of the difference between \LEEwitnesses\ (witnesses of the loop existence and elimination condition \LEE), and witnesses of `layered \LEE' (\LLEEwitnesses) 
      by an example that is based on different runs of the loop elimination procedure (see~Figure~\ref{fig:nonex:ex:LLEEw}). 
      
  \item
    A section in which we informally link the concept of `provable solution' (Definition~\ref{def:provable:solution})
       that is the basis for our concept of `coinductive proof' 
         to a coalgebraic formulation of this concept by Schmid, Rot, and Silva in \cite{schm:rot:silv:2021}
           (see Section~\ref{solutions}).
\end{itemize}

\paragraph{Relation with the completeness proof of Milner's system in \cite{grab:2022:lics}.}
  The com\-plete\-ness proof of Milner's system \milnersys\ summarized in \cite{grab:2022:lics} with report \cite{grab:2022:lics:arxiv}
    was finished and written only after the article \cite{grab:2021:calco} for CALCO~2021. 
  Indeed, we found the results in Section~6 of \cite{grab:2021:calco} and here in Section~\ref{milnersys:2:coindmilnersys}
      (that instances of the fixed-point rule \RSPstar\ can be mimicked by \LLEEwitnessed\ coinductive proofs) 
    in an effort to prepare for that completeness proof. 
  In particular, we wanted to be able to argue for the expedience of the use of \LLEEonecharts\ (see Definition~\ref{def:LLEEwitness})
    despite of the fact that reasoning with \LLEEonecharts\ towards a completeness proof of \milnersys\
      encounters a crucial obstacle%
           \footnote{Namely the fact that \LLEEonecharts\ are not closed under `\onebisimulation\ collapse',
                     an observation that is central for the crystallization procedure sketched in \cite{grab:2022:lics}.}.
  Without any argumentation that links derivations in Milner's system closely to \LLEEonecharts,
    it could be conceivable that this obstacle does not have any wider significance.
      Namely, it could be entirely specific to the use of \LLEEonecharts,
        while a completeness proof might possibly be based on quite different concepts.   
  The situation changed, however, after we realized
    that instances of the fixed-point rule can always be modeled (see Lemma~\ref{lem:RSPstar:derivable:coindmilnersysone})  
      by cyclic proofs of the shape of guarded \LLEEonecharts\ (see Definition~\ref{def:coindproof}),
        and proofs in Milner's system can be transformed (see Theorem~\ref{thm:main}) into meta-prooftrees of such cyclic proofs 
          (derivations in the system \CLC, see Definition~\ref{def:CLC:CC}).
  On the basis of these results we could argue that in principle every completeness proof of Milner's system \milnersys\ can be routed through (see Section~\ref{conclusion}) 
    arguments in which \LLEEonecharts\ appear front and central.
           
  The completeness proof of \milnersys\ in \cite{grab:2022:lics,grab:2022:lics:arxiv}
    uses additional observations and concepts (above all, a `crystallization procedure' of \LLEEonecharts\ for minimization under \onebisimilarity),
      and is not formulated in terms of the cyclic proof systems that we introduce here.
  However, the results of Section~\ref{coindmilnersys:2:milnersys}, 
    the transformation of \LLEEwitnessed\ coinductive proofs into derivations in \milnersys\ (see Proposition~\ref{prop:LLEEcoindproofeq:impl:milnersyseq})
      are of central importance for formulating the completeness proof in \cite{grab:2022:lics,grab:2022:lics:arxiv}.
    Indeed, they prove the lemmas {\bf (E)} (\underline{e}xtraction of provable solutions from guarded \LLEEonecharts)
      and {\bf (SE)} (provable \underline{s}olution \underline{e}quality in guarded \LLEEonecharts)
        of the completeness proof as listed in Section~5 of \cite{grab:2022:lics,grab:2022:lics:arxiv}.
        
  Conversely, we here use another one of the lemmas in Section~5 of \cite{grab:2022:lics,grab:2022:lics:arxiv},
    the lemma {\bf (T)} (transfer of provable solutions conversely along functional \onebisimulations),
      for illustrating the results of Section~\ref{coindmilnersys:2:milnersys}:
        we prove two (specific) completeness properties of \LLEEwitnessed\ coinductive proofs in relation to Milner's system
          (see Corollary~\ref{cor:milnersys:compl:expansion:minimization}).
          
  The completeness proof for Milner's system \milnersys\ with respect to process semantics equality of star expressions
    implies that the coinductive versions \coindmilnersys\ and \CLC\ of \milnersys\ that we introduce here are complete (in the same sense) as well.
  This is because our main result (see Theorem~\ref{thm:main}) 
    states that \coindmilnersys\ and \CLC\ have the same derivational power as \milnersys.

\tableofcontents
  %
\paragraph{Overview.}
We start in Section~\ref{prelims}
  with introducing basic definitions concerning the process semantics of regular expressions, and concepts that we will need.
  We define star expressions, finite process graphs with \onetransitions,
    ($\sone$-)bisimulations and \onebisimilarity,
      and the process semantics of star expressions. 
  Then we introduce \equationallogic\ based, and equation-based proof systems,
    with Milner's system \milnersys\ and two variants as first examples. 
  Also, we define when inference rules are derivable or admissible in such a proof system,
    and establish easy interconnections.
  Finally we define the concept of solution for \onecharts\ with respect to an equational proof system.
In Section~\ref{solutions}
  we link to an insightful coalgebraic characterization of provable solutions of \onecharts\ that is due to Schmid, Rot, and Silva in \cite{schm:rot:silv:2021}.
  We reformulate it in our terminology, but do not prove it in detail, as our development
    does not depend~on~it.    
In Section~\ref{LEE}
  we explain concepts and definitions concerning the (layered) loop existence and elimination property (L)LEE 
    from \cite{grab:fokk:2020:lics:arxiv,grab:fokk:2020:lics},
      and recall the `\onechart\ interpretation' of star expressions from \cite{grab:2020:TERMGRAPH2020-postproceedings:arxiv,grab:2021:TERMGRAPH2020-postproceedings},
        which guarantees \LLEE.
        
In Section~\ref{coindmilnersys}
  we introduce `coinductive proofs' over equational proof systems. 
  We formulate proof systems \CC\ and \CLC\ with appropriate rule schemes 
    that permit to use and combine coinductive proofs, and respectively, \LLEEwitnessed\ coinductive proofs.
  Then we introduce the coinductive reformulation \coindmilnersys\ of \milnersys\
    as an extension of the equational part \milnersysmin\ of \milnersys. 
  We also establish basic \prooftheoretic\ connections between these new systems.

In Section~\ref{coindmilnersys:2:milnersys}
  we show that coinductive proofs over proof systems with derivational power not greater than Milner's system \milnersys\
    can be transformed into derivations in \milnersys.
  We use this fact to obtain a proof transformation from \coindmilnersys\ to \milnersys.  
  
In Section~\ref{milnersys:2:coindmilnersys}
  we demonstrate that every instance $\ainst$ of the fixed-point rule \RSPstar\ of \milnersys\
    can be mimicked by a coinductive proof of the conclusion of $\ainst$
      where (correctness conditions of) that proof may use the equational part \milnersysmin\ of \milnersys\ plus the premise equation of $\ainst$.
  We apply this central observation for defining a proof transformation from \milnersys\ to \coindmilnersys.   
  With this transformation and the one constructed in Section~\ref{coindmilnersys:2:milnersys}
    we prove that the proof systems \CLC\ and \coindmilnersys\ are \theoremequivalent\ with Milner's system \milnersys. 
       
In the final section, Section~\ref{conclusion},
  we recapitulate our motivation for introducing coinductive circular proofs,
    and summarize our results.
We argue that the coinductive proof systems \CLC\ and \coindmilnersys\
  can be viewed as being located roughly half-way 
    in between 
      derivations in \milnersys\
    and   
    bisimulations between process interpretations of star expressions.
We conclude with initial ideas about a proof strategy for a completeness proof of \CLC\ and \coindmilnersys,
  which would yield a completeness of \milnersys. 

\section{Process semantics for star expressions, and Milner's proof system}%
  \label{prelims}

Here we fix terminology concerning star expressions, \onecharts, \onebisimulations;
  we exhibit Milner's system (and a few variants), and recall the chart interpretation of star expressions.

\smallskip

\begin{defi}[star expressions]
  Let $\actions$ be a set 
                          of \emph{actions}.
  The set $\StExpover{\actions}$ of \emph{star expressions over actions in $\actions$} 
  are strings that are defined by the following grammar: 
  \begin{equation*}
    \astexp, \astexpi{1}, \astexpi{2}
      \:\BNFdefdby\:
    \stexpzero
      \BNFor
    \stexpone
      \BNFor  
    \aact
      \BNFor
    ( \stexpsum{\astexpi{1}}{\astexpi{2}} )
      \BNFor
    ( \stexpprod{\astexpi{1}}{\astexpi{2}} )
      \BNFor
      \stexpit{\astexp}  \qquad\quad \text{(where $\aact\in\actions$)}
  \end{equation*} 
  We will drop outermost brackets, and those that are expendable according to the precedence of star~$\stexpit{}$ over composition~$\sstexpprod$ and choice~$\sstexpsum$,
    and of composition~$\sstexpprod$ over choice~$\sstexpsum$.
  We use $\astexp,\bstexp,\cstexp,\dstexp$, possibly indexed and/or decorated, as identifiers (for reasoning on the meta-level 
                                                                                               like with `syntactical variables' \cite{shoe:1967})
                                                                                  for star expressions.
  We write $\ssynteq$ for syntactic equality between star expressions denoted by such identifiers,
    and values of star expression functions, in a given context,
  but we permit $\sformeq$ in formal equations between star expressions.
  We denote by $\StExpEqover{\actions}$ the set of formal equations $\astexp \formeq \bstexp$ 
  between two star expressions $\astexp,\bstexp\in\StExpover{\actions}$. 
  
  We define sum expressions $\sum_{i=1}^{n} \astexpi{i}$ inductively
    as $\stexpzero$ if $n=0$,
    as $\astexpi{1}$ if $n=1$,
    and as $\stexpsum{(\sum_{i=1}^{n-1} \astexpi{i})}{\astexpi{n}}$ if $n>0$,
    for $n\in\nat = \setexp{0,1,2,\ldots}$.
  The \emph{(syntactic) star height} $\sth{\astexp}$ of a star expression $\astexp\in\StExpover{\actions}$
  is the maximal nesting depth of stars in $\astexp$, defined inductively 
  by: $\sth{\stexpzero} \defdby \sth{\stexpone}  \defdby \sth{\aact} \defdby 0$, 
       $\sth{\stexpsum{\astexpi{1}}{\astexpi{2}}} \defdby \sth{\stexpprod{\astexpi{1}}{\astexpi{2}}}
                                                  \defdby \max\setexp{\sth{\astexpi{1}}, \sth{\astexpi{2}}}$, 
       and $\sth{\stexpit{\astexp}} \defdby 1 + \sth{\astexp}$.
\end{defi}

\smallskip

\begin{defi}[\protect\onecharts, and charts]
  A \emph{\onechart}
  is a 6\nb-tuple $\tuple{\verts,\actions,\sone,\start,\transs,\sterminates}$ 
    where $\verts$ is a \ul{finite} set of \emph{vertices},
    $\actions$ is a set of \emph{(proper)} $\emph{action labels}$,
    $\sone\notin\actions$ is the specified \emph{empty step label},
    $\start\in\verts$ is the \emph{start vertex} (hence $\verts \neq \emptyset$),
    $\transs \subseteq \verts\times\oneactions\times\verts$ is the \emph{labeled transition relation},
    where $\oneactions \defdby \actions \cup \setexp{\sone}$ is the set of action labels including $\sone$, 
    and $\sterminates \subseteq \verts$ is a set of \emph{vertices with immediate termination}. 
  In such a \onechart, 
  we call a transition in $\transs\cap(\verts\times\actions\times\verts)$ (labeled by a \emph{proper action} in $\actions$) 
          a \emph{proper transition},
  and a transition in $\transs\cap(\verts\times\setexp{\sone}\times\verts)$ (labeled by the \emph{empty-step symbol}~$\sone$)
          a \emph{\onetransition}.  
  Reserving non-underlined action labels like $\aact,\bact,\ldots$ for proper actions,
  we use underlined action label symbols like $\aoneact$ for actions labels in the set $\oneactions$;
  in doing so we highlight also in \firebrick{firebrick} transition labels that may~involve~$\sone$. 
  
  We say that a \onechart\ is \emph{weakly guarded} if it does not contain cycles of \onetransitions.
  
  By a \emph{chart} we mean a \onechart\ $\aonechart$ 
    that is \emph{\onetransition\ free} in the sense that all of its transitions are proper. 
  We will use the symbols $\aonechart$ and $\achart$
    (also with subscripts) as identifiers 
                              for \onecharts, and charts, respectively.
  We use the notations $\vertsof{\aonechart}$, and $\vertsof{\achart}$
    for quick reference to the set of vertices of a \onechart~$\aonechart$, and of a chart~$\achart$.
\end{defi}

\smallskip

Below we define the process semantics of star (regular) expressions as (\sonefree) charts, 
  and hence as finite, rooted labeled transition systems,
  which will be compared with ($\sone$-)bi\-si\-mi\-la\-rity. 
The charts that will be obtained in this way correspond to non-deterministic finite-state automata that are defined
  by iterating partial~derivatives \cite{anti:1996} (1996) of Antimirov (who did not aim at a process semantics). 
Indeed, Antimirov's result that every regular expression only has finitely many iterated partial derivatives (Corollary~3.5 in \cite{anti:1996})
  guarantees finiteness of chart interpretations as defined below.
We will use the notation $\chartof{\astexp}$ with as meaning `the chart induced by (the process interpretation of) the star expression~$\astexp$'.

\begin{defi}
    \label{def:StExpTSS}\label{def:chartof}
  The \emph{chart interpretation of} a star expression~$\astexp\in\StExpover{\actions}$
  is the $\sone$\nb-tran\-si\-tion free chart
    $\chartof{\astexp} = \tuple{\vertsof{\astexp},\actions,\sone,\astexp,\transs\cap(\vertsof{\astexp}{\times}{\actions}{\times}{\vertsof{\astexp}}),\termexts\cap\vertsof{\astexp}}$,
  where $\vertsof{\astexp}$ consists of all star expressions that are reachable from $\astexp$
  via the labeled transition relation $ \transs \subseteq \StExpover{\actions}\times\actions\times\StExpover{\actions}$
  that is defined, together with the \immediatetermination\ relation $\sterminates \subseteq \StExpover{\actions}$,
  via derivability in the transition system specification (TSS)~$\StExpTSSover{\actions}$,
  for $\aact\in\actions$, $\astexp,\astexpi{1},\astexpi{2},\astexpacc,\astexpacci{1},\astexpacci{2}\in\StExpover{\actions}$:
  \begin{center}
    $
    \begin{aligned}
      &
      \AxiomC{\phantom{$\terminates{\stexpone}$}}
      \UnaryInfC{$\terminates{\stexpone}$}
      \DisplayProof
      & \hspace*{0.5ex} &
      \AxiomC{$ \terminates{\astexpi{1}} $}
      \UnaryInfC{$ \terminates{(\stexpsum{\astexpi{1}}{\astexpi{2}})} $}
      & \hspace*{4ex} &
      \AxiomC{$ \terminates{\astexpi{i}} $}
      \UnaryInfC{$ \terminates{(\stexpsum{\astexpi{1}}{\astexpi{2}})} $}
      \DisplayProof
      & \hspace*{4ex} &
      \AxiomC{$\terminates{\astexpi{1}}$}
      \AxiomC{$\terminates{\astexpi{2}}$}
      \BinaryInfC{$\terminates{(\stexpprod{\astexpi{1}}{\astexpi{2}})}$}
      \DisplayProof
      & \hspace*{4ex} &
      \AxiomC{$\phantom{\terminates{\stexpit{\astexp}}}$}
      \UnaryInfC{$\terminates{(\stexpit{\astexp})}$}
      \DisplayProof
    \end{aligned} 
    $
    \\[1ex]
    $
    \begin{aligned}
      &  
      \AxiomC{$\phantom{a \:\lt{a}\: \stexpone \rule{0pt}{12.4pt}}$}
      \UnaryInfC{$a \:\lt{a}\: \stexpone$}
      \DisplayProof
      & & \:
      \AxiomC{$ \astexpi{i} \:\lt{a}\: \astexpacci{i} $}
      \UnaryInfC{$ \stexpsum{\astexpi{1}}{\astexpi{2}} \:\lt{a}\: \astexpacci{i} $}
      \DisplayProof 
      & & \:
      \AxiomC{$ \astexpi{1} \:\lt{a}\: \astexpacci{1} $}
      \UnaryInfC{$ \stexpprod{\astexpi{1}}{\astexpi{2}} \:\lt{a}\: \stexpprod{\astexpacci{1}}{\astexpi{2}} $}
      \DisplayProof
      & &
      \AxiomC{$\terminates{\astexpi{1}}$}
      \AxiomC{$ \astexpi{2} \:\lt{a}\: \astexpacci{2} $}
      \BinaryInfC{$ \stexpprod{\astexpi{1}}{\astexpi{2}} \:\lt{a}\: \astexpacci{2} $}
      \DisplayProof
      & & \:
      \AxiomC{$\astexp \:\lt{a}\: \astexpacc \rule{0pt}{11.5pt}$}
      \UnaryInfC{$\stexpit{\astexp} \:\lt{a}\: \stexpprod{\astexpacc}{\stexpit{\astexp}}$}
      \DisplayProof
    \end{aligned}
    $
  \end{center}
  If $\astexp \lt{\aact} \astexpacc$ is derivable in $\StExpTSSover{\actions}$, for $\astexp,\astexpacc\in\StExpover{\actions}$,
  $\aact\in\actions$, then we say that $\astexpacc$ is a \emph{derivative} of $\astexp$.
  If $\terminates{\astexp}$ is derivable in $\StExpTSSover{\actions}$, 
  then we say that $\astexp$ \emph{permits immediate termination}.
\end{defi}

In Section~\ref{LEE} we define a refinement of this interpretation from \cite{grab:2021:TERMGRAPH2020-postproceedings} into a \onechart\ interpretation.
In both versions, ($\sone$-)charts obtained will be compared with respect to \onebisimilarity\
that relates the behavior of `induced transitions' of \onecharts.
By an \emph{induced \transitionact{\aact}} $\avert \ilt{\aact} \bvert$, for a proper action $\aact\in\actions$, in a \onechart~$\aonechart$ we mean 
  a path $\avert \lt{\sone} \cdots \lt{\sone} \cdot \lt{\aact} \bvert$ in $\aonechart$
  that consists of a finite number of \onetransitions\ that ends with a proper \transitionact{\aact}.
By \emph{induced termination} $\oneterminates{\avert}$, for $\avert\in\verts$ we mean that there is a path
  $\avert \lt{\sone} \cdots \lt{\sone} \averttilde$ with $\terminates{\averttilde}$~in~$\aonechart$.      

\begin{defi}[($\protect\sone$-)bisimulation]\label{def:onebisim}
  Let $\aonecharti{i} = \tuple{\vertsi{i},\actions,\sone,\starti{i},\sredi{i},\termextsi{i}}$
      be \onecharts, for $i\in\setexp{1,2}$. 
  
  By a \emph{\onebisimulation\ between $\aonecharti{1}$ and $\aonecharti{2}$}
  we mean a binary relation $\abisim \subseteq \vertsi{1}\times\vertsi{2}$ 
  such that $\pair{\starti{1}}{\starti{2}}\in\abisim$ holds (that is, $\abisim$ relates the start vertices of $\acharti{1}$ and $\acharti{2}$),
  and for every $\pair{\averti{1}}{\averti{2}}\in\abisim$ the following three conditions hold:
  \begin{description}
    \item[(forth)] \quad
      $ \forall \avertacci{1}\in\vertsi{1}
          \forall \aact\in\actions
              \bigl(\,
                \averti{1} \ilti{\aact}{1} \avertacci{1}
                  \;\;\Longrightarrow\;\;
                    \exists \avertacci{2}\in\vertsi{2}
                      \bigl(\, \averti{2} \ilti{\aact}{2} \avertacci{2} 
                                 \logand
                               \pair{\avertacci{1}}{\avertacci{2}}\in\abisim \,)
            \,\bigr) \punc{,} $
      
    \item[(back)] \quad
      $ \forall \avertacci{2}\in\vertsi{2}
          \forall \aact\in\actions
            \bigr(\,
              \exists \avertacci{1}\in\vertsi{1}
                \bigl(\, \averti{1} \ilti{\aact}{1} \avertacci{1} 
                           \logand
                         \pair{\avertacci{1}}{\avertacci{2}}\in\abisim \,)
                  \;\;\Longleftarrow\;\;
                \averti{2} \ilti{\aact}{2} \avertacci{2}
              \,\bigr) \punc{,} $
      
    \item[(termination)] \quad
      $ \oneterminatesi{1}{\averti{1}}
          \;\;\Longleftrightarrow\;\;
            \oneterminatesi{2}{\averti{2}} \punc{.}$
  \end{description}
  We write $\aonecharti{1} \onebisimvia{\abisim} \aonecharti{2}$ 
    if $\abisim$ is a \onebisimulation\ between $\aonecharti{1}$ and $\aonecharti{2}$.
  We denote by $\aonecharti{1} \onebisim \aonecharti{2}$ 
    and say that $\aonecharti{1}$ and $\aonecharti{2}$ are \emph{\onebisimilar},
      if there is a \onebisimulation\ between  $\aonecharti{1}$ and $\aonecharti{2}$. 
      
  By a \emph{functional} \onebisimulation\ \emph{from $\aonecharti{1}$ to $\aonecharti{2}$} 
    we mean a \onebisimulation\ $\abisim$ between $\aonecharti{1}$ and $\aonecharti{2}$ 
      that is defined by a function $\sphifun \funin \vertsi{1} \to \vertsi{2}$
      as its graph, that is, by $\abisim = \graphof{\sphifun} = \descsetexp{\pair{\avert}{\phifun{\avert}}}{\avert\in\verts}$;
  in this case we write $\aonecharti{1} \funonebisimvia{\sphifun} \aonecharti{2}$. 
  We write $\aonecharti{1} \funonebisim \aonecharti{2}$
    if there is a functional \onebisimulation\ from $\aonecharti{1}$ to $\aonecharti{2}$.
        
  We note that for \onetransitionfree\ \onecharts\
    the bisimulation conditions specialize to their\vspace*{-1.5pt} usual form:
    the induced transitions $\silt{\cdot}$ in 
                                              (forth) and (back) specialize\vspace*{-1.5pt} 
      to proper transitions $\slt{\cdot}$, and induced termination $\soneterminates$ in 
                                                                                        (termination) specializes to immediate termination~$\sterminates$.
  Let $\acharti{1}$ and $\acharti{2}$ be charts (\onetransitionfree\ \onecharts).
    We write $\acharti{1} \bisimvia{\abisim} \acharti{2}$, and say that $\abisim$ is a \emph{bisimulation between} $\acharti{1}$ and $\acharti{2}$ 
      if $\abisim$ is a \onebisimulation\ between $\acharti{1}$ and $\acharti{2}$. 
  We write $\acharti{1} \bisim \acharti{2}$,
    and say that $\acharti{1}$ and $\acharti{2}$ are \emph{bisimilar}
      if there is a bisimulation\ between  $\acharti{1}$ and $\acharti{2}$. 
  We write $\acharti{1} \funbisim \acharti{2}$
      if there is \onebisimulation\ from chart $\acharti{1}$~to~chart~$\acharti{2}$). 
        
  Let $\aonechart$ be a \onechart, and $\abisim \subseteq\vertsof{\aonechart}\times\vertsof{\aonechart}$.
    We say that $\abisim$ is a \onebisimulation\ \emph{on}~$\aonechart$ 
    if $\abisim$ is a \onebisimulation\ between $\aonechart$ and $\aonechart$.
  Let $\achart$ be a chart, and $\abisim \subseteq\vertsof{\achart}\times\vertsof{\achart}$.
    We say that $\abisim$ is a bisimulation \emph{on}~$\achart$
     if $\abisim$ is a bisimulation between $\achart$ and $\achart$.
\end{defi}

We now define `process semantics equality' of two star expressions 
  as bisimilarity of their chart interpretations.
We do not 
          introduce the process semantics of star expressions as `star behaviors'
  (bisimilarity equivalence classes of their chart interpretations) as Milner in \cite{miln:1984},
    but only the relation that two star expressions denote the same~star~behavior.

\begin{defi}[process semantics equality]
  We define \emph{process semantics equality} 
    as the binary relation $\sprocsemeq \subseteq \StExpover{\actions}\times\StExpover{\actions}$ 
      by stipulating it, for all $\astexp,\bstexp\in\StExpover{\actions}$, 
        as bisimilarity of the (\onefree) chart interpretations of $\astexp$ and $\bstexp\,$ :
  \begin{equation*}
    \astexp
      \procsemeq
    \bstexp
        \;\; \funin \: \Longleftrightarrow \;\;
           \chartof{\astexp}
             \bisim  
           \chartof{\bstexp} \punc{.}  
  \end{equation*}
\end{defi}

\begin{defi}[proof system $\protect\eqlogic$, $\protect\StExpEqover{\cdot}$-based/\protect\eqlogicbased\ proof systems]\label{def:eqlogic:eqlogicbased}
  Let $\actions$ be a set.
  
  By an \emph{\equationbasedover{\actions} proof system} we will mean a Hilbert-style proof system
    whose \emph{formulas} are the equations in $\StExpEqover{\actions}$ between star expressions over $\actions$. 
  For an \equationbasedover{\actions} proof system $\asys$ 
    and a set $\aseteqs \subseteq \StExpEqover{\actions}$
      we denote by $\thplus{\asys}{\aseteqs}$ 
        the \equationbasedover{\actions} proof system whose \emph{rules} are those of $\asys$,
          and whose \emph{axioms} are those of $\asys$ plus the equations in $\aseteqs$. 
  
  The \emph{basic proof system~$\eqlogicover{\actions}$ of equational logic}
    for star expressions over $\actions$ is an \equationbasedover{\actions} proof system 
  that has the following \emph{rules}:  
  %
  %
  \begin{gather*}\renewcommand{\fCenter}{\formeq}
  \begin{alignedat}{3}
    &
  \begin{aligned}[c]
    \AxiomC{$\phantom{\astexp   \fCenter   \astexp\rule{0pt}{7.5pt}}$}
    \RightLabel{\REFL}
    \UnaryInfC{$\astexp   \fCenter   \astexp\rule{0pt}{7.5pt}$} 
    \DisplayProof
  \end{aligned}  
    & \qquad & 
  \begin{aligned}[c]
    \Axiom$\astexp   \fCenter   \bstexp$
    \RightLabel{\SYMM}
    \UnaryInf$\bstexp   \fCenter   \astexp$
    \DisplayProof
  \end{aligned}  
    & \qquad & 
  \begin{aligned}[c]
    \AxiomC{$\astexp \formeq \bstexp$}
    \AxiomC{$\bstexp \formeq \cstexp$}
    \RightLabel{\TRANS}
    \BinaryInfC{$\astexp \formeq \cstexp\rule{0pt}{7.5pt} $}
    \DisplayProof
  \end{aligned}  
    & \qquad
  \begin{aligned}
    \Axiom$\astexp   \fCenter   \bstexp$
    \RightLabel{\CXT}
    \UnaryInf$\acxtap{\astexp}   \fCenter   \acxtap{\bstexp}$
    \DisplayProof
  \end{aligned}
  \end{alignedat}
  \end{gather*}%
  that is, the rules~$\REFL$ (for reflexivity), and the rules~$\SYMM$ (for symmetry), $\TRANS$ (for transitivity), and $\CXT$ (for filling a context),
  where $\acxtap{\iets}$ is a 1-hole star expression context. 

  By an \emph{\eqlogicbasedover{\actions} system} 
  we mean an \equationbasedover{\actions} proof system  
  whose rules include the rules of the basic system $\eqlogicover{\actions}$ of equational logic
  (additionally, it may specify an arbitrary set of axioms). 
  We will use the letter $\asys$ as identifier
                         for \eqlogicbased\ proof systems. 
\end{defi}
  
\begin{defi}  
  Let $\asys$ be an \equationbasedover{\actions} proof system. Let $\astexp,\bstexp\in\StExpover{\actions}$.
  We say that $\astexp \formeq \bstexp$ is \emph{derivable in $\asys$}, 
    which we denote here by $\astexpi{1} \eqin{\asys} \astexpi{2}$ (instead of the more commonly used notation $\derivablein{\asys} \astexpi{1} \formeq \astexpi{2}$),
      if there is a derivation without assumptions in $\asys$ that has conclusion $\astexp \formeq \bstexp$.    
  If $\astexp \formeq \bstexp$ is derivable in $\asys$, we also say that $\astexp \formeq \bstexp$ is a \emph{theorem} of $\asys$.
\end{defi}

\begin{defi}[sub-system, theorem equivalence/subsumption of \protect\equationbasedover{\cdot} proof systems]\label{def:subsystem:isthmsubsumedby:thmequiv}
  Let $\asysi{1}$ and $\asysi{2}$ be \equationbasedover{\actions} proof systems. 
  
  We say that $\asysi{1}$ is a \emph{sub-system} of $\asysi{2}$, denoted by $\asysi{1} \subsystem \asysi{2}$,
    if every axiom of $\asysi{1}$ is an axiom of $\asysi{2}$, and every rule of $\asysi{1}$ is also a rule of $\asysi{2}$. 
  We say that \emph{$\asysi{1}$ is \theoremsubsumed\ by $\asysi{2}$}, denoted by $\asysi{1} \isthmsubsumedby \asysi{2}$,
    if every theorem of $\asysi{1}$ is also a theorem of $\asysi{2}$,
       that is, if $\astexp \eqin{\asysi{1}} \bstexp$ implies $\astexp \eqin{\asysi{2}} \bstexp$, for all $\astexp,\bstexp\in\StExpover{\actions}$.
  We say that \emph{$\asysi{1}$ and $\asysi{2}$ are \theoremequivalent}, denoted by $\asysi{1} \thmequiv \asysi{2}$,
    if $\asysi{1}$ and $\asysi{2}$ have the same theorems (that is, if $\asysi{1} \isthmsubsumedby \asysi{2}$, and~$\asysi{2} \isthmsubsumedby \asysi{1}$).
\end{defi}

For the definitions of the concept of `derivable', `correct', and `admissible' rule in Definition~\ref{def:derivable:admissible:rules} below
  for an \equationbasedover{\actions} proof system we introduce an informal concept of derivation rule that will suffice for our purpose.
For abstract formulations of rules, and for the concepts of derivability, correctness, and admissibility of rules we refer: 
\begin{enumerate}[label={(\roman{*})}]
  \item
   to \cite{grab:2005:prfschr},
    where these concepts have been gathered and formally treated for Hilbert-style proof systems
      (as well as for natural-deduction style proof systems),
  \item 
    and to \cite{grab:2004:CollegiumLogicum},
      where for `abstract Hilbert systems', systematic connections between these concepts of rules have been studied,
      also with respect to how rules can be eliminated from derivations. 
\end{enumerate}    

Let $\asys$ be an \equationbasedover{\actions} proof system for star expressions over $\actions$. 
  Let $n\in\nat$. By a(n) \emph{(\premisen{n}) rule $\arule$ for $\asys$} we mean an inference scheme
    all of whose instances are of the form:
  \begin{equation*}
    \AxiomC{$\astexpi{1} \formeq \bstexpi{1}$}
    \AxiomC{$\astexpi{n} \formeq \bstexpi{n}$}
    \insertBetweenHyps{$\;\;\ldots\;\;$}
    \RightLabel{$\arule$}
    \BinaryInfC{$\astexp \formeq \bstexp$}
    \DisplayProof
  \end{equation*}
with star expressions $\astexp,\bstexp,\astexpi{1}, \ldots, \astexpi{n}, \bstexpi{1}, \ldots, \bstexpi{n}\in\StExpover{\actions}$. 
For such a rule $\arule$ for $\asys$ we denote by $\thplus{\asys}{\arule}$ the \equationbasedover{\actions} proof system
  that extends $\asys$ by adding $\arule$ as an additional rule.

\begin{defi}[derivable, correct, and admissible rules]\label{def:derivable:admissible:rules}
  Let $\asys$ be an \equationbasedover{\actions} proof system. 
    Let $\arule$ be a rule for $\asys$.
  
  We say that $R$ is \emph{derivable in $\asys$}
    if every instance $\ainst$ of $R$ can be mimicked by a derivation $\aDerivi{\ainst}$ in $\asys$
      by which we mean that the set of assumptions of $\aDerivi{\ainst}$ is contained in the set of premises of $\ainst$,
        and the conclusion of $\aDerivi{\ainst}$ is the conclusion of $\ainst$. 
    
  We say that $\arule$ is \emph{correct for $\asys$}
    if instances of $\arule$ can be eliminated from derivations in $\thplus{\asys}{\arule}$
      in the following limited sense: 
    for every derivation $\aDeriv$ in $\thplus{\asys}{\arule}$ without assumptions that terminates with an instance of $\arule$
      but all of whose immediate subderivations are derivations in $\asys$  
        there is a derivation $\aDerivacc$ in $\asys$ without assumptions, and with the same conclusion as $\aDeriv$.   
    
  We say that $\arule$ is \emph{admissible in $\asys$}
    if $\thplus{\asys}{\arule} \thmequiv \asys$ holds,
      that is, the addition of $\arule$ to $\asys$ does not extend the derivable formulas (the theorems) of $\asys$.
\end{defi}

The definition of `$\arule$ is admissible in $\asys$'
  is easily understood to be equivalent with the statement that instances of $\arule$ can be eliminated from derivations in $\thplus{\asys}{\arule}$ without assumptions
      in the \underline{\smash{un}}limited sense:
        for every derivation $\aDeriv$ in $\thplus{\asys}{\arule}$ without assumptions
          there is a derivation $\aDerivacc$ in $\asys$ without assumptions, and with the same conclusion as $\aDeriv$.  
Therefore rule admissibility implies rule correctness. This justifies the implication ``$\Rightarrow$'' 
  in item~\ref{it:1:lem:derivable:admissible:rules} of the lemma below
    that gathers basic relationships between the three properties of rules with respect to a proof system as defined above.

\begin{lem}\label{lem:derivable:admissible:rules}
  Let $\arule$ be a rule for an \eqlogicbased\ proof system~$\asys$ for star expressions over $\actions$.
  Then the following statements link 
    derivability, correctness, and admissibility of $\arule$ in/for $\asys\,$:
  \begin{enumerate}[label={(\roman{*})}]
    \item{}\label{it:1:lem:derivable:admissible:rules}
      $\arule$ is admissible in $\asys$
        if and only if  
      $\arule$ is correct for $\asys$.
    \item{}\label{it:2:lem:derivable:admissible:rules}
      If $\arule$ is derivable in $\asys$, then $\arule$ is also correct for $\asys$,
        and due to \ref{it:1:lem:derivable:admissible:rules} also admissible~in~$\asys$.
      However, rule admissibility and correctness does not imply derivability in general.  
  \end{enumerate}
\end{lem}

\begin{proof}
  Concerning statement \ref{it:1:lem:derivable:admissible:rules} of the lemma
    we have already argued for the direction ``$\Rightarrow$'' just above.
  The direction ``$\Leftarrow$'' can be established by showing 
    that, if $\arule$ is correct for $\asys$,
      then every given derivation $\aDeriv$ in $\thplus{\asys}{\arule}$ 
      can be transformed into a derivation $\aDerivacc$ in $\asys$ with the same conclusion
        by eliminating instances of $\arule$ in top-down direction,
          using derivation replacements as guaranteed by the defining statement of `$\arule$ is correct for $\asys$'.
      
  For showing the main part of statement \ref{it:2:lem:derivable:admissible:rules}, we consider an \premisen{n} rule $\arule$ for $\asys$
    that is a derivable rule of $\asys$.
  In order to show that $\arule$ is correct for $\asys$, we have to show that every derivation $\aDeriv$ in $\thplus{\asys}{\arule}$
    that terminates with an instance $\ainst$ of $\arule$ but has immediate subderivations in~$\asys$
      can be transformed into a derivation $\aDerivacc$ in $\asys$ with the same conclusion. 
  Let $\aDeriv$ be such a derivation in $\thplus{\asys}{\arule}$ with instance $\ainst$ of $\arule$ at the bottom,
    as illustrated on the right below. 
  Since $\arule$ is derivable in $\asys$ there is a derivation $\aDerivi{\ainst}$ in $\asys$ that derives
    the conclusion of $\ainst$ from its $n$ premises. 
  Then $\aDeriv$ can be transformed according to the~following~step:       
  \begin{equation}\label{eq:prf:it:1:lem:derivable:admissible:rules}
    \aDeriv\left\{\quad
    \begin{gathered}
      \AxiomC{$\aDerivi{1}$}
      \noLine
      \UnaryInfC{$\astexpi{1} \formeq \bstexpi{1}$}
      \AxiomC{$\aDerivi{n}$}
      \noLine
      \UnaryInfC{$\astexpi{n} \formeq \bstexpi{n}$}
      \insertBetweenHyps{$\;\;\ldots\;\;$}
      \RightLabel{$\arule$}
      \LeftLabel{$\ainst$}
      \BinaryInfC{$\astexp \formeq \bstexp$}
      \DisplayProof
    \end{gathered}
    \quad\right.
      \quad\Longmapsto\quad
    \left.\quad  
    \begin{gathered}
      \AxiomC{$\aDerivi{1}$}
      \noLine
      \UnaryInfC{$[\astexpi{1} \formeq \bstexpi{1}]$}
      \AxiomC{$\aDerivi{n}$}
      \noLine
      \UnaryInfC{$[\astexpi{n} \formeq \bstexpi{n}]$}
      \insertBetweenHyps{$\;\;\ldots\;\;$}
      \noLine
      \BinaryInfC{$\aDerivi{\ainst}$}
      \noLine
      \UnaryInfC{$\astexp \formeq \bstexp$}
      \DisplayProof
    \end{gathered}
    \quad\right\} \aDerivacc
  \end{equation}
  (where $[\astexpi{1} \formeq \bstexpi{1}]$, \ldots $[\astexpi{n} \formeq \bstexpi{n}]$
    denote the assumption classes of $\astexpi{1} \formeq \bstexpi{n}$, \ldots, $\astexpi{1} \formeq \bstexpi{n}$
      in leafs at the top of the prooftree $\aDerivi{\ainst}$).
  The result of this step is a derivation $\aDerivacc$ in $\asys$ with the same conclusion as $\aDeriv$.
  Since $\aDeriv$ was chosen arbitrary in this statement but with a bottommost instance of $\arule$ and immediate subderivations in $\asys$,
    we have shown the desired transformation statement, which guarantees that $\arule$ is correct for $\asys$.  
    
  If a rule $\arule$ is correct and admissible in an \eqlogicbased\ proof system $\asys$,
    then $\arule$ does not need to be derivable. 
  This is because correctness of $\arule$ in $\asys$
    cannot be used to mimic such instances of $\arule$ that do not have theorems of $\asys$ as conclusion by derivations in $\asys$.
  As a trivial example we take $\asys = \eqlogicover{\actions}$. In $\asys$ only reflexivity axioms are theorems.
    A \premisen{1} rule that leaves its premise unchanged is clearly admissible in $\asys$,
      but not derivable, because instances with formulas $\astexp \formeq \bstexp$ where $\astexp \notsynteq \bstexp$ cannot be mimicked
        by derivations in $\asys$. 
\end{proof}

Now we introduce Milner's proof system \milnersys, and two of its variants \milnersysacc\ and \milnersysaccbar.
 Afterwards we gather basic connections between these systems.

\begin{defi}[Milner's system \protect\milnersys, variants and subsystems]
            \label{def:milnersys}
  Let $\actions$ be a set of actions.
  
  By the proof system $\milnersysminover{\actions}$
    we mean  
  the \eqlogicbasedover{\actions} proof system for star expressions over $\actions$
  with the following \emph{axiom schemes}:
  \begin{alignat*}{4}
    (\assocstexpsum) \quad & &
        \stexpsum{(\stexpsum{\astexp}{\bstexp})}{\cstexp}
          & \formeq
        \stexpsum{\astexp}{(\stexpsum{\bstexp}{\cstexp})}
    & \qquad
    (\leftidstexpprod) \quad & &
        \stexpprod{\stexpone}{\astexp}  
          & \formeq \astexp
    \displaybreak[0]\\ 
    (\neutralstexpsum) \quad & &
        \stexpsum{\astexp}{\stexpzero}
          & \formeq
        \astexp
    & \qquad
    (\rightidstexpprod) \quad & &   
        \stexpprod{\astexp}{\stexpone}
          & \formeq \astexp
    \displaybreak[0]\\ 
    (\commstexpsum) \quad & &
        \stexpsum{\astexp}{\bstexp} 
          & \formeq
        \stexpsum{\bstexp}{\astexp}
    & \qquad
    (\stexpzerostexpprod) \quad & &  
        \stexpprod{\stexpzero}{\astexp}
          & \formeq
        \stexpzero
    \displaybreak[0]\\ 
    (\idempotstexpsum) \quad & &
        \stexpsum{\astexp}{\astexp}
          & \formeq
        \astexp 
    & \qquad
    (\recdefstexpit) \quad & &
        \stexpit{\astexp}
          & \formeq \stexpsum{\stexpone}{\stexpprod{\astexp}{\stexpit{\astexp}}}
    \displaybreak[0]\\ 
    (\assocstexpprod) \quad & &
        \stexpprod{(\stexpprod{\astexp}{\bstexp})}{\cstexp}
          & \formeq
        \stexpprod{\astexp}{(\stexpprod{\bstexp}{\cstexp})}   
    & \qquad\quad
    (\termstexpit) \quad & & 
        \stexpit{\astexp}
          & \formeq \stexpit{(\stexpsum{\stexpone}{\astexp})}
    \displaybreak[0]\\ 
    (\rdistr) \quad & & 
        \stexpprod{(\stexpsum{\astexp}{\bstexp})}{\cstexp}
          & \formeq
        \stexpsum{\stexpprod{\astexp}{\cstexp}}{\stexpprod{\bstexp}{\cstexp}} 
  \end{alignat*}
  where $\astexp,\bstexp,\cstexp\in\StExpover{\actions}$, 
  and with the \emph{rules} of the system $\eqlogicover{\actions}$ of equational logic.
  
  The \emph{recursive specification principle for star iteration} $\RSPstar$,
  the \emph{unique solvability principle for star iteration} $\USPone$,
  and the \emph{general unique solvability principle} $\USP$ are the schematically defined rules with \sideconditions\ of the following forms:
  \begin{gather*}
    \begin{aligned}\renewcommand{\fCenter}{\formeq}
      \Axiom$ \astexp  \fCenter  \stexpsum{\stexpprod{\bstexp}{\astexp}}{\cstexp} $
      \RightLabel{\RSPstar\ $\,${\small (if $\notterminates{\bstexp}$)}}
      \UnaryInf$ \astexp   \fCenter   \stexpprod{\stexpit{\bstexp}}{\cstexp} $
      \DisplayProof
      & \quad & 
      \AxiomC{$ \astexpi{1}   \formeq   \stexpsum{\stexpprod{\bstexp}{\astexpi{1}}}
                                       {\cstexp} $}
      \AxiomC{$ \astexpi{2}   \formeq   \stexpsum{\stexpprod{\bstexp}{\astexpi{2}}}
                                       {\cstexp} $}
      \RightLabel{\USPone\ $\,${\small (if $\notterminates{\bstexp}$)}}
      \BinaryInfC{$ \astexpi{1}   \formeq   \astexpi{2} \rule{0pt}{7pt} $}                                 
      \DisplayProof
    \end{aligned}
    \displaybreak[0]\\[0.25ex]
    \begin{aligned}  
      \AxiomC{$ \Bigl\{\,
                  \astexpi{i,1}   \formeq    \bigl(\, \sum_{j=1}^{n_i} \bstexpi{i,j} \prod \astexpi{j,1} \,\bigr) + \cstexpi{i} \,\bigr)
                    \quad
                  \astexpi{i,2}   \formeq    \bigl(\, \sum_{j=1}^{n_i} \bstexpi{i,j} \prod \astexpi{j,2} \,\bigr) + \cstexpi{i} \,\bigr) 
                \,\Bigr\}_{i=1,\ldots,n} $}
      \RightLabel{\USP\ $\,$ \parbox{\widthof{for all $i,j$)}}
                                    {\small (if $\notterminates{\bstexpi{i,j}}$
                                     \\\phantom{(}%
                                     for all $i,j$)}}
      \UnaryInfC{$ \astexpi{1,1}   \formeq    \astexpi{1,2} \rule{0pt}{7.5pt} $}
      \DisplayProof
    \end{aligned}
  \end{gather*}
  
  \emph{Milner's proof system} $\milnersysover{\actions}$ is the extension of $\milnersysminover{\actions}$ 
  by adding the rule $\RSPstar$. Its variant systems $\milnersysaccover{\actions}$, and $\milnersysaccbarover{\actions}$,
  arise from $\milnersysminover{\actions}$ by adding (instead of $\RSPstar$) the rule $\USPone$, and respectively, the rule $\USP$.
  $\ACIover{\actions}$ is the system with the axioms for \ul{a}ssociativity, \ul{c}ommutativity, and \ul{i}dempotency for $+$. 
  We will keep the action set $\actions$ implicit in the notation. 
\end{defi}

\begin{lem}\label{lem:milnersys:thmequiv:milnersysacc:isthmsubsumedby:milnersysaccbar}
  Milner's system $\milnersys$ and its variants $\milnersysacc$ and $\milnersysaccbar$ are related as follows:
  \begin{enumerate}[label={(\roman{*})},itemsep=0.25ex]
    \item{}\label{it:1:lem:milnersys:thmequiv:milnersysacc:isthmsubsumedby:milnersysaccbar}
      $\milnersysacc \isthmsubsumedby \milnersysaccbar$,
    \item{}\label{it:2:lem:milnersys:thmequiv:milnersysacc:isthmsubsumedby:milnersysaccbar}
      $\milnersys \thmequiv \milnersysacc$.
  \end{enumerate}
\end{lem}

\begin{proof}
  Statement~\ref{it:1:lem:milnersys:thmequiv:milnersysacc:isthmsubsumedby:milnersysaccbar}
    of the lemma is due to the fact that instances the rule $\USPone$ are also instances of $\USP$,
    and therefore 
  $\milnersysacc 
     = 
   \thplus{\milnersysmin}{\USPone}
     \isthmsubsumedby
   \thplus{\milnersysmin}{\USP} = \milnersysaccbar$ follows. 
  
  \smallskip
  For establishing statement~\ref{it:2:lem:milnersys:thmequiv:milnersysacc:isthmsubsumedby:milnersysaccbar}
    we show that $\USPone$ is a derivable rule in $\milnersys$,
             and that $\RSPstar$ is a derivable rule in $\milnersysacc$.
  Then, by Lemma~\ref{lem:derivable:admissible:rules}, \ref{it:2:lem:derivable:admissible:rules},
    $\USPone$ is an admissible rule of $\milnersys$, 
      thus $\thplus{\milnersys}{\USPone} \thmequiv \milnersys$,
     and $\RSPstar$ is an admissible rule of $\milnersysacc$,
      hence $\thplus{\milnersysacc}{\USPone} \thmequiv \milnersysacc$.
  With this we can argue as follows: 
  \begin{equation*}
    \milnersys
      \thmequiv
    \thplus{\milnersys}{\USPone}
      =
    \thplus{(\thplus{\milnersysmin}{\RSPstar})}{\USPone}
      =
    \thplus{(\thplus{\milnersysmin}{\USPone})}{\RSPstar}
      =
    \thplus{\milnersysacc}{\RSPstar}
      \thmequiv
    \milnersysacc     
  \end{equation*} 
  Then we obtain $\milnersys \thmequiv \milnersysacc$ by transitivity of theorem equivalence $\sthmequiv$.
  
  For showing that $\RSPstar$ is derivable in $\milnersysacc$, 
    we consider an instance of $\RSPstar$ as in Definition~\ref{def:milnersys}, 
      for fixed star expressions $\astexp$, $\bstexp$ with $\notterminates{\bstexp}$, and $\cstexp$. 
  From its premise 
  $\astexp = \stexpsum{\stexpprod{\bstexp}{\astexp}}
                      {\cstexp}$
  we have to show that the conclusion  
  $\astexp = \stexpprod{\stexpit{\bstexp}}{\cstexp}$ of the $\RSPstar$ instance can be derived by inferences in $\milnersysacc = \thplus{\milnersysmin}{\USPone}$.
  By stepwise use of axioms of $\milnersysmin$ we obtain:
  \begin{align*}
    \stexpprod{\stexpit{\bstexp}}{\cstexp}
      & {} \milnersysmineq
    \stexpprod{(\stexpsum{\stexpone}
                         {\stexpprod{\bstexp}{\stexpit{\bstexp}}})}{\cstexp}
      \milnersysmineq
    \stexpsum{\stexpprod{\stexpone}{\cstexp}}
             {\stexpprod{(\stexpprod{\bstexp}{\stexpit{\bstexp}})}{\cstexp}}
    \\         
      & {} \milnersysmineq
    \stexpsum{\stexpprod{\stexpone}{\cstexp}}
             {\stexpprod{\bstexp}{(\stexpprod{\stexpit{\bstexp}}{\cstexp})}}
      \milnersysmineq
    \stexpsum{\cstexp}
             {\stexpprod{\bstexp}{(\stexpprod{\stexpit{\bstexp}}{\cstexp})}}
      \milnersysmineq
    \stexpsum{\stexpprod{\bstexp}{(\stexpprod{\stexpit{\bstexp}}{\cstexp})}}
             {\cstexp} 
  \end{align*}           
  Hence there is a derivation of
  $ \stexpprod{\stexpit{\bstexp}}{\cstexp}
      = 
    \stexpsum{\stexpprod{\bstexp}{(\stexpprod{\stexpit{\bstexp}}{\cstexp})}}
             {\cstexp} $
  in $\milnersysacc$.
  This derivation can be extended, due to $\notterminates{\bstexp}$,
    by an instance of $\USPone$ that is applied  
    to $\astexp = \stexpsum{\stexpprod{\bstexp}{\astexp}}
                           {\cstexp}$
  and 
  $ \stexpprod{\stexpit{\bstexp}}{\cstexp}
      = 
    \stexpsum{\stexpprod{\bstexp}{(\stexpprod{\stexpit{\bstexp}}{\cstexp})}}
             {\cstexp} $.
  We obtain a derivation of
  $\astexp = \stexpprod{\stexpit{\bstexp}}{\cstexp}$~in~$\milnersysacc$
    from the assumption 
  $\astexp = \stexpsum{\stexpprod{\bstexp}{\astexp}}
                      {\cstexp}$.
  
  For showing that $\USPone$ is derivable in $\milnersys$, 
    we consider an instance of $\USPone$ as in Definition~\ref{def:milnersys}, 
  with premises
  $\astexpi{1} = \stexpsum{\stexpprod{\bstexp}{\astexpi{1}}}
                          {\cstexp}$,
  and
  $\astexpi{2} = \stexpsum{\stexpprod{\bstexp}{\astexpi{2}}}
                          {\cstexp}$,
  for fixed star expressions $\astexpi{1}$, $\astexpi{2}$, $\bstexp$ with $\notterminates{\bstexp}$, and $\cstexp$.
  By two instances of $\RSPstar$ we get
  $\astexpi{1} = \stexpprod{\stexpit{\bstexp}}{\cstexp}$,
  and
  $\astexpi{2} = \stexpprod{\stexpit{\bstexp}}{\cstexp}$.
  By applying $\SYMM$ below $\astexpi{2} = \stexpprod{\stexpit{\bstexp}}{\cstexp}$, we obtain
  $\stexpprod{\stexpit{\bstexp}}{\cstexp} = \astexpi{2}$.
  Then by applying $\TRANS$ to $\astexpi{1} = \stexpprod{\stexpit{\bstexp}}{\cstexp}$ and $\stexpprod{\stexpit{\bstexp}}{\cstexp} = \astexpi{2}$
  we obtain the conclusion $\astexpi{1} = \astexpi{2}$ of the $\USPone$ instance. 
\end{proof}

Now we define soundness and completeness of equation-based proof systems for star expressions
  with respect to equivalence relations on star expressions.
Then we formulate soundness of Milner's system from \cite{miln:1984}, 
    and recall Milner's completeness question.

\begin{defi}
  Let $\asys$ be an \equationbasedover{\actions} proof system 
    and let $\aeqrel$ be an equivalence relation on $\StExpover{\actions}$. 
  We say that \emph{$\asys$ is sound for $\aeqrel$} 
    if, for all $\astexp,\bstexp\in\StExpover{\actions}$, $\astexp \eqin{\asys} \bstexp$ implies $\astexp \aeqrel \bstexp$.
  We say that \emph{$\asys$ is complete for $\aeqrel$} 
    if, for all $\astexp,\bstexp\in\StExpover{\actions}$, $\astexp \aeqrel \bstexp$ implies $\astexp \eqin{\asys} \bstexp$.
\end{defi}

\begin{propC}[\cite{miln:1984}] 
             \label{prop:milnersys:sound}
  $\milnersys$ is sound for process semantics equality $\sprocsemeq$ on regular expressions. 
  That is, for all $\astexp,\bstexp\in\StExpover{\actions}$ it holds: 
  $(\,
     \astexp \milnersyseq \bstexp
       \;\;\Longrightarrow\;\;
     \astexp \procsemeq \bstexp  
   \,) \punc{,}$
  and hence
  $(\,
     \astexp \milnersyseq \bstexp
       \;\;\Longrightarrow\;\;
     \astexp \procsemeq \bstexp  
   \,) \punc{.}$ 
\end{propC}

\theoremstyle{defC}
\newtheorem{questionC}[thm]{Question}

\begin{questionC}[{\cite{miln:1984}}]
  Is $\milnersys$ complete for bisimilarity of process interpretations?
  That is, does for all $\astexp,\bstexp\in\StExpover{\actions}$ the implication
  $(\,
     \astexp \milnersyseq \bstexp
       \;\;\Longleftarrow\;\;
     \chartof{\astexp} \bisim \chartof{\bstexp}
   \,)$ hold?
\end{questionC}

Finally we define the crucial concept of provable solution of a \onechart\ with respect to an \eqlogicbased\ proof system. 
Intuitively, a `provable solution' of a \onechart~$\aonechart$ is a provable solution of some recursive specification $\fap{\aspec}{\aonechart}$
  that is associated with $\aonechart$ in a natural way (see for example the two examples on page~\pageref{fig:milner-bosscher-expressible}).
Since associating specifications $\fap{\aspec}{\aonechart}$ to \onecharts~$\aonechart$ 
  presupposes the use of some list representation for the set $\transitionsinfrom{\aonechart}{\avert}$ of transitions from vertex~$\avert$,
    for every vertex $\avert$ of $\aonechart$,
    any such association map cannot be unique. 
The definition of provable solutions of \onecharts\ below
  uses such list representations implicitly, and assumes that associativity, commutativity, and reflexivity axioms are present in the underlying proof system.
In this way the concept of provable solution permits us to avoid defining associated specifications for \onecharts\ in some canonical (but still necessarily arbitrary) way.  
In the next section we explain an alternative characterization of provable solutions. 

\begin{defi}[provable solutions]\label{def:provable:solution}
  Let $\asys$ be an \eqlogicbased\ proof system for star expressions over $\actions$
    that extends \ACI. 
  Let $\aonechart = \tuple{\verts,\actions,\sone,\start,\transs,\exts}$ be a \onechart.
  
  By a \emph{star expression function on $\aonechart$} we mean a function $\sasol \funin \verts \to \StExpover{\actions}$
    on the vertices of $\aonechart$. 
  Let $\avert\in\verts$.   
  We say that such a star expression function $\sasol$ on $\aonechart$ 
    is an \emph{\provablein{\asys} solution of $\aonechart$ at $\avert$} if it holds that:
    \begin{equation}\label{eq:def:provable:solution}
      \asol{\avert}
        \,\eqin{\asys}\,
          \stexpsum{\terminatesconstof{\aonechart}{\avert}}
                   {\sum_{i=1}^{n} \stexpprod{\aoneacti{i}}{\asol{\averti{i}}}} 
            \punc{,}
    \end{equation}
    given the (possibly redundant) list representation
    $\transitionsinfrom{\aonechart}{\avert}
       =
     \descsetexpbig{ \avert \lt{\aoneacti{i}} \averti{i} }{ i \in\setexp{1,\ldots,n} }$,
    of transitions from $\avert$ in~$\aonechart$ 
    and where $\terminatesconstof{\aonechart}{\avert}$
    is the \emph{termination constant $\terminatesconstof{\aonechart}{\avert}$ of $\aonechart$ at $\avert$}
      defined as $\stexpzero$ if $\notterminates{\avert}$,
                      and as $\stexpone$ if $\terminates{\avert}$.
    This definition does not depend on the specifically chosen list representation of $\transitionsinfrom{\aonechart}{\avert}$, 
    because $\asys$ extends \ACI, and therefore it contains the associativity, commutativity, and idempotency axioms~for~$\sstexpsum$.
  
  By an \emph{\provablein{\asys} solution of $\aonechart$} 
    (with \emph{principal value $\asol{\start}$} at the start vertex $\start$)
    we mean a star expression function $\sasol$ on $\aonechart$ that is an \provablein{\asys} solution of $\aonechart$ at every vertex~of~$\aonechart$.
\end{defi}

\section{Characterization of provable solutions of \protect\sone-charts}%
  \label{solutions}

This section is an intermezzo in which we link
  to an elegant coalgebraic formulation of the concept of provable solution
    by Schmid, Rot, and Silva in \cite{schm:rot:silv:2021}.
  Their observation is a crucial first part of a detailed and beautiful coalgebraic analysis
    of the completeness proof in \cite{grab:fokk:2020:lics:arxiv,grab:fokk:2020:lics} 
      by Fokkink and myself for a tailored restriction of Milner's system \milnersys\
        to `\onefree\ star expressions'. 
Here we reformulate their characterization of provable solution by means of the terminology that we are using here,
  and explain the connection, but do not prove the statements in detail.
This is because we will only use this characterization later as 
  an additional motivation for our concept of `coinductive proof',
    but not for developing the proof transformations to and from the coinductive reformulation $\coindmilnersys$ of Milner's system $\milnersys$.      
    
Schmid, Rot, and Silva construe the operational process semantics that a transition system like that in Definition~\ref{def:chartof}
  induces on the set $\StExp$ of all star expressions as a coalgebra (also denoted by) $\StExp$. 
Charts can also be represented as (finite) coalgebras due to their structure as transition graphs.
On this basis,
  they obtain the following characterization of provable solutions for proof systems $\asys$ like \milnersys\ and \milnersysmin.
  
\begin{lem}[$\protect\sim\,$Lemma~2.2 in \cite{schm:rot:silv:2021}]
  For every chart $\achart$, for every \starexpression\ function $\sasol \funin \vertsof{\achart} \to \StExp$,
    and for $\asys \in \setexp{ \milnersys, \milnersysmin }$,
    the following two statements are equivalent:
  \begin{enumerate}[label={(\roman{*})}]
    \item
      $\sasol$
        is an \provablein{\asys} solution of a chart $\achart$.
    \item   
      $\achart \overset{\eqcl{\sasol}{\seqin{\asys}}}{\longrightarrow} \factor{\StExp}{\seqin{\asys}}$ 
        is a coalgebra homomorphism,
        where  $\begin{aligned}[t]
                 \eqcl{\sasol}{\seqin{\asys}} \!\!\funin \vertsof{\achart} & {} \to \factor{\StExp}{\seqin{\asys}}
                 \\[-0.5ex]
                                                                   \avert  & {} \mapsto \eqcl{\asol{\avert}}{\seqin{\asys}} \punc{.}
                \end{aligned}$ 
  \end{enumerate}
\end{lem}
By analyzing the proof of this statement on the basis of terminology we use here, we find the following. 
  Schmid, Rot, and Silva noticed that statements~\ref{it:1:SRS} and \ref{it:2:SRS} below hold, and used them in conjunction with \ref{it:3:SRS}
    (which is analogous to Proposition~2.9 in \cite{grab:fokk:2020:lics:arxiv,grab:fokk:2020:lics}) 
    to obtain the characterization for charts above. We extend it to \onecharts\ here in \ref{it:4:SRS}: 
\begin{enumerate}[label={(\alph{*})},align=right,leftmargin=*,itemsep=0.25ex]
  \item{}\label{it:1:SRS} 
    Provability in a system $\asys$ like Milner's defines a bisimulation relation
    on the the set of star expressions when that is endowed with the process semantics.
      (See Lemma~\ref{lem:provability:is:bisim}).
      
  \item{}\label{it:2:SRS} 
    Due to \ref{it:1:SRS} a factor chart $\factor{\chartof{\astexp}}{\eqin{\asys}}$ can be defined such
      that $\chartof{\astexp} \funbisim \factor{\chartof{\astexp}}{\seqin{\asys}}$ holds,
        that is, there is a functional bisimulation from $\chartof{\astexp}$ to $\factor{\chartof{\astexp}}{\seqin{\asys}}$. 
      (See Lemma~\ref{lem:bisim:factor:onechart}).
      
  \item{}\label{it:3:SRS} 
    Every star expression $\astexp$ is the principal value of a \provablein{\milnersys} solution of 
      the chart interpretation $\chartof{\astexp}$ of $\astexp$.
      (See Lemma~\ref{lem:chart-int:milnersysmin:solvable}). 
    Equivalently, every star expression $\astexp$ can be \provablyin{\milnersys} reconstructed
      from the transitions to its derivatives in the process semantics. (See Lemma~\ref{lem:FT:chart-int}.)
      
  \item{}\label{it:4:SRS}   
    A star expression function $\sasol$ with principal value $\astexp$ is a \provablein{\milnersys} solution of a \onechart~$\aonechart$
      if and only if
    the relativization $\eqcl{\sasol}{\seqin{\milnersys}}$ of $\sasol$ to \equivalenceclasseswrt{\seqin{\milnersys}}
      defines a \onebisimulation\ from $\aonechart$ to $\factor{\chartof{\astexp}}{\seqin{\asys}}$. (See Proposition~\ref{prop:char:solution:SRS}.)
\end{enumerate}
   
We formulate these statements more precisely here below. We start with a general definition of factor charts.
  But the property \ref{it:2:SRS} of factor charts of the chart interpretation with respect to provability 
    will then only be shown in Lemma~\ref{lem:bisim:factor:chart} below, 
      after the formulation of the property~\ref{it:1:SRS} in Lemma~\ref{lem:provability:is:bisim}. 

\begin{defi}[factor chart]\label{def:factor:chart}
  Let $\achart = \tuple{\verts,\actions,\sone,\start,\transs,\termexts}$
    be a chart (that is, a \onetransition\ free \onechart). 
  Let $\sacongrel$ be an equivalence relation on $\verts$. 
  Then we define the \emph{factor chart} $\factor{\achart}{\sacongrel}$ of $\achart$ with respect to $\sacongrel$ by:
  \begin{align*} 
    &
    \factor{\achart}{\sacongrel} = \tuple{\factor{\verts}{\sacongrel},\actions,\sone,\eqcl{\start}{\sacongrel},\factor{\transs}{\sacongrel},\factor{\termexts}{\sacongrel}}
    \displaybreak[0]\\
    & \quad
    \text{ where: }
    \begin{aligned}[t]
      &
      \factor{\verts}{\sacongrel} 
        \defdby 
          \descsetexpbig{ \eqcl{\avert}{\sacongrel} }{ \avert\in\verts } \punc{,}
      \text{ and for all $\avert,\averti{1},\averti{2}\in\verts$ and $\aact\in\actions$: }   
      \\
      & 
      \eqcl{\averti{1}}{\sacongrel} 
        \ltfact{\aact}{\sacongrel}
      \eqcl{\averti{2}}{\sacongrel}
        \; \funin \: \Longleftrightarrow \;
          \begin{aligned}[t]
            & \text{there are $\averttildei{1},\averttildei{2}\in\verts$ such that }
              \averti{1} \acongrel \averttildei{1} \lt{\aact} \averttildei{2} \acongrel \averti{2} \punc{,}
          \end{aligned}
      \\
      &
      \terminatesfact{({\eqcl{\avert}{\sacongrel}})}{\sacongrel}
        \; \funin\: \Longleftrightarrow \;
          \terminates{\avert} \punc{.}
    \end{aligned}
  \end{align*}
  By the projection function from $\achart$ to $\factor{\achart}{\sacongrel}$
    we mean the function $\sprojwrt{\sacongrel} \funin \verts \to \factor{\verts}{\sacongrel}$, $\avert \mapsto \eqcl{\avert}{\sacongrel}$.
\end{defi}

\begin{lem}\label{lem:provability:is:bisim}
  Let $\asys\in\setexp{\milnersysminover{\actions},\milnersysover{\actions}}$.
  Then provability $\eqin{\asys}$ with respect to $\asys$ is a bisimulation on the chart interpretation $\chartof{\astexp}$ of $\astexp$, for every $\astexp\in\StExpover{\actions}$.  
\end{lem}

\begin{proof}[Proof (Idea)]
  By verifying the bisimulation conditions (forth), (back), (termination)
    for conclusions of derivations in $\asys$, proceeding by induction on the depth of derivations in $\asys$.
  In the base case, this is settled for the axioms of \milnersysmin.
  In the induction step, it is settled for the conclusions of the reflexivity, symmetry, and transitivity rules of \eqlogic, 
    and of the fixed-point rule \RSPstar\ in \milnersys.
  (The arguments are similar to the proof of Theorem~2.1 in \cite{schm:rot:silv:2021}.)
\end{proof}

\begin{lem}\label{lem:bisim:factor:chart}
  Let $\achart = \tuple{\verts,\actions,\sone,\start,\transs,\termexts}$
    be a (\sonefree) chart. Let $\sacongrel$ be an equivalence relation on $\verts$ that is a bisimulation on $\achart$.
    
  Then
  $\achart \funbisimvia{\sprojwrt{\sacongrel}} \factor{\achart}{\sacongrel}$
    holds,
  that is, $\sprojwrt{\sacongrel}$ defines a functional bisimulation\ from $\achart$ to $\factor{\achart}{\sacongrel}$.          
\end{lem}

\begin{proof}[Proof (hint)]
  The (forth) and (termination) conditions for the graph $\graphof{\sprojwrt{\sacongrel}}$ of $\sprojwrt{\sacongrel}$ 
    to be a bisimulation are easy to verify.
  For demonstrating also the (back) condition for $\graphof{\sprojwrt{\sacongrel}}$ to be a bisimulation
    it is crucial to use the assumption that $\sacongrel$ is a bisimulation on $\achart$.
\end{proof}

\begin{lem}\label{lem:bisim:factor:onechart}
  $\chartof{\astexp} \funbisimvia{\sprojwrt{\seqin{\asys}}} \factor{\chartof{\astexp}}{\seqin{\asys}}$
    for every $\astexp\in\StExpover{\actions}$, and $\asys\in\setexp{\milnersysmin,\milnersys}$.
\end{lem}

\begin{proof}
  Let $\asys\in\setexp{\milnersysmin,\milnersys}$.
  Due to Lemma~\ref{lem:provability:is:bisim}, $\eqin{\asys}$ is a bisimulation on $\chartof{\astexp}$.
  Then 
    we obtain 
      $\chartof{\astexp} \funbisimvia{\sprojwrt{\seqin{\asys}}} \factor{\chartof{\astexp}}{\seqin{\asys}}$
    by applying Lemma~\ref{lem:bisim:factor:chart}.
\end{proof}

The following lemma states that every star expression $\astexp$ can be reconstructed, provably in \milnersysmin,
  from the transitions that it facilitates in the process semantics, and the targets of these transitions.
Statements like this are frequently viewed as being analogous to the fundamental theorem of calculus,
  which states that every differentiable function can be reconstructed from its derivative function via integration.
  
\begin{lem}\label{lem:FT:chart-int}
  $ \astexp
     \milnersysmineq
   \terminatesconstof{\chartof{\astexp}}{\astexp} 
     +  
   \sum_{i=1}^{n} \stexpprod{\aacti{i}}{\astexpacci{i}} $ holds,
  given a list representation\vspace{-1pt}
  $\transitionsinfrom{\chartof{\astexp}}{\bvert}
     =
   \descsetexpbig{ \astexp \lt{\aacti{i}} \astexpacci{i} }{ i\in\setexp{1,\ldots,n} }$
  of the transitions from $\astexp$ in the chart interpretation $\chartof{\astexp}$ of $\astexp$.
\end{lem}

\begin{proof}[Proof (hint)]
  The proof proceeds by induction on the structure of the star expression $\astexp$.
  All axioms of \milnersysmin\ (and hence all axioms of \milnersys) are necessary in the arguments.
  An analogous statement that can be viewed as the restriction of the statement of Lemma~\ref{lem:FT:chart-int}
    for `\onefree\ star expressions'
    was proved as Lemma~A.2 in \cite{grab:fokk:2020:lics:arxiv} ,
      and as Theorem~2.2 in \cite{schm:rot:silv:2021}.
  Here we will prove an analogous statement, Lemma~\ref{lem:FT:onechart-int}, in the next section.
\end{proof}

\begin{lem}\label{lem:chart-int:milnersysmin:solvable}
  For every star expression $\astexp\in\StExpover{\actions}$ 
    with chart interpretation $\chartof{\astexp} = \tuple{\vertsof{\astexp},\actions,\sone,\astexp,\transs,\exts}$
  the identical \starexpression\ function $\sidfunon{\vertsof{\astexp}} \funin \vertsof{\astexp} \to \StExpover{\actions}$, $\astexp \mapsto \astexp$
    is a \provablein{\milnersysmin} solution of $\chartof{\astexp}$ with principal value $\astexp$.  
\end{lem}

\begin{proof}
  The correctness conditions for the \starexpression\ function $\sidfunon{\vertsof{\astexp}}$ 
      to be a \provablein{\milnersysmin} solution of the chart interpretation $\chartof{\astexp}$ of $\astexp$
    are guaranteed by Lemma~\ref{lem:FT:chart-int}.
\end{proof}

On the basis of these preparations we now reformulate the characterization of provable solutions of charts by Schmid, Rot, and Silva
  as a characterization of provable solutions of \onecharts\ via functional \onebisimulations\ to factor charts of appropriate chart interpretations.

\begin{prop}\label{prop:char:solution:SRS}
  Let $\aonechart = \tuple{\verts,\actions,\start,\sone,\transs,\termexts}$ be a \onechart.
  Let $\asys \in\setexp{\milnersysover{\actions},\milnersysminover{\actions}}$. 
  Then for all star expression functions $\sasol \funin \verts \to \StExpover{\actions}$ it holds:
  \begin{equation}\label{eq:prop:char:solution:SRS}
    \text{$\sasol$ is \provablein{\asys} solution of $\aonechart$}
      \;\;\:\;\Longleftrightarrow\;\;\;\;
    \aonechart
      \:\funonebisimvia{\eqcl{\sasol}{\eqin{\asys}}} \!\!
    \factorset{\chartof{\asol{\start}}}{\eqin{\asys}} \punc{.}
  \end{equation}    
  where $\eqcl{\sasol}{\eqin{\asys}} \defdby \scompfuns{\sprojwrt{\seqin{\asys}}}{\sasol}$
    with the projection $\sprojwrt{\seqin{\asys}} \funin \StExpover{\actions} \to \factor{\StExpover{\actions}}{\seqin{\asys}}$, $\astexp \mapsto \eqcl{\astexp}{\seqin{\asys}}$.
\end{prop}

\begin{proof}[Proof (sketch)]
  A technical part of the proof consists in showing that 
    for every \starexpression\ function $\sasol \funin \verts \to \StExpover{\actions}$
      on a \onechart~$\aonechart = \tuple{\verts,\actions,\sone,\start,\transs,\termexts}$ 
        the following two statements are equivalent:
  \begin{enumerate}[label={(\roman{*})}]
    \item
      $\sasol$ is an \provablein{\asys} solution of $\aonechart$,
    \item  
      $\sasol$ is an \provablein{\asys} solution of the `induced chart' $\indchartof{\aonechart} = \tuple{\verts,\actions,\sone,\onetranss,\onetermexts}$ of $\aonechart$
      that results by using induced transitions as transitions, and induced termination as ter\-mi\-na\-tion, that is, with:
      \begin{align*}
        \onetranss
          & {} \,\defdby\,
                 \descsetexpbig{ \triple{\avert}{\aact}{\avertacc} }{ \aact\in\actions,\, \avert,\avertacc\in\verts,\: \avert \ilt{\aact} \avertacc } \punc{,}
          &
        \onetermexts
          & {} \,\defdby\,
                 \descsetexpbig{ \avert }{ \avert\in\verts,\: \oneterminates{\avert} } \punc{.}     
      \end{align*}
  \end{enumerate}    
  
  For the implication `$\Rightarrow$' in \eqref{eq:prop:char:solution:SRS} 
    we assume that $\sasol$ is a \provablein{\asys} solution of $\aonechart$. 
    By the auxiliary statement above, $\sasol$ is then also a \provablein{\asys} solution of the induced chart $\indchartof{\aonechart}$ of $\aonechart$.
    Then it is not difficult to verify the \onebisimulation\ conditions (forth), (back), and (termination) 
      for $\eqcl{\sasol}{\seqin{\asys}}$ to define a \onebisimulation\ from $\aonechart$ to $\factorset{\chartof{\asol{\start}}}{\eqin{\asys}}$.
      
  For the converse implication `$\Leftarrow$' in \eqref{eq:prop:char:solution:SRS},
    we assume $\sasol \funin \verts \to \StExpover{\actions}$ as a star expression function 
    with $\aonechart 
            \:\funonebisimvia{\eqcl{\sasol}{\eqin{\asys}}} \!\!
          \factorset{\chartof{\asol{\start}}}{\eqin{\asys}}$.
    Here Lemma~\ref{lem:FT:chart-int},
      the possibility to \provablyin{\asys} reconstruct a star expression $\astexp$ from the transitions to its derivatives,
    can be employed in order to recognize $\sasol$ as a \provablein{\asys} solution of the induced chart $\indchartof{\aonechart}$ of $\aonechart$.
    Then by applying the auxiliary statement, 
      we obtain that $\sasol$ is also a \provablein{\asys} solution of $\aonechart$.
\end{proof}

\section{Layered loop existence and elimination, and LLEE-witnesses}
  \label{LEE}

In this subsection we recall definitions from \cite{grab:fokk:2020:lics,grab:2021:TERMGRAPH2020-postproceedings}
  of the loop existence and elimination condition \LEE,
  its `layered' version~\LLEE, and of chart labelings that witness these conditions.
Specifically we will use the adaptation of these concepts to \onecharts\ 
  that has been introduced in \cite{grab:2021:TERMGRAPH2020-postproceedings},
  because the use of \onecharts\ with \onetransitions\ 
    will be crucial for the concept of `\LLEEwitnessed\ coinductive proof' in Section~\ref{coindmilnersys}.
For this purpose we also recall the `\onechart\ interpretation' of star expressions    
  as introduced in \cite{grab:2021:TERMGRAPH2020-postproceedings}
    for which the property \LLEE\ is guaranteed in contrast to the chart interpretation from Definition~\ref{def:chartof}.
We will keep formalities to a minimum as these are necessary for our purpose here,
  and have to refer to 
  \cite{grab:fokk:2020:lics,grab:2021:TERMGRAPH2020-postproceedings}
  and the appertaining reports
  \cite{grab:fokk:2020:lics:arxiv,grab:2020:TERMGRAPH2020-postproceedings:arxiv} for more details.
  
We start with the definitions of loop \onecharts, and of loop sub(-$\sone$)charts, and examples for these concepts.   
  
\begin{defi}[loop \protect\onechart]\label{def:loop:onechart}
  A \onechart~$\aoneloop = \tuple{\verts,\actions,\sone,\start,\transs,\termexts}$ is called a \emph{loop \onechart} if
  it satisfies three conditions:
  \begin{enumerate}[label={{\rm (L\arabic*)}},leftmargin=*,align=left,itemsep=0.5ex]
    \item{}\label{loop:1}
      There is an infinite path from the start vertex $\start$.
    \item{}\label{loop:2}  
      Every infinite path from $\start$ returns to $\start$ after a positive number of transitions.
    \item{}\label{loop:3}
      Immediate termination is only permitted at the start vertex, that is, $\termexts\subseteq\setexp{\start}$.
  \end{enumerate}\vspace*{-0.25ex}
  Transitions from $\start$ are called \emph{\loopentry\ transitions},
  all other transitions \emph{\loopbody\ transitions}.
\end{defi}

  \begin{figure}[t]
\begin{center}  
\begin{tikzpicture}

%
\matrix[anchor=center,row sep=1cm,every node/.style={draw,very thick,circle,minimum width=2.5pt,fill,inner sep=0pt,outer sep=2pt}] at (0,0) {
  \node(C1-0){}; 
  \\
  \node(C1-1){};
  \\
  \node(C1-2){};
  \\
};
\draw[<-,very thick,>=latex,color=chocolate](C1-0) -- ++ (90:0.5cm);
\draw[->](C1-0) to (C1-1);
\draw[->](C1-1) to (C1-2);


  \path (C1-2) ++ (0cm,-1cm) node{{{\colorred{\st{(L1)}}}}};

%
\matrix[anchor=center,row sep=1cm,every node/.style={draw,very thick,circle,minimum width=2.5pt,fill,inner sep=0pt,outer sep=2pt}] at (2,0) {
  \node(C2-0){}; 
  \\
  \node(C2-1){};
  \\
  \node(C2-2){};
  \\
};
\draw[<-,very thick,>=latex,color=chocolate](C2-0) -- ++ (90:0.5cm);

\draw[->,very thick,color=red] (C2-0) to node[right]{${\phantom{\aact}}$} (C2-1); 
              
\draw[->] (C2-1) to node[right]{$\hspace*{-0.65pt}{\phantom{\aact}}$} (C2-2);
             
\draw[->,color=red] (C2-1) to node[right]{$\hspace*{-0.65pt}{\phantom{\aact}}$} (C2-2);
             
\draw[->,shorten >=3pt,shorten <=3pt] ([shift=(270:1.13cm)]C2-1) arc[radius=1.13cm,start angle=-90,end angle=90];
\draw[->,out=180,in=180,distance=1.05cm] (C2-2) to node[left]{$\phantom{\bact}$} (C2-1);
  
\draw[->,color=forestgreen,shorten >=3pt,shorten <=3pt] ([shift=(270:1.13cm)]C2-1) arc[radius=1.13cm,start angle=-90,end angle=90];
\draw[->,color=red,out=180,in=180,distance=1.05cm] (C2-2) to node[left]{$\phantom{\bact}$} (C2-1);
  
\path (C2-1) ++ (1.25cm,0cm) node{${\phantom{\bact}}$};
\path ([shift=(270:0.575cm)]C2-1) ++ (-0.83cm,0cm) node{${\phantom{\bact}}$};

\path (C2-0) ++ (-0.275cm,0cm) node{${\averti{0}}$};
\path (C2-1) ++ (+0.3cm,0cm) node{${\averti{1}}$}; 
\path (C2-2) ++ (+0.3cm,-0.275cm) node{${\averti{2}}$};

\path (C2-2) ++ (1.5cm,-1cm) node{{\forestgreen{(L1)},\colorred{\st{(L2)}},\forestgreen{(L3)}}};

%
\matrix[anchor=center,row sep=0.75cm,column sep=0.924cm,ampersand replacement=\&,
        every node/.style={draw,very thick,circle,minimum width=2.5pt,fill,inner sep=0pt,outer sep=2pt}] at (4.75,0) {
  \node(C-2-1){};  \&                  \&     \node(C-2-2){};
  \\
                   \&                  \&                  
  \\
                   \& \node(C-2-3){};  \&
  \\
};
\draw[<-,very thick,>=latex,color=chocolate](C-2-1) -- ++ (90:0.5cm);  

\draw[->,bend right,distance=0.65cm] (C-2-1) to 
                                                (C-2-2);
\draw[->,bend right,distance=0.65cm,very thick,red] (C-2-1) to 
                                                                             (C-2-3); 

\draw[->,bend right,distance=0.65cm]  (C-2-2) to 
                                                 (C-2-1); 
\draw[->,bend left,distance=0.65cm,red]  (C-2-2) to 
                                                                  (C-2-3);

\draw[->,bend right,distance=0.45cm] (C-2-3) to 
                                                ($(C-2-1)+(0.25cm,-0.2cm)$);
\draw[->,bend left,distance=0.65cm,red]  (C-2-3) to 
                                                                  (C-2-2);

\matrix[anchor=center,row sep=1.75cm,column sep=1cm,every node/.style={draw,thick,circle,minimum width=2.5pt,fill,inner sep=0pt,outer sep=2pt},
        ampersand replacement=\&] at (7.25,0) {
  \node[chocolate](C3-0){}; 
  \\
  \node[chocolate](C3-1){};  
  \\
};
\draw[<-,very thick,>=latex,color=chocolate,shorten <=0.06cm](C3-0) -- ++ (90:0.56cm);

\draw[very thick,chocolate] (C3-0) circle (0.12cm);
\draw[very thick,chocolate] (C3-1) circle (0.12cm);

\draw[->,very thick,red,bend left,distance=0.6cm,shorten <=0.09cm, shorten >=0.09cm] (C3-0) to node[above]{$\phantom{\aact}$} (C3-1); 
\draw[->,forestgreen,bend left,distance=0.6cm,shorten <=0.09cm, shorten >=0.09cm]  (C3-1) to node[below]{$\phantom{\aact}$} (C3-0);

\path (C2-2) ++ (4.95cm,-1cm) node{{\forestgreen{(L1)},\forestgreen{(L2)},{\colorred{\st{(L3)}}}}};

%
\matrix[anchor=center,row sep=1cm,every node/.style={draw,very thick,circle,minimum width=2.5pt,fill,inner sep=0pt,outer sep=2pt}] at (9.75,0) { 
  \node(C5-0){}; 
  \\
  \node(C5-1){};
  \\
  \node(C5-2){};
  \\
};
\draw[<-,very thick,>=latex,color=chocolate](C5-1) -- ++ (135:0.5cm);

\draw[->,forestgreen] (C5-0) to node[right]{${\phantom{\aact}}$} (C5-1);

\draw[->,forestgreen] (C5-0) to node[right]{${\phantom{\aact}}$} (C5-1);

\draw[->,very thick,forestgreen] (C5-1) to node[right]{$\hspace*{-0.65pt}{\phantom{\aact}}$} (C5-2);

\draw[->,out=180,in=180,distance=1.05cm] (C5-2) to node[left]{$\phantom{\bact}$} (C5-1);

\draw[->,color=forestgreen,
         shorten >=3pt,shorten <=3pt] ([shift=(270:1.13cm)]C5-1) arc[radius=1.13cm,start angle=-90,end angle=90];

\draw[->,color=forestgreen,
         out=180,in=180,distance=1.05cm] (C5-2) to node[left]{$\phantom{\bact}$} (C5-1);
        
\path (C5-1) ++ (1.25cm,0cm) node{${\phantom{\bact}}$};
\path ([shift=(270:0.575cm)]C5-1) ++ (-0.83cm,0cm) node{${\phantom{\bact}}$};

\path (C5-0) ++ (-0.275cm,0cm) node{${\averti{0}}$};
\path (C5-1) ++ (+0.3cm,0cm) node{${\averti{1}}$}; 
\path (C5-2) ++ (+0.3cm,-0.275cm) node{${\averti{2}}$};

\path (C5-2) ++ (0cm,-1cm) node{{\forestgreen{(L1)},\forestgreen{(L2)},{\forestgreen{(L3)}}}};
\path (C5-2) ++ (0cm,-1.5cm) node{\forestgreen{loop \onechart}};

%
\matrix[anchor=center,row sep=1cm,every node/.style={draw,very thick,circle,minimum width=2.5pt,fill,inner sep=0pt,outer sep=2pt}] at (12.5,0) { 
  \node(C6-0){}; 
  \\
  \node(C6-1){};
  \\
  \node(C6-2){};
  \\
};
\draw[<-,very thick,>=latex,color=chocolate](C6-1) -- ++ (135:0.5cm);

\draw[->] (C6-0) to node[right]{${\phantom{\aact}}$} (C6-1);

\draw[->,darkcyan] (C6-0) to node[right]{${\phantom{\aact}}$} (C6-1);

\draw[->,darkcyan] (C6-1) to node[right]{$\hspace*{-0.65pt}{\phantom{\aact}}$} (C6-2);

\draw[->,out=180,in=180,distance=1.05cm] (C6-2) to node[left]{$\phantom{\bact}$} (C6-1);

\draw[->,very thick,color=darkcyan,shorten >=3pt,shorten <=3pt] ([shift=(270:1.13cm)]C6-1) arc[radius=1.13cm,start angle=-90,end angle=90];

\draw[->,out=180,in=180,distance=1.05cm] (C6-2) to node[left]{$\phantom{\bact}$} (C6-1);
        
\path (C6-1) ++ (1.25cm,0cm) node{${\phantom{\bact}}$};
\path ([shift=(270:0.575cm)]C6-1) ++ (-0.83cm,0cm) node{${\phantom{\bact}}$};

\draw[<-,very thick,>=latex,color=darkcyan](C6-2) -- ++ (270:0.5cm);

\path (C6-0) ++ (-0.275cm,0cm) node{${\averti{0}}$};
\path (C6-1) ++ (+0.3cm,0cm) node{${\averti{1}}$}; 
\path (C6-2) ++ (+0.3cm,-0.275cm) node{${\averti{2}}$};

\path (C6-2) ++ (0.15cm,-1cm) node{\darkcyan{loop \subonechart}};

\path ($(C5-0)!0.5!(C6-0)$) ++ (0cm,0.25cm) node{\Large $\aonechart$};
\path (C6-2) ++ (0.95cm,-0.15cm) node{\Large $\darkcyan{\aoneloop}$};

\end{tikzpicture}
\end{center}  
  \vspace*{-1.5ex}%
  \caption{\label{fig:exs:nonexs:looponechart}
    Four \protect\onecharts\ (action labels ignored) 
      that violate at least one loop \protect\onechart\ condition \ref{loop:1}, \ref{loop:2}, or \ref{loop:3}, 
    and a loop \protect\onechart\ $\aonechart$ with one of~its~loop~\protect\subonecharts~$\darkcyan{\aoneloop}$.}
\end{figure}%
\begin{exa}\label{exa:loop:chart}
  In Figure~\ref{fig:exs:nonexs:looponechart} we have gathered, on the left,
    four examples of \onecharts\ (with action labels ignored) that are \emph{not} loop \onecharts:
      each of them violates one of the conditions \ref{loop:1}, \ref{loop:2}, or \ref{loop:3}.
  The paths in red indicate violations of \ref{loop:2}, and \ref{loop:3}, respectively,
    where the thicker arrows from the start vertex indicate transitions that would need to be (but are not) \loopentrytransitions.
  However, the \onechart~$\aonechart$ in Figure~\ref{fig:exs:nonexs:looponechart} is~indeed~a~loop~\onechart.   
\end{exa}

\begin{defi}[loop \protect\subonechart\ of \protect\onechart]\label{def:loop:subonechart}
  Let $\aonechart = \tuple{\verts,\actions,\sone,\start,\transs,\termexts}$
    be a \onechart.
    
  A \emph{loop \subonechart\ of} a \onechart~$\aonechart$
    is a loop \onechart~$\aoneloop$
    that is a \subonechart\ of $\aonechart$ 
      with some vertex $\avert\in\verts$ of $\aonechart$ as start vertex
    such that $\aoneloop$ is formed, for a nonempty set $\asettranss$ of transitions of $\aonechart$ from $\avert$,
    by all vertices and transitions on paths that start with a transition in $\asettranss$ and continue onward until $\avert$ is reached again;
    in this case the transitions in $\asettranss$ are the \loopentrytransitions~of~$\aoneloop$,
      and we say that the transitions in $\asettranss$ \emph{induce} $\aoneloop$.
\end{defi}

\begin{exa}
  In the \onechart~$\aonechart$ in Figure~\ref{fig:exs:nonexs:looponechart}
    we have illustrated (in the right copy~of~$\aonechart$)
      a loop \subonechart~$\darkcyan{\aoneloop}$ of $\aonechart$ with start vertex $\averti{2}$
        that is induced by the set $\asettranss \defdby \setexp{ \triple{\averti{2}}{\aact}{\averti{0}} }$ 
          that consists of the single \loopentry\ transition from $\averti{2}$ to $\averti{0}$, assuming that its action label is $\aact$.
  Then $\darkcyan{\aoneloop}$ consists of all \darkcyan{colored} transitions.
  --~We note that also the generated \subonechart~$\gensubchartofby{\aonechart}{\averti{2}}$ of $\aonechart$ 
       that is rooted at $\averti{2}$ is a loop \subonechart\ of $\aonechart$,
         because it is a loop \onechart, and it 
           that is generated by the set of both of the two transitions from $\averti{2}$.
\end{exa}

\begin{defi}[(single-/multi-step) loop elimination]\label{def:loopelimred}
  Let $\aonechart = \tuple{\verts,\actions,\sone,\start,\transs,\termexts}$ be a \onechart.
  Suppose that $\aoneloop$, $\aoneloopi{1}$, \ldots, $\aoneloopi{n}$ are loop \subonecharts\ 
    with sets $\asettranss, \asettranssi{1},\ldots,\asettranssi{n}$ of \loopentry\ transitions 
      from their start vertices $\avert,\averti{1},\ldots,\averti{n}\in\verts$, respectively,
  for $n\in\nat\setminus\setexp{0}$.
 
  The \onechart~$\aonechartacc$ that results by the \emph{elimination of (the loop \subonechart) $\aoneloop$ from $\aonechart$}
    arises by removing all \loopentrytransitions\ in $\asettranss$ of $\aoneloop$ from $\aonechart$, 
      and then also removing all vertices and transitions that become unreachable;
        in this case we write $\aonechart \loopelimred \aonechartacc$,
          and also say that 
            $\aonechartacc$ results by a \emph{single-step loop elimination} from $\aonechart$.
                     
    Suppose that the loop \subonecharts~$\aoneloopi{1}$, \ldots, $\aoneloopi{n}$ 
      satisfy the following two conditions: 
      \begin{enumerate}[label={(ms-\arabic{*})}]
        \item{}\label{msone}
          their sets $\asettranssi{1},\ldots,\asettranssi{n}$ of \loopentry\ transitions are disjoint (that is, $\asettranssi{i} \cap \asettranssi{j} \neq \emptyset$
                                                                                for all $i,j\in\setexp{1,\ldots,n}$ with $i \neq j$),
        \item{}\label{mstwo}   
          no start vertex of a loop \subonechart\ $\aoneloopi{i}$ is in the body of another one $\aoneloopi{j}$,
            for all $i,j\in\setexp{1,\ldots,n}$ with $i\neq j$.
      \end{enumerate}
    Then we say 
      that a \onechart~$\aonechartacc$ results by the \emph{multi-step loop elimination of $\aoneloopi{1}$, \ldots, $\aoneloopi{n}$ from $\aonechart$}
      if $\aonechartacc$ arises from $\aonechart$ 
        by removing all \loopentrytransitions\ in $\asettranssi{1},\ldots,\asettranssi{n}$ of $\aoneloopi{1}$, \ldots, $\aoneloopi{n}$ from $\aonechart$, 
        and then also removing all vertices and transitions that become unreachable;
          in this case we write $\aonechart \loopelimllred \aonechartacc$,
            and say that 
              $\aonechartacc$ results by a \emph{multi-step loop elimination} from $\aonechart$.
\end{defi}

\begin{lem}\label{lem:loopelimred:loopelimllred}
  ${\sloopelimred}  \,\subseteq\,   {\sloopelimllred}   \,\subseteq\,   {\sloopelimredtc}$,
    and consequently \mbox{}
      ${\sloopelimredrtc} = {\sloopelimllredrtc}$.
\end{lem}
\begin{proof}
  First, ${\sloopelimred}  \subseteq {\sloopelimllred}$ holds because
    every single-step loop elimination is also a multi-step loop elimination.
  Crucially, ${\sloopelimllred} \subseteq {\sloopelimredtc}$ holds, because
    every multi-step loop elimination of loop \subonecharts~$\aoneloopi{1},\ldots,\aoneloopi{n}$ in a \onechart~$\aonechart$
      with \loopentry\ transitions $\asettranssi{1},\ldots,\asettranssi{n}$ 
      can be implemented as a sequence of single-step loop eliminations of $\aoneloopi{1},\ldots,\aoneloopi{n}$ irrespective of the chosen order:
        hereby \ref{msone} guarantees that every \loopentry\ transition belongs uniquely to one of $\aoneloopi{1},\ldots,\aoneloopi{n}$
          and thus is removed in precisely one step;
        and \ref{mstwo} ensures that, after the elimination of a loop \subonechart~$\aoneloopi{i}$,
          another one $\aoneloopi{j}$ with $j\neq i$ that has not yet been eliminated is still a loop \subonechart.
  Finally these statements imply that the many-step versions of single-step and multi-step loop elimination coincide.          
\end{proof}

\begin{defi}[LEE and LLEE]\label{def:LEE}
  Let $\aonechart = \tuple{\verts,\actions,\sone,\start,\transs,\termexts}$ be a \onechart.
   
  We say that $\aonechart$ has the \emph{loop existence and elimination property} (\LEE)
    if repeated loop elimination started on~$\aonechart$ leads to a \onechart\ without an infinite path,
      that is, if there is multi-step loop elimination reduction sequence $\aonechart \loopelimllredrtc \aonechartacc$ 
                                                                          (or by Lemma~\ref{lem:loopelimred:loopelimllred} equivalently,
                                                                             a single-step loop elimination reduction sequence
                                                                             $\aonechart \loopelimredrtc \aonechartacc$)
        that leads to a \onechart~$\aonechartacc$ without an infinite path.
    
  If, in a successful elimination sequence from a \onechart~$\aonechart$,
    loop-entry transitions are never removed that depart from a vertex in the body of a previously eliminated loop \subonechart,
    then we say that $\aonechart$ satisfies \emph{layered \LEE} (\LLEE),
    and that $\aonechart$ is a \emph{\LLEEonechart}.
\end{defi}

  \input{figs/fig-ex-LEE-rev.tex}%
  \input{figs/fig-nonex-ex-LLEEw}%
\begin{exa}\label{exa:LEE}
  In Figure~\ref{fig:ex:LEE} we have illustrated a successful run of the loop elimination procedure for the \onechart~$\aonechart$ there. 
  The \loopentry\ transitions of loop \subonecharts\ that are eliminated in the next step, respectively, are marked in bold. 
  We have neglected action labels there, except for indicating \onetransitions\ by dotted arrows.
  Since the graph $\aonechart'''$ that is reached
    by three loop-subgraph elimination steps $\aonechart \loopelimredtc \aonechart'''$ from the \onechart\ $\aonechart$ does not have 
  an infinite path, and since no \loopentry\ transitions have been removed from a previously eliminated loop \subonechart,
  we conclude that $\aonechart$ satisfies \LEE\ and \LLEE. 

  In Figure~\ref{fig:nonex:ex:LLEEw} we illustrate two runs of the loop elimination procedure from a \protect\onechart~$\conechart$:
  The one from $\conechart$ to the left only witnesses \LEE\ but not \LLEE, 
  since in the second elimination step
    a \loopentry\ transition (drawn \textcolor{red}{red}) is removed from the body of the loop \subonechart\ that is eliminated in the first step
    (drawn in \forestgreen{green}).
  The one from $\conechart$ to the right witnesses \LLEE, because transitions are only removed sequentially at the same vertex,
    and hence no \loopentry\ transition is removed from the body of a loop \onechart\ that was eliminated~before. 
    
  The two process graphs $G_1$ and $G_2$ on page~\pageref{fig:milner-bosscher-expressible}, which are not expressible by star expressions modulo bisimilarity,
    do not satisfy \LLEE\ nor \LEE:
      neither of them has a loop subchart (as argued in Example~\ref{exa:loop:chart}),
        yet both of them facilitate infinite paths. 
\end{exa}

\begin{defi}[\protect\LLEEwitness]\label{def:LLEEwitness}
  Let $\aonechart = \tuple{\verts,\actions,\sone,\start,\transs,\exts}$ be a \onechart.\nopagebreak[4]
  
  By an \emph{\entrybodylabeling} of $\aonechart$ 
    we mean a \onechart~$\aonecharthat = \tuple{\verts,\actions\times\nat,\sone,\start,\transshat,\termexts}$ with actions in $\actions\times\nat$
  that results from $\aonechart$ by attaching to every transition of $\aonechart$ an additional \emph{marking label} in $\nat$
    (the transitions in $\transshat$ are marking-labeled versions of the transitions in $\transs$). 
   
  A \emph{\LLEEwitness\ $\aonecharthat$ of} a \onechart~$\aonechart$
    is an \entrybodylabeling\ of $\aonechart$
  that is the recording of a \LLEE-guaranteeing, 
                             successful run 
                             $\aonechart \loopelimllredrtc \aonechartacc$
                             of the multi-step loop elimination procedure on $\aonechart$
  that results by attaching to a transition $\atrans$ of $\aonechart$ the marking label $n$ for $n\in\natplus$ 
   (in pictures indicated as $\looplab{n}$, in steps as $\sredi{\looplab{n}}$) 
   forming a \emph{\loopentry\ transition}
  if $\atrans$ is eliminated in the $n$\nb-th multi-step,
  and by attaching marking label~$0$ to all other transitions of $\aonechart$
   (in pictures neglected, in steps indicated as $\sredi{\bodylab}$)
   forming a \emph{body transition}. 
   
  We say that a \LLEEwitness~$\aonecharthat$ of a \onechart~$\aonechart$ \emph{is guarded}
    if the action labels of the \loopentrytransitions\ of $\aonecharthat$ are proper (different from $\sone$).  
  We say that a \LLEEonechart~$\aonechart$ is \emph{guarded} if $\aonechart$ has a guarded \LLEEwitness.
\end{defi}

The definition above of guardedness for \LLEEwitnesses\ 
  is justified in view of the fact that \loopentrytransitions\ divide infinite paths in \LLEEwitnesses\ into finite segments 
    that consist only of \bodytransitions\ with perhaps a leading \loopentrytransition.
This is a consequence of the fact that \LLEEwitnesses\ do not permit infinite paths of \bodytransitions\ 
    (see Lemma~\ref{lem:descsteps:bodysteps:wf},~%
                   \ref{it:2:lem:descsteps:bodysteps:wf}).                 
Therefore guarded \LLEEwitnesses, in which \loopentrytransitions\ must be proper,
  do not permit infinite paths of \onetransitions.
It also follows that the underlying \mbox{(\LLEE-)\onechart} of a guarded \LLEEwitness\ is weakly guarded.

  \begin{figure}[t]
\begin{center}
\begin{tikzpicture}
  
\matrix[anchor=center,row sep=1cm,column sep=0.75cm,ampersand replacement=\&,
        every node/.style={draw,very thick,circle,minimum width=2.5pt,fill,inner sep=0pt,outer sep=2pt}] at (0,0) {
    \&\ \node[chocolate](v){};
  \\[-0.25ex]              
  \node(v1){};
    \&\ \&\ \node(v2){};
  \\[0.25cm]
  \node(v11){};
    \&\ \&\ \node(v21){};
  \\   
  };   
\path (v) ++ (-1.15cm,0.15cm) node(label){\LARGE $\aonecharthatjp{1}$};

\draw[<-,very thick,>=latex,chocolate,shorten <= 2pt](v) -- ++ (90:0.55cm);   
\draw[thick,chocolate] (v) circle (0.12cm);
\path (v) ++ (0.25cm,-0.3cm) node{\small $\avert$};
\draw[->,thick,darkcyan,shorten >= 0.175cm,shorten <= 2pt] 
  (v) to 
         node[right,pos=0.4,xshift=-0.075cm,yshift=1pt]{\small $\loopnsteplab{3}$}  (v11);
\draw[->,thick,darkcyan,shorten >= 0.175cm,shorten <= 2pt] 
  (v) to 
         node[left,pos=0.6,xshift=0.075cm,yshift=1pt]{\small $\loopnsteplab{3}$} (v21);

\path (v1) ++ (-0.225cm,0.25cm) node{\small $\averti{1}$};
\draw[->,thick,darkcyan,shorten >= 0pt]
  (v1) to 
          node[left,pos=0.6,xshift=0.075cm]{\small $\loopnsteplab{1}$} (v11);
\draw[->,thick,densely dotted,out=90,in=180,distance=0.5cm,shorten >=2pt](v1) to (v);
\draw[->,shorten <= 0pt,shorten >= 0pt] (v1) to 
                                                (v21); 

\path (v11) ++ (0cm,-0.25cm) node{\small $\averti{11}$};
\draw[->,thick,densely dotted,out=180,in=180,distance=0.75cm](v11) to (v1);

\path (v2) ++ (0.25cm,0.25cm) node{\small $\averti{2}$};
\draw[->,thick,darkcyan,shorten >= 0pt]
  (v2) to 
          node[right,pos=0.6,xshift=-0.075cm]{\small $\loopnsteplab{2}$} (v21);
\draw[->,thick,densely dotted,out=90,in=0,distance=0.5cm,shorten >= 2pt](v2) to (v);

\path (v21) ++ (0cm,-0.25cm) node{\small $\averti{21}$};
\draw[->,thick,densely dotted,out=0,in=0,distance=0.75cm](v21) to (v2);

\matrix[anchor=center,row sep=1cm,column sep=0.75cm,ampersand replacement=\&,
        every node/.style={draw,very thick,circle,minimum width=2.5pt,fill,inner sep=0pt,outer sep=2pt}] at (4.5,0) {
    \&\ \node[chocolate](v){};
  \\[-0.25ex]              
  \node(v1){};
    \&\ \&\ \node(v2){};
  \\[0.25cm]
  \node(v11){};
    \&\ \&\ \node(v21){};
  \\   
  };   
\path (v) ++ (-1.15cm,0.15cm) node(label){\LARGE $\aonecharthatjp{2}$};

\draw[<-,very thick,>=latex,chocolate,shorten <= 2pt](v) -- ++ (90:0.55cm);   
\draw[thick,chocolate] (v) circle (0.12cm);
\path (v) ++ (0.25cm,-0.3cm) node{\small $\avert$};
\draw[->,thick,darkcyan,shorten >= 0.175cm,shorten <= 2pt] 
  (v) to 
         node[right,pos=0.4,xshift=-0.075cm,yshift=1pt]{\small $\loopnsteplab{4}$}  (v11);
\draw[->,thick,darkcyan,shorten >= 0.175cm,shorten <= 2pt] 
  (v) to 
         node[left,pos=0.6,xshift=0.075cm,yshift=1pt]{\small $\loopnsteplab{3}$} (v21);

\path (v1) ++ (-0.225cm,0.25cm) node{\small $\averti{1}$};
\draw[->,thick,darkcyan,shorten >= 0pt]
  (v1) to 
          node[left,pos=0.6,xshift=0.075cm]{\small $\loopnsteplab{2}$} (v11);
\draw[->,thick,densely dotted,out=90,in=180,distance=0.5cm,shorten >=2pt](v1) to (v);
\draw[->,shorten <= 0pt,shorten >= 0pt] (v1) to 
                                                (v21); 

\path (v11) ++ (0cm,-0.25cm) node{\small $\averti{11}$};
\draw[->,thick,densely dotted,out=180,in=180,distance=0.75cm](v11) to (v1);

\path (v2) ++ (0.25cm,0.25cm) node{\small $\averti{2}$};
\draw[->,thick,darkcyan,shorten >= 0pt]
  (v2) to 
          node[right,pos=0.6,xshift=-0.075cm]{\small $\loopnsteplab{1}$} (v21);
\draw[->,thick,densely dotted,out=90,in=0,distance=0.5cm,shorten >= 2pt](v2) to (v);

\path (v21) ++ (0cm,-0.25cm) node{\small $\averti{21}$};
\draw[->,thick,densely dotted,out=-0,in=0,distance=0.75cm](v21) to (v2);

\matrix[anchor=center,row sep=1cm,column sep=0.75cm,ampersand replacement=\&,
        every node/.style={draw,very thick,circle,minimum width=2.5pt,fill,inner sep=0pt,outer sep=2pt}] at (9,0) {
    \&\ \node[chocolate](v){};
  \\[-0.25ex]              
  \node(v1){};
    \&\ \&\ \node(v2){};
  \\[0.25cm]
  \node(v11){};
    \&\ \&\ \node(v21){};
  \\   
  };   
\path (v) ++ (-1.15cm,0.15cm) node(label){\LARGE $\aonecharthatjp{3}$};

\draw[<-,very thick,>=latex,chocolate,shorten <= 2pt](v) -- ++ (90:0.55cm);   
\draw[thick,chocolate] (v) circle (0.12cm);
\path (v) ++ (0.25cm,-0.3cm) node{\small $\avert$};
\draw[->,thick,darkcyan,shorten >= 0.175cm,shorten <= 2pt] 
  (v) to 
         node[right,pos=0.4,xshift=-0.075cm,yshift=1pt]{\small $\loopnsteplab{2}$}  (v11);
\draw[->,thick,darkcyan,shorten >= 0.175cm,shorten <= 2pt] 
  (v) to 
         node[left,pos=0.6,xshift=0.075cm,yshift=1pt]{\small $\loopnsteplab{2}$} (v21);

\path (v1) ++ (-0.225cm,0.25cm) node{\small $\averti{1}$};
\draw[->,thick,darkcyan,shorten >= 0pt]
  (v1) to 
          node[left,pos=0.6,xshift=0.075cm]{\small $\loopnsteplab{1}$} (v11);
\draw[->,thick,densely dotted,out=90,in=180,distance=0.5cm,shorten >=2pt](v1) to (v);
\draw[->,shorten <= 0pt,shorten >= 0pt] (v1) to 
                                                (v21); 

\path (v11) ++ (0cm,-0.25cm) node{\small $\averti{11}$};
\draw[->,thick,densely dotted,out=180,in=180,distance=0.75cm](v11) to (v1);

\path (v2) ++ (0.25cm,0.25cm) node{\small $\averti{2}$};
\draw[->,thick,darkcyan,shorten >= 0pt]
  (v2) to 
          node[right,pos=0.6,xshift=-0.075cm]{\small $\loopnsteplab{1}$} (v21);
\draw[->,thick,densely dotted,out=90,in=0,distance=0.5cm,shorten >= 2pt](v2) to (v);

\path (v21) ++ (0cm,-0.25cm) node{\small $\averti{21}$};
\draw[->,thick,densely dotted,out=-0,in=0,distance=0.75cm](v21) to (v2);  
  
\end{tikzpicture}
\end{center}  
  \vspace*{-1.5ex}
  \caption{\label{fig:ex:LLEEw-1-2-3} 
           Three \protect\LLEEwitnesses\ of the \protect\onechart~$\aonechart$ in Figure~\protect\ref{fig:ex:LEE},
             of which $\protect\aonecharthatjp{1}$ and $\protect\aonecharthatjp{3}$ 
                  are the recordings of the successful single-step and multi-step runs of the loop elimination procedure in Figure~\protect\ref{fig:ex:LEE},
                  respectively.}
\end{figure}  %
  \begin{figure}[t]
\begin{center}
  \begin{center}\vspace{-0.75ex}\label{non-ex-LLEEw}%
    \scalebox{0.95}{\begin{tikzpicture}  
\matrix[anchor=center,row sep=0.6cm,column sep=1.15cm,ampersand replacement=\&,
        every node/.style={draw,very thick,circle,minimum width=2.5pt,fill,inner sep=0pt,outer sep=2pt}] at (0,0) {
  \node(v){}; 
  \\
               \& \node(u){};
  \\
  \node(w1){};
  \\
  \\
  \node(w2){}; 
  \\   
  };     
   
\path (v) ++ (0.25cm,0.15cm) node{$\avert$}; 
\draw[<-,very thick,>=latex,chocolate](v) -- ++ (90:0.5cm);
\draw[->] (v) to (u);
\draw[->] (v) to (w1);
\path (v) ++ (-0.7cm,0.45cm) node(label){\LARGE $\conecharti{1}$};

\path (u) ++ (-0.05cm,-0.25cm) node{$\cvert$}; 
\draw[->] (u) to (w1);
\draw[->,out=-50,in=50,distance=1.2cm] (u) to (u);
\path (w1) ++ (-0.325cm,0cm) node{$\bverti{1}$};
\draw[->] (w1) to (w2);
\path (w2) ++ (0.35cm,-0.15cm) node{$\bverti{2}$};
\draw[->,out=180,in=180,distance=1.85cm] (w2) to (v);

\matrix[anchor=center,row sep=0.6cm,column sep=1.15cm,ampersand replacement=\&,
        every node/.style={draw,very thick,circle,minimum width=2.5pt,fill,inner sep=0pt,outer sep=2pt}] at (4.5,0) {
  \node(v){}; 
  \\
               \& \node(u){};
  \\
  \node(w1){};
  \\
  \\
  \node(w2){}; 
  \\   
  };     
   
\path (v) ++ (0.25cm,0.15cm) node{$\avert$}; 
\draw[<-,very thick,>=latex,chocolate](v) -- ++ (90:0.5cm);
\draw[->] (v) to (u);
\draw[->,very thick,darkcyan] (v) to node[right,xshift=-0.075cm,yshift=0.1cm]{$\loopnsteplab{1}$} (w1);
\path (v) ++ (-0.7cm,0.5cm) node(label){\LARGE $\conecharthatijp{1}{1}$};
\path (u) ++ (-0.05cm,-0.25cm) node{$\cvert$}; 
\draw[->] (u) to (w1);
\draw[->,very thick,darkcyan,out=-50,in=50,distance=1.2cm] (u) to node[right,xshift=-0.075cm]{$\loopnsteplab{2}$} (u);
\path (w1) ++ (-0.325cm,0cm) node{$\bverti{1}$};
\draw[->,very thick,darkcyan] (w1) to node[right,xshift=-0.075cm,yshift=0.1cm]{$\loopnsteplab{3}$} (w2);
\path (w2) ++ (0.35cm,-0.15cm) node{$\bverti{2}$};
\draw[->,out=180,in=180,distance=1.85cm] (w2) to (v);

\matrix[anchor=center,row sep=0.6cm,column sep=1.15cm,ampersand replacement=\&,
        every node/.style={draw,very thick,circle,minimum width=2.5pt,fill,inner sep=0pt,outer sep=2pt}] at (9,0) {
  \node(v){}; 
  \\
               \& \node(u){};
  \\
  \node(w1){};
  \\
  \\
  \node(w2){}; 
  \\   
  };     
   
\path (v) ++ (0.25cm,0.15cm) node{$\avert$}; 
\draw[<-,very thick,>=latex,chocolate](v) -- ++ (90:0.5cm);
\draw[->,very thick,darkcyan] (v) to node[above,pos=0.6,xshift=0.025cm,yshift=0.025cm]{$\loopsteplab{3}$} (u);
\draw[->,very thick,darkcyan] (v) to node[right,xshift=-0.075cm,yshift=0.1cm]{$\loopnsteplab{1}$} (w1);
\path (v) ++ (-0.7cm,0.5cm) node(label){\LARGE $\conecharthatijp{1}{2}$};
\path (u) ++ (-0.05cm,-0.25cm) node{$\cvert$}; 
\draw[->] (u) to (w1);
\draw[->,very thick,darkcyan,out=-50,in=50,distance=1.2cm] (u) to node[right,xshift=-0.075cm]{$\loopsteplab{2}$} (u);
\path (w1) ++ (-0.325cm,0cm) node{$\bverti{1}$};
\draw[->] (w1) to (w2);

\path (w2) ++ (0.35cm,-0.15cm) node{$\bverti{2}$};
\draw[->,out=180,in=180,distance=1.85cm] (w2) to (v);

\end{tikzpicture}  }\vspace*{-1.5ex}
  \end{center}
\end{center}
  \caption{\label{fig-non-ex-LLEEw}
           A \protect\LEEwitness\ that is not layered (in the middle),
             and a \protect\LLEEwitness\ (right) for a variation $\conecharti{1}$ 
             of the \protect\LLEEonechart~$\conechart$ in Figure~\ref{fig:nonex:ex:LLEEw}.}
\end{figure}
\begin{exa}  
    The \entrybodylabelings\ $\aonecharthati{1}$ and $\aonecharthati{3}$ in Figure~\ref{fig:ex:LLEEw-1-2-3} of the \onechart~$\aonechart$ from Figure~\ref{fig:ex:LEE}
      are \LLEEwitnesses\ that arise from the successful runs of the loop elimination procedure in Example~\ref{def:LLEEwitness}:
        $\aonecharthati{1}$ is the recording on $\aonechart$ of the three single-step loop eliminations 
          (viewed as trivial multi-steps in order to apply the clause for a \LLEEwitness\ in Definition~\ref{def:LLEEwitness})
        that lead to $\aonechart'''$,
        and $\aonecharthati{3}$ is the recording on $\aonechart$ of the two multi-step loop eliminations from $\aonechart$ to $\aonecharthat'''$.       
    The \entrybodylabeling\ $\aonecharthati{2}$ in Figure~\ref{fig:ex:LLEEw-1-2-3}
      is another \LLEEwitness\ of $\aonechart$ that records the successful process of four elimination steps of four loop \subonecharts\
        each of which is induced by only a single \loopentrytransition.
    The \onechart~$\aonechart$ in Figure~\ref{fig:ex:LEE}
      has a property that does not hold in general: 
        $\aonechart$ only admits layered \LEEwitnesses.

  Indeed, this does not hold for the \onechart~$\conechart$ in Figure~\ref{fig:nonex:ex:LLEEw}:
    the \entrybodylabeling~$\conecharthatjp{1}$ is not a layered \LLEEwitness,
      because it arises from a run of the loop elimination process in which in the second step a \loopentry\ transition\vspace{-1.5pt} 
        is eliminated from the body of a loop \subonechart\ that was eliminated in the first step.
    But the \entrybodylabeling~$\conecharthatjp{2}$ there is a layered \LLEEwitness\ of $\conechart$. 
  The situation is analogous for the two \entrybodylabelings\vspace{-1pt}  $\conecharthatijp{1}{1}$ and $\conecharthatijp{1}{2}$
    of the slightly more involved \LLEEonechart~$\conecharti{1}$ in Figure~\ref{fig-non-ex-LLEEw},\vspace{-1pt}
      where $\conecharthatijp{1}{2}$ is a layered \LLEEwitness\ of $\conecharti{1}$, 
      but $\conecharthatijp{1}{1}$ is a \LEEwitness\ of $\conechart$ that is not layered. 
\end{exa}

\begin{rem}[from \protect\LEEwitness{es} to \protect\LLEEwitness{es}]\label{rem:LLEEwitnesses}       
  It can be shown that every \LEEwitness\ that is not layered can be transformed 
                                                                into a layered \LEEwitness\vspace{-1.5pt} (\LLEEwitness) of the same underlying \onechart.    
  Indeed, the step from the (not layered) \LEEwitness~$\conecharthatijp{1}{1}$\vspace{-2pt}  
                                            to the \LLEEwitness~$\conecharthatijp{1}{2}$ in the example in Figure~\ref{fig-non-ex-LLEEw},
    which transfers the \loopentrytransition\ marking label $\loopsteplab{3}$
    from the transition from $\bverti{1}$ to $\bverti{2}$ over to the transition from $\avert$ to $\cvert$,
  hints at the proof of this statement.
  However, we do not need this result, because
    we will be able to use the guaranteed existence of \LLEEwitnesses\ (see Theorem~\ref{thm:onechart-int:LLEEw})
    for the \onechart\ interpretation below (see~Definition~\ref{def:onechartof}).
\end{rem}

For the proofs in Section~\ref{coindmilnersys:2:milnersys} we will need 
  the `\txtdescendsinloopto' relation $\sdescendsinloopto$ as defined below,
    and the fact that it constitutes a `descent' in a \LLEEwitness.
The latter is expressed by the subsequent lemma
  together with termination of the body-step relation $\sredi{\bodylab}$. 
Both of these properties can be established by arguing with the successful runs of the loop \subonechart\ elimination procedure
  that underlies a \LLEEwitness.

\begin{defi}  
  Let $\aonecharthat$ be a \LLEEwitness\ of a \onechart~$\aonechart = \tuple{\verts,\actions,\sone,\start,\transs,\termexts}$. 
  
  Let $\avert,\bvert\in\verts$.
  We denote by $\avert \descendsinloopto \bvert$, and by $\bvert \convdescendsinloopto \avert$, 
    and say that $\avert$ \emph{descends in a loop to} $\bvert$, 
    if $\bvert$ is in the body of the loop \subonechart\ at $\avert$,
    which means that there is a path $\avert \redi{\looplab{n}} \avertacc \redrtci{\bodylab} \bvert$ 
      from $\avert$ via a \loopentrytransition\ and subsequent body transitions \ul{without} encountering $\avert$ again. 
  
\end{defi}

\begin{lem}\label{lem:descsteps:bodysteps:wf}
  The relations $\sdescendsinloopto$ and $\sredi{\bodylab}$ defined by a \LLEEwitness~$\aonecharthat$
  of a \onechart~$\aonechart$ satisfy: 
  \begin{enumerate}[(i)]
    \item{}\label{it:1:lem:descsteps:bodysteps:wf}
      $\sconvdescendsinlooptotc$ is a \wellfounded, strict partial order on $\verts$.
    \item{}\label{it:2:lem:descsteps:bodysteps:wf}
      $\sconvredtci{\bodylab}$ is a \wellfounded\ strict partial order on $\verts$.  
  \end{enumerate}
\end{lem}

\begin{proof}
  Well-foundedness and irreflexivity of each of $\sconvdescendsinlooptotc$ and $\sconvredtci{\bodylab}$
    follows from termination of $\sdescendsinloopto$ and $\sredi{\bodylab}$, respectively.
  These termination properties can be established in the same way
    as for \LLEEcharts\ without \onetransitions,
      for which they follow immediately from Lemma~5.2 in \cite{grab:fokk:2020:lics,grab:fokk:2020:lics:arxiv}.
  Since $\sconvdescendsinlooptotc$ and $\sconvredtci{\bodylab}$ are transitive by definition,
    it follows that both relations are \wellfounded\ strict partial orders.
\end{proof}

While chart interpretations of `\onefree' star expressions always satisfy \LEE, see \cite{grab:fokk:2020:lics:arxiv}),
  we observed in \cite{grab:2021:TERMGRAPH2020-postproceedings}
    that this is not true for the chart interpretations of star expressions in general.   
As a remedy for this failure of \LEE\ for chart interpretations, 
  we introduced `\onechart\ interpretations' of star expressions
    \cite{grab:2021:TERMGRAPH2020-postproceedings}.    
For such \onechart\ interpretations we showed that \LEE\ is guaranteed,
  and that they refine chart interpretations in the sense that there always is a functional \onebisimulation\ from the \onechart\ interpretation of a star expression
    to its chart interpretation (see Theorem~\ref{thm:onechart-int:LLEEw} below).
For the definition of \onechart\ interpretations we extended the syntax of star expressions
  to obtain `stacked star expressions', see the definition below. 
The intuition behind the use of the `stack product' symbol $\sstackprod$ 
  is to keep track of when a transition has descended into the body of an iteration expression
    such that the iteration can be interpreted as a loop \subonechart\ or a tower of nested (and possibly partially overlapping) loop \subonecharts. 
This feature of $\sstackprod$ makes a transition system specification possible (see Definition~\ref{def:onechartof})
    which introduces \onetransitions\ only as `backlinks'
      that lead from the body (the internal vertices) of some loop \subonechart\ $\aoneloop$ of a \onechart\ interpretation back to the start vertex of $\aoneloop$.

\begin{defi}[stacked star expressions]
            \nf\label{def:stackStExp}
  Let $\actions$ be a set whose members we call \emph{actions}.
  The set $\stackStExpover{\actions}$ of \emph{stacked star expressions over (actions in) $\actions$} is defined by the grammar:
  \begin{center}
    $
    \asstexp
      \;\;\BNFdefdby\;\;
        \astexp
          \BNFor
        \stexpprod{\asstexp}{\astexp}
          \BNFor
        \stexpstackprod{\asstexp}{\stexpit{\astexp}}
          \qquad\text{(where $\astexp\in\StExpover{\actions}$)} \punc{.}  
          $
  \end{center}
  Note that the set $\StExpover{\actions}$ of star expressions would arise if the clause $\stexpstackprod{\asstexp}{\stexpit{\astexp}}$ were dropped.
  
  The \emph{star height} $\sth{\asstexp}$ of stacked star expressions $\asstexp$ is defined by adding the two clauses
  $\sth{ \stexpprod{\asstexp}{\astexp} } \defdby \max \setexp{ \sth{\asstexp},\, \sth{\astexp} }$,
  and
  $\sth{ \stexpstackprod{\asstexp}{\stexpit{\astexp}} } \defdby \max \setexp{ \sth{\asstexp},\, \sth{\stexpit{\astexp}} }$
  to the definition of the star height of star expressions.
 
  The \emph{projection function} $\sproj \funin \stackStExpover{\actions} \to \StExpover{\actions}$  
  is defined by interpreting $\sstexpstackprod$ as~$\sstexpprod$ by the clauses:
  $\proj{\stexpprod{\asstexp}{\astexp}} \defdby \stexpprod{\proj{\asstexp}}{\astexp}$, \mbox{}
  $\proj{\stexpstackprod{\asstexp}{\stexpit{\astexp}}} \defdby \stexpprod{\proj{\asstexp}}{\stexpit{\astexp}}$, \mbox{}
  and $\proj{\astexp} \defdby \astexp$,
  for all stacked star expressions $\asstexp\in\stackStExpover{\actions}$, and star expressions $\astexp\in\StExpover{\actions}$.
\end{defi}

In line with \cite{grab:2021:TERMGRAPH2020-postproceedings} we introduce the \onechart\ interpretation of a star expression~$\astexp$
  with notation $\onechartof{\astexp}$ as `the \onechart\ induced by (the process interpretation of) $\astexp$'.  
For understanding the TSS in its definition below
  it is key to note that, by the rules for iterations, the stacked product operation $\sstackprod$ 
    helps to record that an expression has descended from the iteration expression on the right-hand side of $\sstackprod$. 
This feature is used by the rule for stacked product to introduce \onetransitions\ only as \backlinks\ to expressions from which they have descended.  
The rules for iteration expressions define \loopentrytransitions\ and \bodytransitions, respectively,
  dependent on whether $e$ is `strongly normed' (symbolically denoted by $\normedplus{e}$)
    in the sense of facilitating a process trace to termination,
      and hence dependent on whether an iteration induces a loop \subonechart\ (outside of inner loop~\subonecharts).

\begin{defi}[$\protect\sone$-chart interpretation of star expressions]\label{def:onechartof}
  By the \emph{\onechart\ interpretation} $\onechartof{\astexp}$ of a star expression $\astexp$
  we mean the \onechart\ that arises together with the \entrybodylabeling\ $\onecharthatof{\astexp}$ 
  as the $\astexp$\nb-rooted labeled transition system with \onetransitions\ (\oneLTS) generated by $\setexp{\astexp}$
  according to the following TSS on the set $\stackStExp$ of stacked star expressions, 
  where $\darkcyan{\alab} \in \setexp{\bodylabcol} \cup \descsetexp{ \darkcyan{\loopnsteplab{\aLname}} }{ \aLname\in\natplus }$
    are marking labels:
  \begin{center}
    $
    \begin{gathered}
      \begin{aligned}
        & 
        \AxiomC{$\phantom{a_i \:\lti{a_i}{\darkcyan{\bodylabcol}}\: \stexpone}$}
        \UnaryInfC{$a \:\lti{a}{\bodylabcol}\: \stexpone$}
        \DisplayProof
        & & 
        \AxiomC{$ \astexpi{i} \:\lti{\aact}{\darkcyan{\alab}}\: \asstexpacci{i} $}
        \AxiomC{\scriptsize ($i\in\setexp{1,2}$)}
        \insertBetweenHyps{\hspace*{-0ex}}
        \BinaryInfC{$ \stexpsum{\astexpi{1}}{\astexpi{2}} \:\lti{\aact}{\bodylabcol}\: \asstexpacci{i} $}
        \DisplayProof 
        & &
        \AxiomC{$   \phantom{\astexpi{1}}
                 \astexp \:\lti{a}{\darkcyan{\alab}}\: \asstexpacc 
                   \phantom{\asstexpacci{1}} $}
        \AxiomC{\scriptsize (if $\fap{\textit{nd$^+$}}{\astexp}$)}   
        \insertBetweenHyps{\hspace*{-2ex}}        
        \BinaryInfC{$\stexpit{\astexp} \:\lti{\aact}{\darkcyan{\loopnsteplab{\sth{\stexpit{\astexp}}}}}\: \stexpstackprod{\asstexpacc}{\stexpit{\astexp}}$}
        \DisplayProof
        & &
        \AxiomC{$   \phantom{\astexpi{1}}
                 \astexp \:\lti{a}{\darkcyan{\alab}}\: \asstexpacc 
                   \phantom{\astexpacci{1}} $}
        \insertBetweenHyps{\hspace*{-2ex}}     
        \AxiomC{\scriptsize (if $\lognot{\fap{\textit{nd$^+$}}{\astexp}}$)} 
        \BinaryInfC{$\stexpit{\astexp} \:\lti{a}{\bodylabcol}\: \stexpstackprod{\asstexpacc}{\stexpit{\astexp}}$}
        \DisplayProof 
      \end{aligned}
      \\
      \begin{aligned}
        &   
        \AxiomC{$ \asstexpi{1} \:\lti{\alert{\aoneact}}{\darkcyan{\alab}}\: \asstexpacci{1} $}
        \UnaryInfC{$ \stexpprod{\asstexpi{1}}{\astexpi{2}} \:\lti{\alert{\aoneact}}{\darkcyan{\alab}}\: \stexpprod{\asstexpacci{1}}{\astexpi{2}} $}
        \DisplayProof
        & & 
        \AxiomC{$ \asstexpi{1} \:\lti{\alert{\aoneact}}{\darkcyan{\alab}}\: \asstexpacci{1} $}
        \UnaryInfC{$ \stexpstackprod{\asstexpi{1}}{\stexpit{\astexpi{2}}} \:\lti{\alert{\aoneact}}{\darkcyan{\alab}}\: \stexpstackprod{\asstexpacci{1}}{\stexpit{\astexpi{2}}} $}
        \DisplayProof
        & &
        \AxiomC{$\terminates{\astexpi{1}}$}
        \AxiomC{$ \astexpi{2} \:\lti{a}{\darkcyan{\alab}}\: \asstexpacci{2} $}
        \BinaryInfC{$ \stexpprod{\astexpi{1}}{\astexpi{2}} \:\lti{a}{\bodylabcol}\: \asstexpacci{2} $}
        \DisplayProof
        & &
        \AxiomC{$ \hspace*{5ex} \terminates{\astexpi{1}}\rule{0pt}{13pt} \hspace*{5ex} $}
        \UnaryInfC{$ \stexpstackprod{\astexpi{1}}{\stexpit{\astexpi{2}}} \:\lti{\sone}{\bodylabcol}\: \stexpit{\astexpi{2}} $}
        \DisplayProof
      \end{aligned}
    \end{gathered}
    $
  \end{center}
  The condition $\normedplus{e}$ means a strengthening of normedness, namely, that $e$ permits a positive length path to an expression $f$~with~$\terminates{f}\,$;
    it is definable by induction.
  Immediate termination for expressions of $\onechartof{\astexp}$ is defined by the same rules 
  as in Definition~\ref{def:StExpTSS} (for star expressions only, preventing immediate termination for expressions with stacked product~$\sstackprod$).
  We note that finiteness of $\onechartof{\astexp}$ as a \onechart\ is guaranteed by Theorem~\ref{thm:onechart-int:LLEEw}, \ref{it:2:thm:onechart-int:LLEEw}, below.
    
  We also extend the \onechart\ interpretation of star expressions in the obvious way to all stacked star expressions $E\in\stackStExp$:
  by $\onechartof{E}$ we mean the $\asstexp$\nb-rooted sub-\oneLTS\ generated by $\setexp{\asstexp}$
  in the \oneLTS\ generated by the TSS above.
\end{defi}

\begin{defi}
  For every stacked star expression $\asstexp\in\stackStExpover{\actions}$, we define 
  the set~$\oneactderivs{\asstexp}$ of \emph{action (partial) \onederivatives\ of} $\asstexp$,
    and the set~$\partonederivs{\asstexp}$ of \emph{(partial) \onederivatives\ of} $\asstexp$ by:
  \begin{align*}
    \oneactderivs{\asstexp} 
      & {} \defdby
        \descsetexpbig{ \pair{\aoneact}{\asstexpacc} }{ \asstexp \lt{\aoneact} \asstexpacc } 
          \subseteq
            \oneactions\times\stackStExpover{\actions} \punc{,}
    \\
    \partonederivs{\asstexp}
      & {} \defdby
        \descsetexpbig{ \asstexpacc }{ \asstexp \lt{\aoneact} \asstexpacc \text{ for some $\aoneact\in\oneactions$}}
          \subseteq
            \stackStExpover{\actions} \punc{,}
  \end{align*}
  where the transitions are defined by the TSS in Definition~\ref{def:onechartof}.   
\end{defi}

\begin{lem}\label{lem:actonederivs}
  The action \onederivatives\ $\oneactderivs{\asstexp}$ of a stacked star expression $\asstexp$ over actions in $\actions$
    satisfy the following recursive equations, for all $\aact\in\actions$, $\astexp,\astexpi{1},\astexpi{2}\in\StExpover{\actions}$,
    and stacked star expressions~$E_1$ over actions in $\actions$:
  \begin{align*}
    \oneactderivs{0} 
      & {} \defdby \oneactderivs{1} 
           \defdby \emptyset \punc{,}
    \\
    \oneactderivs{\aact}
       & {} \defdby \setexp{ \pair{\aact}{1} } \punc{,}
    \displaybreak[0]\\
    \oneactderivs{e_1 + e_2} 
      & {} \defdby  
             \oneactderivs{e_1} \cup \oneactderivs{e_2} \punc{,}
    \displaybreak[0]\\
    \oneactderivs{E_1 \prod e_2}
      & {} \defdby
           \begin{cases}
             \descsetexpbig{ \pair{\aoneact}{E'_1 \prod e_2} }{ \pair{\aoneact}{E'_1}\in\oneactderivs{E_1} }
               & \text{if $\notterminates{E_1}$,}
             \\
             \descsetexpbig{ \pair{\aoneact}{E'_1 \prod e_2} }{ \pair{\aoneact}{E'_1}\in\oneactderivs{E_1} }
               \cup
             \oneactderivs{E_2} 
               & \text{if $\terminates{E_1}$,}  
           \end{cases}  
    \displaybreak[0]\\
    \oneactderivs{E_1 \stackprod e^*_2}
      & {} \defdby
           \begin{cases}
             \descsetexpbig{ \pair{\aoneact}{E'_1 \prod e^*_2} }{ \pair{\aoneact}{E'_1}\in\oneactderivs{E_1} }
               & \text{if $\notterminates{E_1}$,}
             \\
             \descsetexpbig{ \pair{\aoneact}{E'_1 \prod e^*_2} }{ \pair{\aoneact}{E'_1}\in\oneactderivs{E_1} }
               \cup
             \setexp{ \pair{\sone}{e^*_2} } 
               & \text{if $\terminates{E_1}$,}  
           \end{cases}  
    \\
    \oneactderivs{e^*}
      & {} \defdby
             \descsetexpbig{ \pair{\aact}{E' \stackprod e^*} }
                           { \pair{\aact}{E'} \in\oneactderivs{e} } \punc{.}  
  \end{align*}
\end{lem}

\begin{proof}
  By case-wise inspection of the definition of the TSS in Def.~\ref{def:onechartof}.
\end{proof}

\begin{thmC}[\protect\cite{grab:2020:TERMGRAPH2020-postproceedings:arxiv,grab:2021:TERMGRAPH2020-postproceedings}]\label{thm:onechart-int:LLEEw}
  For every $\astexp\in\StExpover{\actions}$, the following statements hold
    for the concepts as introduced in Definition~\ref{def:onechartof}: 
  \begin{enumerate}[label=(\roman{*}),leftmargin=*,align=right,itemsep=0.25ex]
    \item{}\label{it:1:thm:onechart-int:LLEEw}
      The \entrybodylabeling~$\onecharthatof{\astexp}$ of $\onechartof{\astexp}$ 
        is a guarded \LLEEwitness\ of $\onechartof{\astexp}$.
    \item{}\label{it:2:thm:onechart-int:LLEEw} 
      The projection function $\sproj$ defines a \onebisimulation\ from the \onechart\ interpretation $\onechartof{\astexp}$ of $\astexp$
        to the chart interpretation $\chartof{\astexp}$ of $\astexp$,
        that is symbolically, $\onechartof{\astexp} \funonebisimvia{\sproj} \chartof{\astexp}$, 
          and hence also $\onechartof{\astexp} \funonebisim \chartof{\astexp}$.
      Since the set of stacked star expressions that form the pre-image of a star expression under the projection function is always finite, 
        it follows that \onechart\ interpretations of star expressions are always finite as well.
  \end{enumerate}        
\end{thmC}

For the proof of Lemma~\ref{lem:onechart-int:milnersysmin:solvable} below 
  we will need the second of the two subsequent lemmas, Lemma~\ref{lem:FT:onechart-int}.
Its proof uses the first lemma, which crucially states that every star expression $\astexp$ with immediate termination
  can \provablyin{\milnersysmin} be written as a star expression $1 + \bstexp$ where $\bstexp$ does not permit immediate termination.

\begin{lem}\label{lem:terminates:2:notterminates}
  If $\terminates{e}$ for a star expression $e\in\StExpover{\actions}$,
    then there is a star expression $f\in\StExpover{\actions}$
    with $\notterminates{f}$, $e \milnersysmineq 1 + f$, $\sth{f} = \sth{e}$, 
    and $\compfuns{((\sidfun\times\sproj)}{\soneactderivs)}{f} = \compfuns{((\sidfun\times\sproj)}{\soneactderivs)}{e}$.
\end{lem}     
        
\begin{proof}
  By a proof by induction on structure of $e$,
    in which all axioms of $\milnersysmin$ are used.
\end{proof}

\begin{lem}\label{lem:FT:onechart-int}
  $ \proj{\asstexp}
     \milnersysmineq
   \terminatesconstof{\onechartof{\asstexp}}{\asstexp}  
     +  
   \sum_{i=1}^{n} \stexpprod{\aoneacti{i}}{\proj{\asstexpacci{i}}} $,
  given a list representation\vspace{-1pt}
  $\transitionsinfrom{\onechartof{\asstexp}}{\bvert}
     =
   \descsetexpbig{ \asstexp \lt{\aoneacti{i}} \asstexpacci{i} }{ i\in\setexp{1,\ldots,n} }$
  of the transitions from $\asstexp$ in $\onechartof{E}$.
\end{lem}

\begin{proof}
  We establish the lemma by induction on the star height $\sth{\asstexp}$ of $\asstexp$
    with a subinduction on the syntactical structure of $E$.
  All cases of stacked star expressions can be dealt with in a quite straightforward manner, 
    except for the case of star expressions with an outermost iteration. 
  There, an appeal to Lemma~\ref{lem:terminates:2:notterminates} is crucial.
  We treat this case in detail below. 
 
  \smallskip
  Suppose that $E \synteq e^*$ for some star expression $e\in\StExpover{\actions}$ (without occurrences of stacked product $\sstexpstackprod$).
  For showing the representation of $\proj{E}$ as stated by the lemma, we assume 
    that the transitions from $e$ in $\onechartof{e}$ as defined in Definition~\ref{def:onechartof} are as follows:
  \begin{equation}\label{eq:1:lem:FT:onechart-int}
    \transitionsinfrom{\onechartof{e}}{e}
      = \descsetexpbig{ e \lt{\aacti{i}} E'_i }{ i = 1, \ldots n } \punc{,}
  \end{equation}
  for some stacked star expression $E'_1,\ldots,E'_n$, which are \onederivatives\ of $e$. 
  Note that according to the TSS in Definition~\ref{def:onechartof} only proper transitions (those with proper action labels in $\actions$) 
    can depart from the star expression $e$ (which does not contain stacked products~$\sstexpstackprod$).
  Then it follows, again from the TSS in Definition~\ref{def:onechartof} that:
  \begin{equation}\label{eq:3:lem:FT:onechart-int}
    \transitionsinfrom{\onechartof{E}}{E} =
    \transitionsinfrom{\onechartof{e^*}}{e^*}
       = \descsetexpbig{ e^* \lt{\aacti{i}} E'_i \stackprod e^* }{ i = 1, \ldots n } \punc{.}
  \end{equation} 
  We assume now that $\terminates{e}$ holds. (We will see that if $\notterminates{e}$ holds, the argumentation below becomes easier).
  Then by Lemma~\ref{lem:terminates:2:notterminates} there is
    a star expression $f\in\StExpover{\actions}$ 
    with $\notterminates{f}$, and such that
  $ 1 + f \milnersysmineq e$, $\sth{f} = \sth{e}$, and 
  $\compfuns{((\sidfun\times\sproj)}{\soneactderivs)}{f} = \compfuns{((\sidfun\times\sproj)}{\soneactderivs)}{e}$ hold.
  From the latter it follows with~\eqref{eq:1:lem:FT:onechart-int}: 
  \begin{align}
    \transitionsinfrom{\onechartof{f}}{f}
      & {} = \descsetexpbig{ f \lt{\aacti{i}} F'_i }{ i = 1, \ldots n } \punc{,}
        \label{eq:2:lem:FT:onechart-int}
    \\    
    \proj{F'_i}
      & {} = \proj{E'_i} \qquad \text{(for all $i = 1, \ldots n$)} \punc{,}
        \label{eq:4:lem:FT:onechart-int}
  \end{align}
  for some stacked star expression $F'_1,\ldots,F'_n$, which are \onederivatives\ of $f$. 
  Note again that only proper transitions can depart from the star expression $f$ according to Definition~\ref{def:onechartof}.
  On the basis of these assumptions we can now argue as follows: 
  \begin{alignat*}{2}
    \proj{E}
      & \;\,\parbox[t]{\widthof{$\eqin{\milnersysmin}$}}{$\synteq$}\:
        \proj{e^*}
        & & \quad\text{(due to $E \synteq e^*$ in this case)} 
      \displaybreak[0]\\ 
      & \;\,\parbox[t]{\widthof{$\eqin{\milnersysmin}$}}{$\synteq$}\:
        e^*
        & & \quad\text{(since $e^*$ does not contain $\sstackprod$)}  
      \displaybreak[0]\\ 
      & \;\,\parbox[t]{\widthof{$\eqin{\milnersysmin}$}}{$\milnersysmineq$}\:
        (1 + f)^*
        & & \quad\text{(by the choice of $f$ with $1 + f \milnersysmineq e$)}   
      \displaybreak[0]\\ 
      & \;\,\parbox[t]{\widthof{$\eqin{\milnersysmin}$}}{$\milnersysmineq$}\:
        f^*
        & & \quad\text{(by axiom ($\termstexpit$))} 
      \displaybreak[0]\\ 
      & \;\,\parbox[t]{\widthof{$\eqin{\milnersysmin}$}}{$\milnersysmineq$}\:
        1 + f \prod f^*
        & & \quad\text{(by axiom ($\recdefstexpit$))}  
      \displaybreak[0]\\ 
      & \;\,\parbox[t]{\widthof{$\eqin{\milnersysmin}$}}{$\milnersysmineq$}\:
        1 + f \prod (1 + f)^*
        & & \quad\text{(by axiom ($\termstexpit$))} 
      \displaybreak[0]\\ 
      & \;\,\parbox[t]{\widthof{$\eqin{\milnersysmin}$}}{$\milnersysmineq$}\:
        1 + f \prod e^*
        & & \quad\text{(by the choice of $f$ with $1 + f \milnersysmineq e$)}  
      \displaybreak[0]\\ 
      & \;\,\parbox[t]{\widthof{$\eqin{\milnersysmin}$}}{$\milnersysmineq$}\:
        1 + \bigl( \terminatesconstof{\onechartof{f}}{f}
                     +
                   \sum_{i=1}^n
                     \aacti{i} \prod \proj{ F'_i } 
            \bigr) \prod e^*         
        & & \quad\parbox{\widthof{($\sth{f} = \sth{e} < \sth{e} + 1 = \sth{e^*} = \sth{E} $,)}}
                        {(by the induction hypothesis, due to
                         \\\phantom{(}%
                         $\sth{f} = \sth{e} < \sth{e} + 1 = \sth{e^*} = \sth{E}$,
                         \\\phantom{(}
                         and in view of \eqref{eq:2:lem:FT:onechart-int})}  
      \displaybreak[0]\\ 
      & \;\,\parbox[t]{\widthof{$\eqin{\milnersysmin}$}}{$\synteq$}\:
        1 + \bigl( 0
                     +
                   \sum_{i=1}^n
                     \aacti{i} \prod \proj{ E'_i } 
            \bigr) \prod e^*         
        & & \quad\parbox{\widthof{(by $\terminatesconstof{\onechartof{f}}{f} \synteq 0$ due to $\notterminates{f}$}}
                        {(by $\terminatesconstof{\onechartof{f}}{f} \synteq 0$ due to $\notterminates{f}$
                         \\
                         and by using \eqref{eq:4:lem:FT:onechart-int})}
      \displaybreak[0]\\ 
      & \;\,\parbox[t]{\widthof{$\eqin{\milnersysmin}$}}{$\milnersysmineq$}\:
        \terminatesconstof{\onechartof{e^*}}{e^*}
          + \sum_{i=1}^n
              \aacti{i} \prod ( \proj{ E'_i } \prod e^* ) 
        & & \quad\parbox{\widthof{(by axioms ($\neutralstexpsum$), ($\rdistr$),}}
                        {(by axioms ($\neutralstexpsum$), ($\rdistr$),
                         \\
                          ($\assocstexpprod$), $\terminatesconstof{\onechartof{e^*}}{e^*}\synteq 1$ due to $\terminates{(e^*)}$}  
      \displaybreak[0]\\ 
      & \;\,\parbox[t]{\widthof{$\eqin{\milnersysmin}$}}{$\synteq$}\:
        \terminatesconstof{\onechartof{E}}{E}
          + \sum_{i=1}^n
              \aacti{i} \prod  \proj{ E'_i \stackprod e^* }  
        & & \quad\text{(by the definition of the projection $\sproj$).}  
  \end{alignat*}
  In view of \eqref{eq:3:lem:FT:onechart-int}, 
  this chain of \provablein{\milnersysmin} equalities verifies the statement in the lemma
    in this case $E = e^*$ with $\terminates{e}$. 
  If $\notterminates{e  }$ holds, then the detour via $f$ is not necessary,
    and the argument is much simpler. The statement of the lemma holds true then as well. 
\end{proof}

\begin{lem}\label{lem:onechart-int:milnersysmin:solvable}
  For every star expression $\astexp\in\StExpover{\actions}$ 
    with \onechart\ interpretation $\onechartof{\astexp} = \tuple{\vertsof{\astexp},\actions,\sone,\astexp,\transs,\exts}$
  the \starexpression\ function $\sasol \funin \vertsof{\astexp} \to \StExpover{\actions}$, $\asstexp \mapsto \proj{\asstexp}$
    is a \provablein{\milnersysmin} solution of $\onechartof{\astexp}$ with principal value $\astexp$.  
\end{lem}

\begin{proof}
  The statement of the lemma is an immediate consequence of Lemma~\ref{lem:FT:onechart-int}.
\end{proof}

\section{Coinductive version of Milner's proof system}%
  \label{coindmilnersys}

In this section we motivate and define `coinductive proofs', introduce coinductive versions of Milner's system $\milnersys$,
and establish first interconnections between these proof systems.

\begin{figure}[tb!]
\begin{center}
\begin{tikzpicture}
  \matrix[anchor=center,row sep=1cm,column sep=1cm,ampersand replacement=\&] {
      \& \node(top){\large $\aonechart$};
    \\
    \node(left){\large $\factorset{\chartof{\astexpi{1}}}{\seqin{\milnersys}}$};
      \&
        \& \node(right){\large $\factorset{\chartof{\astexpi{2}}}{\seqin{\milnersys}}$}; 
    \\    
    }; 
  \path (top) ++ (0.1cm,0.3cm) node[right]{\forestgreen{LLEE}, guarded}; 
  \draw[funonebisimleft]  (top) to node[above,xshift=-1.5em]{$\eqcl{\saeqfuni{1}}{\seqin{\milnersys}}$} (left);
  \draw[funonebisimright] (top) to node[above,xshift=1.5em]{$\eqcl{\saeqfuni{2}}{\seqin{\milnersys}}$} (right);  
    
\end{tikzpicture}  
\end{center}
  \vspace*{-1.5ex}
  \caption{\label{fig:motivation:coindproofs}
           A coinductive proof $\pair{\aonechart}{\saeqfun}$ over $\milnersys$ of $\astexpi{1} \formeq \astexpi{2}$ 
             gives rise to functional \protect\onebisimulations\ to factor charts of $\chartof{\astexpi{1}}$ and $\chartof{\astexpi{2}}$ 
             with respect to provability~in~$\milnersys$.}
\end{figure}%
As the central concept we now introduce `coinductive proofs' over \eqlogicbased\ proof systems~$\asys$.
We have seen examples for such circular deductions earlier in Figure~\ref{fig:ex:1:coindproof} and Figure~\ref{fig:ex:1:RSPstar:to:coindproof}.
We define a coinductive proof over $\asys$ as a \weaklyguarded\ \onechart~$\aonechart$ 
    whose vertices are labeled by equations between star expressions
      such that the left-, and the right-hand sides of the equations
        form \provablein{\asys} solutions of $\aonechart$.
  The conclusion of such a proof is the equation that labels the start vertex of $\aonechart$.        
If $\asys$ is \theoremsubsumed\ by $\milnersys$ (formally, $\asys \isthmsubsumedby \milnersys$ holds), 
  then a coinductive proof with \onechart~$\aonechart$, conclusion $\astexpi{1} \formeq \astexpi{2}$,
  and left- and right-hand side labeling functions $\saeqfuni{1}$ and $\saeqfuni{2}$
    can be viewed, due to Proposition~\ref{prop:char:solution:SRS}, 
      as a pair of \onebisimulations\ defined by $\eqcl{\saeqfuni{1}}{\seqin{\milnersys}}$ and by $\eqcl{\saeqfuni{2}}{\seqin{\milnersys}}$
        from $\aonechart$ to~$\factor{\chartof{\astexpi{1}}}{\seqin{\milnersys}}$, and to~$\factor{\chartof{\astexpi{2}}}{\seqin{\milnersys}}$
          (see Figure~\ref{fig:motivation:coindproofs}).
In this case we can show that the conclusion $\astexpi{1} \formeq \astexpi{2}$ of the coinductive proof is semantically sound
  (see~Proposition~\ref{prop:coindproofeq:2:onebisim}). 
Indeed a stronger statement holds, and its proof will form the central part of Section~\ref{coindmilnersys:2:milnersys}: 
  if $\asys \isthmsubsumedby \milnersys$ holds,
  and the underlying chart of the coinductive proof of $\astexpi{1} \formeq \astexpi{2}$ over $\asys$ is a \LLEEonechart,
    then that proof can be transformed into a proof of $\astexpi{1} \formeq \astexpi{2}$ in Milner's system $\milnersys$
    (see~Proposition~\ref{prop:LLEEcoindproofeq:impl:milnersyseq} in the next section).

In order to guarantee that coinductive proofs over a proof system $\asys$ can only derive semantically valid equations,
  it is necessary to demand that $\asys$ is sound for \processsemantics\ equality $\sprocsemeq$
    (for example if $\asys \isthmsubsumedby \milnersys$, using Proposition~\ref{prop:milnersys:sound}). 
This notwithstanding, we do not include this requirement in the definition below, but add it later to statements when it is needed.
The reason is that in Section~\ref{milnersys:2:coindmilnersys}
  we want to be able to show (see Lemma~\ref{lem:mimic:RSPstar}) that even instances of the fixed-point rule \RSPstar\ with premises that are not semantically valid
    can be mimicked by coinductive proofs over appropriate proof systems that are unsound.

\begin{defi}[(\protect\text{\protect\nf LLEE}-witnessed) coinductive proofs]\label{def:coindproof}
  Let $\actions$ be a set of actions.
  Let $\asys$ be an \eqlogicbasedover{\actions} proof system 
    with $\ACI \subsystem \asys$. 
  Let $\astexpi{1},\astexpi{2}\in\StExpover{\actions}$ be star expressions.
  
  \smallskip
  By a \emph{coinductive proof over $\asys$ of $\astexpi{1} \formeq \astexpi{2}$}  
  we mean a pair $\aCoProof = \pair{\aonechart}{\saeqfun}$
    that consists of 
      a \weaklyguarded\ \onechart~$\aonechart = \tuple{\verts,\actions,\sone,\start,\transs,\termexts}$,
      and 
      a labeling function $\saeqfun \funin \verts \to \StExpEqover{\actions}$ of vertices of $\aonechart$ by formal equations over $\actions$
      such that for the functions
      $\saeqfuni{1},\saeqfuni{2} \funin \verts \to \StExpover{\actions}$
      that denote the star expressions $\aeqfuni{1}{\avert}$, and $\aeqfuni{2}{\avert}$,
      on the left-, and on the right-hand side of the equation $\aeqfun{\avert}$, respectively, the following conditions hold: 
  \begin{enumerate}[label={\text{\nf (cp\arabic{*})}},align=right,leftmargin=*,itemsep=0.5ex]
    \item
      $\saeqfuni{1}$ and $\saeqfuni{2}$ are \provablein{\asys} solutions of $\aonechart$,
    \item
      $\astexpi{1} \synteq \aeqfuni{1}{\start}$ and $\astexpi{2} \synteq \aeqfuni{2}{\start}$
        ($\astexpi{1}$ and $\astexpi{2}$ are principal values of $\saeqfuni{1}$ and $\saeqfuni{2}$, respectively).
  \end{enumerate}
  We denote by $\astexpi{1} \coindproofeqin{\asys} \astexpi{2}$ 
    that there is a coinductive proof over $\asys$ of $\astexpi{1} \formeq \astexpi{2}$.
    
  \smallskip  
  By a \emph{\LLEEwitnessed\ coinductive proof} over $\asys$ we mean a coinductive proof $\aCoProof = \pair{\aonechart}{\saeqfun}$\vspace{-2pt}
    where $\aonechart$ is a guarded \LLEEonechart.
  We denote by $\astexpi{1} \LLEEcoindproofeqin{\asys} \astexpi{2}$
    that there is a \LLEEwitnessed\ coinductive proof over $\asys$ of $\astexpi{1} \formeq \astexpi{2}$. 
\end{defi}

While the restriction to guardedness of the \LLEEonechart\ underlying \LLEEwitnessed\ coinductive proofs could be relaxed to weak guardedness,
  we have required guardedness in this definition in order to (somewhat) reduce technicalities in the proofs in Section~\ref{coindmilnersys:2:milnersys}. 
 
We provide two examples of \LLEEwitnessed\ coinductive proofs. First we develop a new one, and 
then we revisit and justify the example in Figure~\ref{fig:ex:1:coindproof} from the introduction.

\begin{exa}\label{ex:2:LLEEcoindproof}
  In Figure~\ref{fig:ex:2:LLEEcoindproof} we have illustrated
    a \LLEEwitnessed\ coinductive proof over $\milnersysmin$
  of the statement 
  $\stexpit{(\stexpprod{\stexpit{\aact}}
                       {\stexpit{\bact}})}
     \LLEEcoindproofeqin{\milnersysmin}
   \stexpit{(\stexpsum{\aact}{\bact})}$.
  Formally this proof is of the form $\aCoProof = \pair{\aonechart}{\saeqfun}$  
    where $\aonechart = \onechartof{(a^* \prod b^*)^*}$ 
    has the guarded \LLEEwitness~$\onecharthatof{(a^* \prod b^*)^*}$ (see Theorem~\ref{thm:onechart-int:LLEEw}) as indicated in Figure~\ref{fig:ex:2:LLEEcoindproof}
    where framed boxes contain vertex names.
    
  In this illustration we have drawn the \onechart~$\aonechart$ that carries the equations with its start vertex below
    in order to adhere to the prooftree intuition for the represented derivation, namely with the conclusion at the bottom. 
  We will do so repeatedly also below.
  Solution correctness for the left-hand sides $\saeqfuni{1}$ of the equations $\saeqfun$ on $\aonechart$ in Figure~\ref{fig:ex:2:LLEEcoindproof}
    follow from Lemma~\ref{lem:onechart-int:milnersysmin:solvable}, because $\aonechart = \onechartof{(a^* \prod b^*)^*}$
      where $(a^* \prod b^*)^*$ is the left-hand side of the conclusion.  
  This notwithstanding, below we verify the correctness conditions in $\aonechart$
      for the left-hand side~$\saeqfuni{1}$ and the right-hand side~$\saeqfuni{2}$
              of the equation labeling function $\saeqfun$
    for the (most involved) case of vertex $\averti{1}$ as follows
  (we neglect some associative brackets, and combine some applications of axioms in $\milnersysmin$):%
  \begin{figure}[t!]
\begin{center}
\begin{tikzpicture}\renewcommand{\stexpprod}[2]{{#1}\hspace*{1pt}{\sstexpprod}\hspace*{1pt}{#2}}
  
\matrix[anchor=center,row sep=0.4cm,column sep=-0.3cm] {
  \node(v11){$ \stexpprod{(\stexpprod{(\stexpprod{\stexpone}{\stexpit{\aact}})}{\stexpit{\bact}})}
                         {\stexpit{(\stexpprod{\stexpit{\aact}}{\stexpit{\bact}})}} 
                 \formeq
               \stexpprod{\stexpone}{\stexpit{(\stexpsum{\aact}{\bact})}} $};
    & & \node(v21){$ \stexpprod{(\stexpprod{\stexpone}{\stexpit{\bact}})}
                               {\stexpit{(\stexpprod{\stexpit{\aact}}{\stexpit{\bact}})}} 
                       \formeq
                     \stexpprod{\stexpone}{\stexpit{(\stexpsum{\aact}{\bact})}} $};
  \\                 
  \node(v1){$ \stexpprod{(\stexpprod{\stexpit{\aact}}{\stexpit{\bact}})}{\stexpit{(\stexpprod{\stexpit{\aact}}{\stexpit{\bact}})}} 
                \formeq
              \stexpit{(\stexpsum{\aact}{\bact})} $};
    & & \node(v2){$ \stexpprod{\stexpit{\bact}}{\stexpit{(\stexpprod{\stexpit{\aact}}{\stexpit{\bact}})}} 
                      \formeq
                    \stexpit{(\stexpsum{\aact}{\bact})} $};
  \\
    & \node(v){$\stexpit{(\stexpprod{\stexpit{\aact}}{\stexpit{\bact}})}
                  \formeq
                \stexpit{(\stexpsum{\aact}{\bact})}$};
  \\
  };   
 
\path (v11) ++ (0cm,0.5cm) node{\tightfbox{$\averti{11}$}};
\draw[->,thick,densely dotted,out=180,in=180,distance=0.65cm](v11) to (v1);

\path (v1) ++ (-0.85cm,-0.5cm) node{\tightfbox{$\averti{1}$}}; 
\draw[->,thick,darkcyan,shorten >= -2pt]
  (v1) to node[left,pos=0.35]{\small $\black{\aact}$} node[right,pos=0.35]{\small $\loopnsteplab{1}$} (v11);
\draw[->,thick,densely dotted,out=-90,in=180,distance=0.5cm,shorten >= 0.2cm](v1) to (v);
\draw[->,shorten <= 0pt,shorten >= 0pt,out=12.5,in=185] (v1) to node[above]{\small $\bact$} (v21); 

\path (v21) ++ (0cm,0.5cm) node{\tightfbox{$\averti{21}$}};
\draw[->,thick,densely dotted,out=0,in=0,distance=0.65cm](v21) to (v2);  
  
\path (v2) ++ (0.85cm,-0.5cm) node{\tightfbox{$\averti{2}$}}; 
\draw[->,thick,darkcyan,shorten >= -2pt]
  (v2) to node[right,pos=0.35]{\small $\black{\bact}$} node[left,pos=0.35]{\small $\loopnsteplab{1}$} (v21);
\draw[->,thick,densely dotted,out=-90,in=0,distance=0.5cm,shorten >= 0.2cm](v2) to (v);

\path (v) ++ (1.5cm,0.5cm) node{\tightfbox{$\start$}}; 
\draw[thick,chocolate,double] (v) ellipse (1.7cm and 0.4cm); 
\draw[<-,very thick,>=latex,chocolate,shorten <= 5pt](v) -- ++ (90:0.85cm); 
\draw[->,thick,darkcyan,out=110,in=-10,shorten <= 0.2cm] 
  (v) to node[below,pos=0.36]{\small $\black{\aact}$} node[above,pos=0.26,xshift=1pt,yshift=1pt]{\small $\loopnsteplab{2}$}  (v11);
\draw[->,thick,darkcyan,out=70,in=190,shorten <= 0.2cm] 
  (v) to node[below,pos=0.37]{\small $\black{\bact}$} node[above,pos=0.25,xshift=-1pt,yshift=1pt]{\small $\loopnsteplab{2}$} (v21);


\end{tikzpicture}  
\end{center}
  \vspace*{-1ex}
  \caption{\label{fig:ex:2:LLEEcoindproof}%
           \protect\LLEEwitnessed\ coinductive proof 
             of $\stexpit{(\stexpprod{\stexpit{\aact}}
                                     {\stexpit{\bact}})}
                   \formeq
                 \stexpit{(\stexpsum{\aact}{\bact})}$ over \milnersysmin\ (see Ex.~\ref{ex:2:LLEEcoindproof}).}
\end{figure}  %
  \begin{align*}
    \aeqfuni{1}{\averti{1}}
      & {} \mathrel{\parbox{\widthof{$\milnersysmineq$}}{$\synteq$}}
      (a^* \cdot b^*) \cdot (a^* \cdot b^*)^*
    \\
      & {} \milnersysmineq
    ((1 + a \cdot a^*) \cdot (1 + b \cdot b^*)) \cdot (a^* \cdot b^*)^*
    \displaybreak[0]\\
      & {} \milnersysmineq
    (1 \cdot 1 + a \cdot a^* \cdot 1 + 1 \cdot b \cdot b^* + a \cdot a^* \cdot b \cdot b^*) \cdot (a^* \cdot b^*)^*
    \displaybreak[0]\\
      & {} \milnersysmineq
    (1 + a \cdot a^* + a \cdot a^* \cdot b \cdot b^* + b \cdot b^*) \cdot (a^* \cdot b^*)^*
    \displaybreak[0]\\
      & {} \milnersysmineq
    (1 + a \cdot a^* \cdot (1 + b \cdot b^*) + b \cdot b^*) \cdot (a^* \cdot b^*)^*
    \displaybreak[0]\\
      & {} \milnersysmineq
    (1 + a \cdot a^* \cdot b^* + b \cdot b^*) \cdot (a^* \cdot b^*)^*
    \displaybreak[0]\\
      & {} \milnersysmineq
    1 \cdot (a^* \cdot b^*)^*  +  a \cdot (((1 \cdot a^*) \cdot b^*) \cdot (a^* \cdot b^*)^*)  +  b \cdot ((1 \cdot b^*) \cdot (a^* \cdot b^*)^*)
    \\
      & {} \mathrel{\parbox{\widthof{$\milnersysmineq$}}{$\synteq$}}
    1 \cdot \aeqfuni{1}{\start} + a \cdot \aeqfuni{1}{\averti{11}} + b \cdot \aeqfuni{1}{\averti{21}}
    \displaybreak[0]\\[0.75ex]
    \aeqfuni{2}{\averti{1}}
      & {} \mathrel{\parbox{\widthof{$\milnersysmineq$}}{$\synteq$}}
    (a + b)^* 
      \\
      & {} \milnersysmineq
    (a + b)^*  +  (a + b)^* 
      \milnersysmineq
    1 + (a + b) \cdot (a + b)^*  +  1 + (a + b) \cdot (a + b)^*
    \displaybreak[0]\\
      & {} \milnersysmineq
    1 + 1 + (a + b) \cdot (a + b)^*  + a \cdot (a + b)^* + b \cdot (a + b)^*
    \displaybreak[0]\\
      & {} \milnersysmineq
    1 + (a + b) \cdot (a + b)^*  + a \cdot (1 \cdot (a + b)^*) + b \cdot (1 \cdot (a + b)^*)
    \displaybreak[0]\\
      & {} \milnersysmineq
    1 \cdot (a + b)^*  + a \cdot (1 \cdot (a + b)^*) + b \cdot (1 \cdot (a + b)^*) 
    \\
      & {} \mathrel{\parbox{\widthof{$\milnersysmineq$}}{$\synteq$}}
    1 \cdot \aeqfuni{2}{\start} + a \cdot \aeqfuni{2}{\averti{11}} + b \cdot \aeqfuni{2}{\averti{21}}  
  \end{align*}
  Note that the form of these two correctness conditions at $\averti{1}$ 
    arise from the outgoing transitions from $\averti{1}$ in $\aonechart$ in Figure~\ref{fig:ex:2:LLEEcoindproof}:
      the \onetransition\ from $\averti{1}$ to $\start$, the \transitionact{\aact} from $\averti{1}$ to $\averti{11}$,
        and the \transitionact{\bact} from~$\averti{1}$ to~$\averti{21}$.

  The solution conditions for $\saeqfun = \pair{\saeqfuni{1}}{\saeqfuni{2}}$ at the vertices $\avert$ and $\averti{2}$ can be verified analogously.
  At $\averti{11}$ and~at~$\averti{21}$ the solution conditions follow by using the axiom $\leftidstexpprod$~of~$\milnersysmin$.
\end{exa}
 
\begin{exa}
           \label{ex:1:coindproof}
  For the statement 
  $g^* \prod 0 \synteq (a + b)^* \prod 0 
     \LLEEcoindproofeqin{\milnersysmin}
   (a \prod (a + b) + b)^* \prod 0 \synteq h^* \prod 0$,
  we illustrated in Figure~\ref{fig:ex:1:RSPstar:to:coindproof}
  the coinductive proof $\aCoProof = \pair{\onechartof{h^* \prod 0}}{\saeqfun}$ over $\milnersysmin$
    with underlying guarded \LLEEwitness~$\onecharthatof{h^* \prod 0}$,
  where $\onechartof{h^* \prod 0}$ and $\onecharthatof{h^* \prod 0}$
    is the \onechart\ interpretation as defined according to Definition~\ref{def:onechartof},
    and the equation-labeling function $\saeqfun$ on $\onechartof{h^* \prod 0}$ is defined as in the figure.
    
  The correctness conditions at the start vertex (at the bottom) can be verified as follows:
  \begin{align*}
    g^* \prod 0 
      & {} \mathrel{\parbox{\widthof{$\milnersysmineq$}}{$\synteq$}}
        (a + b)^* \prod 0 
      \milnersysmineq
        (1 + (a + b) \prod (a + b)^*) \prod 0
      \milnersysmineq
        1 \prod 0 + ((a + b) \prod g^*) \prod 0 
      \displaybreak[0]\\
      & {} \milnersysmineq
        0 + (a \prod g^* + b \prod g^*) \prod 0 
      \milnersysmineq
        (a \prod g^* + b \prod g^*) \prod 0 
      \displaybreak[0]\\
      & {} \milnersysmineq
        (a \prod g^*) \prod 0 + (b \prod g^*) \prod 0
      \milnersysmineq
        a \prod (g^* \prod 0) + b \prod (g^* \prod 0)
      \displaybreak[0]\\
      & {} \milnersysmineq
        a \prod ((1 \prod g^*) \prod 0) + b \prod ((1 \prod g^*) \prod 0) \punc{,}
    \displaybreak[0]\\[0.75ex]   
    h^* \prod 0
    & {} \mathrel{\parbox{\widthof{$\milnersysmineq$}}{$\synteq$}}
    (a \prod (a + b) + b)^* \prod 0
      \milnersysmineq
    (1 + (a \prod (a + b) + b) \prod (a \prod (a + b) + b)^*) \prod 0  
    \displaybreak[0]\\
    & {} \milnersysmineq
    1 \prod 0 + ( (a \prod (a + b) + b) \prod h^*) \prod 0 
      \milnersysmineq
    0 + ( (a \prod (a + b)) \prod h^*
          + b \prod h^* ) \prod 0     
    \displaybreak[0]\\
    & {} \milnersysmineq
    ( a \prod ((a + b) \prod h^*)
          + b \prod h^* ) \prod 0  
      \milnersysmineq
    ( (a \prod ((a + b) \prod h^*)) \prod 0 
          + ( b \prod h^* ) \prod 0        
    \displaybreak[0]\\
    & {} \milnersysmineq
    a \prod ((a + b) \prod h^*) \prod 0) 
          + b \prod (h^* \prod 0) 
    \displaybreak[0]\\
    & {} \milnersysmineq
    a \prod ((1 \prod (a + b)) \prod h^*) \prod 0)    
          + b \prod ((1 \prod h^*) \prod 0) \punc{.}
  \end{align*}
  From the provable equality for $g^* \prod 0$
    the correctness condition for $(1 \prod g^*) \prod 0$ at the left upper vertex of $\onechartof{h^* \prod 0}$ 
      can be obtained by additional uses of the axiom ($\leftidstexpprod$). 
  The correctness condition for $((1 \prod (a + b)) \prod h^*) \prod 0$  at the left upper vertex of $\onechartof{h^* \prod 0}$ 
  can be verified~as~follows:
  \begin{align*}
    ((1 \prod (a + b)) \prod h^*) \prod 0
      & {} \mathrel{\parbox{\widthof{$\milnersysmineq$}}{$\milnersysmineq$}}
        ((a + b) \prod h^*) \prod 0 
        \mathrel{\parbox{\widthof{$\milnersysmineq$}}{$\milnersysmineq$}}
          (a \prod h^*  +  b \prod h^*) \prod 0 
      \\
      & {} \mathrel{\parbox{\widthof{$\milnersysmineq$}}{$\milnersysmineq$}}
        (a \prod h^*) \prod 0  +  (b \prod h^*) \prod 0 
        \mathrel{\parbox{\widthof{$\milnersysmineq$}}{$\milnersysmineq$}}
        a \prod (h^* \prod 0)  +  b \prod (h^* \prod 0) 
      \\
      & {} \mathrel{\parbox{\widthof{$\milnersysmineq$}}{$\milnersysmineq$}}
        a \prod ((1 \prod h^*) \prod 0)  +  b \prod ((1 \prod h^*) \prod 0) \punc{.}
  \end{align*}
  Finally, the correctness conditions at the right upper vertex of $\onechartof{h^* \prod 0}$ 
    can be obtained by applications of the axiom~($\leftidstexpprod$)~only.  
\end{exa}  
  
As a direct consequence of Definition~\ref{def:coindproof}, the following lemma 
  states that \LLEEwitnessed\ coinductive provability of an equation implies its coinductive provability.
The subsequent lemma states easy observations about the composition of 
  coinductive provability over a proof system~$\asys$ with provability $\seqin{\asys}$ in $\asys$.  

\begin{lem}\label{lem:LLEEcoindproofeq:2:coindproofeq}
  $\astexpi{1} \LLEEcoindproofeqin{\asys} \astexpi{2}$
    implies 
  $\astexpi{1} \coindproofeqin{\asys} \astexpi{2}$,
    for all $\astexpi{1},\astexpi{2}\in\StExpover{\actions}$, where $\asys$ is an \eqlogicbased\ proof system over $\StExpover{\actions}$
      with $\ACI \subsystem \asys$ that is sound with respect to $\procsemeq$.
\end{lem}

\begin{lem}\label{lem:props:coindproofeq:LLEEcoindproofeq}
  Let $\sabinrel \in \setexpbig{ \scoindproofeqin{\asys}, \sLLEEcoindproofeqin{\asys} }$ 
    for some \eqlogicbased\ proof system $\asys$ with $\ACI \subsystem \asys$. 
    Then $\abinrel$ is reflexive, symmetric, and
    satisfies $ \seqin{\asys} \circ \sabinrel \subseteq \sabinrel $, 
              $ \sabinrel \circ \seqin{\asys} \subseteq \sabinrel $, and
              $ \seqin{\asys} \subseteq \sabinrel $. 
\end{lem}

The proposition below, and the subsequent remark
  are evidence for the fact, mentioned at the start of this section,
    that coinductive proofs over proof systems that are sound with respect to $\sprocsemeq$
      derive semantically valid conclusions themselves. 
Rather than formulating their statements for all semantically sound proof systems,
  we restrict our attention to systems that are \theoremsubsumed\ by \milnersys.      

\begin{prop}\label{prop:coindproofeq:2:onebisim}
  Let $\asys$ be an \eqlogicbased\ proof system over $\StExpover{\actions}$ with $\ACI \subsystem \asys \isthmsubsumedby \milnersys$. 
  Then for all $\astexpi{1},\astexpi{2}\in\StExpover{\actions}$ it holds:
  \begin{equation}
    \astexpi{1}
      \coindproofeqin{\asys}
    \astexpi{2}  
      \;\;\Longrightarrow\;\;
    \chartof{\astexpi{1}}
      \bisim
    \chartof{\astexpi{2}} \punc{,}
  \end{equation}
  That is, if there is a coinductive proof over $\asys$ of $\astexpi{1} \formeq \astexpi{2}$,
    then the chart interpretations of $\astexpi{1}$ and $\astexpi{2}$ are bisimilar. 
\end{prop}

\begin{figure}[t]
\begin{center}
\begin{tikzpicture}
  \matrix[anchor=center,row sep=0.45cm,column sep=0.57cm,ampersand replacement=\&] {
      \& 
        \& 
          \& \node(top){\large $\aonechart$};
    \\
    \node(assm-pos){};
      \&
        \&
          \&
            \&
              \&
                \& \node(concl-pos){};
    \\
     \& \node(leftleft){\large $\chartof{\astexpi{1}}$};
        \& \node(left){\large $\factorset{\chartof{\astexpi{1}}}{\eqin{\milnersys}}$};
          \&
            \& \node(right){\large $\factorset{\chartof{\astexpi{2}}}{\eqin{\milnersys}}$}; 
              \& \node(rightright){\large $\chartof{\astexpi{2}}$};
    \\    
    }; 
  \path (top) ++ (0.1cm,0.3cm) node[right]{\forestgreen{LLEE}, guarded};   
  \draw[funbisimright] (leftleft) to node[below,yshift=-1ex]{\tiny (by Lem.~\ref{lem:bisim:factor:onechart})} (left);  
  \draw[funonebisimleft] (top) to node[above,xshift=-1.6em,pos=0.5]{\black{$\eqcl{\saeqfuni{1}}{\seqin{\milnersys}}$}} (left);
  \draw[funonebisimright] (top) to node[above,xshift=1.6em,pos=0.5]{\black{$\eqcl{\saeqfuni{2}}{\seqin{\milnersys}}$}} (right);  
  \draw[funbisimleft] (rightright) to node[below,yshift=-1ex]{\tiny (by Lem.~\ref{lem:bisim:factor:onechart})} (right);   
    
  \path (assm-pos) ++ (-0.8cm,0.4cm) node(assm)
         {$\left.
           \parbox{\widthof{\large coinductive proof}}
                  {{\large \hspace*{\fill}$\pair{\aonechart}{\saeqfun}$\hspace*{\fill}\mbox{}}
                   \\[-0.1ex]
                   \hspace*{\fill}coinductive proof\hspace*{\fill}\mbox{}
                   \\[-0.1ex]
                   \hspace*{\fill}of $\astexpi{1} \formeq \astexpi{2}$\hspace*{\fill}\mbox{}
                   \\[0.1ex]
                   \hspace*{\fill}over $\asys \isthmsubsumedby \milnersys$\hspace*{\fill}\mbox{}}
           \right\}$}; 
    
  \path (concl-pos) ++ (0.5cm,0.4cm) node(concl)
                                         {\large
                                          $\left\{\;
                                           \begin{gathered}
                                             \chartof{\astexpi{1}} \bisim \chartof{\astexpi{2}}
                                             \\
                                             \astexpi{1} \procsemeq \astexpi{2} 
                                           \end{gathered}
                                           \;\right.$};  
    
  \draw[-implies,double equal sign distance]
    (assm)
      to node[below,yshift=-0.125cm,xshift=0.2cm]{\tiny (Prop.$\,$\ref{prop:char:solution:SRS})}  
    ($(assm.east) + (0.75cm,0cm)$);  
  \draw[-implies,double equal sign distance,shorten >=4pt] 
    ($(concl.west) + (-0.85cm,0cm)$) 
      to 
    (concl); 
    
\end{tikzpicture}  
\end{center}
  \vspace*{-2ex}
  \caption{\label{fig:prop:coindproofeq:2:onebisim}%
           Coinductive proofs 
             over $\asys \isthmsubsumedby \milnersys$ guarantee
             equality of process semantics.}
\end{figure}
\begin{proof}
  We have illustrated the proof of this proposition in Figure~\ref{fig:prop:coindproofeq:2:onebisim}:
  In every coinductive proof $\pair{\aonechart}{\saeqfun}$ over $\asys$ with $\asys \isthmsubsumedby \milnersys$ of an equation $\astexpi{1} \formeq \astexpi{2}$,
    the \starexpression\ functions $\saeqfuni{1}$ and $\saeqfuni{2}$ are \provablein{\milnersys} solutions of $\aonechart$.
  Then by using Proposition~\ref{prop:char:solution:SRS} and Lemma~\ref{lem:bisim:factor:onechart}
    we get the link of functional ($\sone$-)bi\-si\-mu\-la\-tions between $\chartof{\astexpi{1}}$ and $\chartof{\astexpi{2}}$ 
      as drawn in that picture. 
  Since (functional) bisimulations compose with (functional) \onebisimulations\ to \onebisimulations,
    and \onebisimulations\ between charts are bisimulations,
  we obtain that $\chartof{\astexpi{1}} \bisim \chartof{\astexpi{2}}$ holds, and consequently, that $\astexpi{1} \procsemeq \astexpi{2}$ holds.
\end{proof}

\begin{rem}\label{rem:coindproof:defs:onebisim:up:to}
  For every coinductive proof $\aCoProof = \pair{\aonechart}{\saeqfun}$, whether $\aCoProof$ is \LLEEwitnessed\ or not, 
    over an \eqlogicbased\ proof system $\asys$ with $\ACI \subsystem \asys \isthmsubsumedby \milnersys$
  the finite relation defined by: 
  \begin{center}
    $
    \abisim
      \,\defdby\,
        \descsetexpBig{\displaystyle
          \Big\langle
            \terminatesconstof{\aonechart}{\avert}
              +
            \sum_{i=1}^{n} 
              \aoneacti{i} \prod \aeqfuni{1}{\averti{i}}
            ,\,
            \terminatesconstof{\aonechart}{\avert}
              +
            \sum_{i=1}^{n} 
              \aoneacti{i} \prod \aeqfuni{2}{\averti{i}}
          \Big\rangle
                       }{ \begin{gathered}
                             \transitionsinfrom{\aonechart}{\avert}
                               =
                             \descsetexpbig{ \avert \lt{\aoneacti{i}} \averti{i} }{ i \in\setexp{1,\ldots,n} },
                             \\[-0.75ex]
                             \avert\in\vertsof{\aonechart},\, \text{the set of vertices of $\aonechart$}
                          \end{gathered} }
    $
  \end{center}
  is a \onebisimulation\ up to $\seqin{\asys}$ 
  with respect to the labeled transition system on all star expressions that is defined by the TSS in Definition~\ref{def:chartof}. 
  This can be shown by using that the left-hand sides $\aeqfuni{1}{\avert}$, and respectively the right-hand sides $\aeqfuni{2}{\avert}$,
  of the equations $\aeqfun{\avert}$ in $\aCoProof$, for $\avert\in\vertsof{\aonechart}$,
    form \provablein{\asys} solutions of the \onechart\ $\aonechart$ that underlies $\aCoProof$.
\end{rem}

We now define two proof systems \CLC\ and \CC\ for combining \LLEEwitnessed\ coinductive provability. 
Each of these systems consists of a single rule scheme, a more specific one for \CLC, and a more liberal one for \CC. 
Instances of rules of these two schemes formalize \LLEEwitnessed\ coinductive provability in \CLC, 
                                                  and respectively coinductive provability in \CC,
  of equations between star expressions from assumed equations. 
Different from the exposition in \cite{grab:2021:calco:arxiv,grab:2021:calco},
  where we permitted entire coinductive proofs as formulas and as premises of rules,
    we here keep the proof systems \equationbasedover{\actions}
      by externalizing coinductive proofs from the rules by `hiding' them in \sideconditions.%
      \footnote{Keeping the systems equation-based (by avoiding coinductive proofs as formulas as in \cite{grab:2021:calco:arxiv,grab:2021:calco})
                  permits us to then compare the coinductive proof systems \CLC, \CC, and later \coindmilnersys, and \coindmilnersysbar\
                  via rule derivability and admissibility to Milner's system \milnersys\ and its variants $\milnersysacc$, and $\milnersysaccbar$.}                                                                       
The more restricted proof system \CLC\ 
  will form the core of our coinductive reformulation of Milner's~system. 

\begin{defi}[proof systems \protect\CLC, \protect\CC\ for combining (\protect\LLEE-witn.) coinductive~provability]\label{def:CLC:CC}
  Let $\actions$ be a set of actions.
  We define \equationbasedover{\actions} proof systems \CLCover{\actions} and \CCover{\actions}.
  
  The \emph{proof system \CLCover{\actions} for combining \LLEEwitnessed\ 
            coinductive provability (over extensions of $\milnersysminover{\actions}\!$) of equations between star expressions over $\actions$}
  is an \equationbasedover{\actions} proof system 
    without \emph{axioms},
    but with the \emph{rules} of the scheme:  
  \begin{gather}
    \mbox{
      \AxiomC{$ \cstexpi{1}   \formeq   \dstexpi{1} $}
      \AxiomC{\ldots}
      \AxiomC{$ \cstexpi{n}   \formeq   \dstexpi{n} $}
      \insertBetweenHyps{\hspace*{0.25ex}}
      \RightLabel{\LCoindProofi{n} \mbox{} {\small (if \eqref{LCoindProof:sidecond} holds)}}
      \TrinaryInfC{$ \astexp \formeq \bstexp $}           
      \DisplayProof
      }
      \label{def:LCoProof}
    \\[0.5ex]
    \text{$\astexp \LLEEcoindproofeqin{\thplus{\milnersysmin}{\aseteqs}} \bstexp$
          for $\Gamma = \setexp{ \cstexpi{1} \formeq \dstexpi{1}, \ldots, \cstexpi{n} \formeq \dstexpi{n} }$
          with $n\in\nat$ (including $n=0$).}
      \label{LCoindProof:sidecond}
  \end{gather} 
  
  The \emph{proof system \CCover{\actions} for combining coinductive provability (over extensions of $\milnersysminover{\actions}$)} 
    is an \equationbasedover{\actions} proof system 
      without \emph{axioms},
        but with the \emph{rules} of the scheme:  
  \begin{gather}
    \mbox{
      \AxiomC{$ \cstexpi{1}   \formeq   \dstexpi{1} $}
      \AxiomC{\ldots}
      \AxiomC{$ \cstexpi{n}   \formeq   \dstexpi{n} $}
      \insertBetweenHyps{\hspace*{0.25ex}}
      \RightLabel{\CoindProofi{n} \mbox{} {\small (if \eqref{CoindProof:sidecond} holds)}}
      \TrinaryInfC{$ \astexp \formeq \bstexp $}           
      \DisplayProof
      }
      \label{def:CoProof}
    \notag
    \\[0.5ex]
    \text{$\astexp \coindproofeqin{\thplus{\milnersysmin}{\aseteqs}} \bstexp$
          for $\Gamma = \setexp{ \cstexpi{1} \formeq \dstexpi{1}, \ldots, \cstexpi{n} \formeq \dstexpi{n} }$
          with $n\in\nat$ (including $n=0$).}
      \label{CoindProof:sidecond}
  \end{gather} 
  
  We will keep the set $\actions$ implicit, and write \CLC\ and \CC\ for \CLCover{\actions} and \CCover{\actions}, respectively. 
\end{defi}

Note that the systems \CLC\ and \CC\ do not contain 
  the rules of $\eqlogic$ nor any axioms.
Instead, derivations in these systems have to start with 0-premise instances of $\LCoindProofi{0}$ or~$\CoindProofi{0}$.
Due to Lemma~\ref{lem:LLEEcoindproofeq:2:coindproofeq} every instance of the rule $\LCoindProofi{n}$ of $\CLC$ for some $n\in\nat$ 
  is also an instance of the rule $\CoindProofi{n}$ of $\CC$.
It follows that derivability of an equation $\astexp \formeq \bstexp$ in \CLC\ implies derivability of $\astexp \formeq \bstexp$ in \CC,
  that is, $\CLC \isthmsubsumedby \CC$ holds, see Lemma~\ref{lem:thmsubsumption:rels:coindproofsystems}, \ref{it:1:lem:thmsubsumption:rels:coindproofsystems}, below.
  
Based on \CLC, we now define the system that we call the coinductive variant $\coindmilnersys$ of Milner's proof system $\milnersys$.
For this, we replace the fixed-point rule 
                                in $\milnersys$
  by the rule scheme 
                     of \CLC,
    or equivalently, by adding this rule scheme to the purely equational part $\milnersysmin$ of $\milnersys$. 
By adding the rule scheme 
                          of \CC\ to $\milnersysmin$,
  we also define an extension $\coindmilnersysbar$ of $\coindmilnersys$.

\begin{defi}[proof systems $\protect\coindmilnersys$, $\protect\coindmilnersysone$, and $\protect\coindmilnersysbar$]\label{def:coindmilnersys}
  Let $\actions$ be a set of actions.\nopagebreak[4]
  
  The proof system $\coindmilnersysover{\actions}$,
    the \emph{coinductive variant of $\milnersysover{\actions}$},
  is an \eqlogicbased\ proof system whose \emph{axioms} are those of $\milnersysminover{\actions}$,
  and whose \emph{rules} are those of $\eqlogicover{\actions}$,
    plus the \emph{rule scheme} $\family{\LCoindProofi{n}}{n\in\nat}$ from $\CLCover{\actions}$.
  By $\coindmilnersysoneover{\actions}$ we mean
    \emph{the simple coinductive variant of $\milnersysover{\actions}$},
  an \eqlogicbased\ proof system
    that arises by only adding the rule $\LCoindProofi{1}$ of $\CLCover{\actions}$ to the rules and axioms of $\milnersysminover{\actions}$.
  
  By $\coindmilnersysbarover{\actions}$ we mean the variant of $\coindmilnersysover{\actions}$
    in which the more general \emph{rule scheme} $\family{\CoindProofi{n}}{n\in\nat}$ from $\CCover{\actions}$
    is used (instead of $\family{\LCoindProofi{n}}{n\in\nat}$ from $\CLCover{\actions}$).
    
  We again permit to write
             $\coindmilnersys$, $\coindmilnersysone$, $\coindmilnersysbar$
         for $\coindmilnersysover{\actions}$, $\coindmilnersysoneover{\actions}$, and $\coindmilnersysbarover{\actions}$, respectively.
\end{defi}

We now prove a lemma (Lemma~\ref{lem:thmsubsumption:rels:coindproofsystems} below) 
  that gathers elementary theorem equivalence and theorem subsumption statements 
    between the coinductive variants of Milner's system defined above. 
For its proof we argue
  with subsystem relationships as gathered by Lemma~\ref{lem:subsystem:rels:coindproofsystems} below, 
    and we explain basic proof transformations between these systems. 

\begin{lem}\label{lem:subsystem:rels:coindproofsystems}
  The following subsystem relationships hold
  between the coinductive proof systems defined above: 
  \begin{enumerate*}[label={(\roman{*})}]
    \item{}\label{it:1:lem:subsystem:rels:coindproofsystems}
      $\CLC \subsystem \coindmilnersys$,
    \item{}\label{it:2:lem:subsystem:rels:coindproofsystems} 
      $\coindmilnersysone \subsystem \coindmilnersys$,
    \item{}\label{it:3:lem:subsystem:rels:coindproofsystems}
      $\CC \subsystem \coindmilnersysbar$.   
  \end{enumerate*}
\end{lem}

\begin{lem}\label{lem:thmsubsumption:rels:coindproofsystems}
  The following theorem subsumption and theorem equivalence statements hold:
  \begin{enumerate}[label={(\roman{*})},align=right,leftmargin=*,itemsep=0.5ex]
    \item{}\label{it:1:lem:thmsubsumption:rels:coindproofsystems}
      $ \CLC \isthmsubsumedby \CC $,
      
    \item{}\label{it:2:lem:thmsubsumption:rels:coindproofsystems}
      $ \coindmilnersysone
          \isthmsubsumedby
        \coindmilnersys 
          \isthmsubsumedby
        \coindmilnersysbar$,
        
    \item{}\label{it:3:lem:thmsubsumption:rels:coindproofsystems}
      $ \CLC \thmequiv \coindmilnersys $, 
      
    \item{}\label{it:4:lem:thmsubsumption:rels:coindproofsystems}
      $ \CC \thmequiv \coindmilnersysbar $.
  \end{enumerate}
\end{lem}

\begin{proof}
  We have argued for statement~\ref{it:1:lem:thmsubsumption:rels:coindproofsystems} above, below Definition~\ref{def:CLC:CC}: 
      every instance of the rule $\LCoindProofi{n}$ of $\CLC$, for $n\in\nat$,
        is also an instance of the rule $\CoindProofi{n}$ of $\CC$.
  This also implies the part $\coindmilnersys \isthmsubsumedby \coindmilnersysbar$ of statement~\ref{it:2:lem:thmsubsumption:rels:coindproofsystems},
    because it shows in every derivation $\aDeriv$ of an equation $\astexp \formeq \bstexp$ in $\coindmilnersys$
    every instance of $\LCoindProofi{n}$, for $n\in\nat$, can be replaced by an instance of $\CoindProofi{n}$
      with as result a derivation $\aDerivacc$ of $\astexp \formeq \bstexp$ in $\coindmilnersysbar$.
  The part $\coindmilnersysone \isthmsubsumedby \coindmilnersys$ of statement~\ref{it:2:lem:thmsubsumption:rels:coindproofsystems}
    follows from the fact that $\coindmilnersysone$ is a subsystem of $\coindmilnersys$ by Lemma~\ref{lem:subsystem:rels:coindproofsystems}, \ref{it:2:lem:subsystem:rels:coindproofsystems}.
          
  \smallskip      
  For statement~\ref{it:3:lem:thmsubsumption:rels:coindproofsystems}, 
  the \theoremsubsumption\ part $\CLC \isthmsubsumedby \coindmilnersys$ follows from $\CLC \subsystem \coindmilnersys$, 
    see Lemma~\ref{lem:subsystem:rels:coindproofsystems},~\ref{it:1:lem:subsystem:rels:coindproofsystems}. 
  For showing the converse implication, $\CLC \thmsubsumes \coindmilnersys$, 
  we will demonstrate the proof-transformation statement \eqref{eq:1:prf:it:3:lem:thmsubsumption:rels:coindproofsystems} below
    by first showing its special case \eqref{eq:2:prf:it:3:lem:thmsubsumption:rels:coindproofsystems}:
  \begin{align}
    &
    \left.\;
    \parbox{0.825\textwidth}
           {Every derivation $\aDeriv$ in \coindmilnersys\ can be transformed into a derivation $\aDerivacc$ in \CLC\ 
            \\[-0.1ex]
            with the same conclusion.}
      \label{eq:1:prf:it:3:lem:thmsubsumption:rels:coindproofsystems}
    \;\right\}   
    \displaybreak[0]\\
    &
    \left.\;
    \parbox{0.825\textwidth}
           {The transformation statement \eqref{eq:1:prf:it:3:lem:thmsubsumption:rels:coindproofsystems} 
            holds for every derivation $\aDeriv$ in \coindmilnersys\
            with an instance of $\LCoindProofi{n}$, for some $n\in\nat$, at the bottom.}
    \;\right\}
      \label{eq:2:prf:it:3:lem:thmsubsumption:rels:coindproofsystems}        
  \end{align} 
  
  The idea for both of these proof transformation statements 
      is to `hide' derivation parts that consist of axioms and rules in \milnersysmin\ 
        into the correctness statements of coinductive proofs that appear as \sideconditions\ in instances of the coinductive rule in \CLC.
  More precisely, the idea is to replace derivation parts $\aDerivi{0}$ in $\milnersysmin$ of a derivation $\aDeriv$ in \coindmilnersys,
    where $\aDerivi{0}$ consists of the inference in \milnersysmin\ of an equation $\cstexp \formeq \dstexp$ 
        from a set $\aseteqs$ of $m$ assumption equations,
    by an instance of the coinductive rule $\LCoindProofi{m}$ that has the $m$ assumptions of $\aDerivi{0}$ in $\aseteqs$ as its premises.
  Then whenever $\cstexp \formeq \dstexp$ is needed to justify a correctness conditions
    it can be reconstructed from the premises, provably in \milnersysmin, on the basis of the derivation $\aDerivi{0}$ in \milnersysmin. 
  
  We start by showing \eqref{eq:2:prf:it:3:lem:thmsubsumption:rels:coindproofsystems}, 
     and proceed for this purpose by induction on the depth~$\depth{\aDeriv}$ of~$\aDeriv$. 
  Suppose that $\aDeriv$ is a derivation in \coindmilnersys\ with an occurrence of an instance $\ainst$ of $\LCoindProofi{n}$ at the bottom.
  To perform the induction step for $\aDeriv$, 
    we need to transform $\aDeriv$ into a derivation $\aDerivacc$ in \CLC\ with the same conclusion.  
  We may assume that the immediate subderivations $\aDerivi{1}$, \ldots, $\aDerivi{n}$ of $\aDeriv$ (just above the instance $\ainst$)
    contain axioms and/or rules of $\milnersysmin$, because otherwise $\aDeriv$ is already a derivation in $\CLC$. 
  To keep the illustration of the transformation step simple, 
    we assume that only the $i$\nb-th subderivation $\aDerivi{i}$ contains axioms and/or rules of $\milnersys$;
      the general case will become clear through this example.
  We assume that $\aDeriv$ is~of~the~form:  
  \begin{align}
    \mbox{%
      \AxiomC{$\aDerivi{1}$}
      \dashedLine
      \UnaryInfC{$\cstexpi{1} \formeq \dstexpi{1}$}
      \AxiomC{$\aDerivi{i1}$}
      \dashedLine
      \UnaryInfC{$(\cstexpi{i1} \formeq \dstexpi{i1})$}
      \AxiomC{$\aDerivi{im}$}
      \dashedLine
      \UnaryInfC{$(\cstexpi{im} \formeq \dstexpi{im})$}
      \insertBetweenHyps{$\;\cdots\;$}
      \noLine
      \BinaryInfC{$\aDerivi{i0}$}
      \noLine
      \UnaryInfC{$\ldots \hspace*{3ex} \cstexpi{i} \formeq \dstexpi{i} \hspace*{3ex} \ldots $}
      \AxiomC{$\aDerivi{n}$}
      \dashedLine
      \UnaryInfC{$\cstexpi{n} \formeq \dstexpi{n}$}
      \RightLabel{$\LCoindProofi{n}$}
      \insertBetweenHyps{\hspace*{1ex}}
      \TrinaryInfC{$\astexp \formeq \bstexp$}
      \DisplayProof
      }
    \label{eq:3:prf:it:3:lem:thmsubsumption:rels:coindproofsystems}
    \\[1ex]
    \text{with \sidecondition} \;\;
      \astexp
        \LLEEcoindproofeqin{\thplus{\milnersysmin}{\aseteqs}}
      \bstexp
    \;\;  
    \text{for $\aseteqs = \setexp{ \cstexpi{1} \formeq \dstexpi{1},
                                      \ldots
                                   \cstexpi{n} \formeq \dstexpi{n} }$\punc{,}}
    \label{eq:4:prf:it:3:lem:thmsubsumption:rels:coindproofsystems}                               
  \end{align}  
  where $\aDerivi{1}, \ldots, \aDerivi{i-1}, \aDerivi{i+1}, \ldots \aDerivi{n}$ are already derivations in \CLC\
    (with bottommost instances of $\LCoindProofi{n_i}$ that are suggested by dashed lines),
  but $\aDerivi{i}$ contains axioms and/or rule instances of \milnersysmin,
    and can be construed with a bottom part $\aDerivi{i0}$ in \milnersysmin\ below $m$ conclusions
      $ \cstexpi{i1} \formeq \dstexpi{i1}$, \ldots, $\cstexpi{im} \formeq \dstexpi{im} $
        of instances of coinductive rules from $\family{\LCoindProofi{j}}{j\in\nat}$.
  
  Then we apply the induction hypothesis to the subderivations $\aDerivi{i1},\ldots,\aDerivi{im}$ of $\aDerivi{0}$, 
      which is possible due to $\depth{\aDerivi{i1}}, \ldots, \depth{\aDerivi{im}} < \depth{\aDerivi{i}} < \depth{\aDeriv}$,
    to obtain derivations $\aDerivacci{i1},\ldots,\aDerivacci{im}$ in $\CLC$ with the same conclusions 
      $\cstexpi{i1} \formeq \dstexpi{i1}$, \ldots, $\cstexpi{im} \formeq \dstexpi{im}$, respectively.  
  Then we transform $\aDeriv$
    by replacing the instance of $\LCoindProofi{n}$ at the bottom
      by an instance of $\LCoindProofi{n+m-1}$,
        keeping the first $i-1$ and the last $n+1-i$ premises and their subderivations,
          but replacing the $i$-th premise and its immediate subderivation $\aDerivi{i}$
            by $m$ additional premises with immediate subderivations $\aDerivacci{i1}$, \ldots, $\aDerivacci{im}$,
              thereby obtaining:
  \begin{gather}\renewcommand{\fCenter}{}
    \mbox{
      \AxiomC{$\aDerivi{1}$}
      \dashedLine
      \UnaryInfC{$\cstexpi{1} \formeq \dstexpi{1}$}
      \Axiom$\fCenter\aDerivacci{i1}$
      \dashedLine
      \UnaryInf$\cstexpi{i1} \fCenter{}\formeq \dstexpi{i1} $
      \Axiom$\fCenter\aDerivacci{im}$
      \dashedLine
      \UnaryInf$\cstexpi{im} \fCenter{}\formeq \dstexpi{im} $
      \AxiomC{$\aDerivi{n}$}
      \dashedLine
      \UnaryInfC{$\cstexpi{n} \formeq \dstexpi{n}$}
      \insertBetweenHyps{$\;\;\ldots\;\;$}
      \RightLabel{$\LCoindProofi{n+m-1}$}
      \QuaternaryInfC{$\astexp \formeq \bstexp$}
      \DisplayProof
      }
    \label{eq:5:prf:it:3:lem:thmsubsumption:rels:coindproofsystems}
  \end{gather}
  However, we need to show the \sidecondition\ for the displayed instance of $\LCoindProofi{n+m-1}$:
  \begin{equation}
  \hspace*{-10cm}
  \begin{minipage}{5cm}
  \[
      \astexp
        \LLEEcoindproofeqin{\thplus{\milnersysmin}{\aseteqsacc}}
      \bstexp
    \;\;  
    \text{with }
    \begin{aligned}[t]
      &
      \aseteqsacc = \bigl\{ \cstexpi{1} \formeq \dstexpi{1},
                              \ldots
                            \cstexpi{i-1} \formeq \dstexpi{i-1},
                            \cstexpi{i+1} \formeq \dstexpi{i+1},
                              \ldots,  
                            \cstexpi{n} \formeq \dstexpi{n} \bigr\}
                     \cup \bseteqs \punc{,}
      \\
      &  
      \text{ where: }                
      \bseteqs 
        \defdby
          \begin{aligned}[t]
            &
            \bigl\{ \cstexpi{i1} \formeq \dstexpi{i1},
                      \ldots
                    \cstexpi{im} \formeq \dstexpi{im} \bigr\} 
            \\[-0.5ex]                
            &
            \text{ the set of conclusions of $\aDerivacci{i1},\ldots,\aDerivacci{im}$.}
         \end{aligned}          
    \end{aligned}                   
    \label{eq:6:prf:it:3:lem:thmsubsumption:rels:coindproofsystems}
    \]                     
\end{minipage}  
  \end{equation}
  Now due to \eqref{eq:4:prf:it:3:lem:thmsubsumption:rels:coindproofsystems}
  there is a \LLEEwitnessed\ coinductive proof $\aLLEECoProof$ of $\astexp \formeq \bstexp$ over $\thplus{\milnersysmin}{\aseteqs}$. 
  But now $\aLLEECoProof$ is also a \LLEEwitnessed\ coinductive proof 
    of $\astexp \formeq \bstexp$ over $\thplus{\milnersysmin}{\aseteqsacc}$,
      and thus over a different set of premises, which we recognize as follows.
  The equations in $\bseteqs$, which have been added to $\aseteqs$ in order to get $\aseteqsacc$ after removing $\cstexpi{i} \formeq \dstexpi{i}$,
    are derivable in $\thplus{\milnersysmin}{\aseteqs}$ by means for the derivation $\aDerivi{0}$.  
  Therefore the correctness conditions for $\aLLEECoProof$ as a \LLEEwitnessed\ coinductive proof over $\thplus{\milnersysmin}{\aseteqs}$
    imply the correctness conditions for $\aLLEECoProof$ as a \LLEEwitnessed\ coinductive proof over $\thplus{\milnersysmin}{\aseteqsacc}$.
  This shows \eqref{eq:6:prf:it:3:lem:thmsubsumption:rels:coindproofsystems}.
  Therefore the  resulting derivation $\aDerivacc$ in \eqref{eq:5:prf:it:3:lem:thmsubsumption:rels:coindproofsystems}
     is a derivation in \CLC\ with the same conclusion as $\aDeriv$.
  (This transformation step can obviously be generalized to the situation in which not just $\aDerivi{i}$, 
    but also others among the immediate subderivations $\aDerivi{1}, \ldots, \aDerivi{n}$ of $\aDeriv$ contain axioms and/or rules of \milnersysmin.) 
  In this way we have performed the induction step.
      
  Finally we establish \eqref{eq:1:prf:it:3:lem:thmsubsumption:rels:coindproofsystems} in full generality.
  For this we consider the remaining situation in which the derivation $\aDeriv$ in \coindmilnersys\ does not terminate with an instance of a coinductive rule.
  Then $\aDeriv$ can be construed with a part derivation $\aDerivi{0}$ in \milnersysmin\ above its conclusion, 
    and below $m\in\nat$ subderivations $\aDerivi{1}$, \ldots, $\aDerivi{m}$, each of which terminates with a coinductive rule
      (see below on the left). Note that $m=0$ is possible if $\aDeriv$ is a derivation in \milnersysmin.
  By applying \eqref{eq:2:prf:it:3:lem:thmsubsumption:rels:coindproofsystems} to $\aDerivi{1}, \ldots, \aDerivi{m}$
    we obtain derivations $\aDerivacci{1}, \ldots, \aDerivacci{m}$ in \CLC\ with the same conclusions, respectively.
  By combining these derivations in \CLC\ with an instance of $\LCoindProofi{m}$ we can perform the following transformation step
    in order to obtain a derivation $\aDerivacc$ in $\CLC\,$:  
  \begin{align}
    \mbox{
      \AxiomC{$\aDerivi{1}$}
      \dashedLine
      \UnaryInfC{$(\astexpi{1} \formeq \bstexpi{1})$}
      \AxiomC{$\aDerivi{m}$}
      \dashedLine
      \UnaryInfC{$(\astexpi{m} \formeq \bstexpi{m})$}
      \insertBetweenHyps{$\ldots$}
      \noLine
      \BinaryInfC{$\aDerivi{0}$}
      \noLine
      \UnaryInfC{$\astexp \formeq \bstexp$}
      \DisplayProof
          } 
      \;\;\Longrightarrow\;\;
        \mbox{
          \renewcommand{\fCenter}{}
            \Axiom$\fCenter\aDerivacci{1}$
            \dashedLine
            \UnaryInf$\astexpi{1} \fCenter{}\formeq \bstexpi{1}$
            \AxiomC{\ldots}
            \Axiom$\fCenter\aDerivacci{m}$
            \dashedLine
            \UnaryInf$\astexpi{m} \fCenter{}\formeq \bstexpi{m}$
            \insertBetweenHyps{\hspace*{1ex}}
            \RightLabel{$\LCoindProofi{m}$}
            \TrinaryInfC{$\astexp \formeq \bstexp$}
            \DisplayProof
               }
    \label{eq:7:prf:it:3:lem:thmsubsumption:rels:coindproofsystems}           
  \end{align}  
  Here we need to establish the following \sidecondition\ for the instance of $\LCoindProofi{m}$ at the bottom of $\aDerivacc$:
  \begin{equation}
    \astexp \LLEEcoindproofeqin{\thplus{\milnersysmin}{\aseteqs}} \bstexp
      \quad \text{where } \aseteqs \defdby \setexp{ \astexpi{1} \formeq \bstexpi{1}, \ldots, \astexpi{m} \formeq \bstexpi{m} } \punc{.} 
        \label{eq:8:prf:it:3:lem:thmsubsumption:rels:coindproofsystems}
  \end{equation}
  We can establish this coinductive-provability statement 
    by recognizing that $\pair{\onechartof{\astexp}}{\saeqfun}$ 
    with $\onechartof{\astexp}$ the \onechart\ interpretation of $\astexp$, and with equation labeling function:
  \begin{equation*}
    \saeqfun \funin \vertsof{\onechartof{\astexp}} \to \StExpover{\actions} \punc{,}
      \;\; 
      \csstexp \mapsto \begin{cases}
                          \csstexp \formeq \csstexp   & \text{if $\csstexp\notsynteq\astexp$}
                          \\[-0.5ex]
                          \astexp \formeq \bstexp     & \text{if $\csstexp\synteq\astexp$}
                        \end{cases}  
  \end{equation*}    
  is a \LLEEwitnessed\ coinductive proof over $\thplus{\milnersysmin}{\aseteqs}$ of $\astexp \formeq \bstexp$.
  To verify this statement, 
    we use that $\onechartof{\astexp}$ is a guarded \LLEEonechart\ by Theorem~\ref{thm:onechart-int:LLEEw}, \ref{it:1:thm:onechart-int:LLEEw}, and
    we have to check the correctness conditions for $\saeqfuni{1}$ and $\saeqfuni{2}$ with $\saeqfun = \pair{\saeqfuni{1}}{\saeqfuni{2}}$ 
    to be \provablein{\thplus{\milnersysmin}{\aseteqs}} solutions of $\onechartof{\astexp}$. 
  We first note that $\astexp \formeq \bstexp$ is provable in $\thplus{\milnersysmin}{\aseteqs}$ (i.e.\ $\astexp \eqin{\thplus{\milnersysmin}{\aseteqs}} \bstexp$)
    since $\aDerivi{0}$ is a derivation of $\astexp \formeq \bstexp$ in \milnersysmin\ from the assumptions in $\aseteqs$. 
  Then~we~argue~as~follows:  
  \begin{itemize}  
    \item
      The correctness conditions for $\saeqfuni{1}$ to be a \provablein{(\thplus{\milnersysmin}{\aseteqs})} solution of $\onechartof{\astexp}$ 
        follow from the fact that $\saeqfuni{1}$ is a \provablein{\milnersysmin} solution of $\onechartof{\astexp}$ 
          due to Lemma~\ref{lem:onechart-int:milnersysmin:solvable}.
    \item   
      $\saeqfuni{2}$ differs from $\saeqfuni{1}$ only in the value for the start vertex $\astexp$,
        where $\aeqfuni{2}{\astexp} \synteq \bstexp$, but $\aeqfuni{1}{\astexp} \synteq \astexp$. 
      From this it follows, in view of $\astexp \eqin{\thplus{\milnersysmin}{\aseteqs}} \bstexp$,
        that the correctness conditions for $\saeqfuni{1}$ 
          imply the correctness conditions for $\saeqfuni{2}$,
            because they differ only up to expressions that are provably equal in $\thplus{\milnersysmin}{\aseteqs}$
              and also need to hold up to \provabilityin{(\thplus{\milnersysmin}{\aseteqs})}.
  \end{itemize} 
  This argument shows \eqref{eq:8:prf:it:3:lem:thmsubsumption:rels:coindproofsystems},
    which justifies the \sidecondition\ of the instance of $\LCoindProofi{m}$ 
      at the bottom of the derivation $\aDerivacc$ on the right in \eqref{eq:7:prf:it:3:lem:thmsubsumption:rels:coindproofsystems}.
  Therefore we have indeed obtained from $\aDeriv$ a derivation in \CLC\ with the same conclusion as $\aDeriv$.
  In this way we have completed the proof of $\coindmilnersys \isthmsubsumedby \CLC$, the remaining part of \ref{it:3:lem:thmsubsumption:rels:coindproofsystems}.
  
  \smallskip
  Statement~\ref{it:4:lem:thmsubsumption:rels:coindproofsystems} can be shown entirely analogously as statement~\ref{it:3:lem:thmsubsumption:rels:coindproofsystems}.
\end{proof}

\begin{rem}[completeness of $\protect\CC$, $\protect\coindmilnersysbar$, $\protect\milnersysaccbar$]\label{rem:CC:coindmilnersysbar:milnersyaccbar}
  The proof systems $\CC$ and $\coindmilnersysbar$,
    as well as the variant $\milnersysaccbar$ of Milner's system with the general unique solvability principle \USP\
  are complete for bisimilarity of star expressions.
  This can be established along Salomaa's completeness proof
    for his inference system for language equality of regular expressions \cite{salo:1966}, 
    by an argument that we can outline as follows.
  Given star expressions $e$ and $f$ with $\chartof{e} \bisim \chartof{f}$, 
    $e$ and $f$ can be shown to be principal values of \provablein{\milnersysmin} solutions of $\chartof{e}$ and $\chartof{f}$, respectively
    (by a lemma for the chart interpretation similar to Lemma~\ref{lem:onechart-int:milnersysmin:solvable}). 
  These solutions can be transferred to the (\sonefree) product chart~$\achart$ of $\chartof{e}$ and $\chartof{f}$,
    with $e$ and $f$ as principal values
      of \provablein{\milnersysmin} solutions $\saeqfuni{1}$ and $\saeqfuni{2}$ of $\achart$, respectively. 
    From this we obtain a (not necessarily \LLEEwitnessed) coinductive proof $\pair{\achart}{\saeqfun}$ of $e = f$ over $\milnersysmin$. 
    It follows that $e = f$ is provable in $\CC$, and in $\coindmilnersysbar$.
    Now since the correctness conditions for the \provablein{\milnersysmin} solutions $\saeqfuni{1}$ and $\saeqfuni{2}$ of $\achart$
      at each of the vertices of $\achart$
      together form a guarded system of linear equations to which the rule \USP\ can be applied (as $\achart$ is \sonefree),
      we obtain that $e \formeq f$ is also provable in $\milnersysaccbar$. 
\end{rem}

\section{From LLEE-witnessed coinductive proofs to Milner's system}
  \label{coindmilnersys:2:milnersys}

In this Section we develop a proof-theoretic interpretation of the coinductive variant system \coindmilnersys\ of \milnersys\
  in Milner's original system $\milnersys$.
Since $\coindmilnersys$ and $\milnersys$ differ only 
  by the coinductive rule scheme $\family{\LCoindProofi{n}}{n\in\nat}$ (which is part of $\coindmilnersys$, but not of $\milnersys$),
  and 
  by the fixed-point rule \RSPstar\ (which is part of $\milnersys$, but not of $\coindmilnersys$),
  the crucial step for this proof transformation is to show that instances of $\LCoindProofi{n}$, for $n\in\nat$,
  can be mimicked in $\milnersys$
  if their premise equations 
  are derivable in \milnersys.
We will do so by showing that the rules $\LCoindProofi{n}$, for $n\in\nat$, 
  are correct for \milnersys.
This implies that these rules are admissible in \milnersys\ (by Lemma~\ref{lem:derivable:admissible:rules}), 
  and also, that they can be eliminated from derivations in $\thplus{\milnersys}{\family{\LCoindProofi{n}}{n\in\nat}}$,
    which is an extension~of~\coindmilnersys. 
In this way we obtain the proof-theoretic interpretation of $\coindmilnersys$ in $\milnersys$.    

\begin{figure}[t]
\begin{center}
\begin{tikzpicture}
  \matrix[anchor=center,row sep=0.45cm,column sep=0.6cm,ampersand replacement=\&] {
      \& 
        \& 
          \& \node(top){\large $\aonechart$};
    \\
    \node(assm-pos){};
      \&
        \&
          \&
            \&
              \&
                \& \node(concl-pos){};
    \\
     \& \node(leftleft){\large \phantom{$\chartof{\astexpi{1}}$}};
        \& \node(left){\large $\factorset{\chartof{\astexpi{1}}}{\eqin{\milnersys}}$};
          \&
            \& \node(right){\large $\factorset{\chartof{\astexpi{2}}}{\eqin{\milnersys}}$}; 
              \& \node(rightright){\large \phantom{$\chartof{\astexpi{2}}$}};
    \\    
    }; 
  \draw[funonebisimleft] (top) to node[above,xshift=-1.55em,pos=0.5]{\black{$\eqcl{\saeqfuni{1}}{\seqin{\milnersys}}$}} (left);
  \draw[funonebisimright] (top) to node[above,xshift=1.55em,pos=0.5]{\black{$\eqcl{\saeqfuni{2}}{\seqin{\milnersys}}$}} (right);  
    
  \path (assm-pos) ++ (-0.5cm,0.4cm) node(assm)
         {$\left.
           \parbox{\widthof{\large coinductive proof}}
                  {{\large \hspace*{\fill}$\pair{\aonechart}{\saeqfun}$\hspace*{\fill}\mbox{}}
                   \\
                   \hspace*{\fill}\forestgreen{\LLEE}-witnessed\hspace*{\fill}\mbox{}
                   \\[-0.1ex]
                   \hspace*{\fill}coinductive proof\hspace*{\fill}\mbox{}
                   \\[-0.1ex]
                   \hspace*{\fill}of $\astexpi{1} \formeq \astexpi{2}$\hspace*{\fill}\mbox{}
                   \\[-0.1ex]
                   \hspace*{\fill}over $\asys\isthmsubsumedby\milnersys$\hspace*{\fill}\mbox{}}
           \right\}$}; 
           
  \path (top) ++ (0.1cm,0.3cm) node[right]{\forestgreen{LLEE}, guarded}; 
    
  \path (concl-pos) ++ (0cm,0.4cm) node(concl)
                                         {\large $\astexpi{1} \eqin{\milnersys} \astexpi{2}$};  
    
  \draw[-implies,double equal sign distance] 
    (assm) 
      to node[below,yshift=-0.125cm,xshift=0.2cm]{\tiny (Prop.$\,$\ref{prop:char:solution:SRS})} 
    ($(assm.east) + (0.75cm,0cm)$);  
  \draw[-implies,double equal sign distance,shorten >=4pt] 
    ($(concl.west) + (-0.85cm,0cm)$) 
      to node[below,yshift=-0.125cm,xshift=1cm]{\parbox{\widthof{\tiny (together with ``$\Leftarrow$'' in Prop.~\ref{prop:char:solution:SRS}}}
                                                          {\tiny (Prop.$\,$\ref{prop:LLEEcoindproofeq:impl:milnersyseq}, together
                                                                 \\\phantom{(}%
                                                                   with ``$\Leftarrow$'' 
                                                                   in Prop.~\ref{prop:char:solution:SRS}
                                                                 \\\phantom{(}%
                                                                   for $\asys = \milnersys$)}} 
    (concl); 
    
\end{tikzpicture}  
\end{center}
  \vspace*{-2.5ex}
  \caption{\label{fig:coindmilnersys:2:milnersys}%
           Statement underlying the proof transformation from \protect\coindmilnersys\ to \protect\milnersys.
           }
\end{figure}%
For proving correctness of $\LCoindProofi{n}$ for \milnersys, where $n\in\nat$,
  we will show that every \LLEEwitnessed\ coinductive proof $\pair{\aonechart}{\saeqfun}$
  over a proof system $\asys$ that is \theoremsubsumed\ by \milnersys\
  of an equation $\astexpi{1} \formeq \astexpi{2}$ 
  can also be established by a proof of $\astexpi{1} \formeq \astexpi{2}$ in Milner's system~$\milnersys$. 
Informally, this statement is illustrated in Figure~\ref{fig:coindmilnersys:2:milnersys}
  by informally employing the characterization in Proposition~\ref{prop:char:solution:SRS} of the provable solutions
  in such \LLEEwitnessed\ coinductive proofs. We establish the indicated informal second step in this section,
    where it will be guaranteed by Proposition~\ref{prop:LLEEcoindproofeq:impl:milnersyseq}.
In particular, our proof of this step will use the following two statements:
\begin{enumerate}[label={{\bf (SE)$_{\protect\onescriptbs}$}},leftmargin=*,align=right,labelsep=0.75ex,itemsep=0.25ex] 
  \item[{\crtcrossreflabel{{\bf (SE)}}[SE]}] \emph{(solution extraction)}
    from a \LLEEwitness~$\aonecharthat$ of $\aonechart$ a \provablein{\milnersys} solution $\sextrsolof{\aonechart}$ of $\aonechart$
      can be extracted (Lemma~\ref{lem:extrsol:is:sol}), and
  \item[{\crtcrossreflabel{{\bf (SU)}}[SU]}] \emph{(solution uniqueness)} 
    every \provablein{\milnersys} solution of the \LLEEonechart~$\aonechart$ (such as $\saeqfuni{1}$ and $\saeqfuni{2}$) is \provablyin{\milnersys} equal
      to the solution $\sextrsolof{\aonechart}$ extracted from $\aonecharthat$ (Lemma~\ref{lem:sols:provably:equal:LLEE}). 
\end{enumerate}
By these statements we will obtain
  for every \LLEEwitnessed\ coinductive proof $\pair{\aonechart}{\saeqfun}$ of $\astexpi{1} \formeq \astexpi{2}$,
    assuming that $\start$ is the star vertex of $\aonechart$ and hence $\aeqfuni{1}{\start} \synteq \astexpi{1}$ as well as $\aeqfuni{2}{\start} \synteq \astexpi{2}$ hold,
that     
$\astexpi{1} \synteq \aeqfuni{1}{\start} \milnersyseq \extrsolof{\aonecharthat}{\start}
                                         \milnersyseq \aeqfuni{2}{\start} \synteq \astexpi{2}$ holds,
and therefore $\astexpi{1} \milnersyseq \astexpi{2}$.                                             
  
The proofs of the statements \ref{SE} and \ref{SU} 
                                                   below are adaptations to \LLEEonecharts\ 
  of proofs of analogous statements for \LLEEcharts\ in Section~5 of \cite{grab:fokk:2020:lics,grab:fokk:2020:lics:arxiv}.
We have to, at places substantially, refine the extraction technique of star expressions from process graphs with \LLEE\
  that was first described in \cite{grab:fokk:2020:lics,grab:fokk:2020:lics:arxiv}.
However, we will use the simplification to only reason about {guarded} \LLEEwitnesses,
  in which \loopentrytransition\ are proper transitions. 
We can do so because the \onechart\ interpretation $\onechartof{\astexp}$ of a star expression $\astexp$ 
  is guaranteed to have a guarded \LLEEwitness\ $\onecharthatof{\astexp}$ by Theorem~\ref{thm:onechart-int:LLEEw}. 

For developing and proving the extraction statement \ref{SE} 
  we use that the  hierarchical loop structure of a \onechart~$\aonechart$
  with \LLEEwitness~$\aonecharthat$ 
  facilitates the extraction of a \provablein{\milnersysmin} solution of $\aonechart$
  (see Lemma~\ref{lem:extrsol:is:sol}). The reason is as follows. 
The process behavior at every vertex $\bvert$ in $\aonechart$ can be split 
  into an iteration part that is induced via the \loopentry\ transitions from $\bvert$ in $\aonecharthat$
                                         (which induce loop \subonecharts\ with inner loop \subonecharts\ whose behavior can be synthesized recursively),
  and an exit part that is induced via the body transitions from $\bvert$ in $\aonecharthat$.
This intuition leads us to the definition below.
We define the `extraction function' $\sextrsolof{\aonecharthat}$\vspace*{-2.5pt} of $\aonecharthat$
  from a `relative extraction function' $\sextrsoluntilof{\aonecharthat}$\vspace{-1.5pt} of $\aonecharthat$, 
  whose values $\extrsoluntilof{\aonecharthat}{\bvert}{\avert}$ capture the behavior at $\bvert$ 
  in a loop \subonechart\ at $\avert$ until $\avert$ is reached.

\begin{defi}[(relative) extraction function]\label{def:extrsoluntil:extrsol}
  Let $\aonechart = \tuple{\verts,\actions,\sone,\start,\transs,\termexts}$ be a (guarded) \LLEEonechart\ with guarded \LLEEwitness~$\aonecharthat$. 
  
  The \emph{extraction function 
            $\sextrsolof{\aonecharthat} \funin \verts \to \StExpover{\actions}$ 
            of $\aonecharthat$}
  is defined from the\vspace{-1pt} 
  \emph{relative extraction function  
        $\sextrsoluntilof{\aonecharthat} \funin \descsetexp{\pair{\bvert}{\avert}}
                                                        {\bvert,\avert\in\vertsof{\aonechart},\,\bvert\convdescendsinlooptosc\avert} 
                                                \to \StExpover{\actions} $
        of $\aonecharthat$}
  for $\bvert,\avert\in\verts$:\vspace{-1.5ex}
  \begin{align*}
    \extrsoluntilof{\aonecharthat}{\bvert}{\avert}
      & {} \defdby
        \begin{cases}
          \stexpone
            & \text{if $\bvert = \avert$,}
          \\[-0.5ex] 
          \stexpprod{\displaystyle
            \stexpit{\Bigl(
                       \sum_{i=1}^{n}
                         \stexpprod{\aacti{i}}{\extrsoluntilof{\aonecharthat}{\bverti{i}}{\bvert}}
                     \Bigr)}
                     }{\displaystyle
            \Bigl(
              {\sum_{i=1}^{m}
                 \stexpprod{\boneacti{i}}{\extrsoluntilof{\aonecharthat}{\cverti{i}}{\avert}}}
            \Bigr)              
                       }
            & \text{if $\bvert \convdescendsinloopto \avert$,}
        \end{cases}
      \displaybreak[0]\\[-0.5ex]  
      \extrsolof{\aonecharthat}{\bvert}
        & {} \defdby
          \stexpprod{%
            \stexpit{\Bigl(
                       \sum_{i=1}^{n}
                         \stexpprod{\aacti{i}}{\extrsoluntilof{\aonecharthat}{\bverti{i}}{\bvert}}
                     \Bigr)}
                     }{
            \Bigl(
              \stexpsum{\terminatesconstof{\aonechart}{\bvert}}
                       {\sum_{i=1}^{m}
                          \stexpprod{\boneacti{i}}{\extrsolof{\aonecharthat}{\cverti{i}}}}
            \Bigr)              
                       } \punc{,}
     \\[-0.5ex]
     & \hspace*{-8ex} \text{provided: } 
       \transitionsinfrom{\aonecharthat}{\bvert}
         =
       \begin{aligned}[t]
         &
         \descsetexpbig{ \bvert \lti{\aacti{i}}{\looplab{\alabi{i}}} \bverti{i} }
                       { \alabi{i}\in\natplus,\,  i\in\setexp{1,\ldots,n} }
           {\cup} 
         \descsetexpbig{ \bvert \lti{\boneacti{i}}{\looplab{\bodylab}} \cverti{i} }
                       { i\in\setexp{1,\ldots,m} } \text{,}  
       \end{aligned}  
     \\[-0.5ex]
     & \hspace*{-8ex}\text{induction for $\sextrsoluntilof{\aonecharthat}$ on: }
     \begin{aligned}[t]  
       \pair{\bverti{1}}{\averti{1}}
          \lexspo
        \pair{\bverti{2}}{\averti{2}}
          \;\funin\: \Longleftrightarrow\:
            \averti{1} \convdescendsinlooptotc \averti{2}
              \logor
            (\, \averti{1} = \averti{2}
                  \logand 
                \bverti{1} \convredtci{\bodylab} \bverti{2} \,) \punc{,}
     \end{aligned}                        
     \\[-0.75ex]
     & \hspace*{-8ex}\text{induction for $\sextrsolof{\aonecharthat}$ on the strict partial order $\sconvredtci{\bodylab}$ (see Lemma~\ref{lem:descsteps:bodysteps:wf})}
         \punc{,}
  \end{align*}
  where $\slexspo$ is a well-founded strict partial order due to Lemma~\ref{lem:descsteps:bodysteps:wf}.
  The choice of the list representations of \actiontarget\ sets of $\aonecharthat$ 
  changes the definitions of these functions 
                          only up to provability~in~\ACI.
\end{defi}

We exemplify the extraction process defined above by a concrete example.

\begin{exa}\label{ex:extraction}\mbox{}%
  \begin{figure}[tbp]
\begin{center}
  \begin{equation*}
    \hspace*{-3ex}
    \begin{aligned}[c]
      \scalebox{1.25}{\begin{tikzpicture}\renewcommand{\stexpprod}[2]{{#1}\hspace*{1pt}{\sstexpprod}\hspace*{1pt}{#2}}
  
\matrix[anchor=center,row sep=1cm,column sep=0.75cm,
        every node/.style={draw,very thick,circle,minimum width=2.5pt,fill,inner sep=0pt,outer sep=2pt}] {
  \node(v11){};
    & & \node(v21){};
  \\                 
  \node(v1){};
    & & \node(v2){};
  \\
    & \node[chocolate](v){};
  \\
  };   
\path (v) ++ (0cm,-1cm) node(label){\Large $\aonechart$, $\aonecharthat$};

\path (v11) ++ (0cm,0.25cm) node{\small $\averti{11}$};
\draw[->,thick,densely dotted,out=180,in=180,distance=0.75cm](v11) to (v1);

\path (v1) ++ (-0.225cm,-0.25cm) node{\small $\averti{1}$};
\draw[->,thick,darkcyan,shorten >= 0pt]
  (v1) to node[left,pos=0.25,xshift=0.075cm]{\small $\black{\aact}$} node[left,pos=0.6,xshift=0.075cm]{\small $\loopnsteplab{1}$} (v11);
\draw[->,thick,densely dotted,out=-90,in=180,distance=0.5cm,shorten >=2pt](v1) to (v);
\draw[->,shorten <= 0pt,shorten >= 0pt] (v1) to node[above]{$\bact$} (v21); 

\path (v21) ++ (0cm,0.25cm) node{\small $\averti{21}$};
\draw[->,thick,densely dotted,out=-0,in=0,distance=0.75cm](v21) to (v2);  
  
\path (v2) ++ (0.25cm,-0.25cm) node{\small $\averti{2}$};
\draw[->,thick,darkcyan,shorten >= 0pt]
  (v2) to node[right,pos=0.25,xshift=-0.075cm]{\small $\black{\bact}$} node[right,pos=0.6,xshift=-0.075cm]{\small $\loopnsteplab{1}$} (v21);
\draw[->,thick,densely dotted,out=-90,in=0,distance=0.5cm,shorten >= 2pt](v2) to (v);

\draw[<-,very thick,>=latex,chocolate,shorten <= 2pt](v) -- ++ (-90:0.55cm);   
\draw[thick,chocolate] (v) circle (0.12cm);
\path (v) ++ (0.25cm,-0.3cm) node{\small $\start$};
\draw[->,thick,darkcyan,shorten >= 0.175cm,shorten <= 2pt] 
  (v) to node[left,pos=0.36,xshift=0.075cm]{\small $\black{\aact}$} node[right,pos=0.4,xshift=-0.075cm,yshift=1pt]{\small $\loopnsteplab{2}$}  (v11);
\draw[->,thick,darkcyan,shorten >= 0.175cm,shorten <= 2pt] 
  (v) to node[right,pos=0.36,xshift=-0.05cm]{\small $\black{\bact}$} node[left,pos=0.6,xshift=0.075cm,yshift=1pt]{\small $\loopnsteplab{2}$} (v21);
  
\end{tikzpicture} }
    \end{aligned}
    \hspace*{1ex}
    \begin{aligned}[c]
      \extrsoluntilof{\aonecharthat}{\averti{21}}{\averti{2}}
        & \;\parbox[t]{\widthof{$\milnersysmineq$}}{$\defdby$}\,
      \stexpprod{\stexpit{\stexpzero}}  
                {(\stexpprod{\stexpone}{\extrsoluntilof{\aonecharthat}{\averti{2}}{\averti{2}}})}
        \,\synteq\,
      \stexpprod{\stexpit{\stexpzero}}{(\stexpprod{\stexpone}{\stexpone})}
        \,\milnersysmineq\,
      \stexpone
      \\[-0.1ex]
      \extrsoluntilof{\aonecharthat}{\averti{2}}{\start}
        & \;\parbox[t]{\widthof{$\milnersysmineq$}}{$\defdby$}\,
      \stexpprod{\stexpit{(\stexpprod{\bact}{\extrsoluntilof{\aonecharthat}{\averti{21}}{\averti{2}}})}}
                {(\stexpprod{\stexpone}{\extrsoluntilof{\aonecharthat}{\start}{\start}})}
        \,\milnersysmineq\,  
      \stexpit{\bact}
      \\[-0.1ex]
      \extrsoluntilof{\aonecharthat}{\averti{21}}{\start}
        & \;\parbox[t]{\widthof{$\milnersysmineq$}}{$\defdby$}\,
      \stexpprod{\stexpit{\stexpzero}}  
                {(\stexpprod{\stexpone}{\extrsoluntilof{\aonecharthat}{\averti{2}}{\start}})}
        \,\milnersysmineq\,
      \stexpprod{\stexpone}{\stexpit{\bact}}
        \,\milnersysmineq\,
      \stexpit{\bact}
      \\[-0.1ex]
      \extrsoluntilof{\aonecharthat}{\averti{11}}{\averti{1}}
        & \;\parbox[t]{\widthof{$\milnersysmineq$}}{$\defdby$}\,
      \stexpprod{\stexpit{\stexpzero}}
                {(\stexpprod{\stexpone}{\extrsoluntilof{\aonecharthat}{\averti{1}}{\averti{1}}})}
        \,\synteq\,
      \stexpprod{\stexpit{\stexpzero}}
                {(\stexpprod{\stexpone}{\stexpone})}
        \,\milnersysmineq\,
      \stexpone
      \\[-0.1ex]
      \extrsoluntilof{\aonecharthat}{\averti{1}}{\start}
        & \;\parbox[t]{\widthof{$\milnersysmineq$}}{$\defdby$}\,
      \stexpprod{\stexpit{(\stexpprod{\aact}{\extrsoluntilof{\aonecharthat}{\averti{11}}{\start}})}}
                {(\stexpsum{\stexpprod{\bact}{\extrsoluntilof{\aonecharthat}{\averti{21}}{\start}}}
                           {\stexpprod{\stexpone}{\extrsoluntilof{\aonecharthat}{\start}{\start}}})}
      \\[-0.5ex] 
        & \;\parbox[t]{\widthof{$\milnersysmineq$}}{$\milnersysmineq$}\,
      \stexpprod{\stexpit{\aact}}
                {(\stexpsum{\stexpprod{\bact}{\stexpit{\bact}}}
                           {\stexpone})}
        \,\milnersysmineq\,
      \stexpprod{\stexpit{\aact}}{\stexpit{\bact}} 
      \\[-0.1ex]
      \extrsoluntilof{\aonecharthat}{\averti{11}}{\start}
        & \;\parbox[t]{\widthof{$\milnersysmineq$}}{$\defdby$}\,
      \stexpprod{\stexpit{\stexpzero}}  
                {(\stexpprod{\stexpone}{\extrsoluntilof{\aonecharthat}{\averti{1}}{\start}})}
        \,\milnersysmineq\,
      \stexpprod{\stexpit{\aact}}{\stexpit{\bact}}
      \\[-0.1ex]
      \extrsolof{\aonecharthat}{\start}
        & \;\parbox[t]{\widthof{$\milnersysmineq$}}{$\defdby$}\,
      \stexpprod{\stexpit{(\stexpsum{\stexpprod{\aact}{\extrsoluntilof{\aonecharthat}{\averti{11}}{\start}}}
                                    {\stexpprod{\bact}{\extrsoluntilof{\aonecharthat}{\averti{21}}{\start}}})}}
                {\stexpone}
      \\[-0.5ex]
        & \;\parbox[t]{\widthof{$\milnersysmineq$}}{$\milnersysmineq$}\,
     \stexpit{(\stexpsum{\stexpprod{\aact}{(\stexpprod{\stexpit{\aact}}{\stexpit{\bact}})}}
                        {\stexpprod{\bact}{\stexpit{\bact}}})}
    \end{aligned}
  \end{equation*}     
\end{center}
    \vspace*{-2.5ex}
  \caption{\protect\label{fig:ex:extraction}%
           Extraction of the principal value 
                                             of a \protect\provablein{\protect\milnersysmin} solution~$\protect\sextrsolof{\protect\aonecharthat}$ 
           from\protect\vspace*{-2pt} the \protect\LLEEwitness~$\protect\acharthat$ in the coinductive proof 
           in Example~\protect\ref{ex:2:LLEEcoindproof},
           with \protect\provablein{\milnersysmin}~simplifications.
           }

\end{figure}      %
  We consider the \onechart~$\aonechart$,
  and the \LLEEwitness~$\aonecharthat$ of $\aonechart$,
  in the \LLEEwitnessed\ coinductive proof $\aCoProof = \pair{\aonechart}{\saeqfun}$ of 
  $\stexpit{(\stexpprod{\stexpit{\aact}}
                       {\stexpit{\bact}})}
     \formeq
   \stexpit{(\stexpsum{\aact}{\bact})}$
  in Example~\ref{ex:2:LLEEcoindproof}.
  We detail in Figure~\ref{fig:ex:extraction} the process of computing
  the principal value $\extrsolof{\aonecharthat}{\start}$ of the
  extraction function $\sextrsolof{\aonecharthat}$ of~$\aonecharthat$. 
  The statement of Lemma~\ref{lem:extrsol:is:sol} below will guarantee that $\sextrsolof{\aonecharthat}$
  is a \provablein{\milnersysmin} solution~of~$\aonechart$.
\end{exa}

In order to show
  that the extraction function of a guarded \LLEEwitness\ of a \onechart~$\aonechart$
  defines a \provablein{\milnersysmin} solution of $\aonechart$,
    see Lemma~\ref{lem:extrsol:is:sol} and its proof later,
we first have to establish a \provablein{\milnersysmin} connection
  between the relative extraction function and the extraction function of a guarded \LLEEwitness. 
For this we prove the following lemma.

\begin{lem}\label{lem:lem:extrsol:is:sol}
  Let $\aonechart$ be a (guarded) \LLEEonechart\ with guarded \LLEEwitness~$\aonecharthat$. 
  Then 
  $\extrsolof{\aonecharthat}{\bvert}
     \milnersysmineq 
   \stexpprod{\extrsoluntilof{\aonecharthat}{\bvert}{\avert}}
             {\extrsolof{\aonecharthat}{\avert}}$
  holds 
  for all vertices $\bvert,\avert\in\vertsof{\aonechart}$ such that $\bvert \convdescendsinlooptosc \avert$.
\end{lem}

\begin{proof}
  Let $\aonechart = \tuple{\verts,\actions,\sone,\start,\transs,\exts}$ be a \LLEEonechart\
    with guarded \LLEEwitness\ $\aonecharthat$.
  
  We have to show that \mbox{}
  $\extrsol{\bvert}
     \eqin{\milnersysmin}
   \stexpprod{\extrsoluntilof{\aonecharthat}{\bvert}{\avert}}
             {\extrsol{\avert}}$ \mbox{}
  holds for all $\bvert,\avert\in\verts$ with $\bvert \convdescendsinlooptosc \avert$.
  We first notice that this statement holds obviously for $\bvert = \avert$,
  due $ \extrsoluntilof{\aonecharthat}{\bvert}{\avert} \synteq \extrsoluntilof{\aonecharthat}{\avert}{\avert} \synteq \stexpone$,
  and the presence of the axiom $(\leftidstexpprod)$ in $\milnersysmin$. 
  Therefore it suffices to show, by also using this fact, that
  $\extrsol{\bvert}
     \eqin{\milnersysmin}
   \stexpprod{\extrsoluntilof{\aonecharthat}{\bvert}{\avert}}
             {\extrsol{\avert}}$ \mbox{}
  holds for all $\bvert,\avert\in\verts$ with $\bvert \convdescendsinloopto \avert$.
  We will show this by using the same induction as for the definition of the relative extraction function $\sextrsoluntilof{\aonecharthat}$ in Definition~\ref{def:extrsoluntil:extrsol}, that is,
    by complete induction on the (converse) lexicographic partial order $\slexspo$ of $\sconvdescendsinlooptotc$ and $\sconvredtci{\bodylab}$
    on $\verts\times\verts$ defined by:
    $\pair{\bverti{1}}{\averti{1}}
          \lexspo
        \pair{\bverti{2}}{\averti{2}}
          \;\funin\: \Longleftrightarrow\:
            \averti{1} \convdescendsinlooptotc \averti{2}
              \logor
            (\, \averti{1} = \averti{2}
                  \logand 
                \bverti{1} \convredtci{\bodylab} \bverti{2} \,) $,
  which is \wellfounded\ by Lemma~\ref{lem:descsteps:bodysteps:wf}.               
  For our argument we assume to have given, 
    underlying the definition of the relative extraction function $\sextrsoluntilof{\aonecharthat}$ 
    and the extraction function $\sextrsolof{\aonecharthat}$,
  list\vspace{-2pt} representations $\transitionsinfrom{\aonecharthat}{\bvert}$ of the transitions from $\bvert$ in $\aonecharthat$ 
  as in Definition~\ref{def:extrsoluntil:extrsol},
  for all $\bvert\in\verts$.\vspace{1pt}
  
  In order to carry out the induction step, 
  we let $\bvert,\avert\in\verts$ be arbitrary, but such that $\bvert \convdescendsinloopto \avert$. 
  On the basis of the form of $\transitionsinfrom{\aonecharthat}{\bvert}$ 
  as in Definition~\ref{def:extrsoluntil:extrsol}
  we argue as follows,
  starting with a step in which we use the definition of $\sextrsolof{\aonecharthat}$,
  and followed by a second step in which we use that $\terminatesconstof{\aonechart}{\bvert} \synteq \stexpzero$ holds,
    because $\bvert$ cannot have immediate termination as due to $\bvert \convdescendsinloopto \avert$ it is
    in the body of the loop at $\avert$ (cf.\ condition~(L3) for loop \onecharts\ in Section~\ref{LEE}):
  \begin{align*}
    \extrsolof{\aonecharthat}{\bvert}
      & \;\,\parbox[t]{\widthof{$\eqin{\milnersysmin}$}}{$\synteq$}\:
      \stexpprod{\Bigl(
                   \sum_{i=1}^{n} \stexpprod{\aacti{i}}{\extrsoluntilof{\aonecharthat}{\bverti{i}}{\bvert}}
                 \Bigr)^{\sstar} 
                 }{\Bigr(\stexpsum{\Bigr(\displaystyle
                                     \sum_{i=1}^{m} 
                                       \stexpprod{\boneacti{i}}{\extrsolof{\aonecharthat}{\cverti{i}}}
                                   \Bigl)}
                                  {\terminatesconstof{\aonechart}{\bvert}}
                   \Bigl)}  
    \displaybreak[0]\\
      & \;\,\parbox[t]{\widthof{$\eqin{\milnersysmin}$}}{$\synteq$}\:
      \stexpprod{\Bigl(
                   \sum_{i=1}^{n} \stexpprod{\aacti{i}}{\extrsoluntilof{\aonecharthat}{\bverti{i}}{\bvert}}
                 \Bigr)^{\sstar} 
                 }{\Bigr(\stexpsum{\Bigr(\displaystyle
                                     \sum_{i=1}^{m} 
                                       \stexpprod{\boneacti{i}}{\extrsolof{\aonecharthat}{\cverti{i}}}
                                   \Bigl)}
                                  {\stexpzero}
                   \Bigl)}  
    \displaybreak[0]\\
      & \;\,\parbox[t]{\widthof{$\eqin{\milnersysmin}$}}{$\:\milnersysmineq$}\:
      \stexpprod{\Bigl(
                   \sum_{i=1}^{n} \stexpprod{\aacti{i}}{\extrsoluntilof{\aonecharthat}{\bverti{i}}{\bvert}}
                 \Bigr)^{\sstar} 
                 }{\Bigr(\displaystyle
                     \sum_{i=1}^{m} 
                       \stexpprod{\boneacti{i}}{\extrsolof{\aonecharthat}{\cverti{i}}}
                   \Bigl)} 
        \\
        & \;\,\parbox[t]{\widthof{$\eqin{\milnersysmin}$\hspace*{3ex}}}{\mbox{}}\:
          \text{(by axiom $(\neutralstexpsum)$)}
    \displaybreak[0]\\
      & \;\,\parbox[t]{\widthof{$\eqin{\milnersysmin}$}}{$\eqin{\milnersysmin}$}\:
      \stexpprod{\Bigl(
                   \sum_{i=1}^{n} \stexpprod{\aacti{i}}{\extrsoluntilof{\aonecharthat}{\bverti{i}}{\bvert}}
                 \Bigr)^{\sstar}
                 }{\sum_{i=1}^{m} 
                     \stexpprod{\boneacti{i}}
                               {\stexpprod{(\extrsoluntilof{\aonecharthat}{\cverti{i}}{\avert}}
                                          {\extrsolof{\aonecharthat}{\avert})}}}
        \notag\\
        & \;\,\parbox[t]{\widthof{$\eqin{\milnersysmin}$\hspace*{3ex}}}{\mbox{}}\:
          \parbox[t]{\widthof{\phantom{(if }%
                              imply $\cverti{i} \convdescendsinloopto \avert$ and $\cverti{i} \convredi{\bodylab} \bvert$,
                              and hence $\pair{\cverti{i}}{\avert}  \lexspo  \pair{\bvert}{\avert}$) holds)}}
                         {$\bigl($if $\cverti{i} = \avert$, then 
                             $\extrsolof{\aonecharthat}{\cverti{i}}
                                \milnersysmineq
                              \stexpprod{\extrsoluntilof{\aonecharthat}{\cverti{i}}{\avert}}
                                        {\extrsolof{\aonechart}{\avert})}$
                             due to $\extrsoluntilof{\aonecharthat}{\avert}{\avert} = 1$;   
                          \\
                          \phantom{(}if $\cverti{i} \neq \avert$, we can apply the induction hypothesis
                          to $\extrsolof{\aonecharthat}{\cverti{i}}$,
                          \\
                          \phantom{(if }%
                          as $\bvert \redi{\bodylab} \cverti{i}$ (see $\transitionsinfrom{\aonecharthat}{\bvert}$ 
                                                                      as in Def.~\ref{def:extrsoluntil:extrsol})
                           and $\cverti{i} \neq \avert$
                          \\
                          \phantom{(if }%
                          imply $\cverti{i} \convdescendsinloopto \avert$,
                          and $\cverti{i} \convredi{\bodylab} \bvert$
                            entails $\pair{\cverti{i}}{\avert}  \lexspo  \pair{\bvert}{\avert}$)} 
    \displaybreak[0]\\  
      & \;\,\parbox[t]{\widthof{$\eqin{\milnersysmin}$}}{$\eqin{\milnersysmin}$}\:
      \stexpprod{\Bigl(\Bigl(
                   \sum_{i=1}^{n} \stexpprod{\aacti{i}}{\extrsoluntilof{\aonecharthat}{\bverti{i}}{\bvert}}
                 \Bigr)^{\sstar} 
                 }{\stexpprod{\Bigr(
                                \sum_{i=1}^{m} 
                                  \stexpprod{\boneacti{i}}
                                            {\extrsoluntilof{\aonecharthat}{\cverti{i}}{\avert}}
                              \Bigl)\Bigr)}
                             {\extrsolof{\aonecharthat}{\avert}}
                   }
        \\
        & \;\,\parbox[t]{\widthof{$\eqin{\milnersysmin}$\hspace*{3ex}}}{\mbox{}}\:
          \text{(by axioms $(\rdistr)$, and $(\assocstexpprod)$)} 
    \displaybreak[0]\\[1ex]
      & \;\,\parbox[t]{\widthof{$\eqin{\milnersysmin}$}}{$\synteq$}\:
      \stexpprod{ \extrsoluntilof{\aonecharthat}{\bvert}{\avert} }
                  { \extrsolof{\aonecharthat}{\avert} }
        \\
        & \;\,\parbox[t]{\widthof{$\eqin{\milnersysmin}$\hspace*{3ex}}}{\mbox{}}\: 
          \text{(by $\bvert \convdescendsinloopto \avert$, 
                    and the definition of $\extrsoluntilof{\aonecharthat}{\bvert}{\avert}$ in Def.~\ref{def:extrsoluntil:extrsol})}
  \end{align*}
  We note that this reasoning also applies for the special cases $n=0$,
   and with a slight change also for $m=0$, 
    where
    $\sum_{i=1}^{m} 
       \stexpprod{\boneacti{i}}{\extrsolof{\aonecharthat}{\cverti{i}}} 
       \synteq
     \sum_{i=1}^{m} 
       \stexpprod{\boneacti{i}}
                 {\stexpprod{(\extrsoluntilof{\aonecharthat}{\cverti{i}}{\avert}}
                            {\extrsolof{\aonecharthat}{\avert})}}
       \synteq
     \sum_{i=1}^{m} 
     \stexpprod{\boneacti{i}}
               {\extrsoluntilof{\aonecharthat}{\cverti{i}}{\avert}}  
       \synteq
     \stexpzero$,
    and then an axiom ($\stexpzerostexpprod$) has to be used.   
  In this way we have shown, due to $\ACI \subsystem \milnersysmin$, the desired \provablein{\milnersysmin} equality 
  $\asol{\bvert}
     \eqin{\milnersysmin}
   \stexpprod{\extrsoluntilof{\aonecharthat}{\bvert}{\avert}}
             {\asol{\avert}}$ \mbox{}
  for the vertices $\avert$ and $\bvert$ that we picked with the property $\bvert \convdescendsinloopto \avert$.  
  
  Since $\bvert,\avert\in\verts$ with $\bvert \convdescendsinloopto \avert$ were arbitrary above, 
  we have successfully carried out the proof by induction that
  \mbox{}
  $\extrsolof{\aonecharthat}{\bvert}
     \eqin{\milnersysmin}
   \stexpprod{\extrsoluntilof{\aonecharthat}{\bvert}{\avert}}
             {\extrsolof{\aonechart}{\avert}}$ \mbox{}
  holds for all $\bvert,\avert\in\verts$ with $\bvert \convdescendsinloopto \avert$.
  As we have argued that the statement also holds for $\bvert = \avert$, we have proved the lemma.
\end{proof}

\begin{lem}[extracted function is provable solution]\label{lem:extrsol:is:sol}
  Let $\aonechart$ 
    be a 
         \LLEEonechart\ with guarded \LLEEwitness\ $\aonecharthat$.
  Then the extraction function $\sextrsolof{\aonecharthat}$ of $\aonecharthat$ 
  is a \provablein{\milnersysmin} solution~of~$\aonechart$. 
\end{lem}

\begin{proof}
  Let $\aonechart = \tuple{\verts,\actions,\sone,\start,\transs,\exts}$ be a (guarded) \LLEEonechart\
    with guarded \LLEEwitness~$\aonecharthat$.
     
  We show that the extraction function $\sextrsolof{\aonecharthat}$ of $\aonecharthat$ 
  is a \provablein{\milnersysmin} solution of $\aonechart$
  by verifying the \provablein{\milnersysmin} correctness conditions for $\sextrsolof{\aonecharthat}$
  at every vertex $\bvert\in\verts$. 
  For the argument we assume to have given, 
    underlying the definition of the relative extraction function $\sextrsoluntilof{\aonecharthat}$ 
    and\vspace{-2pt} the extraction function $\sextrsolof{\aonecharthat}$,
  list representations $\transitionsinfrom{\aonecharthat}{\bvert}$ of the transitions from $\bvert$ in $\aonecharthat$ 
  as written in Definition~\ref{def:extrsoluntil:extrsol},
  for all vertices $\bvert\in\verts$. 
 
  We let $\bvert\in\verts$ be arbitrary. 
  Starting from the definition of $\sextrsolof{\aonecharthat}$ in Definition~\ref{def:extrsoluntil:extrsol} 
    on the basis of the form of $\transitionsinfrom{\aonecharthat}{\bvert}$,
  we argue by the following steps:  
  \begin{align*}
    \extrsolof{\aonecharthat}{\bvert}
      & \;\,\parbox[t]{\widthof{$\eqin{\milnersysmin}$}}{$\:\synteq\:$}\:
      \stexpprod{\Bigl(
                   \sum_{i=1}^{n} \stexpprod{\aacti{i}}{\extrsoluntilof{\aonecharthat}{\bverti{i}}{\bvert}}
                 \Bigr)^{\sstar}
                         }{ 
                     \:       
                     \Bigr(\displaystyle
                     \stexpsum{\Bigl(
                                 \sum_{i=1}^{m} \stexpprod{\boneacti{i}}{\extrsolof{\aonecharthat}{\cverti{i}}}
                               \Bigr)}
                              {\terminatesconstof{\aonechart}{\bvert}}
                     \Bigl)}              
    \displaybreak[0]\\
    & \;\,\parbox[t]{\widthof{$\eqin{\milnersysmin}$}}{$\eqin{\milnersysmin}$}\:
      \stexpprod{\Bigl(
                 \stexpsum{\sone}
                          {\stexpprod{\Bigl(
                                        \sum_{i=1}^{n} \stexpprod{\aacti{i}}{\extrsoluntilof{\aonecharthat}{\bverti{i}}{\bvert}}
                                      \Bigr)\hspace*{-2pt}}
                                     {\hspace*{-2pt}
                                      \Bigl(
                                        \sum_{i=1}^{n} \stexpprod{\aacti{i}}{\extrsoluntilof{\aonecharthat}{\bverti{i}}{\bvert}}
                                      \Bigr)^{\sstar}}}
                 \Bigr)\hspace*{-2pt}}
                {\hspace*{-2pt}
                 \Bigr(\displaystyle
                 \stexpsum{\Bigl(
                             \sum_{i=1}^{m} \stexpprod{\boneacti{i}}{\extrsolof{\aonecharthat}{\cverti{i}}}
                           \Bigr)}
                          {\terminatesconstof{\aonechart}{\bvert}}
                 \Bigl)}
        \\
        & \;\,\parbox[t]{\widthof{$\eqin{\milnersysmin}$\hspace*{3ex}}}{\mbox{}}\:
          \text{(by axiom $(\recdefstexpit)$)}
    \displaybreak[0]\\
    & \;\,\parbox[t]{\widthof{$\eqin{\milnersysmin}$}}{$\eqin{\milnersysmin}$}\:
      \begin{aligned}[t]
        & 
        \Bigr( 
          \stexpsum{\Bigl(
                      \sum_{i=1}^{m} \stexpprod{\boneacti{i}}{\extrsolof{\aonecharthat}{\cverti{i}}}
                    \Bigr)}
                   {\terminatesconstof{\aonechart}{\bvert}}
        \Bigl)                     
        \\[-1ex]
        &
        \stexpsum{}
                 {\stexpprod{\Bigl(
                               \sum_{i=1}^{n} \stexpprod{\aacti{i}}{\extrsoluntilof{\aonecharthat}{\bverti{i}}{\bvert}}
                             \Bigr)}
                            {\stexpprod{\Bigl( 
                                        \Bigl(
                                          \sum_{i=1}^{n} \stexpprod{\aacti{i}}{\extrsoluntilof{\aonecharthat}{\bverti{i}}{\bvert}}
                                        \Bigr)^{\sstar}}
                                       {\Bigl(
                                          \stexpsum{\Bigl(
                                                      \sum_{i=1}^{m} \stexpprod{\boneacti{i}}{\extrsolof{\aonecharthat}{\cverti{i}}}
                                                    \Bigr)}
                                                   {\terminatesconstof{\aonechart}{\bvert}}
                                        \Bigr)}
                                        \Bigr)}}
      \end{aligned}
        \\
        & \;\,\parbox[t]{\widthof{$\eqin{\milnersysmin}$\hspace*{3ex}}}{\mbox{}}\:
          \text{(by axioms $(\rdistr)$, $(\leftidstexpprod)$, and $(\assocstexpprod)$)}
    \displaybreak[0]\\
    & \;\,\parbox[t]{\widthof{$\eqin{\milnersysmin}$}}{$\synteq$}\:
      \stexpsum{\Bigr( 
                  \stexpsum{\Bigl(
                              \sum_{i=1}^{m} \stexpprod{\boneacti{i}}{\extrsolof{\aonecharthat}{\cverti{i}}}
                            \Bigr)}
                           {\terminatesconstof{\aonechart}{\bvert}}
                \Bigl)}
               {\stexpprod{\Bigl(
                             \sum_{i=1}^{n} \stexpprod{\aacti{i}}{\extrsoluntilof{\aonecharthat}{\bverti{i}}{\bvert}}
                           \Bigr)}
                          {\extrsolof{\aonecharthat}{\bvert}}}
        \\
        & \;\,\parbox[t]{\widthof{$\eqin{\milnersysmin}$\hspace*{3ex}}}{\mbox{}}\:
          \text{(by definition of $\extrsolof{\aonecharthat}{\bvert}$ in Def.~\ref{def:extrsoluntil:extrsol})} 
        \displaybreak[0]\\
    & \;\,\parbox[t]{\widthof{$\eqin{\milnersysmin}$}}{$\eqin{\milnersysmin}$}\:
      \stexpsum{\Bigl(
                  \sum_{i=1}^{n} 
                    \stexpprod{\aacti{i}}
                              {\stexpprod{\bigl(\extrsoluntilof{\aonecharthat}{\bverti{i}}{\bvert}}
                                         {\extrsolof{\aonecharthat}{\bvert}\bigr)}}
                \Bigr)}
               {\Bigr( 
                  \stexpsum{\Bigl(
                              \sum_{i=1}^{m} \stexpprod{\boneacti{i}}{\extrsolof{\aonecharthat}{\cverti{i}}}
                            \Bigr)}
                           {\terminatesconstof{\aonechart}{\bvert}}
                \Bigl)}
        \\
        & \;\,\parbox[t]{\widthof{$\eqin{\milnersysmin}$\hspace*{3ex}}}{\mbox{}}\:
          \text{(by axioms $(\commstexpsum)$, $(\rdistr)$, and $(\assocstexpprod)$)} 
        \displaybreak[0]\\
    & \;\,\parbox[t]{\widthof{$\eqin{\milnersysmin}$}}{$\eqin{\milnersysmin}$}\:
      \stexpsum{\Bigl(
                  \sum_{i=1}^{n} 
                    \stexpprod{\aacti{i}}{\extrsolof{\aonecharthat}{\bverti{i}}}
                \Bigr)}
               {\stexpsum{\Bigl(
                            \sum_{i=1}^{m} \stexpprod{\boneacti{i}}{\extrsolof{\aonecharthat}{\cverti{i}}}
                          \Bigr)} 
                         {\terminatesconstof{\aonechart}{\bvert}}}
        \\
        & \;\,\parbox[t]{\widthof{$\eqin{\milnersysmin}$\hspace*{3ex}}}{\mbox{}}\:
          \parbox[t]{\widthof{(by Lemma~\ref{lem:lem:extrsol:is:sol}, due to $\bverti{i} \convdescendsinlooptosc \bvert$, 
                           which follows from $\bvert \loopnstepto{\alabi{i}} \bverti{i}$}}
                    {(by Lemma~\ref{lem:lem:extrsol:is:sol}, due to $\bverti{i} \convdescendsinlooptosc \bvert$, 
                      which follows from $\bvert \loopnstepto{\alabi{i}} \bverti{i}$%
                     \\
                     \phantom{(}(see $\transitionsinfrom{\aonecharthat}{\bvert}$
                                 as in Def.~\ref{def:extrsoluntil:extrsol}),
                                 and by axioms $(\assocstexpsum)\,$).}
    \\
    & \;\,\parbox[t]{\widthof{$\eqin{\milnersysmin}$}}{$\eqin{\ACI}$}\:
      \stexpsum{\terminatesconstof{\aonechart}{\bvert}}
               {\stexpsum{\Bigl(
                            \sum_{i=1}^{n} 
                              \stexpprod{\aacti{i}}{\extrsolof{\aonecharthat}{\bverti{i}}}
                          \Bigr)}
                         {\Bigl(
                            \sum_{i=1}^{m} \stexpprod{\boneacti{i}}{\extrsolof{\aonecharthat}{\cverti{i}}}
                          \Bigr)}} \punc{.}                        
  \end{align*}
  Since $\ACI \subsystem \milnersysmin$,
  this chain of equalities yields a \provablein{\milnersysmin} equality
  that establishes,
  in view of $\transitionsinfrom{\aonecharthat}{\bvert}$ 
             as in Definition~\ref{def:extrsoluntil:extrsol},
  the correctness condition for $\sextrsolof{\aonecharthat}$ to be a \provablein{\milnersysmin} solution 
  at the vertex $\bvert$ that we picked. 
  
  Since $\bvert\in\verts$ was arbitrary,
  we have established that the extraction function $\sextrsolof{\aonecharthat}$ of $\aonecharthat$
  is a \provablein{\milnersysmin} solution of $\aonechart$.
\end{proof}

For showing the solution uniqueness statement \ref{SU} 
  we can also use the hierarchical loop structure of a \onechart~$\aonechart$ with \LLEEwitness~$\aonecharthat$
    for carrying out proofs by induction.
We repurpose the two-step approach of the proof of Lemma~\ref{lem:extrsol:is:sol}
  that used the \provablein{\milnersys} relationship between
    the extraction function $\sextrsolof{\aonecharthat}$ of~$\aonecharthat$ and the relative\vspace*{-2pt} extraction function $\sextrsoluntilof{\aonecharthat}$ of~$\aonecharthat$
      in Lemma~\ref{lem:lem:extrsol:is:sol}. 
In doing so\vspace*{-0.5pt} we first establish, for every \provablein{\milnersys} solution $\sasol$ of $\aonechart$,
  a connection with the relative extraction function $\sextrsoluntilof{\aonechart}$ of $\aonecharthat$,
    see Lemma~\ref{lem:lem:sols:provably:equal:LLEE} below. 
The proof of this lemma proceeds by an induction that starts at innermost loop \subonecharts\ of the given \LLEEwitness,
  and then progresses to outer loop \subonecharts.
Different from Lemma~\ref{lem:lem:extrsol:is:sol},
  it will be crucial here to employ the fixed-point rule \RSPstar\ of \milnersys\ in the proof.

\begin{lem}\label{lem:lem:sols:provably:equal:LLEE}
  Let $\aonechart$ be a \LLEEonechart\ with guarded \LLEEwitness~$\aonecharthat$.
  Furthermore, let $\asys$ be an \eqlogicbasedover{\actions}\ proof system 
    such that $\ACI \subsystem \asys \isthmsubsumedby \milnersys$.
  
  Let $\sasol \funin \vertsof{\aonechart} \to \StExpover{\actions}$
      be an \provablein{\asys} solution of $\aonechart$.
  Then 
  $\asol{\bvert}
     \eqin{\milnersys}
   \stexpprod{\extrsoluntilof{\aonecharthat}{\bvert}{\avert}}
             {\asol{\avert}}$ \mbox{}
  holds for all vertices $\bvert,\avert\in\vertsof{\aonechart}$ with $\bvert \convdescendsinlooptosc \avert$.
\end{lem}

\begin{proof}
  Let $\aonecharthat$ be a guarded \LLEEwitness\ of a (guarded) \LLEEonechart~$\aonechart = \tuple{\verts,\actions,\sone,\start,\transs,\exts}$.
  Let $\sasol \funin \verts \to \StExpover{\actions}$ be an \provablein{\asys} solution of $\aonechart$.
  
  We have to show that \mbox{}
  $\asol{\bvert}
     \eqin{\milnersys}
   \stexpprod{\extrsoluntilof{\aonecharthat}{\bvert}{\avert}}
             {\asol{\avert}}$ \mbox{}
  holds for all $\bvert,\avert\in\verts$ with $\bvert \convdescendsinlooptosc \avert$.
  We first notice that this statement holds obviously for $\bvert = \avert$,
  due $ \extrsoluntilof{\aonecharthat}{\bvert}{\avert} \synteq \extrsoluntilof{\aonecharthat}{\avert}{\avert} \synteq \stexpone$. 
  Therefore it suffices to show, by also using this fact, that
  $\asol{\bvert}
     \eqin{\milnersys}
   \stexpprod{\extrsoluntilof{\aonecharthat}{\bvert}{\avert}}
             {\asol{\avert}}$ \mbox{}
  holds for all $\bvert,\avert\in\verts$ with $\bvert \convdescendsinloopto \avert$.
  We will show this by using the same induction as for the definition of the relative extraction function in Definition~\ref{def:extrsoluntil:extrsol}, that is,
    by complete induction on the (converse) lexicographic partial order $\slexspo$ of $\sconvdescendsinlooptotc$ and $\sconvredtci{\bodylab}$
    on $\verts\times\verts$ defined by:
    $\pair{\bverti{1}}{\averti{1}}
          \lexspo
        \pair{\bverti{2}}{\averti{2}}
          \;\funin\: \Longleftrightarrow\:
            \averti{1} \convdescendsinlooptotc \averti{2}
              \logor
            (\, \averti{1} = \averti{2}
                  \logand 
                \bverti{1} \convredtci{\bodylab} \bverti{2} \,) $,
  which is \wellfounded\ due to Lemma~\ref{lem:descsteps:bodysteps:wf}.               
  For our argument we suppose to have given, 
    underlying the definition of the relative extraction function $\sextrsoluntilof{\aonecharthat}$ 
    and the extraction function $\sextrsolof{\aonecharthat}$,
  list representations $\transitionsinfrom{\aonecharthat}{\bvert}$ of the transitions from $\bvert$ in $\aonecharthat$ 
  as written in Definition~\ref{def:extrsoluntil:extrsol},
  for all $\bvert\in\verts$. 
   
  In order to carry out the induction step, 
  we let $\bvert,\avert\in\verts$ be arbitrary such that $\bvert \convdescendsinloopto \avert$. 
  On the basis of the form of $\transitionsinfrom{\aonecharthat}{\bvert}$ 
                              as in Definition~\ref{def:extrsoluntil:extrsol}
  we argue as follows,
  starting with a step in which we use that $\sasol$ is an \provablein{\asys} solution of $\aonechart$,
  and followed by a second step in which we use that $\terminatesconstof{\aonechart}{\bvert} \synteq \stexpzero$ holds,
    because $\bvert$ cannot have immediate termination as due to $\bvert \convdescendsinloopto \avert$ it is
    in the body of the loop at $\avert$ (see condition~(L3) for loop \onecharts\ in Section~\ref{LEE}):
  \begin{align*}
    \asol{\bvert}
        &
        \;\,\parbox[t]{\widthof{$\eqin{\milnersysmin}$}}{$\eqin{\asys}$}\:
      \stexpsum{\terminatesconstof{\aonechart}{\bvert}}
               {\stexpsum{\Bigl(\Bigl(
                            \sum_{i=1}^{n}
                              \stexpprod{\aacti{i}}{ \asol{\bverti{i}} }
                          \Bigr)}
                         {\Bigl(
                            \sum_{i=1}^{m}
                              \stexpprod{\boneacti{i}}{ \asol{\cverti{i}} }
                          \Bigr)\Bigr)}}
    \displaybreak[0]\\
        &
        \;\,\parbox[t]{\widthof{$\eqin{\milnersysmin}$}}{$\eqin{\milnersysmin}$}\:
      \stexpsum{\Bigl(
                  \sum_{i=1}^{n}
                    \stexpprod{\aacti{i}}{ \asol{\bverti{i}} }
                \Bigr)}
               {\Bigl(
                  \sum_{i=1}^{m}
                    \stexpprod{\boneacti{i}}{ \asol{\cverti{i}} }
                \Bigr)}
    \displaybreak[0]\\                      
      &
      \;\,\parbox[t]{\widthof{$\eqin{\milnersysmin}$}}{$\eqin{\milnersys}$}\:  
    \stexpsum{\Bigl(
                \sum_{i=1}^{n}
                  \stexpprod{\aacti{i}}
                            {\bigl(\stexpprod{\extrsoluntilof{\aonecharthat}{\bverti{i}}{\bvert}}
                                             {\asol{\bvert}}\bigr)}
              \Bigr)}
             {\Bigl(
                \sum_{i=1}^{m}
                  \stexpprod{\boneacti{i}}
                            {\asol{\cverti{i}}}
              \Bigr)}
        \\
        & \;\,\parbox[t]{\widthof{$\eqin{\milnersysmin}$\hspace*{3ex}}}{\mbox{}}\:
          \parbox[t]{\widthof{\phantom{(if }%
                              because then $\bvert \loopnstepto{\alabi{i}} \bverti{i}$  
                              (see $\transitionsinfrom{\aonecharthat}{\bvert}$ 
                                   as in Def.~\ref{def:extrsoluntil:extrsol}) implies
                              $\bverti{i} \convdescendsinloopto \bvert$,}}
                         {$\bigl($if $\bverti{i} = \bvert$, then 
                             $\asol{\bverti{i}}
                                \milnersyseq
                              \stexpprod{\extrsoluntilof{\aonecharthat}{\bverti{i}}{\bvert}}
                                        {\asol{\bvert})}$
                           due to $\extrsoluntilof{\aonecharthat}{\bvert}{\bvert} = 1$; 
                          \\
                          \phantom{(}if $\bverti{i} \neq \bvert$, we can apply the induction hypothesis
                          to $\asol{\bverti{i}}$,
                          \\
                          \phantom{(if }%
                          because then $\bvert \loopnstepto{\alabi{i}} \bverti{i}$  
                           (see $\transitionsinfrom{\aonecharthat}{\bvert}$ 
                            as in Def.~\ref{def:extrsoluntil:extrsol}) implies
                          $\bverti{i} \convdescendsinloopto \bvert$,
                          \\
                          \phantom{(if }%
                          and due to $\bvert \convdescendsinloopto \avert$ we get $\pair{\bverti{i}}{\bvert} \lexspo \pair{\bvert}{\avert}$)}
    \displaybreak[0]\\[0.75ex]
      &
      \;\,\parbox[t]{\widthof{$\eqin{\milnersysmin}$}}{$\eqin{\milnersys}$}\:  
    \stexpsum{\Bigl(
                \sum_{i=1}^{n}
                  \stexpprod{\aacti{i}}
                            {\bigl(\stexpprod{\extrsoluntilof{\aonecharthat}{\bverti{i}}{\bvert}}
                                             {\asol{\bvert}}\bigr)}
              \Bigr)}
             {\Bigl(
                \sum_{i=1}^{m}
                  \stexpprod{\boneacti{i}}
                            {\bigl(\stexpprod{\extrsoluntilof{\aonecharthat}{\cverti{i}}{\avert}}
                                             {\asol{\avert}}\bigr)}
              \Bigr)}
        \\
        & \;\,\parbox[t]{\widthof{$\eqin{\milnersysmin}$\hspace*{3ex}}}{\mbox{}}\:
          \parbox[t]{\widthof{\phantom{(if }%
                              imply $\cverti{i} \convdescendsinloopto \avert$ and $\cverti{i} \convredi{\bodylab} \bvert$,
                              and hence $\pair{\cverti{i}}{\avert}  \lexspo  \pair{\bvert}{\avert}$) holds)}}
                         {$\bigl($if $\cverti{i} = \avert$, then 
                             $\asol{\cverti{i}}
                                \milnersyseq
                              \stexpprod{\extrsoluntilof{\aonecharthat}{\cverti{i}}{\avert}}
                                        {\asol{\avert})}$
                             due to $\extrsoluntilof{\aonecharthat}{\avert}{\avert} = 1$;   
                          \\
                          \phantom{(}if $\cverti{i} \neq \avert$, we can apply the induction hypothesis
                          to $\asol{\cverti{i}}$,
                          \\
                          \phantom{(if }%
                          as $\bvert \redi{\bodylab} \cverti{i}$ (see $\transitionsinfrom{\aonecharthat}{\bvert}$ 
                                                                  as in Def.~\ref{def:extrsoluntil:extrsol})
                           and $\cverti{i} \neq \avert$
                          \\
                          \phantom{(if }%
                          imply $\cverti{i} \convdescendsinloopto \avert$,
                          and $\cverti{i} \convredi{\bodylab} \bvert$
                            entails $\pair{\cverti{i}}{\avert}  \lexspo  \pair{\bvert}{\avert}$)}
    \displaybreak[0]\\[0.5ex]
      &
      \;\,\parbox[t]{\widthof{$\eqin{\milnersysmin}$}}{$\eqin{\milnersysmin}$}\:  
    \stexpsum{\stexpprod{\Bigl(
                           \sum_{i=1}^{n}
                             \stexpprod{\aacti{i}}
                                       {\extrsoluntilof{\aonecharthat}{\bverti{i}}{\bvert}}
                         \Bigr)}                 
                        {\asol{\bvert}}}
             {\stexpprod{\Bigl(
                           \sum_{i=1}^{m}
                             \stexpprod{\boneacti{i}}
                                       {\extrsoluntilof{\aonecharthat}{\cverti{i}}{\avert}}
                         \Bigr)}              
                        {\asol{\avert}}}
        \\
        & \;\,\parbox[t]{\widthof{$\eqin{\milnersysmin}$\hspace*{3ex}}}{\mbox{}}\:
          \text{(by axioms $(\assocstexpprod)$, and {$(\rdistr)$}).}
  \end{align*}       
  We note that these equalities also hold for the special cases in which $n=0$ or/and $m=0$,
  where in the case $m = 0$ an axiom ($\stexpzerostexpprod$) needs to be used in the last step.
  Since $\ACI \subsystem \asys \isthmsubsumedby \milnersys$, and $\milnersysmin \subsystem \milnersys$, 
  we have obtained the following provable equality:  
  \begin{align*}     
    \asol{\bvert}
        &
        \;\,\parbox[t]{\widthof{$\eqin{\milnersysmin}$}}{$\eqin{\milnersys}$}\: 
    \stexpsum{\stexpprod{\Bigl(
                           \sum_{i=1}^{n}
                             \stexpprod{\aacti{i}}
                                       {\extrsoluntilof{\aonecharthat}{\bverti{i}}{\bvert}}
                         \Bigr)}                 
                        {\asol{\bvert}}}
             {\stexpprod{\Bigl(
                           \sum_{i=1}^{m}
                             \stexpprod{\boneacti{i}}
                                       {\extrsoluntilof{\aonecharthat}{\cverti{i}}{\avert}}
                         \Bigr)}              
                        {\asol{\avert}}} \punc{,}
  \end{align*}
  Since 
  $\notterminates{\bigl(
                     \sum_{i=1}^{n}
                       \stexpprod{\aacti{i}}
                                 {\extrsoluntilof{\aonecharthat}{\bverti{i}}{\bvert}}
                   \bigr)}$
  holds, we can apply $\RSPstar$ in order to obtain, and reason further: 
  \begin{align*}
    \asol{\bvert}
      &
      \;\,\parbox[t]{\widthof{$\eqin{\milnersysmin}$}}{$\eqin{\milnersys}$}\:
    \stexpprod{
                        {\Bigl(
                           \sum_{i=1}^{n}
                             \stexpprod{\aacti{i}}
                                       {\extrsoluntilof{\aonecharthat}{\bverti{i}}{\bvert}}
                         \Bigr)^{\!\sstar}\!}
               }           
              {\!
               \Bigl(
               \stexpprod{\Bigl(
                             \sum_{i=1}^{m}
                               \stexpprod{\boneacti{i}}
                                         {\extrsoluntilof{\aonecharthat}{\cverti{i}}{\avert}}
                          \Bigr)}              
                         {\asol{\avert}}
               \Bigr)}
    \displaybreak[0]\\[0.75ex]
      &
      \;\,\parbox[t]{\widthof{$\eqin{\milnersysmin}$}}{$\eqin{\milnersysmin}$}\:
    \stexpprod{\Bigl(
                 \stexpprod{\Bigl(
                              \sum_{i=1}^{n}
                                \stexpprod{\aacti{i}}
                                          {\extrsoluntilof{\aonecharthat}{\bverti{i}}{\bvert}}
                            \Bigr)^{\!\sstar}\!}
                           {\!
                            \Bigl(
                              \sum_{i=1}^{m}
                                \stexpprod{\boneacti{i}}
                                          {\extrsoluntilof{\aonecharthat}{\cverti{i}}{\avert}}
                            \Bigr)}  
               \Bigr)}                                        
              {\asol{\avert}}  
        \\
        & \;\,\parbox[t]{\widthof{$\eqin{\milnersysmin}$\hspace*{3ex}}}{\mbox{}}\:
          \text{(by axiom $(\assocstexpprod)$)}
    \displaybreak[0]\\[0.75ex]
      &
      \;\,\parbox[t]{\widthof{$\eqin{\milnersysmin}$}}{$\synteq$}\:
    \stexpprod{\extrsoluntilof{\aonecharthat}{\bvert}{\avert}}
              {\asol{\avert}}
        \\
        & \;\,\parbox[t]{\widthof{$\eqin{\milnersysmin}$\hspace*{3ex}}}{\mbox{}}\:
          \text{(by $\bvert \convdescendsinloopto \avert$, and the definition of $\extrsoluntilof{\aonecharthat}{\bvert}{\avert}$
                 in Def.~\ref{def:extrsoluntil:extrsol})}
  \end{align*} 
  In this way we have shown, due to $\milnersysmin \subsystem \milnersys$, the desired \provablein{\milnersys} equality 
  $\asol{\bvert}
     \eqin{\milnersys}
   \stexpprod{\extrsoluntilof{\aonecharthat}{\bvert}{\avert}}
             {\asol{\avert}}$ \mbox{}
  for the vertices $\avert$ and $\bvert$ that we picked with the property $\bvert \convdescendsinloopto \avert$.  
  
  Since $\bvert,\avert\in\verts$ with $\bvert \convdescendsinloopto \avert$ were arbitrary for this argument,
  we have successfully carried out the proof by induction 
  \mbox{}
  $\asol{\bvert}
     \eqin{\milnersys}
   \stexpprod{\extrsoluntilof{\aonecharthat}{\bvert}{\avert}}
             {\asol{\avert}}$ \mbox{}
  holds for all $\bvert,\avert\in\verts$ with $\bvert \convdescendsinloopto \avert$.
  As we have argued that the statement also holds for $\bvert = \avert$, we have proved the lemma.
\end{proof}

\begin{defi}
  For an \equationbasedover{\actions} proof system $\asys$ 
    we say
  that two star expression functions $\sasoli{1},\sasoli{2} \funin \verts\to\StExpover{\actions}$
  are \emph{\provablyin{\asys} equal} if $\asoli{1}{\avert} \eqin{\asys} \asoli{2}{\avert}$ holds for all $\avert\in\verts$. 
\end{defi}

Now we use the relationship of arbitrary \provablein{\milnersys} solutions of a guarded \LLEEwitness~$\aonecharthat$
  with the relative extraction function $\sextrsolof{\aonecharthat}$ of $\aonecharthat$ as stated by Lemma~\ref{lem:descsteps:bodysteps:wf}\vspace{-1.5pt}
    in order to demonstrate the solution uniqueness statement {\bf (SU)}.
The proof can be viewed as proceeding on the maximal length of \bodytransition\ paths from vertices~$\avert$,
  where for descents from $\avert$ via a \loopentrytransition\ into an inner loop the statement of Lemma~\ref{lem:lem:sols:provably:equal:LLEE} is used. 
Again the use of the fixed-point rule \RSPstar\ of \milnersysmin\ is crucial,
  because any two \provablein{\milnersysmin} solutions of a guarded \LLEEonechart\ cannot be expected to be \provablyin{\milnersysmin} equal in general.%
    \footnote{As a simple example, the use of \RSPstar\ is necessary for proving equal in \milnersys\ the two \provablein{\milnersysmin} solutions
              of the guarded \LLEEonechart\ $\onechartof{(a \cdot a)^* \cdot 0}$ 
              with the principal values $a^* \cdot 0$ and $(a \cdot a)^* \cdot 0$, respectively.}

\begin{lem}[provable equality of solutions of \protect\LLEEonecharts]\label{lem:sols:provably:equal:LLEE}
  Let $\aonechart$ be a guarded \LLEEonechart,
  and let $\asys$ be an \eqlogicbased\ proof system over $\StExpover{\actions}$ such that $\ACI \subsystem \asys \isthmsubsumedby \milnersys$.\nopagebreak[4]
 
  Then any two \provablein{\asys} solutions of $\aonechart$ are \provablyin{\milnersys} equal.
\end{lem}

\begin{proof}
  Let $\aonechart = \tuple{\verts,\actions,\sone,\start,\transs,\exts}$ 
    be a \LLEEonechart\ with guarded \LLEEwitness~$\aonecharthat$,
  and let $\asys$ an \eqlogicbased\ proof system as assumed in the lemma. 
  In order to show that any two \provablein{\asys} solutions of $\aonechart$
  are \provablyin{\milnersys} equal, 
  it suffices to show that every \provablein{\asys} solution of $\aonechart$ is \provablyin{\milnersys} equal
  to the extraction function $\sextrsolof{\aonecharthat}$ of $\aonecharthat$. 

  For demonstrating this,   
  let $\sasol \funin \verts \to \StExpover{\actions}$ be an \provablein{\asys} solution of $\aonechart$.
  We have to show that
  $\asol{\bvert}
     \milnersyseq
   \extrsolof{\aonecharthat}{\bvert}$ 
  holds for all $\bvert\in\verts$. 
  We proceed\vspace*{-2pt} by complete induction on the \wellfounded\ relation $\sconvredtci{\bodylab}$ (see Lemma~\ref{lem:descsteps:bodysteps:wf},
                                                                                                                   \ref{it:2:lem:descsteps:bodysteps:wf}),
    which does not require us to treat base cases separately. 
  For our argument we assume to have given, 
    underlying the definition of the relative extraction function $\sextrsoluntilof{\aonecharthat}$ 
    and the extraction function $\sextrsolof{\aonechart}$,
  list\vspace{-2.5pt} representations $\transitionsinfrom{\aonecharthat}{\bvert}$ of the transitions from $\bvert$ in $\aonecharthat$ 
  as written in Definition~\ref{def:extrsoluntil:extrsol},
  for all $\bvert\in\verts$. 
  
  Let $\bvert\in\verts$ be arbitrary.
  On the basis of $\transitionsinfrom{\aonecharthat}{\bvert}$ 
                              as in Definition~\ref{def:extrsoluntil:extrsol}
  we argue as follows,
  starting with a step in which we use that $\sasol$ is an \provablein{\asys} solution of $\aonechart$
    in view of the assumed form of $\transitionsinfrom{\aonecharthat}{\bvert}\,$: 
  \begin{align}
    \asol{\bvert}
      &
      \;\,\parbox[t]{\widthof{$\eqin{\milnersysmin}$}}{$\eqin{\asys}$}\:
    \stexpsum{\terminatesconstof{\aonechart}{\bvert}}
             {\Bigl(
                \stexpsum{\Bigl(
                            \sum_{i=1}^{n}
                              \stexpprod{\aacti{i}}{ \asol{\bverti{i}} }
                          \Bigr)}
                         {\Bigl(
                            \sum_{i=1}^{m}
                              \stexpprod{\boneacti{i}}{ \asol{\cverti{i}} }
                          \Bigr)}
              \Bigr)}
        \displaybreak[0]
    \notag\displaybreak[0]\\
      &
      \;\,\parbox[t]{\widthof{$\eqin{\milnersysmin}$}}{$\eqin{\ACI}$}\:
    \stexpsum{\Bigl(
                \sum_{i=1}^{n}
                  \stexpprod{\aacti{i}}{ \asol{\bverti{i}} }
              \Bigr)}
             {\Bigl(
                \stexpsum{\Bigl(
                            \sum_{i=1}^{m}
                              \stexpprod{\boneacti{i}}{ \asol{\cverti{i}} }
                          \Bigr)}
                         {\terminatesconstof{\aonechart}{\bvert}
                          \Bigr)}
              \Bigr)}
        \displaybreak[0]
    \notag\\
      &
      \;\,\parbox[t]{\widthof{$\eqin{\milnersysmin}$}}{$\eqin{\milnersys}$}\:  
    \stexpsum{\Bigl(
                \sum_{i=1}^{n}
                  \stexpprod{\aacti{i}}
                            {(\stexpprod{\extrsoluntilof{\aonecharthat}{\bverti{i}}{\bvert}}
                                        {\asol{\bvert}})}
              \Bigr)}
             {\Bigl(
                \stexpsum{\Bigl(
                            \sum_{i=1}^{m}
                              \stexpprod{\boneacti{i}}{ \asol{\cverti{i}} }
                          \Bigr)}
                         {\terminatesconstof{\aonechart}{\bvert}}
              \Bigr)}
        \notag\\[-0.5ex]
        & \;\,\parbox[t]{\widthof{$\eqin{\milnersysmin}$\hspace*{3ex}}}{\mbox{}}\:
          \parbox{\widthof{%
                     ($\bvert \loopnstepto{\alabi{i}} \bverti{i}$ 
                     (due to $\transitionsinfrom{\aonecharthat}{\bvert}$ 
                             as in Definition~\ref{def:extrsoluntil:extrsol}) implies $\bverti{i} \convdescendsinloopto \bvert$,}}
                 {($\bvert \loopnstepto{\alabi{i}} \bverti{i}$ 
                    (due to $\transitionsinfrom{\aonecharthat}{\bvert}$ 
                            as in Definition~\ref{def:extrsoluntil:extrsol}) implies $\bverti{i} \convdescendsinlooptosc \bvert$,
                  \\
                  \phantom{(}%
                  from which Lemma~\ref{lem:lem:sols:provably:equal:LLEE} yields
                  $\asol{\bverti{i}} \milnersyseq \stexpprod{\extrsoluntilof{\aonecharthat}{\bverti{i}}{\bvert}}{\asol{\bvert}}\,$)}
      \displaybreak[0]              
    \notag\\
      &
      \;\,\parbox[t]{\widthof{$\eqin{\milnersysmin}$}}{$\:\eqin{\milnersysmin}$}\:  
    \stexpsum{\stexpprod{\Bigl(
                           \sum_{i=1}^{n}
                             \stexpprod{\aacti{i}}
                                       {\extrsoluntilof{\aonecharthat}{\bverti{i}}{\bvert}}
                         \Bigr)}
                        {\asol{\bvert}}}
             {\Bigl(
               \stexpsum{\Bigl(
                           \sum_{i=1}^{m}
                             \stexpprod{\boneacti{i}}{ \asol{\cverti{i}} }
                         \Bigr)}
                        {\terminatesconstof{\aonechart}{\bvert}}
              \Bigr)}
        \notag\\
        & \;\,\parbox[t]{\widthof{$\eqin{\milnersysmin}$\hspace*{3ex}}}{\mbox{}}\:
          \text{(by axioms $(\assocstexpprod)$ and $(\rdistr)$)}
      \displaybreak[0]      
    \notag\\
      &
      \;\,\parbox[t]{\widthof{$\eqin{\milnersysmin}$}}{$\:\eqin{\milnersys}$}\:  
    \stexpsum{\stexpprod{\Bigl(
                           \sum_{i=1}^{n}
                             \stexpprod{\aacti{i}}
                                       {\extrsoluntilof{\aonecharthat}{\bverti{i}}{\bvert}}
                         \Bigr)}
                        {\asol{\bvert}}}
             {\Bigl(
               \stexpsum{\Bigl(
                           \sum_{i=1}^{m}
                             \stexpprod{\boneacti{i}}{ \extrsolof{\aonecharthat}{\cverti{i}} }
                         \Bigr)}
                        {\terminatesconstof{\aonechart}{\bvert}}
              \Bigr)}
        \notag\\
        & \;\,\parbox[t]{\widthof{$\eqin{\milnersysmin}$\hspace*{3ex}}}{\mbox{}}\:
          \parbox{\widthof{(due to $\bvert \redi{\bodylab} \cverti{i}$ (see $\atsiof{\aonecharthat}{\bvert}$ as in Definition~\ref{def:extrsoluntil:extrsol}),
                   and hence $\cverti{i} \convredi{\bodylab} \bvert$,}}
                 {(due to $\bvert \redi{\bodylab} \cverti{i}$ (see $\transitionsinfrom{\aonecharthat}{\bvert}$ 
                                                                   as in Definition~\ref{def:extrsoluntil:extrsol}),
                   and hence $\cverti{i} \convredi{\bodylab} \bvert$,
                  \\[0.25ex]
                  \phantom{(}%
                  $\asol{\cverti{i}} \milnersyseq \extrsolof{\aonecharthat}{\cverti{i}}$ follows from the induction hypothesis).}        
    \notag\displaybreak[0]
    \intertext{Since $\ACI \subsystem \asys \isthmsubsumedby \milnersys$, and $\milnersysmin \subsystem \milnersys$, 
               we have obtained the following provable equality:}
    \displaybreak[0]              
    \asol{\bvert}
      &
      \;\,\parbox[t]{\widthof{$\eqin{\milnersysmin}$}}{$\eqin{\milnersys}$}\:  
    \stexpsum{\stexpprod{\Bigl(
                           \sum_{i=1}^{n}
                             \stexpprod{\aacti{i}}
                                       {\extrsoluntilof{\aonecharthat}{\bverti{i}}{\bvert}}
                         \Bigr)}
                        {\asol{\bvert}}}
             {\Bigl(
               \stexpsum{\Bigl(
                           \sum_{i=1}^{m}
                             \stexpprod{\boneacti{i}}{ \extrsolof{\aonecharthat}{\cverti{i}} }
                         \Bigr)}
                        {\terminatesconstof{\aonechart}{\bvert}}
              \Bigr)}
              \notag
  \end{align}
  Now since 
  $\notterminates{\bigl(
                     \sum_{i=1}^{n}
                       \stexpprod{\aacti{i}}
                                 {\extrsoluntilof{\aonecharthat}{\bverti{i}}{\bvert}}
                   \bigr)}$
  holds, we can apply the rule $\RSPstar$ to this in order to obtain:
  \begin{align*}                                       
    \asol{\bvert}
      &
      \;\,\parbox[t]{\widthof{$\eqin{\milnersysmin}$}}{$\:\eqin{\milnersys}$}\:  
    \stexpprod{
                        \Bigl(
                          \sum_{i=1}^{n}
                            \stexpprod{\aacti{i}}
                                      {\extrsoluntilof{\aonecharthat}{\bverti{i}}{\bvert}}
                        \Bigr)^{\!\sstar}\!
               }{\Bigl(
                   \stexpsum{\Bigl(
                               \sum_{i=1}^{m}
                                 \stexpprod{\boneacti{i}}{ \extrsolof{\aonecharthat}{\cverti{i}} }
                             \Bigr)}
                            {\terminatesconstof{\aonechart}{\bvert}}
                \Bigr)}
      \\                 
      &
      \;\,\parbox[t]{\widthof{$\eqin{\milnersysmin}$}}{$\synteq$}\: 
    \extrsolof{\aonecharthat}{\bvert}
      \;\,
      \parbox[t]{\widthof{$\eqin{\milnersysmin}$\hspace*{3ex}}}{\mbox{}}\:
        \text{(by the definition of $\sextrsolof{\aonecharthat}$ in Definition~\ref{def:extrsoluntil:extrsol})} 
  \end{align*}
  Thus we have verified the proof obligation
  $\asol{\bvert}
     \milnersyseq
   \extrsolof{\aonecharthat}{\bvert}$
  for the induction step, for the vertex $\bvert$ as picked.
 
  By having performed the induction step, we have successfully carried out the proof by complete induction on $\sconvredi{\bodylab}$
  that 
  $\asol{\bvert}
     \milnersyseq
   \extrsolof{\aonecharthat}{\bvert}$ holds for all $\bvert\in\verts$,
  and for an arbitrary \provablein{\asys} solution $\sasol$ of $\aonechart$. 
  This implies the statement of the lemma, that any two \provablein{\asys} solutions of $\aonechart$ are \provablyin{\milnersys} equal.  
\end{proof}

\begin{prop}\label{prop:LLEEcoindproofeq:impl:milnersyseq}
  For every \eqlogicbased\ proof system $\asys$ over $\StExpover{\actions}$ with $\ACI \subsystem \asys \isthmsubsumedby \milnersys$, 
  provability by \LLEEwitnessed\ coinductive proofs over $\asys$ implies derivability~in~$\milnersys$:
  \begin{equation}\label{eq:prop:LLEEcoindproofeq:impl:milnersyseq}
      \bigl(\;
        \astexpi{1}
          \LLEEcoindproofeqin{\asys}
        \astexpi{2}
            \quad\;\Longrightarrow\quad\;
        \astexpi{1}
          \milnersyseq
        \astexpi{2} 
      \;\bigr)
      \qquad \text{for all $\astexpi{1},\astexpi{2}\in\StExpover{\actions}$.} 
  \end{equation}
\end{prop}

\begin{proof}
  For showing \eqref{eq:prop:LLEEcoindproofeq:impl:milnersyseq}, 
  let $\astexp,\bstexp\in\StExpover{\actions}$ be such that 
  $\astexp \LLEEcoindproofeqin{\asys} \bstexp$.
  Then there is a \LLEEwitnessed\ coinductive proof $\aLLEECoProof = \pair{\aonechart}{\saeqfun}$ of $\astexpi{1} \formeq \astexpi{2}$ over $\asys$,
  which consists of a guarded \LLEEonechart\ $\aonechart$ 
                     and 
                    \provablein{\asys} solutions $\saeqfuni{1},\saeqfuni{2} \funin \vertsof{\aonechart} \to \StExpover{\actions}$ of $\aonechart$
                       with $\astexpi{1} \synteq \aeqfuni{1}{\start}$ and $\astexpi{2} \synteq \aeqfuni{2}{\start}$. 
  By applying Lemma~\ref{lem:sols:provably:equal:LLEE} to $\aLLEECoProof$
    we find that $\saeqfuni{1}$ and $\saeqfuni{2}$ are \provablyin{\milnersys} equal.
  This entails
    $\astexpi{1} 
       \synteq
     \aeqfuni{1}{\start}
       \milnersyseq
     \aeqfuni{2}{\start}
       \synteq
     \astexpi{2}$,
   and thus $\astexpi{1} \milnersyseq \astexpi{2}$.    
\end{proof}

\begin{exa}\label{ex:extraction:proof}\mbox{}%
  \begin{figure}[tbp]
\begin{center}
  $
  \hspace*{-3ex}%
  \begin{aligned}
    \scalebox{0.9}{\begin{tikzpicture}\renewcommand{\stexpprod}[2]{{#1}\hspace*{1pt}{\sstexpprod}\hspace*{1pt}{#2}}
  
\matrix[anchor=north,row sep=1.25cm,column sep=0cm,
        ] {
  \node(v11){$\aeqfuni{i}{\averti{11}}$};
    & & \node(v21){$\aeqfuni{i}{\averti{21}}$};
  \\                 
  \node(v1){$\aeqfuni{i}{\averti{1}}$};
    & & \node(v2){$\aeqfuni{i}{\averti{2}}$};
  \\
    & \node(v){$\aeqfuni{i}{\start}$};
  \\
  };   
\path (v) ++ (0cm,
                  +4.75cm) node(label){\Large $\aonechart$, $\aonecharthat$};

\draw[->,thick,densely dotted,out=180,in=165,distance=0.75cm](v11) to (v1);

\draw[->,thick,darkcyan,shorten >= 0pt]
  (v1) to node[left,pos=0.25,xshift=0.075cm]{\small $\black{\aact}$} node[left,pos=0.6,xshift=0.075cm]{\small $\loopnsteplab{1}$} (v11);
\draw[->,thick,densely dotted,out=-90,in=165,distance=0.5cm](v1) to (v);
\draw[->,shorten <= 0pt,shorten >= 0pt] (v1) to node[above]{$\bact$} (v21); 

\draw[->,thick,densely dotted,out=-0,in=15,distance=0.75cm](v21) to (v2);  
  
\draw[->,thick,darkcyan,shorten >= 0pt]
  (v2) to node[right,pos=0.25,xshift=-0.075cm]{\small $\black{\bact}$} node[right,pos=0.6,xshift=-0.075cm]{\small $\loopnsteplab{1}$} (v21);
\draw[->,thick,densely dotted,out=-90,in=15,distance=0.5cm](v2) to (v);

\draw[<-,very thick,>=latex,chocolate,shorten <=2.5pt](v) -- ++ (-90:0.8cm); 
\draw[thick,chocolate,double] (v) ellipse (0.475cm and 0.3cm);   
\draw[->,thick,darkcyan,shorten >= 0.175cm,shorten <= 3pt] 
  (v) to node[left,pos=0.36,xshift=0.075cm]{\small $\black{\aact}$} node[right,pos=0.4,xshift=-0.075cm,yshift=1pt]{\small $\loopnsteplab{2}$}  (v11);
\draw[->,thick,darkcyan,shorten >= 0.175cm,shorten <= 3pt] 
  (v) to node[right,pos=0.36,xshift=-0.05cm]{\small $\black{\bact}$} node[left,pos=0.6,xshift=0.075cm,yshift=1pt]{\small $\loopnsteplab{2}$} (v21);
  
\end{tikzpicture} }
  \end{aligned}
  \hspace*{1ex}
  \begin{aligned}[c]
    \aeqfuni{i}{\averti{21}}
      & \;\parbox[t]{\widthof{$\milnersysmineq$}}{$\eqinsol{\milnersysmin}$}\,
    \stexpprod{\stexpone}{\aeqfuni{i}{\averti{2}}} 
      \,\milnersysmineq\,
    \aeqfuni{i}{\averti{2}} 
    \\[-0.75ex]
      & \phantom{\;\parbox[t]{\widthof{$\milnersysmineq$}}{$\eqinsol{\milnersysmin}$}}\,
        \text{($\eqinsol{\milnersysmin}$ means use of 
               `is \provablein{\milnersysmin} solution')}
    \\
    \aeqfuni{i}{\averti{2}}
      & \;\parbox[t]{\widthof{$\milnersysmineq$}}{$\eqinsol{\milnersysmin}$}\,
    \stexpsum{\stexpprod{\bact}{\aeqfuni{i}{\averti{21}}}}
             {\stexpprod{\stexpone}{\aeqfuni{i}{\start}}} 
      \,\milnersysmineq\,
    \stexpsum{\stexpprod{\bact}{\aeqfuni{i}{\averti{2}}}}
             {\aeqfuni{i}{\start}}
    \\[-0.5ex]
    & \;\parbox[t]{\widthof{$\milnersysmineq$}}{$\Downarrow${\text{~~applying $\RSPstar$}}}
    \\[-0.5ex]
    \aeqfuni{i}{\averti{2}}
      & \;\parbox[t]{\widthof{$\milnersysmineq$}}{$\milnersyseq$}\,
    \stexpprod{\stexpit{\bact}}
              {\aeqfuni{i}{\start}}
    \\[0.5ex]
    \aeqfuni{i}{\averti{11}}
      & \;\parbox[t]{\widthof{$\milnersysmineq$}}{$\eqinsol{\milnersysmin}$}\,
    \stexpprod{\stexpone}{\aeqfuni{i}{\averti{1}}} 
      \,\milnersysmineq\,
    \aeqfuni{i}{\averti{1}} 
    \\
    \aeqfuni{i}{\averti{1}}
      & \;\parbox[t]{\widthof{$\milnersysmineq$}}{$\milnersysmineq$}\,
    \stexpsum{\stexpprod{\aact}{\aeqfuni{i}{\averti{11}}}}
             {\stexpsum{\stexpprod{\bact}{\aeqfuni{i}{\averti{21}}}}
                       {\stexpprod{\stexpone}{\aeqfuni{i}{\start}}}}
    \\[-0.25ex]      
      & \;\parbox[t]{\widthof{$\milnersysmineq$}}{$\milnersyseq$}\,
    \stexpsum{\stexpprod{\aact}{\aeqfuni{i}{\averti{1}}}}
             {\stexpprod{(\stexpsum{\stexpprod{\bact}{\stexpit{\bact}}}
                                   {\stexpone})}
                        {\aeqfuni{i}{\start}}}
    \\[-0.25ex]      
      & \;\parbox[t]{\widthof{$\milnersysmineq$}}{$\milnersysmineq$}\,
    \stexpsum{\stexpprod{\aact}{\aeqfuni{i}{\averti{1}}}}
             {\stexpprod{\stexpit{\bact}}
                        {\aeqfuni{i}{\start}}}
    \\[-0.5ex]
    & \;\parbox[t]{\widthof{$\milnersysmineq$}}{$\Downarrow${\text{~~applying $\RSPstar$}}}
    \\[-0.5ex]
    \aeqfuni{i}{\averti{1}}
      & \;\parbox[t]{\widthof{$\milnersysmineq$}}{$\milnersyseq$}\,
    \stexpprod{\stexpit{\aact}}
              {(\stexpprod{\stexpit{\bact}}{\aeqfuni{i}{\start}})}
      \,\milnersysmineq\,          
    \stexpprod{(\stexpprod{\stexpit{\aact}}{\stexpit{\bact}})}{\aeqfuni{i}{\start}}  
  \end{aligned}
  $
  \\[0.75ex]
  $
  \begin{aligned}
    \aeqfuni{i}{\start}
      & \;\parbox[t]{\widthof{$\milnersysmineq$}}{$\eqinsol{\milnersysmin}$}\,
    1 +  
    \stexpsum{\stexpprod{\aact}{\aeqfuni{i}{\averti{11}}}}
             {\stexpprod{\bact}{\aeqfuni{i}{\averti{21}}}} 
      \,\milnersysmineq\,
    1 +  
    \stexpsum{\stexpprod{\aact}{\aeqfuni{i}{\averti{1}}}}
             {\stexpprod{\bact}{\aeqfuni{i}{\averti{2}}}}   
    \\[-0.25ex]
      & \;\parbox[t]{\widthof{$\milnersysmineq$}}{$\milnersyseq$}\,
    \stexpprod{(\stexpsum{\stexpprod{\aact}{(\stexpprod{\stexpit{\aact}}{\stexpit{\bact}})}}
                         {\stexpprod{\bact}{\stexpit{\bact}}})}
              {\aeqfuni{i}{\start}}
      + 1          
    \\[-0.5ex]
    & \;\parbox[t]{\widthof{$\milnersysmineq$}}{$\Downarrow${\text{~~applying $\RSPstar$}}}
    \\[-0.5ex]
    \aeqfuni{i}{\start}
      & \;\parbox[t]{\widthof{$\milnersysmineq$}}{$\milnersyseq$}\,
    \stexpprod{\stexpit{(\stexpsum{\stexpprod{\aact}{(\stexpprod{\stexpit{\aact}}{\stexpit{\bact}})}}
                                  {\stexpprod{\bact}{\stexpit{\bact}}})}}
              {\stexpone}
      \,\milnersysmineq\,
    \stexpit{(\stexpsum{\stexpprod{\aact}{(\stexpprod{\stexpit{\aact}}{\stexpit{\bact}})}}
                       {\stexpprod{\bact}{\stexpit{\bact}}})}        
      \,\milnersysmineq\,
    \extrsolof{\aonecharthat}{\start}
  \end{aligned}
  $     
\end{center}
    \vspace*{-1ex}
\caption{\protect\label{fig:ex:extraction:proof}%
         Showing 
         for the coinductive proof $\protect\pair{\protect\aonechart}{\protect\saeqfun}$ in Example~\protect\ref{ex:2:LLEEcoindproof} 
         that the principal value $\protect\aeqfuni{i}{\protect\start}$ of the \protect\provablein{\protect\milnersysmin} solution $\protect\saeqfuni{i}$ for $i\protect\in\protect\setexp{1,2}$ 
         is \protect\provablyin{\protect\milnersys} equal to the principal value $\protect\extrsolof{\protect\aonecharthat}{\protect\start}$ of the solution $\protect\sextrsolof{\protect\aonecharthat}$ 
         extracted the underlying \protect\LLEEwitness~$\protect\aonecharthat$.
         }
\end{figure}      %
  We consider again the \LLEEwitnessed\ coinductive proof $\aCoProof = \pair{\aonechart}{\saeqfun}$ of 
  $\stexpit{(\stexpprod{\stexpit{\aact}}
                       {\stexpit{\bact}})}
     \formeq
   \stexpit{(\stexpsum{\aact}{\bact})}$
  in Example~\ref{ex:2:LLEEcoindproof}. 
  In Figure~\ref{fig:ex:extraction:proof} we exhibit the extraction process of derivations in $\milnersys$
  of $\aeqfuni{1}{\start} \formeq \extrsolof{\aonechart}{\avert}$ 
  and $\aeqfuni{2}{\start} \formeq \extrsolof{\aonechart}{\avert}$
  from the guarded \LLEEwitness~$\aonecharthat$ of $\aonechart$.
  These two derivations in \milnersys\ 
  can be combined by using $\eqlogic$ rules 
  in order to obtain a derivation in $\milnersys$ of
  $\stexpit{(\stexpprod{\stexpit{\aact}}{\stexpit{\bact}})} 
     \synteq       \aeqfuni{1}{\start} 
     \formeq \aeqfuni{2}{\start}
     \synteq       \stexpit{(\stexpsum{\aact}{\bact})}$.
\end{exa}

\begin{lem}\label{lem:LCoindProof:admissible:milnersys}
  The rules $\LCoindProofi{n}$ are correct for $\milnersys$, for all $n\in\nat$.
  This statement also holds effectively: 
  Every derivation $\aDeriv$ in $\thplus{\milnersys}{\LCoindProofi{n}}$
    that consists of a bottommost instance of $\LCoindProofi{n}$, where $n\in\nat$,
      whose immediate subderivations are derivations in $\milnersys$
  can be transformed effectively into a derivation $\aDerivacc$ in $\milnersys$ that has the same conclusion~as~$\aDeriv$.
\end{lem}

\begin{proof}
  We let $n\in\nat$.
  In order to show correctness of the rule $\LCoindProofi{n}$ for $\milnersys$,
    we consider a derivation $\aDeriv$ in $\thplus{\milnersys}{\LCoindProofi{n}}$
      that has immediate subderivations $\aDerivi{1}$, \ldots, $\aDerivi{n}$ in \milnersys,
        and that terminates with an instance $\ainst$ of $\LCoindProofi{n}\,$, 
          where $\aseteqs \defdby \setexp{\cstexpi{1}   \formeq   \dstexpi{1}, \ldots, \cstexpi{1}   \formeq   \dstexpi{1}}\,$:
  \begin{equation*}
    \AxiomC{$\aDerivi{1}$}
    \noLine
    \UnaryInfC{$ \cstexpi{1}   \formeq   \dstexpi{1} $}
    \AxiomC{\ldots}
    \AxiomC{$\aDerivi{n}$}
    \noLine
    \UnaryInfC{$ \cstexpi{n}   \formeq   \dstexpi{n} $}
    %
    %
    \RightLabel{$\LCoindProofi{n}$}
    \LeftLabel{$\ainst$}
    \TrinaryInfC{$ \astexp \formeq \bstexp $} 
    \DisplayProof
  \end{equation*}
  and $\astexp \LLEEcoindproofeqin{\thplus{\milnersysmin}{\aseteqs}} \bstexp$ holds as \sidecondition\
    on the instance $\ainst$ of $\LCoindProofi{n}$.
  Then there is a \LLEEwitnessed\ coinductive proof $\aLLEECoProof = \pair{\aonechart}{\saeqfun}$ of $\astexp \formeq \bstexp$ 
  over $\thplus{\milnersysmin}{\aseteqs}$.   
  We have to show that there is a derivation $\aDerivacc$ in \milnersys\ with the same conclusion $ \astexp \formeq \bstexp $. 
  
  Since $\aDerivi{1}$, \ldots, $\aDerivi{n}$ are derivations in \milnersys,
    their conclusions in $\aseteqs$ are derivable in \milnersys. 
  This implies $\thplus{\milnersys}{\aseteqs} \isthmsubsumedby \milnersys$.
  It follows that $\aLLEECoProof$ is also a \LLEEwitnessed\ coinductive proof of $\astexp\formeq\bstexp$ over \milnersys,
  that is, it holds:
  \begin{equation*}
    \astexp \LLEEcoindproofeqin{\milnersys} \bstexp \punc{.}
  \end{equation*}
  From this we obtain $\astexp \milnersyseq \bstexp$ by applying Proposition~\ref{prop:LLEEcoindproofeq:impl:milnersyseq},
    which guarantees a derivation $\aDerivacc$ in \milnersys\ with conclusion $\astexp \formeq \bstexp$ as desired.
  The proof of Proposition~\ref{prop:LLEEcoindproofeq:impl:milnersyseq} furthermore guarantees
    that such a derivation $\aDerivacc$ in \milnersys\ can be constructed effectively from 
      the coinductive proof of $\astexp\formeq\bstexp$ over $\thplus{\milnersysmin}{\aseteqs}$ (and hence over $\milnersys$)
        and the derivations $\aDerivi{1}$, \ldots, $\aDerivi{n}$ in \milnersys.
\end{proof}

\begin{thm}\label{thm:coindmilnersys:isthmsubsumedby:milnersys}
  $\coindmilnersys \isthmsubsumedby \milnersys$. 
  Moreover, every derivation in $\coindmilnersys$ with conclusion $\astexp \formeq \bstexp$ 
    can be transformed effectively
  into a derivation in $\milnersys$ that has the same conclusion.
\end{thm}

\begin{proof}%
  Due to Lemma~\ref{lem:LCoindProof:admissible:milnersys},
    every rule $\LCoindProofi{n}$, for $n\in\nat$, is correct for $\milnersys$. 
  Then by using Lemma~\ref{lem:derivable:admissible:rules}, \ref{it:1:lem:derivable:admissible:rules}, 
    we find that each of these rules are admissible in $\milnersys$. 
  This means that $\thplus{\milnersys}{\LCoindProofi{n}} \thmequiv \milnersys$ holds for all $n\in\nat$,
  which implies, with an argument by induction on the prooftree size of derivations in $\thplus{\milnersys}{\family{\LCoindProofi{n}}{n\in\nat}}$,
  that $\thplus{\milnersys}{\family{\LCoindProofi{n}}{n\in\nat}} \thmequiv \milnersys$ holds as well.
  With this statement we can now argue as follows: 
  \begin{alignat*}{2}
    \coindmilnersys
      & {} = 
    \thplus{\milnersysmin}{\family{\LCoindProofi{n}}{n\in\nat}}
      & & \text{(by Definition~\ref{def:coindmilnersys})}
      \\
      & {} \subsystem
    \thplus{(\thplus{\milnersysmin}{\family{\LCoindProofi{n}}{n\in\nat}})}{\RSPstar}
      \hspace*{1.25em}
      & & \text{(by extension via adding the rule $\RSPstar$)}
      \displaybreak[0]\\
      & {} =
    \thplus{(\thplus{\milnersysmin}{\RSPstar})}{\family{\LCoindProofi{n}}{n\in\nat}}
      & & \text{(by construing the same system differently)}
      \displaybreak[0]\\
      & {} =
    \thplus{{\milnersys}}{\family{\LCoindProofi{n}}{n\in\nat}}
      & & \text{(by Definition~\ref{def:milnersys})}
      \\
      & {} \thmequiv
    \milnersys
      & & \text{(as argued above)} 
    \punc{.}   
  \end{alignat*}  
  From this we obtain $\coindmilnersys \isthmsubsumedby \milnersys$
    in view of
    $(\sbinrelcomp{\ssubsystem}{\sthmequiv}) 
       \:\subseteq\: 
     (\sbinrelcomp{\sisthmsubsumedby}{\sthmequiv}) 
       \:\subseteq\: 
     \sisthmsubsumedby$.

  For demonstrating the effective transformation statement of the theorem 
    we use the transformation from the proof of the implication ``$\Leftarrow$'' 
                                                      in Lemma~\ref{lem:derivable:admissible:rules}, \ref{it:1:lem:derivable:admissible:rules},
                                                which states that correct rules are also admissible. 
  We have to show that every derivation $\aDeriv$ in $\coindmilnersys$ can be transformed effectively 
    into a derivation $\aDerivacc$ in $\milnersys$ with the same conclusion.
  In order to establish this statement by induction we prove it for all derivations $\aDeriv$ 
    in the extension $\thplus{\coindmilnersys}{\RSPstar} 
                        = \thplus{\thplus{\milnersysmin}{\family{\LCoindProofi{n}}{n\in\nat}}}{\RSPstar}
                        = \thplus{\milnersys}{\family{\LCoindProofi{n}}{n\in\nat}}$
           of $\coindmilnersys$. 
             
  We proceed by induction on the number of instances of rules $\LCoindProofi{n}$, for $n\in\nat$, in $\aDeriv$.
  Let $\aDeriv$ be a derivation $\thplus{\thplus{\milnersysmin}{\family{\LCoindProofi{n}}{n\in\nat}}}{\RSPstar}$.
  If $\aDeriv$ does not contain an instance of $\LCoindProofi{n}$ with $n\in\nat$,
    then $\aDeriv$ is already a derivation in $\milnersys = \thplus{\milnersysmin}{\RSPstar}$,
      and no further transformation is necessary.
  Otherwise $\aDeriv$ contains at least one instance of $\LCoindProofi{n}$ with $n\in\nat$.
  We pick an instance $\ainst$ in $\aDeriv$ of a rule $\LCoindProofi{n_0}$ with $n_0\in\nat$
    that is topmost among the instances of the coinductive rule in $\aDeriv$, that is, 
      none of the immediate subderivations of the instance $\ainst$ in $\aDeriv$ 
        contains any instance of a rule $\LCoindProofi{n}$ for $n\in\nat$.
  Let $\aDerivi{0}$ be the subderivation of $\aDeriv$ that ends in $\ainst$.
  Since $\ainst$ is topmost, 
     all of the immediate subderivations of $\ainst$ and $\aDerivi{0}$ in $\aDeriv$ are derivations in $\milnersys = \thplus{\milnersysmin}{\RSPstar}$,
     and $\aDerivi{0}$ is a derivation in $\thplus{\thplus{\milnersysmin}{\LCoindProofi{n_0}}}{\RSPstar}
                                             =
                                           \thplus{\milnersys}{\LCoindProofi{n_0}}$.  
  Therefore we can apply the effective part of Lemma~\ref{lem:LCoindProof:admissible:milnersys}
   to the subderivation $\aDerivi{0}$.
  We obtain a derivation $\aDerivacci{0}$ in $\milnersys$ with the same conclusion as $\ainst$ and $\aDerivi{0}$.
  Then by replacing $\aDerivi{0}$ in $\aDeriv$ with $\aDerivacci{0}$
    we obtain a derivation $\aDerivtilde$ in $\thplus{\thplus{\milnersysmin}{\family{\LCoindProofi{n}}{n\in\nat}}}{\RSPstar}$ 
      that has the same conclusion as $\aDeriv$, but that has one instance of a coinductive rule less than $\aDeriv$.
  Now we can apply the induction hypothesis to $\aDerivtilde$ 
    in order to effectively transform it in a derivation $\aDerivacc$ in $\milnersys$ that has the same conclusion as $\aDeriv$.
  In this way we have performed the induction step.
\end{proof}

We conclude this section
  with an illustrative application of the results obtained here
    that provides in-roads for a completeness proof for Milner's system \milnersys. 
Specifically    
  we apply Proposition~\ref{prop:LLEEcoindproofeq:impl:milnersyseq},
    the trans\-for\-ma\-tion of \LLEEwitnessed\ coinductive proofs over \milnersys\ into derivations in \milnersys.
We show (see Corollary~\ref{cor:milnersys:compl:expansion:minimization} below) 
  that Milner's system is complete
    for bisimilarity of chart interpretations of star expressions
      when bisimilarity is witnessed by joint expansion or joint minimization to a guarded \LLEEonechart\ via functional \onebisimulations\
        (see Definition~\ref{def:compl:expansion:minimization}). 
For showing this statement we must, however, use here without proof a technical result from \cite{grab:2022:lics}:
  \provablein{\milnersys} solutions of \onecharts\ can be transferred backwards over functional \onebisimulations\
    (see Lemma~\ref{lem:transfer:sols:convfunonesbisim}). 
This statement is a generalization to \onecharts\ of Proposition~5.1 in \cite{grab:fokk:2020:lics},
  which states that \provablein{\asys} solutions of 
                                                    charts, for $\ACI\supsystem\asys$, 
    can be transferred backwards over functional bisimulations.

\begin{defi}\label{def:compl:expansion:minimization}
  Let $\aonechart$, $\aonecharti{1}$, and $\aonecharti{2}$ be \onecharts.
    We say that $\aonecharti{1}$ and $\aonecharti{2}$ are \onebisimilar\ \emph{via $\aonechart$ as (their) joint expansion}
         ($\aonecharti{1}$ and $\aonecharti{2}$ are \onebisimilar\ \emph{via $\aonechart$ as (their) joint minimization})
      if $\aonecharti{1} \convfunonebisim  \aonechart \funonebisim \aonecharti{2}$ holds
        (respectively, if $\aonecharti{1} \funonebisim \aonechart \convfunonebisim \aonecharti{2}$ holds).
\end{defi}

\begin{lem}[$\protect\sim\,$Lemma~3.8, (i), in \protect\cite{grab:2022:lics}]\label{lem:transfer:sols:convfunonesbisim} 
  \Provablein{\milnersys} solvability with principal value $\astexp$
    is preserved under converse functional \onebisimilarity,
      on weakly guarded \onecharts, for all star expressions $\astexp\in\StExp$. 
\end{lem}

\begin{figure}[tb!]
\begin{center}
\begin{tikzpicture}
  \matrix[anchor=center,row sep=0.75cm,column sep=1.25cm,ampersand replacement=\&] { 
      \& \node(top-exp){};
         \& \& \& \& \& \node(top-left-min){}; \& \& \node(top-right-min){};
    \\
      \& \& \& \& \node(concl-pos){}; 
    \\
    \node(left-exp){};
      \&
         \& \node(right-exp){}; 
            \& \& \& \& \& \node (min-min){};
    \\    
    };  
    
  \path (top-exp) ++ (0cm,0cm) node{\large $\aonechart$};
    \path (top-exp) ++ (0.1cm,0.3cm) node[right]{\forestgreen{LLEE}, guarded}; 
  \path (left-exp) ++ (0cm,0cm) node{\large $\chartof{\astexpi{1}}$};
  \path (right-exp) ++ (0cm,0cm) node{\large $\chartof{\astexpi{2}}$};
  \draw[funonebisimleft,shorten <=0.25cm,shorten >=0.25cm]  (top-exp) to (left-exp);
  \draw[funonebisimright,shorten <=0.25cm,shorten >=0.25cm] (top-exp) to (right-exp); 
  \path (concl-pos) ++ (0cm,0cm) node(concl){\large $\astexpi{1} \milnersyseq \astexpi{2}$};
  \draw[-implies,double equal sign distance] ($(concl) + (-2.25cm,0cm)$) 
                                               to node[below]{\tiny Cor.~\ref{cor:milnersys:compl:expansion:minimization}, \ref{it:1:cor:milnersys:compl:expansion:minimization}}
                                                  ($(concl.west) + (-0.25cm,0cm)$); 
  \draw[-implies,double equal sign distance] ($(concl) + (2.25cm,0cm)$) 
                                               to node[below]{\tiny Cor.~\ref{cor:milnersys:compl:expansion:minimization}, \ref{it:2:cor:milnersys:compl:expansion:minimization}}
                                                  ($(concl.east) + (0.25cm,0cm)$); 
  \path (top-left-min) ++ (0cm,0cm) node{\large $\chartof{\astexpi{1}}$};
  \path (top-right-min) ++ (0cm,0cm) node{\large $\chartof{\astexpi{2}}$};
  \path (min-min) ++ (0cm,0cm) node{\large $\aonechart$};
    \path (min-min) ++ (0.1cm,-0.25cm) node[right]{\forestgreen{LLEE}, guarded}; 
  \draw[funonebisimleft,shorten <=0.25cm,shorten >=0.25cm]  (top-left-min) to (min-min);
  \draw[funonebisimright,shorten <=0.25cm,shorten >=0.25cm] (top-right-min) to (min-min);  
\end{tikzpicture}  
\end{center}
  \vspace*{-1.5ex}
  \caption{\label{fig:cor:milnersys:compl:expansion:minimization}%
           Illustration of the statements (for arbitrary given $\astexpi{1},\astexpi{2}\in\StExp$) of 
                           Corollary~\protect\ref{cor:milnersys:compl:expansion:minimization}:
             Milner's system $\milnersys$ is complete for \protect\onebisimilarity\ of chart interpretations 
               via guarded \protect\LLEEonecharts\ as joint expansion or as joint minimization.
           }
\end{figure}%
\begin{cor}\label{cor:milnersys:compl:expansion:minimization} 
  The following two statements hold (see Figure~\ref{fig:cor:milnersys:compl:expansion:minimization} for their illustrations):
  \begin{enumerate*}[label={(\roman{*})}]
    \item{}\label{it:1:cor:milnersys:compl:expansion:minimization}
      $\milnersys$ is complete for \onebisimilarity\ of chart interpretations of two star expressions via a guarded \LLEEonechart\ as joint expansion.
    \item{}\label{it:2:cor:milnersys:compl:expansion:minimization}
      $\milnersys$ is complete for \onebisimilarity\ of chart interpretations of two star expressions via a guarded \LLEEonechart\ as joint minimization. 
  \end{enumerate*}
\end{cor}

\begin{proof}
  For showing statement \ref{it:1:cor:milnersys:compl:expansion:minimization},
    let $\astexpi{1},\astexpi{2}\in\StExpover{\actions}$ be such that there is a guarded \LLEEonechart~$\aonechart$
      with $\chartof{\astexpi{1}} \convfunonebisim \aonechart \funonebisim \chartof{\astexpi{2}}$.
    We have to show $\astexpi{1} \milnersyseq \astexpi{2}$.
  
  Due to Lemma~\ref{lem:chart-int:milnersysmin:solvable}, there are \provablein{\milnersysmin} (and thus \provablein{\milnersys}) solutions of $\chartof{\astexpi{1}}$ and $\chartof{\astexpi{2}}$
    with principal values $\astexpi{1}$, and $\astexpi{2}$, respectively.  
  Then we obtain by applying Lemma~\ref{lem:transfer:sols:convfunonesbisim}, 
    in view of the converse functional \onebisimulations\ from $\chartof{\astexpi{1}}$ and $\chartof{\astexpi{2}}$ to $\aonechart$,
    that there are two \provablein{\milnersys} solutions $\saeqfuni{1}$ and $\saeqfuni{2}$ of $\aonechart$ with the principal values $\astexpi{1}$ and $\astexpi{2}$, respectively. 
  These two \provablein{\milnersys} solutions of the guarded \LLEEonechart~$\aonechart$ can be combined
    to obtain a \LLEEwitnessed\ coinductive proof $\pair{\aonechart}{\saeqfun}$ of $\astexpi{1} \formeq \astexpi{2}$ over \milnersys. Thus we obtain:
  \begin{equation*}
    \astexpi{1} \LLEEcoindproofeqin{\milnersys} \astexpi{2} \punc{.}
  \end{equation*}    
  From this we arrive at $\astexpi{1} \milnersyseq \astexpi{2}$ by Proposition~\ref{prop:LLEEcoindproofeq:impl:milnersyseq}.
  
  \smallskip
  For showing statement \ref{it:2:cor:milnersys:compl:expansion:minimization},
    let $\astexpi{1},\astexpi{2}\in\StExpover{\actions}$ be such that there is a guarded \LLEEonechart~$\aonechart$
      with $\chartof{\astexpi{1}} \funonebisim \aonechart \convfunonebisim \chartof{\astexpi{2}}$.
    We have to show $\astexpi{1} \milnersyseq \astexpi{2}$.
  
  We first note that due to Lemma~\ref{lem:lem:extrsol:is:sol} 
    the extraction function $\sextrsolof{\aonecharthat}$ of a (guarded)\vspace*{-2.5pt} \LLEEwitness~$\aonecharthat$ of the guarded \LLEEonechart~$\aonechart$
      yields a \provablein{\milnersysmin} (and hence also a \provablein{\milnersys}) solution $\sasol$ of $\aonechart$ whose principal value we denote by $\astexp$.
  Next we apply Theorem~\ref{thm:onechart-int:LLEEw}, \ref{it:2:thm:onechart-int:LLEEw},
    to extend the assumed functional \onebisimulations\ to $\aonechart$ above the chart interpretations $\chartof{\astexpi{1}}$ and $\chartof{\astexpi{2}}$
      to start from the \onechart\ interpretations $\onechartof{\astexpi{1}}$ and $\onechartof{\astexpi{2}}$:
        we obtain  
          $\onechartof{\astexpi{1}} \funonebisim \chartof{\astexpi{1}} \funonebisim \aonechart \convfunonebisim \chartof{\astexpi{2}} \convfunonebisim \onechartof{\astexpi{2}}$.
  By Theorem~\ref{thm:onechart-int:LLEEw}, \ref{it:1:thm:onechart-int:LLEEw},
    we find that $\onechartof{\astexpi{1}}$ and $\onechartof{\astexpi{2}}$ are guarded \LLEEonecharts.        
  By transitivity of $\sfunonebisim$ we obtain  
    $\onechartof{\astexpi{1}} \funonebisim \aonechart \convfunonebisim \onechartof{\astexpi{2}}$.
  Now we can apply Lemma~\ref{lem:transfer:sols:convfunonesbisim}
    to obtain, from the \provablein{\milnersys} solution $\sasol$ of $\aonechart$, 
      a \provablein{\milnersys} solution $\saeqfuni{\astexpi{1},2}$ of $\onechartof{\astexpi{1}}$ with principal value $\astexp$, 
                                     and $\saeqfuni{\astexpi{2},2}$ of $\onechartof{\astexpi{2}}$ also with principal value $\astexp$.
  By Lemma~\ref{lem:onechart-int:milnersysmin:solvable}
    we furthermore obtain \provablein{\milnersysmin} (and hence \provablein{\milnersys}) solutions 
          $\saeqfuni{\astexpi{1},1}$ of $\onechartof{\astexpi{1}}$ with principal value $\astexpi{1}$, 
      and $\saeqfuni{\astexpi{2},1}$ of $\onechartof{\astexpi{2}}$ with principal value $\astexpi{2}$.
  The two \provablein{\milnersys} solutions $\saeqfuni{\astexpi{1},1}$ and $\saeqfuni{\astexpi{1},2}$ of the guarded \LLEEonechart~$\onechartof{\astexpi{1}}$
     can be combined
    to a \LLEEwitnessed\ coinductive proof $\pair{\onechartof{\astexpi{1}}}{\saeqfuni{\astexpi{1}}}$ of $\astexpi{1} \formeq \astexp$ over \milnersys.
  Analogously, 
    the two \provablein{\milnersys} solutions $\saeqfuni{\astexpi{2},1}$ and $\saeqfuni{\astexpi{2},2}$ of the guarded \LLEEonechart~$\onechartof{\astexpi{2}}$
     can be combined
    to a \LLEEwitnessed\ coinductive proof $\pair{\onechartof{\astexpi{2}}}{\saeqfuni{\astexpi{2}}}$ of $\astexpi{2} \formeq \astexp$ over \milnersys.
  Together we obtain:
  \begin{equation*}
    \astexpi{1} \LLEEcoindproofeqin{\milnersys} \astexp 
      \;\;\text{ and }\;\;
        \astexpi{2} \LLEEcoindproofeqin{\milnersys} \astexp \punc{.}
  \end{equation*}    
  Now by an appeal to Proposition~\ref{prop:LLEEcoindproofeq:impl:milnersyseq}
    we obtain 
      $\astexpi{1} \milnersyseq \astexp$ and $\astexpi{2} \milnersyseq \astexp$.
  Finally, by applying of symmetry and transitivity rules in \milnersys\ we obtain $\astexpi{1} \milnersyseq \astexpi{2}$.
\end{proof}  
  
\begin{rem}\label{rem:cor:milnersys:compl:expansion:minimization}
  While both of the statements \ref{it:1:cor:milnersys:compl:expansion:minimization} and \ref{it:2:cor:milnersys:compl:expansion:minimization}
   of Corollary~\ref{cor:milnersys:compl:expansion:minimization}
     were important pieces of the puzzle for constructing the completeness proof for \milnersys\ as sketched in \cite{grab:2022:lics},
       neither of them yields such a proof directly.
  This is because of the following two facts:     
  First, two bisimilar chart interpretations do not always have a guarded \LLEEonechart\ as their joint expansion along $\sconvfunonebisim$.
    This can be demonstrated with the counterexample of the charts in Example~4.1 of \cite{grab:fokk:2020:lics}.
    As a consequence, \ref{it:1:cor:milnersys:compl:expansion:minimization} is sometimes not applicable.
  Second,
    two bisimilar chart interpretations do not in general have a guarded \LLEEonechart\ as their joint minimization along $\sfunonebisim$.
      This is an easy consequence of the counterexample that is described in \cite[Sect.$\,$6]{grab:2022:lics}.
      Therefore direct application of \ref{it:2:cor:milnersys:compl:expansion:minimization} is also excluded in some situations.
      
  Yet more sophisticated applications of \ref{it:1:cor:milnersys:compl:expansion:minimization} and \ref{it:2:cor:milnersys:compl:expansion:minimization}   
    can still turn out to be expedient. 
  Indeed,
    a strengthening of \ref{it:2:cor:milnersys:compl:expansion:minimization} has helped us settle a subcase,
      and a refinement of \ref{it:2:cor:milnersys:compl:expansion:minimization} 
        has led to completeness proof for \milnersys\ as sketched in \cite{grab:2022:lics}.
  In \cite{grab:fokk:2020:lics}, 
    Fokkink and I have employed the stronger and more specific version 
      of the joint minimization idea
        in Corollary~\ref{cor:milnersys:compl:expansion:minimization},~\ref{it:2:cor:milnersys:compl:expansion:minimization},
          for a completeness proof of an adaptation \BBP\ by Bergstra, Bethke, and Ponse of \milnersys\ to `\onefree\ star expressions' 
           (without $\stexpone$, and with binary iteration instead of unary iteration).  
    Concretely,
       we used that bisimilar chart interpretations of \onefree\ star expressions have \LLEEcharts\ as their bisimulation collapses
         (and thus have guarded \LLEEonecharts\ as joint minimizations). 
  The completeness proof for \milnersys\ as summarized in \cite{grab:2022:lics}
    has been built around a more involved argument that employs `crystallized' \LLEEonechart\ approximations of the joint bisimulation collapse
      of bisimilar chart interpretations.
\end{rem}

\section{From Milner's system to LLEE-witnessed coinductive proofs}
  \label{milnersys:2:coindmilnersys}

In this section we develop a \prooftheoretic\ interpretation of $\milnersys$
  in the subsystem $\coindmilnersysone$ of $\coindmilnersys$, 
    and hence also a \prooftheoretic\ interpretation of \milnersys\ in \coindmilnersys . 
Since $\milnersys$ and $\coindmilnersysone$ differ only by the fixed-point rule \RSPstar\ (which is part of $\milnersys$, but not of $\coindmilnersys$)
  and the rule $\LCoindProofi{1}$ (which is part of $\coindmilnersysone$, but not of $\milnersys$),
  the crucial step for this proof transformation is to show that instances of \RSPstar\ can be mimicked in $\coindmilnersysone$.
  
We will do so by showing that instances of \RSPstar\ are derivable in $\coindmilnersysone$,
  and in particular,  can be mimicked by instances of $\LCoindProofi{1}$.   
More precisely, we will show that every instance~$\ainst$ of the fixed-point rule \RSPstar\ of \milnersys\
  can be mimicked by an instance of $\LCoindProofi{1}$ that has the same premise and conclusion,
    and that uses as its \sidecondition\ a \LLEEwitnessed\ coinductive proof over \milnersysmin\ 
    in which the premise of $\ainst$ may be used.
Still more explicitly, we show that every \RSPstar-instance with premise $e = f \prod e + g$ such that $\terminates{f}$ and with conclusion $e = f^* \prod g$   
  gives rise to a coinductive proof of $e = f^* \prod g$ over $\thplus{\milnersysmin}{\setexp{e = f \prod e + g}}$
    with underlying \onechart~$\onechartof{f^* \prod g}$ and guarded \LLEEwitness~$\onecharthatof{f^* \prod g}$.
  
We first illustrate this mimicking step by a concrete example (see Example~\ref{ex:1:RSPstar:2:LLEEcoindproof}),
  in order to motivate and convey the idea of this proof transformation. 
It will be built on three auxiliary statements
  (see after Example~\ref{ex:1:RSPstar:2:LLEEcoindproof}) two of which we have shown already in Section~\ref{LEE}.
Subsequently we prove the remaining crucial auxiliary statement (Lemma~\ref{lem:lem:mimic:RSPstar}),
  and then establish the transformation by showing that $\RSPstar$ is a derivable rule in $\coindmilnersysone$
    (Lemma~\ref{lem:RSPstar:derivable:coindmilnersysone}, using Lemma~\ref{lem:mimic:RSPstar}).
Finally we use derivability of $\RSPstar$ in $\coindmilnersysone$ in order to obtain the proof transformation from $\milnersys$ to $\coindmilnersysone$
  (see Theorem~\ref{thm:milnersys:isthmsubsumedby:coindmilnersysone}).

\begin{exa}\label{ex:1:RSPstar:2:LLEEcoindproof}
  We consider an instance of \RSPstar\ that corresponds, up to an application of $\rdistr$,
  to the instance of \RSPstar\ at the bottom in Figure~\ref{fig:ex:extraction:proof}:
  \begin{equation}\renewcommand{\fCenter}{\formeq}\label{eq:ex:1:RSPstar:2:LLEEcoindproof}
  \vcenter{\hbox{
    \mbox{%
    \Axiom$ \overbrace{(a + b)^*}^{\chocolate{e}}   \fCenter   \overbrace{((a \cdot a^* + b) \cdot b^*)}^{\alert{f}} \cdot \overbrace{(a + b)^*}^{\chocolate{e}} + \overbrace{1}^{ \forestgreen{g} } $
    \RightLabel{\RSPstar \hspace*{1.5ex} \text{\small (where $\notterminates{\alert{f}}$)}}
    \UnaryInf$ \underbrace{(a + b)^*}_{\smash{\chocolate{e}}}   \fCenter   \underbrace{((a \cdot a^* + b) \cdot b^*)^*}_{\smash{\alert{f^*}} } \cdot \underbrace{1}_{\smash{\forestgreen{g}}} $
    \DisplayProof}
  }}
  \end{equation}%
  We want to mimic this instance of $\RSPstar$ by an instance of $\LCoindProofi{1}$ that uses a \LLEEwitnessed\ coinductive proof of 
  $\chocolate{e} \formeq \alert{f^*} \cdot \forestgreen{g}$ over $\milnersysmin$ plus the premise of the \RSPstar\ instance \eqref{eq:ex:1:RSPstar:2:LLEEcoindproof}.
  We first obtain the \onechart\ interpretation $\onechartof{\alert{\stexpit{\bstexp}}}$ of $\alert{\stexpit{\bstexp}}$
  according to Definition~\ref{def:onechartof}, see Figure~\ref{fig:ex:1:RSPstar:0},
  together with its \LLEEwitness~$\onecharthatof{\alert{\stexpit{\bstexp}}}\,$
    that is guaranteed by Theorem~\ref{thm:onechart-int:LLEEw}.
  \begin{figure}[tb]
\begin{center}
\begin{tikzpicture}
  \renewcommand{\prod}{\,{\cdot}\,}
  \renewcommand{\stackprod}{\,{\chocolate{\varstar}}\,}

\matrix[anchor=center,row sep=1cm,column sep=1.75cm,
        every node/.style={draw=none}
        ] {
    & \node(v11){}; & & & & \node(v21){};
    \\
    \node(v1acc){}; & & \node(v1){}; & & \node(v2acc){}; & & \node(v2){};
    \\
    & & & \node(v){}; 
    \\ 
  };

\path (v11) ++ (0cm,0cm) node{$ ((1 \stackprod a^*) \prod b^*) \stackprod \alert{f^*} $}; 
\draw[->,thick,densely dotted,out=0,in=60] ($(v11) + (1.4cm,0cm)$) to ($(v1) + (0.8cm,0.3cm)$);

\path (v1acc) ++ (0cm,0cm) node{$ ((1 \prod a^*) \prod b^*) \stackprod \alert{f^*} $};
\draw[->,shorten <=0.35cm,shorten >=0.35cm] 
  (v1acc) to node[below,pos=0.525]{\small \black{$\aact$}} (v11);
\draw[->,shorten <=1.25cm,shorten >=1.2cm] (v1acc) to node[below,pos=0.2]{\small $\bact$} (v21);
\draw[-,thick,densely dotted,out=270,in=180,shorten <= 0.15cm,distance=1.15cm] (v1acc) to ($(v) + (-1.4cm,0cm)$);
 
\path (v1) ++ (0cm,0cm) node{$ (a^* \prod b^*) \stackprod \alert{f^*} $};
\draw[->,thick,densely dotted,out=270,in=180,shorten <= 0.15cm] (v1) to ($(v) + (-1.4cm,0cm)$);
\draw[->,thick,darkcyan,shorten <=0.35cm,shorten >=0.4cm] 
  (v1) to node[below,pos=0.65]{\small \black{$\aact$}} node[above,pos=0.4]{\small $\loopsteplab{1}$} (v11);
\draw[->,shorten <=1.25cm,shorten >=1cm] (v1) to node[below,pos=0.45]{\small $\bact$} (v21);

\path (v21) ++ (0cm,0cm) node{$ (1 \stackprod b^*) \stackprod \alert{f^*} $}; 
\draw[->,thick,densely dotted,out=0,in=90,shorten >=0.1cm] ($(v21) + (1.1cm,0cm)$) to (v2);

\path (v2acc) ++ (0cm,0cm) node{$ (1 \prod b^*) \stackprod \alert{f^*} $};
\draw[->,shorten <=0.35cm,shorten >=0.35cm] 
  (v2acc) to node[below,pos=0.525]{\small \black{$\bact$}} (v21); 
\draw[-,thick,densely dotted,out=270,in=0,shorten <= 0.25cm] (v2acc) to ($(v) + (1.4cm,0cm)$);  
 
\path (v2) ++ (0cm,0cm) node{$ b^* \stackprod \alert{f^*} $};
\draw[->,thick,densely dotted,out=225,in=0,shorten <= 0.25cm] (v2) to ($(v) + (1.4cm,0cm)$); 
\draw[->,thick,darkcyan,shorten <=0.3cm,shorten >=0.3cm] 
  (v2) to node[below,pos=0.525]{\small \black{$\bact$}} node[above,pos=0.4]{\small $\loopsteplab{1}$} (v21);

\draw[thick,chocolate,double] (v) ellipse (1.35cm and 0.31cm);
\path (v) ++ (0cm,-0.3cm) node{$ \underbrace{{((a \prod a^* + b) \prod b^*)^*}}_{\alert{f^*}} $}; 
\draw[->,thick,darkcyan,shorten <=1cm,shorten >=1.2cm] 
  (v) to node[below,pos=0.575]{\small \black{$\aact$}} node[above,pos=0.55]{\small $\loopsteplab{2}$} (v1acc);
\draw[->,thick,darkcyan,shorten <=0.35cm,shorten >=0.35cm] 
  (v) to node[below,pos=0.55]{\small \black{$\bact$}} node[above,pos=0.425]{\small $\loopsteplab{2}$} (v2acc);
    
\end{tikzpicture}
\end{center}
  \caption{\label{fig:ex:1:RSPstar:0}%
           The \onechart\ interpretation $\onechartof{\alert{f^*}}$ for $\alert{f^*}$ as 
           in \eqref{eq:ex:1:RSPstar:2:LLEEcoindproof} in Example~\ref{ex:1:RSPstar:2:LLEEcoindproof}. 
           }

\end{figure}
  
  Due to Lemma~\ref{lem:onechart-int:milnersysmin:solvable} 
    the iterated partial \onederivatives\ as depicted define a \provablein{\milnersysmin} solution of $\onechartof{\alert{\stexpit{\bstexp}}}$
    when stacked products $\sstexpstackprod$ are replaced by products $\sstexpprod\,$. 
  From this \LLEEwitness\ that carries a \provablein{\milnersys} solution we now obtain a \LLEEwitnessed\ coinductive proof 
    of $ \alert{f} \prod \chocolate{e} + \forestgreen{g} \formeq \alert{f^*} \prod \forestgreen{g}$
    under the assumption of $\chocolate{e} \formeq \alert{f} \prod \chocolate{e} + \forestgreen{g}$, as follows.
      By replacing parts $({\ldots}) \stackprod \alert{f^*}$ by $\proj{({\ldots})} \prod \chocolate{e}$   
    in the \provablein{\milnersys} solution of $\onechartof{\alert{f^*}}$,
  and respectively, by replacing $({\ldots}) \stackprod \alert{f^*}$ by $(\proj{({\ldots})} \prod \alert{f^*}) \prod \forestgreen{g}$  
  we obtain the left- and the right-hand sides of the formal equations in the cyclic derivation in Figure~\ref{fig:ex-1-RSPstar-4}.
  \begin{figure}[tb]
  \begin{center}
\begin{tikzpicture}
 
\matrix[anchor=center,row sep=1.75cm,column sep=1.35cm,
        every node/.style={draw=none}
        ] {
    & \node(v11){}; & & & & \node(v21){};
    \\ 
    \node(v1acc){}; & & \node(v1){}; & & \node(v2acc){}; & & \node(v2){};
    \\
    & & & \node(v){}; 
    \\
  };

\path (v11) ++ (0cm,0cm) node{$ ((1 \prod a^*) \prod b^*) \prod \chocolate{e}
                                  \formeq
                                (((1 \prod a^*) \prod b^*) \prod \alert{f^*}) \prod \forestgreen{g} $}; 
\draw[->,thick,densely dotted,out=-45,in=90] ($(v11) + (1cm,-0.35cm)$) to ($(v1) + (0.5cm,0.55cm)$);

\path (v1acc) ++ (0cm,-0.3cm) node{$ ((1 \prod a^*) \prod b^*) \prod \chocolate{e} 
                                    \formeq
                                   (((1 \prod a^*) \prod b^*) \prod \alert{f^*}) \prod \forestgreen{g} $};
\draw[->,shorten <=0cm,shorten >=0.3cm] 
  (v1acc) to node[above,pos=0.35]{\small \black{$\aact$}} (v11);
\draw[->,shorten <=0cm,shorten >=1cm,out=35,in=195] (v1acc) to node[above,pos=0.15]{\small $\bact$} (v21);
\draw[->,thick,densely dotted,out=270,in=120,shorten <= 0.25cm,shorten >= 0.325cm] ($(v1acc) + (-0.475cm,-0.3cm)$) to ($(v) + (-2.25cm,0cm)$);

\path (v1) ++ (0cm,0.3cm) node{$ (a^* \prod b^*) \prod \chocolate{e} 
                                 \formeq
                               ((a^* \prod b^*) \prod \alert{f^*}) \prod \forestgreen{g} $};
\draw[->,thick,densely dotted,out=270,in=90,shorten >= 0.2cm] ($(v1) + (0cm,0cm)$) to (v); 
\draw[->,thick,darkcyan,shorten <=0.5cm,shorten >=0.4cm] 
  (v1) to node[below,pos=0.45,xshift=-0.05cm]{\small \black{$\aact$}} node[above,pos=0.25,xshift=0.05cm]{\small $\loopsteplab{1}$} (v11);
\draw[->,shorten <=1.25cm,shorten >=1cm] (v1) to node[below,pos=0.45]{\small $\bact$} (v21);

\path (v21) ++ (0cm,0cm) node{$ (1 \prod b^*) \prod \chocolate{e}
                                  \formeq
                                ((1 \prod b^*) \prod \alert{f^*}) \prod \forestgreen{g} $}; 
\draw[->,thick,densely dotted,out=315,in=90,shorten <=0.5cm,shorten >=0.5cm] ($(v21) + (0.75cm,0cm)$) to (v2); 
 
\path (v2acc) ++ (0cm,-0.3cm) node{$ (1 \prod b^*) \prod \chocolate{e} 
                                    \formeq
                                  ((1 \prod b^*) \prod \alert{f^*}) \prod \forestgreen{g} $};
\draw[->,shorten <=0cm,shorten >=0.35cm] 
  (v2acc) to node[below,pos=0.525]{\small \black{$\bact$}} (v21);  
\draw[->,thick,densely dotted,out=290,in=70,shorten <= 0.2cm,shorten >= 0.2cm,distance=0.5cm] ($(v2acc) + (0cm,-0.3cm)$) to ($(v) + (1cm,0.15cm)$); 
 
\path (v2) ++ (0cm,0.3cm) node{$ b^* \prod \chocolate{e}
                                    \formeq
                                  (b^* \prod \alert{f^*}) \prod \forestgreen{g} $};
\draw[->,thick,densely dotted,out=265,in=30,shorten <= 0cm] (v2) to ($(v) + (4.4cm,0.15cm)$); 
\draw[->,thick,darkcyan,shorten <=0.5cm,shorten >=0.3cm] 
  (v2) to node[below,pos=0.625]{\small \black{$\bact$}} node[above,pos=0.425,xshift=0.05cm]{\small $\loopsteplab{1}$} (v21);

(v)
\path (v) ++ (0cm,-1cm) node{$ \underbrace{
                                   \underbrace{{((a \prod a^* + b) \prod b^*)}\rule[-10pt]{0pt}{11.5pt}}_{\alert{f}} \prod \underbrace{ (a + b)^* \rule[-10pt]{0pt}{11.5pt}}_{\chocolate{e}} + \underbrace{1\rule[-10pt]{0pt}{11.5pt}}_{\forestgreen{g}}
                                             }_{\sformeq\: \chocolate{e} \;\; \text{\small (by rule assumption)}} 
                                   \hspace*{-1.5ex} \formeq
                                 \underbrace{{(((a \prod a^* + b) \prod b^*)^* \rule[-10pt]{0pt}{11.5pt}}}_{\alert{f^*}} \prod \hspace*{-2ex} \underbrace{ 1 \rule[-10pt]{0pt}{11.5pt} }_{\forestgreen{g}} $}; 

\draw[->,thick,darkcyan,shorten <=0.75cm,shorten >=1.5cm] 
  (v) to node[below,pos=0.4]{\small \black{$\aact$}} node[above,pos=0.375]{\small $\loopsteplab{2}$} (v1acc);
\draw[->,thick,darkcyan,shorten <=0.25cm,shorten >=0.6cm] 
  (v) to node[below,pos=0.425]{\small \black{$\bact$}} node[above,pos=0.275]{\small $\loopsteplab{2}$} (v2acc);

\draw[thick,chocolate,double] ($(v) + (-0.15cm,-0.2cm)$) ellipse (5.35cm and 0.385cm);

%
%
%
%
%
%
%

\end{tikzpicture}
\end{center}
  \vspace*{-1.5ex}
  \caption{\label{fig:ex-1-RSPstar-4}%
           \protect\LLEEwitnessed\ coinductive proof 
           of $ \alert{f} \cdot \chocolate{e} + \forestgreen{g} = \alert{f^*} \prod \forestgreen{g}$
           over $\thplus{\milnersysmin}{\setexp{\chocolate{\astexp}   \formeq   \stexpsum{\stexpprod{\alert{\bstexp}}{\chocolate{\astexp}}}{\forestgreen{\cstexp}}}}\,$.
            }
\end{figure}%
  That derivation is a \LLEEwitnessed\ coinductive proof $\aLLEECoProof$ of $ \alert{f} \cdot \chocolate{e} + \forestgreen{g} = \alert{f^*} \prod \forestgreen{g}$
    over $\thplus{\milnersysmin}{\setexp{\chocolate{\astexp}   \formeq   \stexpsum{\stexpprod{\alert{\bstexp}}{\chocolate{\astexp}}}{\forestgreen{\cstexp}}}}\,$:
  The right-hand sides form a \provablein{\milnersys} solution of $\chartof{\alert{f^*}\cdot\forestgreen{g}}$
  due to Lemma~\ref{lem:onechart-int:milnersysmin:solvable}
  (note that $\chartof{\alert{f^*}\cdot\forestgreen{g}}$ is isomorphic to $\chartof{\alert{f^*}}$ due to $\forestgreen{g} \synteq 1$).
  The left-hand sides also form a solution of $\chartof{\alert{f^*}\cdot\forestgreen{g}}$
  (see Lemma~\ref{lem:lem:mimic:RSPstar} below),
  noting that for the \onetransitions\ back to the conclusion the assumption $\chocolate{e} \formeq \alert{f} \prod \chocolate{e} + \forestgreen{g}$
  must be used in addition to $\milnersysmin$.
  By using this assumption again, the result $\aLLEECoProofacc$ of replacing
    $\alert{f} \prod \chocolate{e} + \forestgreen{g}$ in the conclusion of $\aLLEECoProof$
  by $\chocolate{e}$ is also a \LLEEwitnessed\ coinductive proof 
    over $\thplus{\milnersysmin}{\setexp{\chocolate{\astexp}   \formeq   \stexpsum{\stexpprod{\alert{\bstexp}}{\chocolate{\astexp}}}{\forestgreen{\cstexp}}}}$.
  Consequently:
  \begin{equation}
    \begin{aligned}
      \AxiomC{$ \chocolate{\astexp}   \formeq   \stexpsum{\stexpprod{\alert{\bstexp}}{\chocolate{\astexp}}}{\forestgreen{\cstexp}} $}
      \AxiomC{$ \aLLEECoProofoverof{\thplus{\milnersysmin}{\setexp{\chocolate{\astexp}   \formeq   \stexpsum{\stexpprod{\alert{\bstexp}}{\chocolate{\astexp}}}{\forestgreen{\cstexp}}}}}
                                   {\chocolate{\astexp}   \formeq   \stexpprod{\stexpit{\alert{\bstexp}}}{\forestgreen{\cstexp}}} $}
      \insertBetweenHyps{\hspace*{1ex}}                             
      \RightLabel{\LCoindProofi{1}}
      \BinaryInfC{$ \chocolate{\astexp}   \formeq   \stexpprod{\alert{\stexpit{\bstexp}}}{\forestgreen{\cstexp}} $}           
      \DisplayProof
    \end{aligned}
  \end{equation}
  is a rule instance of $\coindmilnersys$ and $\CLC$ by which we have mimicked the $\RSPstar$ instance in \eqref{eq:ex:1:RSPstar:2:LLEEcoindproof}.
\end{exa}
    \afterpage{\FloatBarrier}%

By examining the steps that we used in the example above,
  we find that three main auxiliary statements were used
    for the construction of a \LLEEwitnessed\ coinductive proof
      that mimics an instance of the fixed-point rule \RSPstar\ 
        by an instance of $\LCoindProofi{1}$.
In relation to an instance of \RSPstar\ of the generic form as in Definition~\ref{def:milnersys},
  these are the statements that, for all star expressions $\astexp$, $\bstexp$, and $\cstexp$, it holds: 
\begin{enumerate}[label={(\alph{*})},leftmargin=*,align=right,itemsep=0.25ex]
  \item{}\label{it:thm:onechart-int:LLEEw}
    The \onechart\ interpretation $\onechartof{\astexp}$ of a star expressions $\astexp$ is a guarded \LLEEonechart. 
  \item{}\label{it:lem:onechart-int:milnersysmin:solvable}
    $e$ is the principal value of a \provablein{\milnersysmin} solution of the \onechart\ interpretation $\onechartof{\astexp}$~of~$\astexp$.
  \item{}\label{it:lem:lem:mimic:RSPstar}
    $e$ is the principal value of a \provablein{(\thplus{\milnersysmin}{\setexp{ e = f \prod e +  g}})} solution
      of the \onechart\ interpretation $\onechartof{f^* \prod g}$ of $f^* \prod g$. 
\end{enumerate}
While \ref{it:thm:onechart-int:LLEEw} is guaranteed by Theorem~\ref{thm:onechart-int:LLEEw},
  and \ref{it:lem:onechart-int:milnersysmin:solvable} by Lemma~\ref{lem:onechart-int:milnersysmin:solvable},
we are now going to justify the central statement \ref{it:lem:lem:mimic:RSPstar} by proving the following lemma.

\begin{lem}\label{lem:lem:mimic:RSPstar}
  Let $\astexp,\bstexp,\cstexp\in\StExpover{\actions}$ with $\alert{\notterminates{\bstexp}}$,
  and let $\aseteqs \defdby \setexp{ \astexp = \stexpsum{\stexpprod{\bstexp}{\astexp}}{\cstexp} }$. 
  Then 
       $\astexp$ is the principal value of a
  \provablein{(\thplus{\milnersysmin}{\aseteqs})} solution of
    the \onechart\ interpretation $\onechartof{\stexpprod{\stexpit{\bstexp}}{\cstexp}}$~of~$\stexpit{\bstexp}{\prod}{\cstexp}$.
\end{lem}

\begin{proof}
  First, it can be verified that the vertices of $\onechartof{f^* \cdot g}$ are of either of three forms:
  \begin{equation}\label{eq:1:lem:lem:mimic:RSPstar}
    \vertsof{\onechartof{f^* \cdot g}}
      =
    \setexp{ f^* \cdot g }
      \cup
    \descsetexpbig{ (F \stackprod f^*) \cdot g }{ F \in \itpartonederivs{f} }
      \cup
    \descsetexp{ G }{ G \in \itpartonederivs{g} } \punc{,}  
  \end{equation}
  where $\itpartonederivs{f}$ means the set of iterated \onederivatives\ of $f$ according to the TSS in Def.~\ref{def:onechartof}.
  
  This facilitates to define a function $\sasol \funin \vertsof{\onechartof{f^* \cdot g}} \to \StExpover{\actions}$
    on $\onechartof{f^* \cdot g}$ by:
  \begin{alignat*}{2}
    \asol{f^* \prod g}
      & {} \defdby
             e \punc{,}
    \\
    \asol{(F \stackprod f^*) \prod g}
      & {} \defdby
             \proj{F} \prod e \punc{,}
        & & \qquad \text{(for $F \in \itpartonederivs{f}$),}       
    \\
    \asol{G}
      & {} \defdby
             \proj{G}
        & & \qquad \text{(for $G \in \itpartonederivs{g}$),}   
  \end{alignat*}
  We will show that $\sasol$ is a \provablein{(\thplus{\milnersysmin}{\aseteqs})} solution of $\onechartof{f^* \cdot g}$.
  Instead of verifying the correctness conditions for $\sasol$ for list representations of transitions,
    we will argue more loosely with sums over action \onederivatives\ sets $\oneactderivs{H}$ of stacked star expressions $H$
    where such sums are only \welldefined\ up to \ACI. Due to $\ACI \subsystem \milnersysmin$ such an argumentation 
    is possible. Specifically we will demonstrate, for all $E\in\vertsof{\onechartof{f^* \cdot g}}$, 
    that $\sasol$ is a \provablein{(\thplus{\milnersysmin}{\aseteqs})} solution at $E$, that is, that it holds:
  \begin{equation}\label{eq:2:lem:lem:mimic:RSPstar}
    \asol{E}
      \eqin{\thplus{\milnersysmin}{\aseteqs}}
    \terminatesconstof{\onechartof{E}}{E}
      + \hspace*{-3.75ex}
    \sum_{\pair{\aoneact}{E'}\in\oneactderivs{E}} \hspace*{-3.75ex}
      \aoneact \prod \asol{E'} \punc{,}
  \end{equation}
  where by the sum on the right-hand side we mean an arbitrary representative 
    of the \ACI\ equi\-va\-lence class of star expressions that is described by the sum expression of~this~form.  
  
  \smallskip
  
  For showing \eqref{eq:2:lem:lem:mimic:RSPstar}, 
    we distinguish the three cases of vertices $E\in\vertsof{\onechartof{f^* \prod g}}$  
  according to \eqref{eq:1:lem:lem:mimic:RSPstar},
    that is, $E \synteq f^* \prod g$,
             $E \synteq (F \stackprod f^*) \prod g$ for some $F\in\itpartonederivs{f}$,
             and $E \synteq G$ for~some~$G\in\itpartonederivs{g}$. 
  We will see that the assumption $\aseteqs$ will only be needed for the treatment of the first case.
  
  In the first case, we consider $E \synteq f^* \prod g$.
  We find by Lemma~\ref{lem:actonederivs} 
    (or by inspecting the TSS in Definition~\ref{def:onechartof}), and in view of \eqref{eq:1:lem:lem:mimic:RSPstar}:
  \begin{align}
    \oneactderivs{f^* \prod g}
      & {} =
    \descsetexp{ \pair{\aoneact}{(F \stackprod f^*) \prod g} }
               { \pair{\aoneact}{F} \in \oneactderivs{f} } 
       \cup
    \oneactderivs{g} \punc{,}
      \label{eq:3a:lem:lem:mimic:RSPstar}
    \displaybreak[0]\\
    \begin{split}
      \partonederivs{f^* \prod g}
        & {} =
      \descsetexp{ (F \stackprod f^*) \prod g }
                 { F \in \partonederivs{f} }
        \cup
      \partonederivs{g}
      \\
        & {} \subseteq
      \descsetexp{ (F \stackprod f^*) \prod g }
                 { F \in \itpartonederivs{f} }
        \cup
      \itpartonederivs{g}
         \subseteq 
           \vertsof{ \onechartof{f^* \prod g} } \punc{.}   
    \end{split}           
      \label{eq:3b:lem:lem:mimic:RSPstar}               
  \end{align}%
  Then \eqref{eq:3b:lem:lem:mimic:RSPstar} guarantees that $\sasol$ is defined for all partial \onederivatives\ of $E \synteq f^* \cdot g$. 
    With this knowledge we can argue as follows:
  \begin{alignat*}{2}
    \asol{E}
      & \;\,\parbox[t]{\widthof{$\eqin{\thplus{\milnersysmin}{\aseteqs}}$}}{$\synteq$}\:
        \asol{f^* \prod g}
        \qquad\quad \text{(by $E \synteq f^* \prod g$)}
      \\  
      & \;\,\parbox[t]{\widthof{$\eqin{\thplus{\milnersysmin}{\aseteqs}}$}}{$\synteq$}\:
        e \phantom{ {} \cdot e + g} 
        \qquad\quad \text{(by the definition of $\sasol$)}
      \\
      & \;\,\parbox[t]{\widthof{$\eqin{\thplus{\milnersysmin}{\aseteqs}}$}}{$\eqin{\thplus{\milnersysmin}{\aseteqs}}$}\:
        f \prod e + g
        \qquad\quad \text{(since $\aseteqs = \setexp{ e = f \prod e + g }$)}  
      \displaybreak[0]\\ 
      & \;\,\parbox[t]{\widthof{$\eqin{\thplus{\milnersysmin}{\aseteqs}}$}}{$\milnersysmineq$}\:
        \Big(
          \terminatesconstof{\onechartof{f}}{f}
            + \hspace*{-3.75ex}
          \sum_{\pair{\aoneact}{F}\in\oneactderivs{f}} \hspace*{-3.5ex}
                 \aoneact \prod \proj{F}
             \Big) \prod e
        +
        \Big(
          \terminatesconstof{\onechartof{g}}{g}
            + \hspace*{-3.75ex}
          \sum_{\pair{\aoneact}{G}\in\oneactderivs{g}} \hspace*{-3.5ex}
                 \aoneact \prod \proj{G}
             \Big)   
        \\[-0.25ex]
        & \;\,\parbox[t]{\widthof{$\eqin{\thplus{\milnersysmin}{\aseteqs}}$\hspace*{3ex}}}{\mbox{}}\: 
          \text{(by using Lemma~\ref{lem:FT:onechart-int})}   
      \displaybreak[0]\\ 
      & \;\,\parbox[t]{\widthof{$\eqin{\thplus{\milnersysmin}{\aseteqs}}$}}{$\milnersysmineq$}\:
        \Big(
          \terminatesconstof{\onechartof{f}}{f}
            \prod e
            + \hspace*{-3.75ex}
          \sum_{\pair{\aoneact}{F}\in\oneactderivs{f}} \hspace*{-3.5ex}
                 \aoneact \prod (\proj{F} \prod e)
             \Big) 
        +
        \Big(
          \terminatesconstof{\onechartof{g}}{g}
            + \hspace*{-3.75ex}
          \sum_{\pair{\aoneact}{G}\in\oneactderivs{g}} \hspace*{-3.5ex}
                 \aoneact \prod \proj{G}
             \Big)   
        \\
        & \;\,\parbox[t]{\widthof{$\eqin{\thplus{\milnersysmin}{\aseteqs}}$\hspace*{3ex}}}{\mbox{}}\: 
          \text{(by using ($\rdistr$) and ($\assocstexpprod$))} 
      \displaybreak[0]\\[0.5ex] 
      & \;\,\parbox[t]{\widthof{$\eqin{\thplus{\milnersysmin}{\aseteqs}}$}}{$\milnersysmineq$}\:
        \Big(
            \hspace*{-3.75ex}
          \sum_{\pair{\aoneact}{F}\in\oneactderivs{f}} \hspace*{-3.5ex}
                 \aoneact \prod ( \proj{F} \prod e )
             \Big) 
        + \Bigl(
            \terminatesconstof{\onechartof{f^* \prod g}}{f^* \prod g}
              + \hspace*{-3.75ex}
            \sum_{\pair{\aoneact}{G}\in\oneactderivs{g}} \hspace*{-3.5ex}
                   \aoneact \prod \proj{G} 
          \Bigr)   
        \\
        & \;\,\parbox[t]{\widthof{$\eqin{\thplus{\milnersysmin}{\aseteqs}}$\hspace*{3ex}}}{\mbox{}}\: 
          \parbox{\widthof{(since $\terminatesconstof{\onechartof{f}}{f} \synteq 0$ due to $\alert{\notterminates{f}}$,
                            by using axiom ($\neutralstexpsum$)}}
            {(since $\terminatesconstof{\onechartof{f}}{f} \synteq 0$ due to $\alert{\notterminates{f}}$,
                  by using axiom ($\stexpzerostexpprod$),
             \\
             and $\terminatesconstof{\onechartof{f^* \prod g}}{f^* \prod g} \synteq \terminatesconstof{\onechartof{g}}{g}$)}     
      \displaybreak[0]\\[0.5ex] 
      & \;\,\parbox[t]{\widthof{$\eqin{\thplus{\milnersysmin}{\aseteqs}}$}}{$\ACIeq$}\:
        \terminatesconstof{\onechartof{f^* \prod g}}{f^* \prod g}
           + 
        \Big(
            \hspace*{-3.75ex}
          \sum_{\pair{\aoneact}{F}\in\oneactderivs{f}} \hspace*{-3.5ex}
                 \aoneact \prod ( \asol{(F \stackprod f^*) \cdot g } )
             \Big) 
        +   \hspace*{-3.75ex}
          \sum_{\pair{\aoneact}{G}\in\oneactderivs{g}} \hspace*{-3.5ex}
                 \aoneact \prod \asol{G}    
        \\[-0.25ex]
        & \;\,\parbox[t]{\widthof{$\eqin{\thplus{\milnersysmin}{\aseteqs}}$\hspace*{3ex}}}{\mbox{}}\: 
          \text{(by definition of $\sasol$, axioms ($\commstexpsum$))}     
      \displaybreak[0]\\[0.5ex] 
      & \;\,\parbox[t]{\widthof{$\eqin{\thplus{\milnersysmin}{\aseteqs}}$}}{$\ACIeq$}\:
        \terminatesconstof{\onechartof{E}}{E}
           + 
            \hspace*{-5.75ex}
          \sum_{\pair{\aoneact}{E'}\in\oneactderivs{f^* \cdot g}=\oneactderivs{E}} \hspace*{-5.5ex}
                 \aoneact \prod \asol{E'}  
        \\
        & \;\,\parbox[t]{\widthof{$\eqin{\thplus{\milnersysmin}{\aseteqs}}$\hspace*{3ex}}}{\mbox{}}\: 
          \text{(by $E \synteq f^* \prod g$ and \eqref{eq:3a:lem:lem:mimic:RSPstar}).}
  \end{alignat*}  
  Due to $\ACI \subsystem \milnersysmin \subsystem \thplus{\milnersysmin}{\aseteqs}$
    this chain of equalities is provable in $\thplus{\milnersysmin}{\aseteqs}$, 
    which verifies \eqref{eq:2:lem:lem:mimic:RSPstar} for $E$ as considered here,
  or in other words,    
    $\sasol$ is a \provablein{(\thplus{\milnersysmin}{\aseteqs})} solution of at $E$.
  
  \smallskip
  
  In the second case we consider $E \synteq (F \stackprod f^*) \prod g \in\vertsof{\onechartof{f^* \prod g}}$. 
  Then $F \in \itpartonederivs{f}$, and  
  $\terminatesconstof{\onechartof{E}}{E} 
     \synteq 
   \terminatesconstof{\onechartof{(F \stackprod f^*) \prod g}}{(F \stackprod f^*) \prod g}
     \synteq 0$ holds,
  because expressions with stacked product occurring do not have immediate termination by Definition~\ref{def:onechartof}.   
  We distinguish the subcases $\terminates{F}$~and~$\notterminates{F}$.
  
  For the first subcase we assume $\notterminates{F}$. 
  Then $\terminatesconstof{\onechartof{F}}{F} \synteq 0$ holds,
    and we find by Lemma~\ref{lem:actonederivs} (or by inspecting the TSS in Definition~\ref{def:onechartof}),
      by $F \in \itpartonederivs{f}$, and from \eqref{eq:1:lem:lem:mimic:RSPstar}:
  \begin{align}
    \oneactderivs{(F \stackprod f^*) \prod g}
      & {} =
    \descsetexp{ \pair{\aoneact}{(F' \stackprod f^*) \prod g} }
               { \pair{\aoneact}{F'} \in \oneactderivs{F} } \punc{,}
      \label{eq:2c:lem:lem:mimic:RSPstar}           
    \\
    \begin{split}
      \partonederivs{(F \stackprod f^*) \prod g}
        & {} =
      \descsetexp{ (F' \stackprod f^*) \prod g }
                 { F' \in \partonederivs{F} }
      \\
        & {} \subseteq
      \descsetexp{ (F' \stackprod f^*) \prod g }
                 { F' \in \itpartonederivs{f} }              
        \subseteq 
          \vertsof{ \onechartof{f^* \cdot g} } \punc{.} 
      \label{eq:2d:lem:lem:mimic:RSPstar}                
    \end{split}                          
  \end{align}%
  Due to \eqref{eq:2d:lem:lem:mimic:RSPstar}, $\sasol$ is defined for all partial \onederivatives\ of $E \synteq (F \stackprod f^*) \prod g$.
    We argue as follows:
  \begin{alignat*}{2}
    \asol{E}
      & \;\,\parbox[t]{\widthof{$\eqin{\milnersysmin}$}}{$\synteq$}\:
        \asol{(F \stackprod f^*) \prod g}
        & & \quad\text{(by $E \synteq (F \stackprod f^*) \prod g$)}
      \\  
      & \;\,\parbox[t]{\widthof{$\eqin{\milnersysmin}$}}{$\synteq$}\:
        \proj{F} \prod  e
        & & \quad\text{(by the definition of $\sasol$)} 
      \displaybreak[0]\\ 
      & \;\,\parbox[t]{\widthof{$\eqin{\milnersysmin}$}}{$\milnersysmineq$}\:
        \Big(
          \terminatesconstof{\onechartof{F}}{F}
            + \hspace*{-3.75ex}
          \sum_{\pair{\aoneact}{F'}\in\oneactderivs{F}} \hspace*{-3.5ex}
                 \aoneact \prod \proj{F'}
             \Big) \prod e
        & & \quad\text{(by using Lemma~\ref{lem:FT:onechart-int})} 
      \displaybreak[0]\\ 
      & \;\,\parbox[t]{\widthof{$\eqin{\milnersysmin}$}}{$\milnersysmineq$}\:
        0 \prod e 
          +
            \hspace*{-3.75ex}
          \sum_{\pair{\aoneact}{F'}\in\oneactderivs{F}} \hspace*{-3.5ex}
                 \aoneact \prod ( \proj{F'} \prod e )
        & & \quad\parbox{\widthof{(by $\terminatesconstof{\onechartof{F}}{F} \synteq 0$, due to $\notterminates{F}$, and}}%
                        {(by $\terminatesconstof{\onechartof{F}}{F} \synteq 0$, due to $\notterminates{F}$, and
                         \\[-0.5ex]\phantom{(}%
                         axioms ($\rdistr$), ($\assocstexpprod$))}
      \displaybreak[0]\\ 
      & \;\,\parbox[t]{\widthof{$\eqin{\milnersysmin}$}}{$\milnersysmineq$}\:
        0 +
            \hspace*{-3.75ex}
          \sum_{\pair{\aoneact}{F'}\in\oneactderivs{F}} \hspace*{-3.5ex}
                 \aoneact \prod \asol{ (F' \stackprod f^*) \prod g } 
        & & \quad\text{(by ax.\ ($\stexpzerostexpprod$) and def. of $\sasol$)} 
      \\ 
      & \;\,\parbox[t]{\widthof{$\eqin{\milnersysmin}$}}{$\ACIeq$}\:
        \terminatesconstof{\onechartof{(F \stackprod f^*) \prod g}}{ (F \stackprod f^*) \prod g }
          + 
            \hspace*{-5ex}
          \sum_{\pair{\aoneact}{E'}\in\oneactderivs{(F \stackprod f^*) \prod g}} \hspace*{-8ex}
                 \aoneact \prod \asol{ E' } 
        & & \quad\text{(due to \eqref{eq:2c:lem:lem:mimic:RSPstar}, and $\terminatesconstof{\onechartof{E}}{E} \synteq 0)$}    
      \displaybreak[0]\\[0.5ex] 
      & \;\,\parbox[t]{\widthof{$\eqin{\milnersysmin}$}}{$\synteq$}\:
        \terminatesconstof{\onechartof{E}}{E}
           + 
            \hspace*{-3.75ex}
          \sum_{\pair{\aoneact}{E'}\in\oneactderivs{E}} \hspace*{-3.5ex}
                 \aoneact \prod \asol{E'}  
        & & \quad \text{(by $E \synteq (F \stackprod f^*) \prod g$).}
  \end{alignat*}  
  For the second subcase we assume $\terminates{F}$.
  Then $F\in\StExpover{\actions}$ (that is, $F$ does not contain a stacked product symbol),
    and $\terminatesconstof{\onechartof{F}}{F} \synteq 1$ holds.
  Furthermore, we find, again by inspecting the TSS in Definition~\ref{def:onechartof},
    by $F \in \itpartonederivs{f}$, and in view of \eqref{eq:1:lem:lem:mimic:RSPstar}:
  \begin{align}
    \oneactderivs{(F \stackprod f^*) \prod g}
      & {} =
    \setexp{ \pair{\sone}{f^* \prod g} }
      \cup
    \descsetexp{ \pair{\aoneact}{(F' \stackprod f^*) \prod g} }
               { \pair{\aoneact}{F'} \in \oneactderivs{F} } \punc{.}
    \label{eq:4a:lem:lem:mimic:RSPstar}           
    \\
    \begin{split}
      \partonederivs{(F \stackprod f^*) \prod g}
       & {} =
         \setexp{ f^* \prod g }
           \cup
         \descsetexp{ (F' \stackprod f^*) \prod g }
                    { F' \in \partonederivs{F} } 
         \\
       & {} \subseteq
         \setexp{ f^* \prod g }
           \cup
         \descsetexp{ (F' \stackprod f^*) \prod g }
                    { F' \in \itpartonederivs{f} }
            \subseteq 
          \vertsof{ \onechartof{f^* \cdot g} } \punc{.} 
      \label{eq:4b:lem:lem:mimic:RSPstar}       
    \end{split}            
  \end{align}%
  Due to \eqref{eq:4b:lem:lem:mimic:RSPstar}, $\sasol$ is defined for all partial \onederivatives\ of $E \synteq (F \stackprod f^*) \prod g$
    also in this subcase. Then we can argue as follows:
  \begin{alignat*}{2}
    \asol{E}
      & \;\,\parbox[t]{\widthof{$\eqin{\milnersysmin}$}}{$\synteq$}\:
        \asol{(F \stackprod f^*) \prod e}
        & & \quad\text{(by $E \synteq (F \stackprod f^*) \prod g$)}
      \\  
      & \;\,\parbox[t]{\widthof{$\eqin{\milnersysmin}$}}{$\synteq$}\:
        \proj{F} \prod  e
        & & \quad\text{(by the definition of $\sasol$)}  
      \displaybreak[0]\\ 
      & \;\,\parbox[t]{\widthof{$\eqin{\milnersysmin}$}}{$\milnersysmineq$}\:
        \Big(
          \terminatesconstof{\onechartof{F}}{F}
            + \hspace*{-3.75ex}
          \sum_{\pair{\aoneact}{F'}\in\oneactderivs{F}} \hspace*{-3.5ex}
                 \aoneact \prod \proj{F'}
             \Big) \prod e
        & & \quad\text{(by using Lemma~\ref{lem:FT:onechart-int})}
      \displaybreak[0]\\ 
      & \;\,\parbox[t]{\widthof{$\eqin{\milnersysmin}$}}{$\milnersysmineq$}\:
        1 \prod e 
          +
            \hspace*{-3.75ex}
          \sum_{\pair{\aoneact}{F'}\in\oneactderivs{f}} \hspace*{-3.5ex}
                 \aoneact \prod ( \proj{F'} \prod e )
        & & \quad\parbox{\widthof{(axioms ($\rdistr$), ($\assocstexpprod$))}}%
                        {(by $\terminatesconstof{\onechartof{F}}{F} \synteq 1$, and
                         \\[0ex]\phantom{(}%
                         axioms ($\rdistr$), ($\assocstexpprod$))}
      \displaybreak[0]\\ 
      & \;\,\parbox[t]{\widthof{$\eqin{\milnersysmin}$}}{$\synteq$}\:
        1 \prod \asol{f^* \prod g} +
            \hspace*{-3.75ex}
          \sum_{\pair{\aoneact}{F'}\in\oneactderivs{f}} \hspace*{-3.5ex}
                 \aoneact \prod \asol{ (F' \stackprod f^*) \prod g } 
        & & \quad\text{(by the definition of $\sasol$)}
      \displaybreak[0]\\ 
      & \;\,\parbox[t]{\widthof{$\eqin{\milnersysmin}$}}{$\ACIeq$}\:
        0 +
            \hspace*{-3.75ex}
          \sum_{\pair{\aoneact}{E'}\in\oneactderivs{(F \stackprod f^*) \prod g}} \hspace*{-7.5ex}
                 \aoneact \prod \asol{ E' } 
        & & \quad \parbox{\widthof{(($\commstexpsum$), and ($\assocstexpsum$))}}
                         {(by \eqref{eq:4a:lem:lem:mimic:RSPstar}, using axioms
                          \\\phantom{(}
                          ($\commstexpsum$), and ($\assocstexpsum$))}
      \displaybreak[0]\\  
      & \;\,\parbox[t]{\widthof{$\eqin{\milnersysmin}$}}{$\synteq$}\:
        \terminatesconstof{\onechartof{E}}{E} 
          + \hspace*{-3.75ex}
        \sum_{\pair{\aoneact}{E'}\in\oneactderivs{E}} \hspace*{-3.5ex}
          \aoneact \prod \asol{E'} 
        & & \quad \text{(by $E \synteq (F \stackprod f^*) \prod g$,
                         and $\terminatesconstof{\onechartof{E}}{E} \synteq 0$).}
  \end{alignat*}  
  Due to $\ACI \subsystem \milnersysmin \subsystem \thplus{\milnersysmin}{\aseteqs}$
    the chains of equalities in both subcases are provable in $\thplus{\milnersysmin}{\aseteqs}$,
  and therefore we have now verified \eqref{eq:2:lem:lem:mimic:RSPstar} also in the (entire) second case,
    that is, that $\sasol$ is a \provablein{(\thplus{\milnersysmin}{\aseteqs})} solution of $\onechartof{f^* \prod g}$ at $E$ as in this case. 
  
  \smallskip
  
  In the final case, $E \synteq G$ with $G \in \itpartonederivs{g}$.
  Since then \onederivatives\ of $G$ are in $\itpartonederivs{g}$ as well, and hence by \eqref{eq:1:lem:lem:mimic:RSPstar} also in $\vertsof{ \onechartof{f^* \cdot g} }$,
    it follows that $\sasol$ is defined for all \onederivatives\ of $G$ and $E$. 
  With this knowledge we can argue as follows:
  \begin{alignat*}{2}
    \asol{E}
      & \;\,\parbox[t]{\widthof{$\eqin{\milnersysmin}$}}{$\synteq$}\:
        \asol{G}
        & & \qquad \text{(by $E \synteq G$)}
      \\   
      & \;\,\parbox[t]{\widthof{$\eqin{\milnersysmin}$}}{$\synteq$}\:
        \proj{G} 
        & & \qquad \text{(by the definition of $\sasol$)}
      \displaybreak[0]\\
      & \;\,\parbox[t]{\widthof{$\eqin{\milnersysmin}$}}{$\milnersysmineq$}\:
        \terminatesconstof{\onechartof{G}}{G}
          + \hspace*{-3.75ex}
        \sum_{\pair{\aoneact}{G'}\in\oneactderivs{G}} \hspace*{-3.5ex}
          \aoneact \prod \proj{G'}
        & & \qquad \text{(by using Lemma~\ref{lem:FT:onechart-int})}  
      \displaybreak[0]\\ 
      & \;\,\parbox[t]{\widthof{$\eqin{\milnersysmin}$}}{$\eqin{\ACI}$}\:
        \terminatesconstof{\onechartof{G}}{G} 
          + \hspace*{-3.75ex}
        \sum_{\pair{\aoneact}{G'}\in\oneactderivs{G}} \hspace*{-3.5ex}
          \aoneact \prod \asol{G'} 
        & & \qquad \text{(by the definition\ of $\sasol$)}  
      \displaybreak[0]\\ 
      & \;\,\parbox[t]{\widthof{$\eqin{\milnersysmin}$}}{$\eqin{\ACI}$}\:
        \terminatesconstof{\onechartof{E}}{E} 
          + \hspace*{-3.75ex}
        \sum_{\pair{\aoneact}{E'}\in\oneactderivs{E}} \hspace*{-3.5ex}
          \aoneact \prod \asol{E'} 
        & & \qquad \text{(by $E \synteq G$).}
  \end{alignat*}  
  Due to $\ACI \subsystem \milnersysmin \subsystem \thplus{\milnersysmin}{\aseteqs}$
    this chain of equalities verifies \eqref{eq:2:lem:lem:mimic:RSPstar} also in this case.
    
  By having established \eqref{eq:2:lem:lem:mimic:RSPstar} for the, according to \eqref{eq:1:lem:lem:mimic:RSPstar}, three possible forms of stacked star expressions
    that are vertices of $\onechartof{f^* \prod g}$, we have shown
    that $\sasol$ is indeed a \provablein{(\thplus{\milnersysmin}{\aseteqs})} solution of $\onechartof{f^* \prod g}$.
\end{proof}

After having proved statement~\ref{it:lem:lem:mimic:RSPstar},
  we can combine the statements~\ref{it:thm:onechart-int:LLEEw}, \ref{it:lem:onechart-int:milnersysmin:solvable}, and \ref{it:lem:lem:mimic:RSPstar}
    as above in order construct \LLEEwitnessed\ coinductive proofs
      with which instances of \RSPstar\ can be mimicked by instances of $\LCoindProofi{1}$. 
This leads us to Lemma~\ref{lem:mimic:RSPstar} below,
  and, as it can show derivability of $\RSPstar$ in $\coindmilnersysone$, to Lemma~\ref{lem:RSPstar:derivable:coindmilnersysone}. 
      
\begin{lem}\label{lem:mimic:RSPstar}
  Let $\astexp,\bstexp,\cstexp\in\StExpover{\actions}$ with $\notterminates{\bstexp}$,
  and let $\aseteqs \defdby \setexp{ \astexp   \formeq   \stexpsum{\stexpprod{\bstexp}{\astexp}}{\cstexp} }$. 
  Then it holds that
  $ \astexp \LLEEcoindproofeqin{(\thplus{\milnersysmin}{\aseteqs})} \stexpprod{\stexpit{\bstexp}}{\cstexp} $.
\end{lem}

\begin{proof}
  First, there is a \provablein{\thplus{\milnersysmin}{\aseteqs}} solution 
    $\sasoli{1}$ of $\onechartof{\stexpprod{\stexpit{\bstexp}}{\cstexp}}$ with  
    $\asoli{1}{\stexpprod{\stexpit{\bstexp}}{\cstexp}} \synteq \astexp$,
      due to Lemma~\ref{lem:lem:mimic:RSPstar}.
  Second, there is a \provablein{\milnersysmin} solution $\sasoli{2}$ of $\onechartof{\stexpprod{\stexpit{\bstexp}}{\cstexp}}$
    with $\asoli{2}{\stexpprod{\stexpit{\bstexp}}{\cstexp}} \synteq \stexpprod{\stexpit{\bstexp}}{\cstexp}$,
    due to Lemma~\ref{lem:onechart-int:milnersysmin:solvable}.  
  Then $\pair{\onechartof{f^*{\prod}g}}{\saeqfun}$
    with $\aeqfun{\avert} \defdby \asoli{1}{\avert} \formeq \asoli{2}{\avert}$ 
        for all $\avert\in\vertsof{\onechartof{f^*{\prod}g}}$
    is a \LLEEwitnessed\ coinductive proof 
      of $\astexp \formeq f^*{\prod}g$
        over $\thplus{\milnersysmin}{\aseteqs}$,
        because $\onechartof{\stexpprod{\stexpit{\bstexp}}{\cstexp}}$ 
        has the guarded \LLEEwitness\ $\onecharthatof{\stexpprod{\stexpit{\bstexp}}{\cstexp}}$~%
        by~Theorem~\ref{thm:onechart-int:LLEEw}.
\end{proof}

We note that the equation in the set $\aseteqs$ in the assumption of Lemma~\ref{lem:mimic:RSPstar} 
  does not need to be sound semantically.
Therefore it was crucial for the formulation of this lemma that
  we did not require proof systems $\asys$ to be sound with respect to $\procsemeq$ 
    for the definition of coinductive proofs in Definition~\ref{def:coindproof}.
Indeed we have done so there in order to be able to formulate this lemma,
  which states that also instances of the fixed-point rule \RSPstar\
    with premises that are not semantically sound can be mimicked by appropriate coinductive proofs.

\begin{lem}\label{lem:RSPstar:derivable:coindmilnersysone}
  $\RSPstar$ is a derivable rule in $\coindmilnersys$. 
\end{lem}

\begin{proof}
  Every instance $\ainst$ of $\RSPstar$ can be replaced by a mimicking derivation $\aDerivi{\ainst}$ in $\coindmilnersysone$ 
    according to the following step,
      where $\notterminates{\bstexp}$ holds as the \sidecondition\ of the instance of $\RSPstar\,$:
  \begin{equation}\label{eq:prf:lem:RSPstar:derivable:coindmilnersysone}
    \begin{aligned}\renewcommand{\fCenter}{\formeq}
      \Axiom$ \astexp   \fCenter   \stexpsum{\stexpprod{\bstexp}{\astexp}}{\cstexp} $
      \LeftLabel{$\ainst$}
      \RightLabel{$\RSPstar$} 
      \UnaryInf$ \astexp   \fCenter   \stexpprod{\stexpit{\bstexp}}{\cstexp} \rule[-5pt]{0pt}{7.5pt} $
      \DisplayProof
    \end{aligned}
    \qquad\Longmapsto\qquad
    %
    \begin{aligned}
      \AxiomC{$ \astexp   \formeq   \stexpsum{\stexpprod{\bstexp}{\astexp}}{\cstexp} $}
      \RightLabel{\LCoindProofi{1}}
      \UnaryInfC{$ \astexp   \formeq   \stexpprod{\stexpit{\bstexp}}{\cstexp} \rule[-5pt]{0pt}{7.5pt} $}           
      \DisplayProof
    \end{aligned}%
  \end{equation}
  Here the \sidecondition\
    $ \astexp \LLEEcoindproofeqin{(\thplus{\milnersysmin}{\setexp{\astexp   \formeq   \stexpsum{\stexpprod{\bstexp}{\astexp}}{\cstexp}}})} \stexpprod{\stexpit{\bstexp}}{\cstexp} $
    of the instance of $\LCoindProofi{1}$ in the derivation $\aDerivi{\ainst}$ on the right is guaranteed by Lemma~\ref{lem:mimic:RSPstar}. 
\end{proof}

We now can show the main result of this section,
  the proof transformation from \milnersys\ to \coindmilnersysone. 
We obtain this transformation by using derivability of \RSPstar\ in \coindmilnersysone\ as stated by this lemma,
  and by combining basic proof-theoretic transformations
    that eliminate derivable, and hence correct and admissible, rules from derivations
      as described in the proof of Lemma~\ref{lem:derivable:admissible:rules}.

\begin{thm}\label{thm:milnersys:isthmsubsumedby:coindmilnersysone}
  $\milnersys 
     \isthmsubsumedby
   \coindmilnersysone$.
  What is more,
  every derivation in $\milnersys$ with conclusion $\astexp \formeq \bstexp$ 
    can be transformed effectively
  into a derivation with conclusion $\astexp \formeq \bstexp$ in $\coindmilnersysone$.
\end{thm}

\begin{proof}
  Due to Lemma~\ref{lem:RSPstar:derivable:coindmilnersysone},
    $\RSPstar$ is a derivable rule in $\coindmilnersysone$. 
      Then by Lemma~\ref{lem:derivable:admissible:rules}, \ref{it:2:lem:derivable:admissible:rules}, 
        $\RSPstar$ is also an admissible rule in $\coindmilnersysone$,
          and hence $\thplus{\coindmilnersysone}{\RSPstar} \thmequiv \coindmilnersysone$ holds.
  With that we can now argue as follows: 
  \begin{align*}
    \milnersys
      =
    \thplus{\milnersysmin}{\RSPstar}
      & {} \subsystem
    \thplus{(\thplus{\milnersysmin}{\RSPstar})}{\LCoindProofi{1}}
      \\
      & {} =
    \thplus{(\thplus{\milnersysmin}{\LCoindProofi{1}})}{\RSPstar}
      =
    \thplus{\coindmilnersysone}{\RSPstar} 
      \thmequiv
    \coindmilnersysone \punc{.}   
  \end{align*}  
  From this we can infer $\milnersys \isthmsubsumedby \coindmilnersysone$,
    because $\ssubsystem$ implies $\sisthmsubsumedby$, and $\sbinrelcomp{\sisthmsubsumedby}{\sthmequiv} \subseteq \sthmequiv$. 
  That every derivation $\aDeriv$ in $\milnersys$ can be transformed effectively 
    into a derivation $\aDerivacc$ in $\coindmilnersys$ with the same conclusion 
      follows from derivability of \RSPstar\ in $\coindmilnersys$:
        then in $\aDeriv$ every instance of \RSPstar\ can be replaced by a corresponding instance of $\LCoindProofi{1}$ 
          as described in \eqref{eq:prf:lem:RSPstar:derivable:coindmilnersysone}
              of the proof of Lemma~\ref{lem:RSPstar:derivable:coindmilnersysone}
            with as result a derivation $\aDerivacc$ in $\coindmilnersysone$ with the same conclusion as $\aDeriv$.
  This argument instantiates the implication from rule derivability to rule admissibility,
   and the transformations explained 
      in the proof of Lemma~\ref{lem:derivable:admissible:rules}, \ref{it:1:lem:derivable:admissible:rules} and \ref{it:2:lem:derivable:admissible:rules},
        specifically \eqref{eq:prf:it:1:lem:derivable:admissible:rules}. 
\end{proof}

\begin{exa}
           \label{ex:fig:ex:1:RSPstar:to:coindproof}
  In Figure~\ref{fig:ex:1:RSPstar:to:coindproof}         
    we provided a first illustration for translating 
      an instance of the fixed-point rule into a coinductive proof in Figure~\ref{fig:ex:1:RSPstar:to:coindproof} on page~\pageref{fig:ex:1:RSPstar:to:coindproof}.
  Specifically, we mimicked the instance $\ainst$ (see below) of the fixed-point rule \RSPstar\ in Milner's system $\milnersys = \thplus{\milnersysmin}{\RSPstar}$
    by a coinductive proof 
    over $\thplus{\milnersysmin}{\setexp{\text{premise of $\ainst$}}}$ with \LLEEwitness~$\onecharthatof{\alert{f^*}\cdot \forestgreen{0}}$. 

  The correctness conditions that have to be satisfied for the right-hand sides 
  in order to recognize this prooftree as a \LLEEwitnessed\ coinductive proof $\thplus{\milnersysmin}{\setexp{\text{premise of $\ainst$}}}$
    are the same as those 
      that we have verified for the right-hand sides of the coinductive proof over $\milnersysmin$ with the same \LLEEwitness\ in Example~\ref{ex:1:coindproof}.
    Note that the premise of $\ainst$ is not used for the correctness conditions of the right-hand sides.
  The correctness condition for the left-hand side $\chocolate{{e _0 ^*} \prod 0}$ at the bottom vertex of $\onechartof{\alert{f^*}\cdot \forestgreen{0}}$
  can be verified as follows, now making use of the premise of the considered instance $\ainst$ of $\RSPstar\,$:
  \begin{align*}
    \chocolate{{e _0 ^*} \prod 0}
      & \;\;{\eqin{\setexp{\text{premise of $\ainst$}}}}\;\;
        \alert{f} \prod (\chocolate{{e _0 ^*} \prod 0}) + \forestgreen{0}
      \\
      & \;\;\parbox{\widthof{$\milnersysmineq$}}{$\milnersysmineq$}\;\;
        (a \prod (a + b) + b) \prod (\chocolate{{e _0 ^*} \prod 0})
      \\  
      & \;\;\parbox{\widthof{$\milnersysmineq$}}{$\milnersysmineq$}\;\;
        (a \prod (a + b)) \prod (\chocolate{{e _0 ^*} \prod 0})  +  b \prod (\chocolate{{e _0 ^*} \prod 0})
      \\  
      & \;\;\parbox{\widthof{$\milnersysmineq$}}{$\milnersysmineq$}\;\;
        a \prod ((a + b) \prod (\chocolate{{e _0 ^*} \prod 0})) + b \prod (1 \prod (\chocolate{{e _0 ^*} \prod 0}))
      \\  
      & \;\;\parbox{\widthof{$\milnersysmineq$}}{$\milnersysmineq$}\;\;
        a \prod ((1 \prod (a + b)) \prod (\chocolate{{e _0 ^*} \prod 0})) + b \prod (1 \prod (\chocolate{{e _0 ^*} \prod 0}))
  \end{align*}
  Together this yields the provable equation:
  \begin{align*}
    \chocolate{{e _0 ^*} \prod 0}
      & \;\;{\eqin{\thplus{\milnersysmin}{\setexp{\text{premise of $\ainst$}}}}}\;\;
        a \prod ((1 \prod (a + b)) \prod (\chocolate{{e _0 ^*} \prod 0})) + b \prod (1 \prod (\chocolate{{e _0 ^*} \prod 0})) \punc{,}
  \end{align*}
  which demonstrates the correctness condition for
    the left-hand side $\chocolate{{e _0 ^*} \prod 0}$ at the bottom vertex of $\onechartof{\alert{f^*}\cdot \forestgreen{0}}$.
  The correctness condition for the left-hand side $a \prod ((1 \prod (a + b))$ at the top left vertex of $\onechartof{\alert{f^*}\cdot \forestgreen{0}}$
  can be verified without using the premise of $\ainst$ as follows:
  \begin{align*}
    ((1 \prod (a + b)) \prod (\chocolate{{e _0 ^*} \prod 0})
      & {} \mathrel{\parbox{\widthof{$\milnersysmineq$}}{$\milnersysmineq$}}
        ((a + b) \prod (\chocolate{{e _0 ^*} \prod 0}) 
        \mathrel{\parbox{\widthof{$\milnersysmineq$}}{$\milnersysmineq$}}
          a \prod (\chocolate{{e _0 ^*} \prod 0}) +  b \prod (\chocolate{{e _0 ^*} \prod 0}) 
      \\
      & {} \mathrel{\parbox{\widthof{$\milnersysmineq$}}{$\milnersysmineq$}}
          a \prod (1 \prod (\chocolate{{e _0 ^*} \prod 0})) +  b \prod (1 \prod (\chocolate{{e _0 ^*} \prod 0})) \punc{.}
  \end{align*}
  Finally, the correctness condition of the left-hand side $1 \prod (\chocolate{{e _0 ^*} \prod 0})$ at the right upper vertex of $\onechartof{\alert{f^*} \prod \forestgreen{0}}$ 
    can be obtained by an application of the axiom~($\leftidstexpprod$) only.
\end{exa}

We close this section by giving an example that provides an additional sanity check 
  for the proof transformation from \milnersys\ to \coindmilnersys\ that we developed above. 
The example below shows that the construction of a \LLEEwitnessed\ coinductive proof fails
  for an inference that is \underline{\smash{not}} an instance of the fixed-point rule \RSPstar\ because the \sidecondition\ is violated.   

\begin{nonexa}\label{nonex:RSPstar:2:LLEEcoindproof} 
  We consider the following (not semantically valid) inference
    according to a (semantically unsound) extension $\notRSPstar$ of the fixed-point rule \RSPstar\
      that does not require the guardedness \sidecondition\ $\notterminates{\bstexp}$ (see Definition~\ref{def:milnersys}):
  \begin{center}\renewcommand{\fCenter}{\formeq}
    \Axiom$ \overbrace{(a + c)^*}^{\chocolate{e}}   \fCenter   \overbrace{(a + 1)}^{\alert{f}} \cdot \overbrace{(a + c)^*}^{\chocolate{e}} + \overbrace{1}^{\forestgreen{g}} $
    \RightLabel{$\notRSPstar \;\; \text{(but not \RSPstar\ since $f \synteq \alert{\terminates{\black{(a + 1)}}}$)}$}
    \UnaryInf$ (a + c)^*   \fCenter   \underbrace{(a + 1)^*}_{\alert{f^*}} \cdot \underbrace{1}_{\forestgreen{g}} $
    \DisplayProof
  \end{center}
  The premise is semantically valid by Proposition~\ref{prop:milnersys:sound} because it is provable in \milnersysmin:
  \begin{align*}
    (a + c)^* 
      & \milnersysmineq
          1 + (a + c) \cdot (a + c)^*
        \milnersysmineq
          1 + a \cdot (a + c)^* + c \cdot (a + c)^*
      \displaybreak[0]\\
      & \milnersysmineq
          1 + a \cdot (a + c)^* + 1 + a \cdot (a + c)^* + c \cdot (a + c)^*
      \displaybreak[0]\\
      & \milnersysmineq
          1 + a \cdot (a + c)^* + (a + c)^*
        \milnersysmineq
          1 + a \cdot (a + c)^* + 1 \cdot (a + c)^* 
      \displaybreak[0]\\  
      & \milnersysmineq
          1 + (a + 1) \cdot (a + c)^* 
        \milnersysmineq
          (a + 1) \cdot (a + c)^* + 1  \punc{.}   
  \end{align*}
  But the conclusion of the inference is obviously not valid semantically, 
  because its left-hand side can iterate \transitionsact{\cact}, while its right-hand side does not permit \transitionsact{\cact}.

  \begin{figure}[t]
\begin{center}
\begin{tikzpicture}
  \renewcommand{\prod}{\,{\cdot}\,}
  \renewcommand{\stackprod}{\,\chocolate{\varstar}\,}
 
\matrix[anchor=center,row sep=0.8cm,column sep=4cm,
        every node/.style={draw=none}
        ] {
    \node(C--1){};      &[0.75cm]      \node(nonsol--1){};        &[0.75cm]     \node(solfstarg--1){};      
    \\ 
    \node(C--root){};   &              \node(nonsol--root){};     &          \node(solfstarg--root){};         
    \\
  };

\path (C--1) ++ (-2cm,0cm) node[above]{\Large $\onechartof{\alert{f^*}}$};
\path (C--1) ++ (0cm,0cm) node[above]{$ 1 \stackprod \alert{f^*} $};
\draw[->,thick,densely dotted,out=180,in=180,distance=1cm] ($(C--1) + (-0.55cm,0.25cm)$) to ($(C--root) + (-0.75cm,-0.25cm)$); 

\path (C--root) ++ (0cm,0cm) node[below]{$ \underbrace{(a + 1)^*}_{\alert{f^*}} $};
\draw[->,thick,darkcyan,shorten <= -2pt, shorten >= -2pt] (C--root) to node[left]{\small \black{$\aact$}} node[right]{\small $\loopsteplab{1}$} (C--1);

\draw[thick,chocolate,double] ($(C--root) + (0cm,-0.3cm)$) ellipse (0.675cm and 0.29cm);

\path (nonsol--1) ++ (0cm,0cm) node[above]{$ 1 \prod \overbrace{\chocolate{e}}^{(a + c)^*} $};
\draw[->,thick,densely dotted,out=180,in=180,distance=1.25cm] ($(nonsol--1) + (-0.85cm,0.25cm)$) to ($(nonsol--root) + (-2.175cm,-0.25cm)$);
\path (nonsol--root) ++ (0cm,0cm) node[below]
   {$ \underbrace{
         \underbrace{(a + 1)}_{\alert{f}} \prod \underbrace{(a + c)^*}_{\chocolate{e}} +  \underbrace{1}_{\forestgreen{g}}
                  }_{\:=\: \chocolate{e} \;\; \text{(by rule assumption)}}  $};
\draw[->,thick,darkcyan,shorten <= -2pt, shorten >= -2pt] (nonsol--root) to node[left]{\small \black{$\aact$}} node[right]{\small $\loopsteplab{1}$} (nonsol--1);
\path (nonsol--root) ++ (-0.55cm,0.15cm) node{\LARGE \alert{$\boldsymbol{\times}$}};

\path (solfstarg--1) ++ (0cm,0cm) node[above]{$ (1 \prod \alert{f^*}) \prod \forestgreen{g} $};
\draw[->,thick,densely dotted,out=180,in=180,distance=1cm] ($(solfstarg--1) + (-0.8cm,0.25cm)$) to ($(solfstarg--root) + (-1.4cm,-0.25cm)$); 
\path (solfstarg--root) ++ (0cm,0cm) node[below]{$ \underbrace{(a + 1)^*}_{\alert{f^*}} \prod \underbrace{1}_{\forestgreen{g}} $};
\draw[->,thick,darkcyan,shorten <= -2pt, shorten >= -2pt] (solfstarg--root) to node[left]{\small \black{$\aact$}} node[right]{\small $\loopsteplab{1}$} (solfstarg--1);

\end{tikzpicture}
\end{center}  
  \vspace*{-1.5ex}
  \caption{\label{fig:nonex:RSPstar}
           Failure of the construction a \protect\provablein{\milnersysmin} solution
             for a coinductive proof that would mimic an instance of a fixed-point rule like \protect\RSPstar\ without \protect\sidecondition.
           }
\end{figure}  %
  Now by mechanically performing the same construction of a coinductive proof 
    as we illustrated it in Example~\ref{ex:1:RSPstar:2:LLEEcoindproof} and in Example~\ref{ex:fig:ex:1:RSPstar:to:coindproof}, 
  we obtain the \LLEEonechart~$\onechartof{f^*}$ and the star expression assignments to it as in Figure~\ref{fig:nonex:RSPstar}.  
  There we recognize that, while 
    $\alert{f^*}\prod\forestgreen{g}$
    is the principal value of a \provablein{\milnersysmin} solution of $\onechartof{\alert{f^*}\prod\forestgreen{g}}$ (see on the right),
  we have not obtained a \provablein{\thplus{\milnersys}{\setexp{\chocolate{e} = \alert{f} \prod \chocolate{e} + \forestgreen{g}}}} 
                         solution of $\onechartof{\alert{f^*}\prod\forestgreen{g}}$
  with principal value $\alert{f} \prod \chocolate{e} + \forestgreen{g}$ (see in the middle).
  This is because the correctness condition is violated at the bottom vertex,
    because $(a + 1) \prod (a + c)^* + 1 \milnersysmineq 1 + a \prod (1 \prod (a + c)^*)$
      does not hold: otherwise it would have to be semantically valid by Proposition~\ref{prop:milnersys:sound},
        but it is not, because only the left-hand side permits an initial \transitionact{\cact}.
        
  Therefore the construction does not give rise to a \LLEEwitnessed\ coinductive proof of the (not semantically valid) formal equation
  $ (a + c)^*   
      \formeq
    (a + 1) \cdot (a + c)^* + 1 $.
\end{nonexa}

Based on the proof transformations that we have developed in this section and earlier in Section~\ref{coindmilnersys:2:milnersys},
  we now obtain our main result. 
                                It justifies that we called the proof system \coindmilnersys\
    a reformulation of Milner's system \milnersys.

\begin{thm}\label{thm:main}
  $\milnersys \thmequiv \coindmilnersysone \thmequiv \coindmilnersys \thmequiv \CLC$,
    i.e.\ these proof systems are \theoremequivalent.
\end{thm}

\begin{proof}
  By combining the statements of 
    Theorem~\ref{thm:milnersys:isthmsubsumedby:coindmilnersysone}, 
    Lemma~\ref{lem:thmsubsumption:rels:coindproofsystems},
    and Theorem~\ref{thm:coindmilnersys:isthmsubsumedby:milnersys}
  we obtain the \theoremsubsumption\ statements
  $\milnersys \isthmsubsumedby \coindmilnersysone \isthmsubsumedby \coindmilnersys \: (\thmequiv \CLC) \isthmsubsumedby \milnersys$,
  which together justify the \theoremequivalence\ of \milnersys\
    with each of \coindmilnersysone, \coindmilnersys, and \CLC.
\end{proof}

\section{Summary and Conclusion}%
  \label{conclusion}

We set out on a \prooftheoretic\ investigation of the problem of whether Milner's system $\milnersys$ is complete
  for \processsemantics\ equality $\procsemeq$ on regular expressions 
    (which Milner calls `star expressions' for disambiguation when interpreted according to the process semantics).
Specifically we aimed at characterizing the derivational power that the fixed-point rule $\RSPstar$ in $\milnersys$
  adds to its purely equational part $\milnersysmin$.  
    
In order to define a substitute for the rule $\RSPstar$ we based ours on two
  results that we have obtained earlier (the first in joint work with Fokkink):
\begin{enumerate}[label={{\bf (S/U)}},align=right,leftmargin=*,itemsep=0.25ex]
  \item[{\crtcrossreflabel{{\bf (S/U)}}[SandU]}] 
    Linear specifications of the shape%
        \footnote{The two examples on page~\pageref{fig:milner-bosscher-expressible} illustrate the correspondence
                  between a recursive specifications and its associated transition graph.
                  But note that these graphs do not satisfy \LLEE\ (see Example~\ref{exa:LEE}).}
      of transition graphs that satisfy the layered loop existence and elimination property \LLEE\ 
      are:
      \begin{itemize}[label={{\bf (U)}},itemsep=0.15ex]
        \item[{\crtcrossreflabel{{\bf (S)}}[S]}] 
          solvable by star expressions modulo provability in $\milnersysmin$ (and modulo $\sbisim$).
        \item[{\crtcrossreflabel{{\bf (U)}}[U]}] 
          uniquely solvable by star expressions modulo provability in $\milnersys$ (and~mod.~$\sbisim$).
      \end{itemize}
      These statements were crucial steps
       in the completeness proof \cite{grab:fokk:2020:lics:arxiv,grab:fokk:2020:lics} 
         by Fokkink and myself 
       for a tailored restriction \BBP\ of~$\milnersys$ to `\onefree' star expressions. 
  \item[{\crtcrossreflabel{{\bf (IV)$_{\sone}$}}[IVone]}] 
    While the chart interpretation $\chartof{\astexp}$ of a star expression $\astexp$ 
      does not always satisfy \LLEE, there is a variant \onechart\ interpretation $\onechartof{\astexp}$ (typically with \onetransitions) 
        that is \mbox{($\sone$-)}bisimilar to $\astexp$, and satisfies \LLEE, for all star expressions $\astexp$.
          (See Definition~\ref{def:onechartof}, and \cite{grab:2020:TERMGRAPH2020-postproceedings:arxiv,grab:2021:TERMGRAPH2020-postproceedings}).
    Hereby \onetransitions\ are interpreted as empty-step processes.       
\end{enumerate}

Now in order to obtain results for \onecharts\ that are analogous to \ref{S} and \ref{U},
  we have generalized these statements to guarded \onecharts\ in Section~\ref{coindmilnersys:2:milnersys}.
In doing so, we obtained:
\begin{enumerate}[label={{\bf (S/U)$_{\sone}$}},align=right,leftmargin=*,itemsep=0.25ex]
  \item[{\crtcrossreflabel{{\bf (S/U)$_{\sone}$}}[SUone]}] 
    Guarded linear specifications of the shape of \onecharts\ that satisfy the loop existence and elimination property \LLEE\ 
      are:
      \begin{itemize}[label={{\bf (U)}},itemsep=0.15ex]
        \item[{\crtcrossreflabel{{\bf (S)$_{\sone}$}}[Sone]}] 
          solvable by star expressions modulo provability in $\milnersysmin$ (see \ref{SE}),
        \item[{\crtcrossreflabel{{\bf (U)$_{\sone}$}}[Uone]}] 
          uniquely solvable by star expressions modulo provability in $\milnersys$ (see \ref{SU}).
      \end{itemize}
    These statements correspond to the statements \ref{SE} and \ref{SU} in Section~\ref{coindmilnersys:2:milnersys}
      which we have proved as Lemma~\ref{lem:extrsol:is:sol}, and Lemma~\ref{lem:sols:provably:equal:LLEE}, respectively. 
\end{enumerate}        
\begin{figure}

\begin{center}
  \scalebox{1}{%
\begin{tikzpicture}
  \matrix[anchor=center,row sep=1cm,column sep=5.25cm,ampersand replacement=\&] {
    \node(milnersysvars-label){};
      \& \node(coindmilnersysvars-label){};
        \& \node(CCvars-label){};
    \\[2ex]
    \node(milnersys-label){};
      \& \node(coindmilnersys-label){}; 
        \& \node(CLC-label){};
    \\
    \node(milnersysacc-label){};
      \& \node(coindmilnersysone-label){};
    \\
    \\
    \node(milnersysaccbar-label){};
      \& \node(coindmilnersysbar-label){};
        \& \node(CC-label){}; 
    \\ 
    }; 
  %
  %
  \path(milnersysvars-label) ++ (0cm,0cm) node{\parbox{\widthof{Milner's system}}{\em Milner's system\\\hspace*{\fill}and variants\hspace*{\fill}\mbox{}}};
  \node (coindsyss-label) at ($(coindmilnersysvars-label)!0.5!(CCvars-label)$) 
          {\parbox{\widthof{coinductive reformulation of Milner's system}}{\em  Coinductive reformulation of Milner's system\\\hspace*{\fill}and variant systems\hspace*{\fill}\mbox{}}};

  \draw[-,densely dotted]
    ($($(milnersysvars-label)!0.46!(coindmilnersysvars-label)$) + (0cm,0.5cm)$) 
      to
    ($($(milnersysaccbar-label)!0.46!(coindmilnersysbar-label)$) + (0cm,-0.75cm)$);

  %
  \path (milnersys-label) ++ (0cm,0cm) node(milnersys){\tightfbox{$\milnersys = \thplus{\milnersysmin}{\RSPstar} $}}; 
  \path (coindmilnersys-label) ++ (0cm,0cm) node(coindmilnersys){\tightfbox{$\coindmilnersys = \thplus{\milnersysmin}{\family{\LCoindProofi{n}}{n}} $}};  
  \path (CLC-label) ++ (0cm,0cm) node(CLC){\tightfbox{$ \CLC = \family{\LCoindProofi{n}}{n} $}};   
  \path (milnersysacc-label) ++ (0cm,0cm) node(milnersysacc){\tightfbox{$\milnersysacc = \thplus{\milnersysmin}{\USPone} $}}; 
  \path (coindmilnersysone-label) ++ (0cm,0cm) node(coindmilnersysone){\tightfbox{$\coindmilnersysone = \thplus{\milnersysmin}{\LCoindProofi{1}} $}}; 
  \path (milnersysaccbar-label) ++ (0cm,0cm) node(milnersysaccbar){\tightfbox{$\milnersysaccbar = \thplus{\milnersysmin}{\USP} $}}; 
  \path (coindmilnersysbar-label) ++ (0cm,0cm) node(coindmilnersysbar){\tightfbox{$\coindmilnersysbar = \thplus{\milnersysmin}{\family{\CoindProofi{n}}{n}} $}};  
  \path (CC-label) ++ (0cm,0cm) node(CC){\tightfbox{$ \CC = \family{\CoindProofi{n}}{n} $}};

  %
  %
  %
  \draw[|->,thick,out=0,in=180,distance=1cm] 
    (milnersys) 
      to node[left]{{\small Thm.\hspace*{1pt}\ref{thm:milnersys:isthmsubsumedby:coindmilnersysone}}}
    (coindmilnersysone); 
  \draw[|->,thick] 
    (milnersys) 
      to node[above,yshift=0.05cm]{{\small Thm.\hspace*{1pt}\ref{thm:milnersys:isthmsubsumedby:coindmilnersysone}}}
    (coindmilnersys);  
  \draw[|->,thick,out=255,in=105] 
    (milnersys) 
      to node[left]{{\small Lem.\hspace*{1pt}\ref{lem:milnersys:thmequiv:milnersysacc:isthmsubsumedby:milnersysaccbar},%
                                             \ref{it:1:lem:milnersys:thmequiv:milnersysacc:isthmsubsumedby:milnersysaccbar}}} 
    (milnersysacc);
  %
  %
  \draw[|->,thick,out=167.5,in=25,distance=1cm] 
    (coindmilnersys) 
      to node[above]{{\small Thm.\hspace*{1pt}\ref{thm:coindmilnersys:isthmsubsumedby:milnersys}}} 
    (milnersys);
  \draw[|->,thick,out=12.5,in=165,distance=0.5cm] 
    (coindmilnersys) 
      to node[above,yshift=-0.1cm]{{\small Lem.\hspace*{1pt}\ref{lem:thmsubsumption:rels:coindproofsystems},%
                                                \ref{it:3:lem:thmsubsumption:rels:coindproofsystems}}} 
    (CLC);  
  \draw[|->,thick,out=190,in=160,distance=1.25cm] 
    (coindmilnersys) 
      to node[left,pos=0.575]{{\small Lem.\hspace*{1pt}\ref{lem:thmsubsumption:rels:coindproofsystems},\ref{it:1:lem:thmsubsumption:rels:coindproofsystems}}} 
    (coindmilnersysbar);

  %
  %
  \draw[|->,thick] 
    (CLC) 
      to node[below,yshift=-0.25cm]{{\small Lem.\hspace*{1pt}\ref{lem:thmsubsumption:rels:coindproofsystems},%
                                                 \ref{it:3:lem:thmsubsumption:rels:coindproofsystems}}} 
    (coindmilnersys);    
  \draw[|->,thick] 
    (CLC) 
      to node[right]{{\small Lem.\hspace*{1pt}\ref{lem:thmsubsumption:rels:coindproofsystems},\ref{it:1:lem:thmsubsumption:rels:coindproofsystems}}} 
    (CC);

  %
  \draw[|->,thick,out=75,in=285] (milnersysacc) to (milnersys);
  \draw[|->,thick] (milnersysacc) to node[left]{{\small Lem.\hspace*{1pt}\ref{lem:milnersys:thmequiv:milnersysacc:isthmsubsumedby:milnersysaccbar},%
                                                             \ref{it:2:lem:milnersys:thmequiv:milnersysacc:isthmsubsumedby:milnersysaccbar}}} (milnersysaccbar);

  %
  \draw[|->,thick] 
    (coindmilnersysone) 
      to node[left]{{\small Lem.\hspace*{1pt}\ref{lem:thmsubsumption:rels:coindproofsystems},%
                                             \ref{it:2:lem:thmsubsumption:rels:coindproofsystems}}} 
    (coindmilnersys); 
  \draw[|->,thick] 
    (coindmilnersysone) 
      to node[right,pos=0.4]{{\small Lem.\hspace*{1pt}\ref{lem:thmsubsumption:rels:coindproofsystems},\ref{it:2:lem:thmsubsumption:rels:coindproofsystems}}} 
    (coindmilnersysbar);

  %
  \draw[|->,thick,out=12.5,in=165,distance=0.5cm] 
    (coindmilnersysbar) 
      to node[above,yshift=-0.1cm]{{\small Lem.\hspace*{1pt}\ref{lem:thmsubsumption:rels:coindproofsystems},%
                                                \ref{it:4:lem:thmsubsumption:rels:coindproofsystems}}} 
    (CC);  
  \draw[|->,thick,out=167.5,in=25,distance=1cm] 
    (coindmilnersysbar) 
      to node[above]{{\small Rem.~\ref{rem:CC:coindmilnersysbar:milnersyaccbar}}} 
    (milnersysaccbar);
    
  %
  %
  \draw[|->,thick] 
    (CC) 
      to node[below,yshift=-0.25cm]{{\small Lem.\hspace*{1pt}\ref{lem:thmsubsumption:rels:coindproofsystems},%
                                                 \ref{it:4:lem:thmsubsumption:rels:coindproofsystems}}} 
    (coindmilnersysbar);    
  
\end{tikzpicture} 
  }
\end{center} 
  \vspace*{-1.5ex}
  \caption{\label{fig:transformations}%
           Web of proof transformations between Milner's system \protect\milnersys\ and variants,
           and the coinductive reformulation \protect\coindmilnersys\ of \protect\milnersys\ with kernel system \CLC\ and variants.}
\end{figure}%
These statements motivated us to define `coinductive proofs' as pairs of solutions for guarded \LLEEonecharts,
  because every pair of \provablein{\milnersysmin}, or \provablein{\milnersys} solutions can be proved equal in \milnersys.
On the basis of this idea we developed the following concepts and results
  (we emphasize earlier introduced terminology again here below):
\begin{enumerate}[label={\protect\mycirc{\arabic{*}}},leftmargin=*,align=right,itemsep=0.25ex]
  \item
    \emph{\LLEEwitnessed\ coinductive proof over an equational proof system $\asys$} (Definition~\ref{def:coindproof}):
      We defined a coinductive proof as a \weaklyguarded\ \LLEEonechart\ $\aonechart$ whose vertices are labeled by equations 
        between the values of two \provablein{\asys} solutions of $\aonechart$.
  \item
    \emph{Coinductive reformulation $\coindmilnersys$} of Milner's system $\milnersys$ (Definition~\ref{def:coindmilnersys}): 
      We defined $\coindmilnersys$ as the result of replacing the fixed-point rule \RSPstar\ in \milnersys\
        with rules of the scheme $\family{\LCoindProofi{i}}{i\in\nat}$
          that formalize provability by \LLEEwitnessed\ coinductive proofs.\vspace*{0.25ex}
    \begin{itemize}[label=$\triangleright$]
      \item
        As the `kernel' of $\coindmilnersys$ we defined the system \CLC\ (Definition~\ref{def:CLC:CC})
        for combining \LLEEwitnessed\ coinductive proofs 
        with only the rules of the scheme $\family{\LCoindProofi{i}}{i\in\nat}$.
    \end{itemize}
  \item
    \emph{Proof transformations} between $\coindmilnersys$ and $\milnersys$ 
      that show their \theoremequivalence.\vspace*{0.25ex}
    \begin{itemize}[label=$\triangleright$]
      \item
        We showed that the rules of $\family{\LCoindProofi{i}}{i\in\nat}$ are correct and admissible for \milnersys,
          and that this implies that every derivation in \coindmilnersys\ can be transformed into one in \milnersys\ with the same conclusion
            by eliminating occurrences of these rules. (See Section~\ref{coindmilnersys:2:milnersys}.)
      \item  
        We showed that the fixed-point rule \RSPstar\ of \milnersys\ is derivable in \coindmilnersys,
          because every instance $\ainst$ can be mimicked by an instance of $\LCoindProofi{1}$
            that uses the premise of~$\ainst$,
              and, as a \sidecondition, a \LLEEwitnessed\ coinductive proof over \milnersysmin.
        As a consequence we showed that every derivation in \milnersys\ can be transformed into one in \coindmilnersys\ with the same conclusion
            by eliminating occurrences of \RSPstar. (See Section~\ref{milnersys:2:coindmilnersys}.)  
    \end{itemize} 
    
  \item
    \emph{Coinductive reformulation \coindmilnersysbar} 
            of the variant $\milnersysaccbar$ (with the powerful rule \USP, see Def\-i\-ni\-tion~\ref{def:milnersys}) of $\milnersys$:
       We formulated systems $\coindmilnersysbar$ (Definition~\ref{def:coindmilnersys}) and $\CC$ (Definition~\ref{def:CLC:CC})
         based on coinductive proofs without \LLEEwitnesses.
       The systems \coindmilnersysbar\ and \CC\ can be recognized as complete variants of \coindmilnersys\ and \CLC.
         We have, however, only argued that in Remark~\ref{rem:CC:coindmilnersysbar:milnersyaccbar} (on the basis of earlier work \cite{grab:2006}),
         but not proved it in detail here.
         
  \item  
    \emph{Proof transformations} between \coindmilnersys\ and \CLC, and \coindmilnersysbar\ and \CC\
      that show their \theoremequivalence.
        Here the idea of the \nontrivial\ transformations from \coindmilnersys\ to \CLC\
          was to `hide' all derivation parts that consist of axioms or rules of the purely equational part \milnersysmin\ of \milnersys\
            into the correctness conditions of solutions in \LLEEwitnessed\ coinductive proofs,
              which occur as \sideconditions\ of instances of rules of $\family{\LCoindProofi{i}}{i\in\nat}$ in \CLC.
        The transformation from \coindmilnersysbar\ to \CC\ operates analogously.
        (See Lemma~\ref{lem:thmsubsumption:rels:coindproofsystems},~%
             \ref{it:3:lem:thmsubsumption:rels:coindproofsystems},~%
             and~\ref{it:4:lem:thmsubsumption:rels:coindproofsystems}.)
\end{enumerate}
In Figure~\ref{fig:transformations} we have illustrated the web of proof transformations between,
  on the one hand,
    Milner's systems and its variants as defined in Definition~\ref{def:milnersys},
  and, on the other hand,
    the coinductive reformulations \coindmilnersys\ and \coindmilnersysone\ of \milnersys, and \coindmilnersysbar\ of \milnersysaccbar\ in Definition~\ref{def:coindmilnersys},
    as well as the coinductive kernel systems \CLC\ of \coindmilnersys, and \CC\ of \coindmilnersysbar\ in Definition~\ref{def:CLC:CC}.
   
\begin{figure}
\begin{center}
\begin{tikzpicture}
  \matrix[anchor=center,column sep=2cm,ampersand replacement=\&] {
    \node(milnersyseq){$ \astexpi{1} \eqin{\milnersys} \astexpi{2} $};
      \&\ \node(coindmilnersyseq){$ \astexpi{1} \eqin{\coindmilnersys} \astexpi{2} $};
        \&[-1cm] \node(CLCeq){$ \astexpi{1} \eqin{\CLC} \astexpi{2} $};
          \&\ \node(procsemeq){\parbox{\widthof{$\chartof{\astexpi{1}} \bisim \chartof{\astexpi{2}}$}}
                                      {$\chartof{\astexpi{1}} \bisim \chartof{\astexpi{2}}$%
                                       \\[0.25ex]
                                       $ \astexpi{1}  \procsemeq  \astexpi{2} $}};
    \\                                   
    };
  \draw[|->,thick,out=-30,in=210,distance=1cm] 
    (milnersyseq) 
      to node[below]{{\small Thm.\hspace*{1pt}\ref{thm:milnersys:isthmsubsumedby:coindmilnersysone}}} 
    (coindmilnersyseq);
  \draw[|->,thick,out=150,in=30,distance=1cm] 
    (coindmilnersyseq) 
      to node[above]{{\small Thm.\hspace*{1pt}\ref{thm:coindmilnersys:isthmsubsumedby:milnersys}}} 
    (milnersyseq);  
  \draw[|->,thick,out=-30,in=210,distance=0.75cm] 
    (coindmilnersyseq) 
      to node[below]{{\small Lem.\hspace*{1pt}\ref{lem:thmsubsumption:rels:coindproofsystems},%
                                              \ref{it:3:lem:thmsubsumption:rels:coindproofsystems}}}  
    (CLCeq);
  \draw[|->,thick,out=150,in=30,distance=0.75cm] 
    (CLCeq) 
      to node[above]{{\small Lem.\hspace*{1pt}\ref{lem:thmsubsumption:rels:coindproofsystems},%
                                              \ref{it:3:lem:thmsubsumption:rels:coindproofsystems}}}  
    (coindmilnersyseq);   
  \draw[|->,thick,out=-35,in=200,distance=1cm] 
    (CLCeq) 
      to node[below]{{\small due to Prop.\hspace*{1pt}\ref{prop:milnersys:sound}}} 
    (procsemeq);
  \draw[|->,thick,densely dashed,out=160,in=35,distance=1cm]  
    (procsemeq) 
      to node[above]{{\bf ?}} 
    (CLCeq);   
\end{tikzpicture}  
\end{center}  
  \vspace*{-1.5ex}
  \caption{\label{fig:beachhead}%
           The coinductive proof system \protect\CLC\ (next to \protect\coindmilnersys) as potential beachhead for a completeness proof of \protect\milnersys.
           We label with a question mark the transformation that is not defined here
             (it can be extracted from the completeness proof of \protect\milnersys).
           }
\end{figure}%
The coinductive reformulation \coindmilnersys\ of \milnersys, and its coinductive kernel \CLC,  
  can be looked upon as being situated roughly half-way in between \milnersys\ and bisimulations between chart interpretations of star expressions.
We illustrate this in Figure~\ref{fig:beachhead}.
This picture arises from the highest level in Figure~\ref{fig:transformations}, 
  when instantiated for specified conclusions $\astexpi{1} \milnersyseq \astexpi{2}$, $\astexpi{1} \eqin{\coindmilnersys} \astexpi{2}$,
                                               and $\astexpi{1} \eqin{\CLC} \astexpi{2}$, 
  and by extending it further to $\chartof{\astexpi{1}} \bisim \chartof{\astexpi{2}}$, and hence to $\astexpi{1} \procsemeq \astexpi{2}$.
For the last step we use soundness of \CLC, 
  which follows from soundness of \milnersys\ with respect to $\procsemeq$, see Proposition~\ref{prop:milnersys:sound},
    in view of the proof transformations that link \CLC\ via \coindmilnersys\ to \milnersys.
Now we note that derivations of \CLC\ represent proof-trees of coinductive proofs
  each of which defines a bisimulation up to provability
    by Proposition~\ref{prop:coindproofeq:2:onebisim}, 
       and Remark~\ref{rem:coindproof:defs:onebisim:up:to}.
Therefore we can argue that prooftrees in \CLC, and by \prooftheoretic\ association also prooftrees in \coindmilnersys,
  are situated roughly half-way in between prooftrees in \milnersys\ and bisimulations between chart interpretation of star expressions.      
We think that Figure~\ref{fig:beachhead} provides a reasonable suggestion of the \prooftheoretic\ closeness of these systems.
    
The \prooftheoretic\ connections of \CLC\ with \coindmilnersys\ and \milnersys\
  guarantee that completeness of \milnersys\ (with respect to $\procsemeq$)
    is equivalent to completeness of \coindmilnersys, 
      and also to completeness of \CLC. 
Stronger still, the proof transformations between \CLC, \coindmilnersys, and \milnersys\
  guarantee that every completeness proof of \milnersys\ can be `routed through' \CLC\ (and also through \coindmilnersys). 
Such a `rerouting' through \CLC\ of a completeness proof of \milnersys\
  does, however, not need to be equally natural as a `direct' completeness proof of \milnersys.
But since \CLC\ is intuitively much closer to $\procsemeq$ than \milnersys\ (as suggested by Figure~\ref{fig:beachhead}),
    much hope was warranted to obtain a completeness proof for \milnersys\
      by finding a completeness proof for \CLC\ first,
        or to at least to use concepts that we have introduced here.
        (The latter hope turned out to be justified.)

Indeed, since the proof systems \coindmilnersys\ and \CLC\
  are tied to process graphs via the circular deductions they permit and to bisimulations up to provability,
    and since \coindmilnersys\ and \CLC\ are \theoremequivalent\ with $\milnersys$,
      they were conceived as (and still can be) natural beachheads for a completeness proof of Milner's system.  
In Figure~\ref{fig:beachhead} we have indicated the step to \CLC\ that is missing \emph{here} for a completeness proof of \milnersys\ with a question mark:
  a completeness proof of \CLC, and that is, an argument
    that, for every bisimulation between $\chartof{\astexpi{1}}$ and $\chartof{\astexpi{2}}$ where $\astexpi{1}$ and $\astexpi{2}$ are given star expressions,
      yields a prooftree in \CLC\ with conclusion $\astexpi{1} \formeq \astexpi{2}$.
   
Due to their feature of permitting derivations that can be construed as combinations of \onebisimulations\ (up to provability),    
  we were confident that the proof system \coindmilnersys\ and its coinductive kernel \CLC\
    substantially increase the space for graph-based approaches to finding a completeness proof of \milnersys.
A concrete indication for this expectation was the following.
  By closely analyzing the completeness proof in \cite{grab:fokk:2020:lics:arxiv,grab:fokk:2020:lics} for the tailored restriction \BBP\ of \milnersys\ to `\onefree' star expressions 
  we find:
    \begin{equation}\label{cons:grab:fokk:2020:lics}
      \left.\;
      \parbox{0.85\textwidth}
             {Valid equations between \onefree\ star expressions
              admit derivations in \CLC\ of depth less than or equal to 2.}
      \;\:\right\}        
    \end{equation}  
    This fact suggests the following research question:
    \begin{equation}\label{research:question}
      \left.\;
      \parbox{0.85\textwidth}
             {Can derivations in \CLC\ (derivations in $\coindmilnersys$) always be simplified to some kind of normal form
              that is of bounded depth (respectively, of bounded nesting depth of \LLEEwitnessed\ coinductive proofs)?}
      \;\:\right\}        
    \end{equation}    
Despite of the fact that this question admits a trivial answer in view of the completeness proof of \milnersys\ in \cite{grab:2022:lics},
  see \ref{cons:2} and \ref{cons:3} below, 
    it can be desirable to also find an answer
      that is based on a proof-theoretic analysis,
        and that leads to a 
           workable concept of `normal form' for derivations in \CLC\ or in $\coindmilnersys$.
Intuitions for finding such a concept
  may be found in developing simplification steps of \onecharts\ with \LLEE\ under \onebisimilarity,
    as those are used in the completeness proofs for \BBP\ with respect to `\onefree' star expressions in \cite{grab:fokk:2020:lics:arxiv,grab:fokk:2020:lics},
      and for \milnersys\ with respect to all star expressions in \cite{grab:2022:lics,grab:2022:lics:arxiv}.
We close by listing consequences of the latter (one is direct, the other two are a bit technical).

\paragraph{Consequences of the completeness proof of \milnersys\ in \cite{grab:2022:lics}.}
  Below we list the most important consequences that the completeness proof of \milnersys\ as summarized in \cite{grab:2022:lics}
    has for the line of investigation on coinductive versions of \milnersys\ as reported here:
    \begin{enumerate}[label={(\text{C}\arabic{*})},itemsep=0.5ex]
      \item{}\label{cons:1}
        The coinductive versions \coindmilnersys\ and \CLC\ of Milner's system \milnersys\ are (just like \milnersys) 
          complete with respect to process semantics equality $\procsemeq$ of star expressions. 
          (This follows directly from Theorem~5.1 in \cite{grab:2022:lics}
             in view of Theorem~\ref{thm:main} here.)
      \item{}\label{cons:2}  
        Valid equations between star expressions 
          admit derivations in \CLC\ of depth less than or equal to $2$.
            (This generalization of \eqref{cons:grab:fokk:2020:lics}
             can be shown by a close analysis of the structure of the completeness proof of \milnersys\ in Section~5 of \cite{grab:2022:lics}.
             The proof is similar to that of Corollary~\ref{cor:milnersys:compl:expansion:minimization},~\ref{it:2:cor:milnersys:compl:expansion:minimization},
               and also similar to the argumentation for the analogous statement concerning `\onefree' star expressions and the system \BBP\
                 as reported above.) 
      \item{}\label{cons:3}
        As a consequence of \ref{cons:2}
          the research question \eqref{research:question} admits the trivial answer ``yes'',
            albeit one that avoids a close proof-theoretic analysis. 
            (Yet we still think that an answer that is based on a fine-grained proof-theoretic analysis would be more desirable).
    \end{enumerate}

\vspace{-0.25ex}
\subsection*{Acknowledgement}
%
First I want to thank the editors Alexandra Silva and Fabio Gadducci very much:
  without their leniency with the deadline 
    I would not have been able to finish this extended article version of my CALCO paper \cite{grab:2021:calco}.
Second, I am greatly indebted to the reviewers of this article for reading it carefully,
  spotting distracting problems, asking spot-on questions, and suggesting changes and additions.
To mention two important examples:
  an acute pointer to a shortcoming of the more informal definition of \LLEEwitness\ in terms of recordings of loop elimination steps 
    (now rectified by using multi-step loop elimination steps, see Definition~\ref{def:LLEEwitness})
  and a perceptive question of whether the proof system \milnersys\ is complete with respect to bisimulations that are witnessed by guarded \LLEEonecharts\    
    (it is indeed, and the question has led me to formulating and proving Corollary~\ref{cor:milnersys:compl:expansion:minimization}). 

I am very thankful to {Luca Aceto} for giving me the chance to continue my work on Milner's problems on the process interpretation
 of regular expressions in the PRIN project \emph{{\nf IT MATTERS} - Methods and Tools for Trustworthy Smart Systems}
  (ID: \text{2017FTXR7S\_005}).

\enlargethispage{2ex}
\bibliographystyle{alphaurl}
\bibliography{cMil-lmcs.bib}         

\end{document}